\documentclass[lettersize,journal]{IEEEtran}
\IEEEoverridecommandlockouts
\usepackage{amsmath,amsfonts}
\usepackage{algorithmicx}
\usepackage{algorithm}
\usepackage[usenames, dvipsnames]{color}
\usepackage{array}
\usepackage[caption=false,font=normalsize,labelfont=sf,textfont=sf]{subfig}
\usepackage{textcomp}
\usepackage{stfloats}
\usepackage{url}
\usepackage{verbatim}
\usepackage{graphicx}
\usepackage{cite}
\usepackage{times}
\usepackage{xcolor}
\usepackage[utf8]{inputenc}
\usepackage[T1]{fontenc}
\usepackage{import}
\usepackage{bbm}
\usepackage{makecell}
\usepackage[export]{adjustbox}

\usepackage{amsthm}
\usepackage{cite}
\usepackage{overpic}
\usepackage{complexity}
\usepackage{balance}
\usepackage{siunitx}
\usepackage[colorlinks=true,allcolors=black]{hyperref}
\graphicspath{{./figures/}}

\usepackage{algpseudocode}

\usepackage{bbm}
\usepackage{algorithm}
\usepackage{float}
\newfloat{algorithm}{t}{lop}
\newcommand{\norm}[1]{\left\lVert#1\right\rVert}
\PassOptionsToPackage{noend}{algpseudocode}% comment out if want end's to show
\usepackage{algpseudocode}
\usepackage{cleveref}
\newtheorem{theorem}{Theorem}

\crefname{figure}{Fig.}{Figs.}
\crefname{theorem}{Theorem}{Theorems}
\crefname{corollary}{Corollary}{Corollaries}
\crefname{observation}{Observation}{Observations}

\usepackage[normalem]{ulem}

\usepackage[switch]{lineno}  %IDEAL FOR IEEE TWO COLUMN PAPERS
\makeatletter
\newcommand*{\algrule}[1][\algorithmicindent]{\makebox[#1][l]{\hspace*{.5em}\vrule height .75\baselineskip depth .25\baselineskip}}%

\newcount\ALG@printindent@tempcnta
\def\ALG@printindent{%
    \ifnum \theALG@nested>0% is there anything to print
        \ifx\ALG@text\ALG@x@notext% is this an end group without any text?
            \addvspace{-3pt}% FUDGE for cases where no text is shown, to make the rules line up
        \else
            \unskip
            \ALG@printindent@tempcnta=1
            \loop
                \algrule[\csname ALG@ind@\the\ALG@printindent@tempcnta\endcsname]%
                \advance \ALG@printindent@tempcnta 1
            \ifnum \ALG@printindent@tempcnta<\numexpr\theALG@nested+1\relax% can't do <=, so add one to RHS and use < instead
            \repeat
        \fi
    \fi
    }%
\usepackage{etoolbox}
\patchcmd{\ALG@doentity}{\noindent\hskip\ALG@tlm}{\ALG@printindent}{}{\errmessage{failed to patch}}
\makeatother
\newsavebox{\ieeealgbox}

\begin{document}
\author{\vspace{-7mm}Francesco Bernardini$^{1}$, Daniel Biediger$^{1}$, Ileana Pineda$^{1}$, Linda Kleist$^{2}$
and Aaron T. Becker$^{1,3}$\vspace{-10mm}}\vspace{-15mm}%
\title{\LARGE\bf
 Using Mobile Relays to Strongly Connect a Minimum-Power Network between Terminals Complying with No-Transmission Zones
  %Strongly-Connected Minimal-Cost Radio-Networks Among Fixed Terminals Using Mobile Relays and Avoiding No-Transmission Zones 
\thanks{\copyright 2025 IEEE. Personal use of this material is permitted. Permission 
from IEEE must be obtained for all other uses, in any current or future media, including reprinting/republishing this material for advertising or promotional purposes, creating new collective works, for resale or redistribution to servers or lists, or reuse of any copyrighted component of this work in other works.}\thanks{
This work was supported by the National Science Foundation under  \href{http://nsf.gov/awardsearch/showAward?AWD_ID=1553063}{[IIS-1553063},
\href{https://www.nsf.gov/awardsearch/showAward?AWD_ID=1932572}{1932572},
\href{https://nsf.gov/awardsearch/showAward?AWD_ID=2130793}{2130793]}, the Alexander von Humboldt Foundation, and the Army Research Laboratory under Cooperative Agreement Number W911NF-23-2-0014. The views and conclusions contained in this document are those of the authors and should not be interpreted as representing the official policies, either expressed or implied, of the Army Research Laboratory or the U.S. Government. The U.S. Government is authorized to reproduce and distribute reprints for Government purposes notwithstanding any copyright notation herein.
}% <-this % stops a space
\thanks{
$^{1}$ University of Houston, TX USA {\tt\small\{fbernardini,  debiedig, idpineda, atbecker\}@uh.edu}}
\thanks{
$^{2}$ Universtität Potsdam, Germany {\tt\small kleist@cs.uni-potsdam.de}}
\thanks{
$^{3}$ TU Braunschweig, Germany }
}% 
\vspace{-15mm}
\maketitle
\vspace{-15mm}
\begin{abstract}
%\todo{Fix this up}
We present strategies for placing a swarm of mobile relays to provide a bi-directional wireless network that connects fixed (immobile) terminals. 
 Neither terminals nor relays are permitted to transmit into disk-shaped no-transmission zones. 
We assume a planar environment and that each transmission area is a disk centered at the transmitter. 
We seek a strongly connected network between all terminals with minimal total cost, where the cost is the sum area of the transmission disks. 
%We assume a planar environment and that all transmission areas are defined as a disk-shaped region centered on the transmitter. The optimal solution's cost has the minimal total broadcast power needed to maintain the network between all terminals and it is calculated as the sum of the areas of the broadcast disks in the network.

Results for networks with increasing levels of complexity are provided. The solutions for local networks containing low numbers of relays and terminals are applied to larger networks. 
For more complex networks, algorithms for a minimum-spanning tree (MST) based procedure are implemented to reduce the solution cost.  A procedure to  characterize and determine the possible homotopies of a system of terminals and obstacles is described, and used to initialize the evolution of the network under the presented algorithms.

%Results for networks with varying levels of difficulty are put forward:
%1.) the optimal solution for three terminals and one relay, 
%2.) optimization and analytical results for two terminals with $n$ relays with one obstacle, 
%3.) algorithms for a minimum-spanning tree (MST) based procedure handling $n$ relays and $m$ terminals with and without obstacles

% problem definition
%optimal strategies
% simples networks
% two terminals one obstacle  
% 3 terminals one relay
%% The 
% multiple terminals multiple obstacles 
% harder netwroks
%% two terminals many relays many obstacles
%%     bitangent, k optimal yen, 2d links
%% m terminals n relays no obstacles 
% hardest network
%% m terminals n relays no obstacles 

\end{abstract}

%##################################################################
\textbf{Note to Practitioners: } This paper focuses on the optimal construction of networks in presence of impairments to transmission. The use case in mind is a scenario where point-to-point transmission between drones is exploited to relay signals between distant fixed stations (ground, or air) that are not in line of sight. The impairments can be represented by buildings, or by sensing areas of potential eavesdroppers. The presented method builds upon previous works, shifting the focus from a randomization of the initial conditions, to the creation of a list of pre-defined prototypes to initialize the network optimizers. These prototypes are related to the different homotopies that the network can represent, given a distribution of obstacles. Once the list of prototypes is available, it enhances the control over the final result of a simulation, by enabling the estimation of the its outcome, on the basis of the homotopy chosen. Furthermore, this enables faster network optimizations, thus allowing the application of the method to dynamic scenarios, where obstacles (and potentially, also terminals) may move from their initial position. A current limitation of the method is the rate at which prototypes are determined: an alternate approach that focuses on the reduction of the complexity of the graph associated to the problem (but keeping the essential features of homotopy), is currently under investigation.

\vspace{-2mm}
\section{Introduction}\label{sec:Intro}%Image only
%################################################################## 
The paper investigates strategies to establish a low-cost network between immobile terminals by placing relays. The network must not transmit into no-transmission regions defined as disks. A sample solution is shown in Fig.~\ref{fig:CirclePacking}. 
The problem of providing a connected network, subject to constraints, is closely related to the \emph{minimum range assignment problem for radio networks}. This is a non-deterministic polynomial-time  (NP) complete problem~\cite{clementi2004power} where the 2D positions of $n$ terminals are given, and the goal is to assign a transmission radius (correlating to a transmission power) to each terminal such that the resulting network is strongly connected, while minimizing the sum of the squared radii. When applied to aerial relays, the problem is called ``constructing a Flying Ad-hoc Network (FANET) with minimum transmission power''~\cite{10139317}. 

%This paper has a promising title, but doesn't optimize relay positions: \cite{Guanghua2023optimizeNetwork}.

%\cite{Caillouet2020uavNetworks} discusses recent work (survey editorial).  This is a poorly written MRPI article.

% We switched to '(mobile) relays' and '(fixed) terminals'.  These are all nodes.
% Removed 'drone' --> relay and 'agent' --> node 

Other related efforts include the \emph{relay placement problem} which seeks to place the minimum number of relays to connect a set of stationary terminals. Each terminal is  assumed to have a transmission radius of 1 while each relay has a radius of~$r$.  
%\linda{many new terms here: relay, terminal (vs node/point?), broadcast range (vs transmission power/radius). How are these different. Can we unify?} 
For this problem a 3.11-approximation algorithm is developed along with proof showing that no polynomial-time approximation scheme exists~\cite{efrat2008improved,efrat2015improved}.
%https://link.springer.com/chapter/10.1007/978-3-540-87744-8_30
% The transmission radii correspond to transmission disks for the mobile relays.

 A necessary network condition is that the union of transmission disks must contain all terminals, similar to \emph{minimum-cost coverage of point sets by disks}~\cite{alt2006minimum}, but this is insufficient to generate a \emph{connected} network.

\begin{figure}[t]
\centering
\includegraphics[width=0.73\columnwidth]{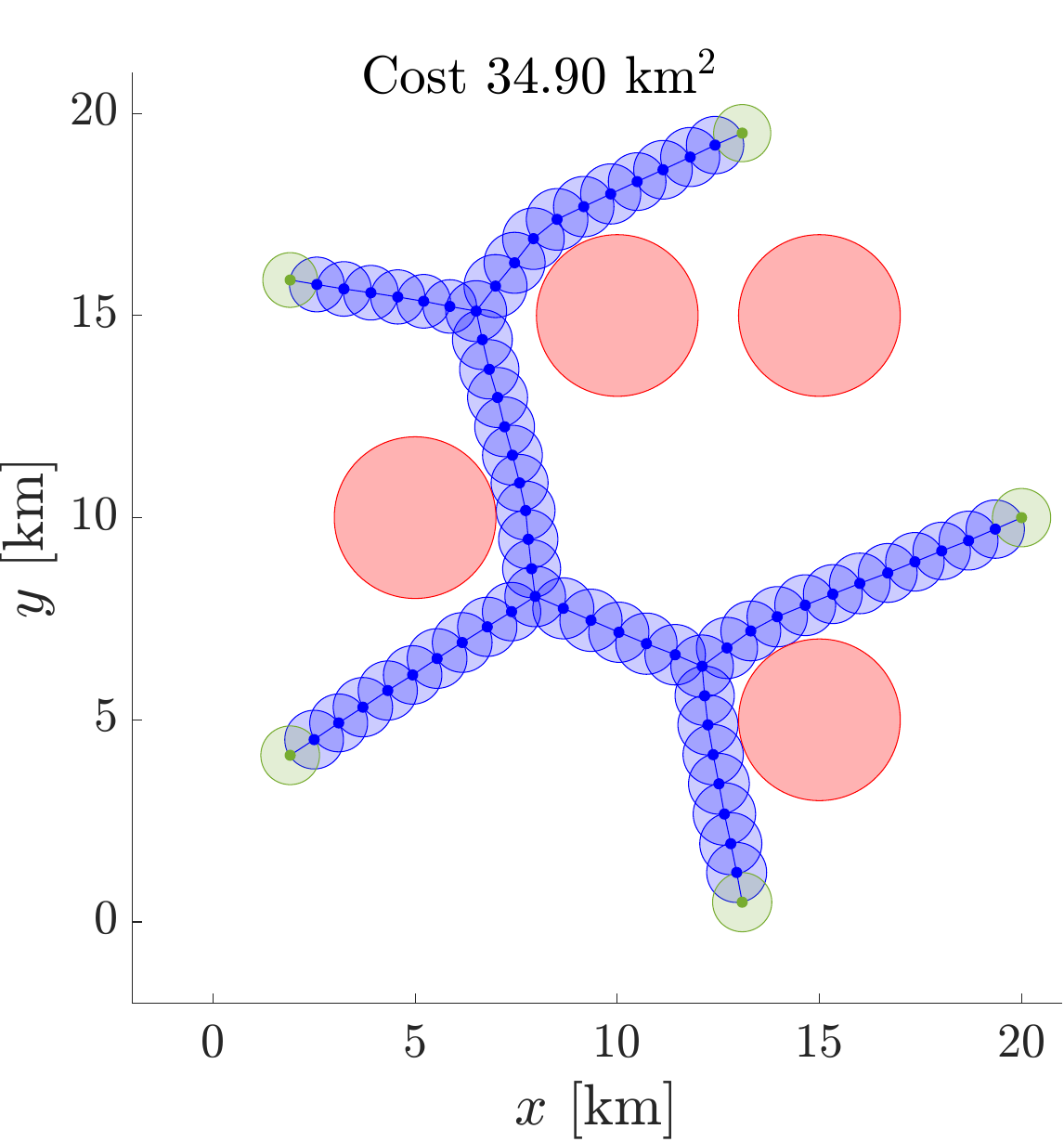}
\vspace{-0.3cm}
\caption{A network that strongly connects $m=5$ fixed terminals (in green) by placing $n=60$ mobile relays (in blue) to establish a strongly connected mesh network that avoids broadcasting into $\phi=4$ no-transmission disks (in red). The solution shown minimizes the sum of the terminal and relay broadcast areas (the green and blue disks).
\label{fig:CirclePacking}}
\vspace{-0.3cm}
\end{figure}  

A network is implied by minimum spanning tree (MST) and Euclidean Steiner tree problems.
%A Steiner tree is an undirected tree graph that contains all terminal nodes (but may include additional vertices) and minimizes the total weight of its edges.
The Steiner tree is an undirected graph that connects a set of terminal nodes and minimizes the total cost of its edges. It does this by introducing additional internal nodes called \emph{Steiner points}, if this helps reduce the total cost of the edges.  In general, these Steiner points have three incident edges, arranged at \SI{120}{\degree} angles. The cost metric for the Steiner tree and the MST is the Euclidean distance of the length of all links in the network. The Steiner tree problem typically generates a network with the minimum length, but for our broadcast model, we use a cost function that increases with the square of the link length. Up to a constant factor, the square link metric is also considered in the \emph{Minimum Area Spanning Tree (MAST)} problem~\cite{GUIMARAES2021101771}, where the area of a tree is given by $\pi/4$
the sum of the squared edge lengths.
%placing a disk on each edge and summing the area of all disks
%the total area of the smallest enclosing edge-disks. %agreed! %\francesco{This amounts to assign an area to each edge and to a constant factor of $\frac\pi4$ this approach is equivalent to the %\linda{clarify the def of MAST
%Given a set P of n points in the plane, find a spanning tree of T of minimum “area”, where the area of a spanning tree T is the area of the union of the n−1 disks whose diameters are the edges in T
%} The squared link length metric is used by 
%\emph{Minimum Area Spanning Tree (MAST)} problem~\cite{GUIMARAES2021101771}, 
However, a MAST does not guarantee a strongly connected network and further development is required to obtain the actual range assignment. 

This paper extends our conference paper~\cite{Bernardini2024CASE}.
In Sec.~\ref{sec:problem} we mathematically formulate the problem and describe its computational complexity. In Sec.~\ref{sec:optimalstrat} we provide optimal strategies for two simple scenarios.
These solutions inspire the strategies developed in later sections. Section~\ref{Sec:AlgsForMultTerminalsAndObstacles} defines two heuristic algorithms and applies them to the general problem without and with obstacles, and presents simulation results.
Then Sec.~\ref{sec:PreScan} presents a method to generate candidate networks with different homotopies and evaluate them. 
Finally, Sec.~\ref{sec:concl} summarizes the paper and outlines possible paths forward for this research.

% \begin{figure}
% \centering
% \includegraphics[width=1.0\columnwidth]{figures/CirclePacking.pdf}
% \caption{This shows a related problem, but not our problem.  FOr 
% \label{fig:CirclePacking}}
% \end{figure}  

%Given a graph $G = (V, E)$ and a set $T \subset V$ of terminals, a Steiner tree is a tree spanning $T$ with the minimum cost. 
%\linda{delete the above sentence?}
%We do not all the leaves to be terminals. Let 1 be a cost function defined on the edge set $E$. The Steiner tree problem is the problem of finding a tree of minimum

% \subsection{Related works}

% \todo{Talk about the problem space including the Steiner trees and the radio relay problem. Connected network components.}

\vspace{-3mm}
%#################################################################
\section{Problem Definition\label{sec:problem}}
%#################################################################

% https://scholar.google.com/scholar?hl=de&as_sdt=0%2C5&q=%22radio+networks%22+power+assignment+problem+kirousis&btnG=#d=gs_qabs&t=1686311454508&u=%23p%3DwcCX3wXxOLAJ

%\linda{I think, we should be more precise here.  Do we restrict to $\mathbb R^2$ or $\mathbb R^d$? 
%Possibly something along the following lines:} 
A directed graph is \emph{strongly connected} if every node is reachable from every other node in the graph. In this paper, we assume a node %\linda{point?} 
at position $A$ with transmission radius $r_A$ can communicate with a node at position $B$ if the Euclidean distance between $A$ and $B$ is not greater than $r_A$. 
A problem instance has a  
set $T$ of $m$ fixed terminals in $\mathbb R^{2\times m}$, a 
set of $\phi$ obstacles specified by disks with radii $R$ in $\mathbb R^{\phi}$ and centers $O$ in $\mathbb R^{2\times \phi}$ and a set $D$ of $n$ mobile relays in $\mathbb R^{2\times n}$. 
We search for a placement of the relays $D$  and an assignment of transmission radii in $\mathbb R^{m+n}$ under the constraints that the directed communication graph is strongly connected and no transmission disk overlaps the obstacles.
 Each transmission radius defines a transmission disk centered at a terminal or relay.
The cost $C(T,O,R,D)$ for a network is 
% \begin{align}
%     C&(T,O,R,D) = \nonumber \\ 
%     &\begin{cases}
%     \infty, & |O_i - T_j | < r_{O,i} \forall i\in[1,o],j\in[0,m] \textbf{ or } 
%     |O_i - D_j | < r_{o,i} \forall i\in[1,o],j\in[0,n] \\
%     \sum_{j=1}^{m} r_{T,j}^2 +\sum_{j=1}^{n} r_{D,j}^2, 
%     & \textrm{otherwise }
%     \end{cases}
% \end{align}
\begin{align}
        C(T,O,R,D) &= \sum_{j=1}^{m} r_{T,j}^2 +\sum_{j=1}^{n} r_{D,j}^2, \label{eqn:TransmissionCostForNetwork}\\
    \textrm{if ~~~~~~} \qquad |O_i - T_j | &\ge R_i +r_{T,j} \quad \forall i\in[1,\phi],j\in[0,m] \nonumber \\
    \textbf{ and }~ 
    |O_i - D_j | &\ge R_i +r_{D,j} \quad \forall i\in[1,\phi],j\in[0,n] \nonumber \\
    \textrm{else ~~~}  C(T,O,R,D) &= \infty.\nonumber
\end{align}
% \begin{align}
%     C(T,O,R,D) &= \sum_{j=1}^{m} r_{T,j}^2 +\sum_{j=1}^{n} r_{D,j}^2, \label{eqn:TransmissionCostForNetwork}
%     \\
%     \textrm{if } \qquad &|O_i - T_j | \ge R_i \quad \forall i\in[1,k],j\in[0,m] \nonumber \\
%     &\textbf{ and }\nonumber \\
%     &|O_i - D_j | \ge R_i \quad \forall i\in[1,k],j\in[0,n] \\
%     \textrm{otherwise } \qquad & \infty.\label{eq:fundamental}
% \end{align}
where $r_{T,j}$ is the transmission radius of terminal $j$, and $r_{D,j}$ is the transmission radius of relay $j$.

%Given a strongly connected network topology among the terminals and relays with a finite cost, we wish to optimize the network.
We seek to minimize Eq.~\eqref{eqn:TransmissionCostForNetwork}. 
We start by describing solutions to problems with small numbers of terminals, where it is possible to generate optimal solutions.
For more complicated instances we start by computing the MST of the network with squared Euclidean distances as weights
%with square link length weights
 to determine the graph topology.
To optimize the position of relay $i$ locally, we determine the network neighbors and then move $D_i$ to minimize the required transmission power. %The solutions are subject to local minima, so the global optimizer is sensitive to initial conditions.
\vspace{-4mm}
%#################################################################
\subsection{Problem Difficulty and Approximation}
%#################################################################

The range assignment problem of setting transmission powers with fixed transmitters and no obstacles to provide strong connectivity in $\mathbb R^2$
%, as presented in \cite{fuchs2008hardness},
is NP-hard
%, already for well-spread instances.
\cite[Thm 10]{fuchs2008hardness},
and approximating the range assignments in
$\mathbb R^3$ better than $1 + 1/
50$ is NP-hard \cite[Thm 13]{fuchs2008hardness}.

%\linda{This paragraph is confusing}
In contrast there exists a 2-approximation~\cite{kirousis2000power} which also holds in 3D.  First, compute an MST $\mathcal{T}$ where the weight of each edge is given by the squared Euclidean distances. %\francesco{Ok I checked the Matlab code and I think this is what we do: first we compute a 1-dimensional tree computing a MST on a graph whose weights are the SQUARED distances (which is a long sentence to say that we compute a MAST). Then, on that graph we discard the 2-dimensional component and apply new disks but on the nodes, this time. In other words, we use the squared distance to force the topology of the network to be that of a MAST, and then we apply new disks to ensure connectivity. Or again, we use the MAST as a first step but the outcome is not a MAST. However, ref. [9] uses a MST as they use the distance not squared.
%If we agree on this we can finally have this paragraph sorted out xD}
 The cost of $\mathcal{T}$ is a lower bound of the optimal solution for the range assignment problem. 
 %here comes a nice argument: for each range assignment, we can consider an mst of the communication graph. Then the range assignment wrt to this mst costs at least as much as the mst: Let's suppose each node points to its most costly edge. By redirecting some assignments, the cost only shrinks and we can ensure that each edge is pointed to at least once.
Second, each terminal is assigned a transmission range equal to the Euclidean length of the longest incident edge of $\mathcal{T}$. This range assignment ensures strong connectivity of the communication graph between the terminals. As every edge of $\mathcal{T}$ is changed at most once for each incident terminal, a 2-approximation is realized.

Moreover, the analysis is tight, as instances of the type shown in Fig.~\ref{fig:areaMSTapprox} illustrate; the MST %\linda{MAST?}
 yields a range assignment where each terminal has radius $1$ and, therefore, a total cost of $n$.  The optimal solution has cost $(n/2+1)+(n/2-1)\varepsilon^2$. For $n\to \infty$ and $\varepsilon\to 0$, the ratio $\to$ 2.

\begin{figure}[tb]
\centering
\includegraphics[width=1.0\columnwidth]{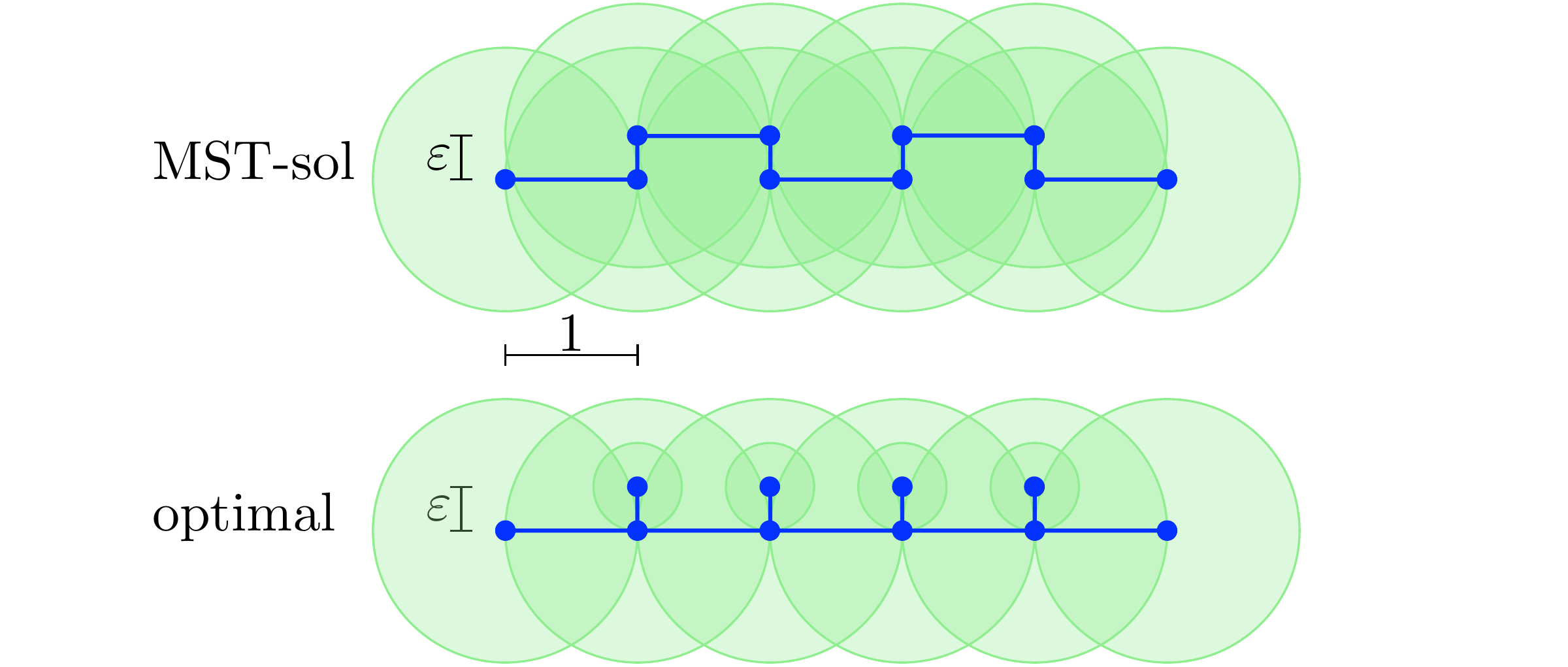}
\vspace{-6mm}
\caption{
A 2-approximation using a MST for the range assignment problem. The optimal assignment requires lower transmission power for the offset nodes.
\label{fig:areaMSTapprox}}
\vspace{-0.6cm}
\end{figure}  

\vspace{-3mm}
%#################################################################
\subsection{Optimal Strategies and Heuristics}
%#################################################################
By restricting the number of terminals, relays and obstacles we construct problem classes with optimal solutions. We describe two optimal placement cases in the following section.  Adding additional terminals and placing obstacles render the problem harder to compute; Section \ref{Sec:AlgsForMultTerminalsAndObstacles} describes heuristic solvers for this problem.

\section{Optimal Strategies\label{sec:optimalstrat}}

% The triangle case is covered in "Optimal solution: three terminals,  one relay" -- but it doesn't cover the colinear case.  Can you add it?  ( Is this trivially the L-shaped case?
% \subsection[Three fixed-terminals, One Relay ]{}
% Three termminals-colinear case
% Three terminals Equilateral case
% Three terminals Isosceles case
% Three terminals Scalene case

\subsection[Two terminals separated by unit radius obstacle]{Two terminals separated by unit radius obstacle}

We begin with two terminals, $t_1$ at $[-d,0]$ and $t_2$ at $[d,0]$, separated by a unit radius obstacle disk centered at $[0,0]$. Given $n$ mobile relays, what is the lowest cost network according to \eqref{eqn:TransmissionCostForNetwork} Several sample solutions, solved numerically, are shown in Fig.~\ref{fig:varyDvalue2TermPlots}.

\begin{figure}[bht]
\centering
\begin{overpic}[width=1.0\columnwidth]{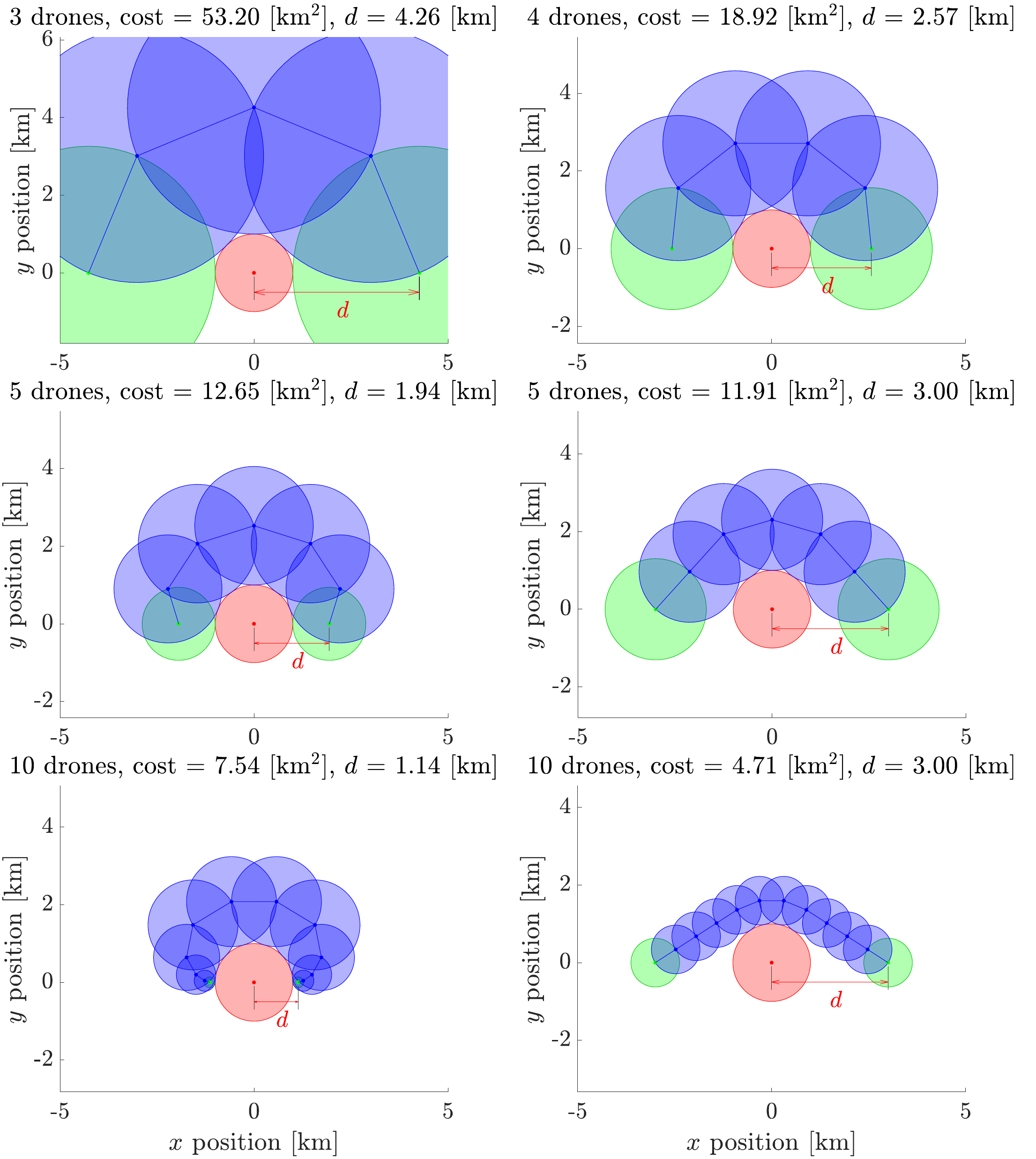}
\put(0,94){\textcolor{black}{a)}}
\put(0,61){\textcolor{black}{c)}}
\put(0,29){\textcolor{black}{e)}}
\put(44,94){\textcolor{black}{b)}}
\put(44,61){\textcolor{black}{d)}}
\put(44,29){\textcolor{black}{f)}}
 \end{overpic}
 \vspace{-0.6cm}
% Plot is generated by 2023Freedom/NetworkingPositionOptimization/RepresentativeOptimizationProblems/plotGatherDataCircleDroneFminconOpt.m
\caption{Building the lowest cost network between two terminals at $[\pm d,0]$ (in green) separated by a unit radius obstacle (in red) using $n$ relays (in blue). Each of these subplots is shown by a marker in Fig.~\ref{fig:varyDvalue2Term}.
%\linda{this is very tiny.. maybe 2+2+2? or increase font at axes. Also: maybe the x-axes could always be -4/5 to +4/5 so that the unit disk has the same size?}
}
\label{fig:varyDvalue2TermPlots}
\vspace{-0.4cm}
\end{figure}
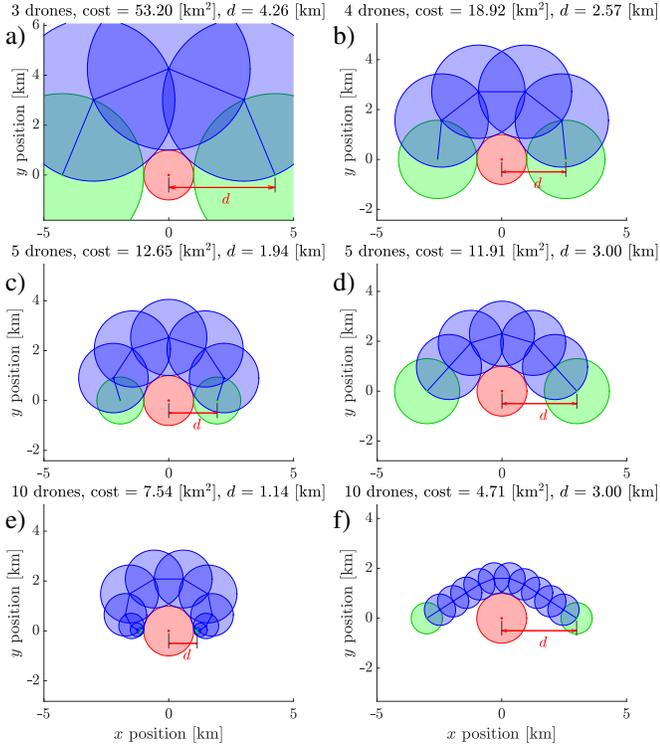  

The numerical solver adjusts the positions of the $n$ relays and the transmission radii of the relays and the terminals.
The transmission radii are constrained to be positive and the distance from any transmitter to the obstacle must be no less than $1$ $+$ the transmission radius. 
Clearly, an optimal network  connecting two terminals forms a single chain. 
%Hence, we may assume relays induce the following paths $t_1,d_1,\dots,d_n,t_2$.

%Additionally, terminal 1 and relay 1 must be within each other's transmission range, the same for relays $i$ and $i+1$, and also for terminal 2 and relay $n$.
For a given $n$ there is a $d$ that minimizes Eq.~\eqref{eqn:TransmissionCostForNetwork}. This occurs if all the transmission radii are equal sized and the relays evenly distributed on a semicircle of radius 
\begin{align}
    d_\textrm{min}(n) =  \frac{1}{1-2 \sin\left( \frac{\pi}{2+2n}\right)}\, \label{eq:OPTIMALdFORn}.
\end{align}
%\linda{change to radii consistently?}

%Francesco: I changed the boldface letters to "a)" as the rest of the paper
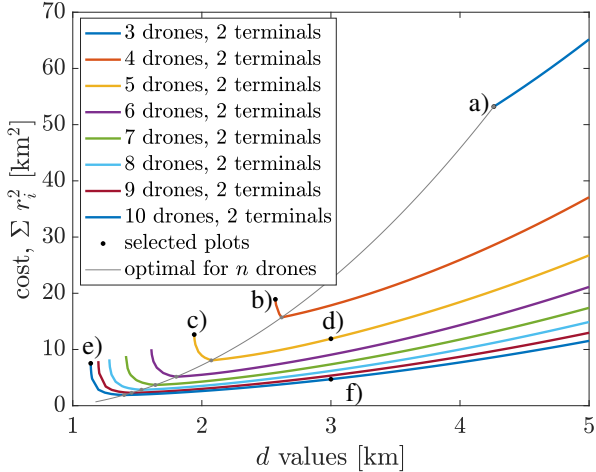
\begin{figure}[tb]
\centering
\begin{overpic}
    [width=1.0\columnwidth]{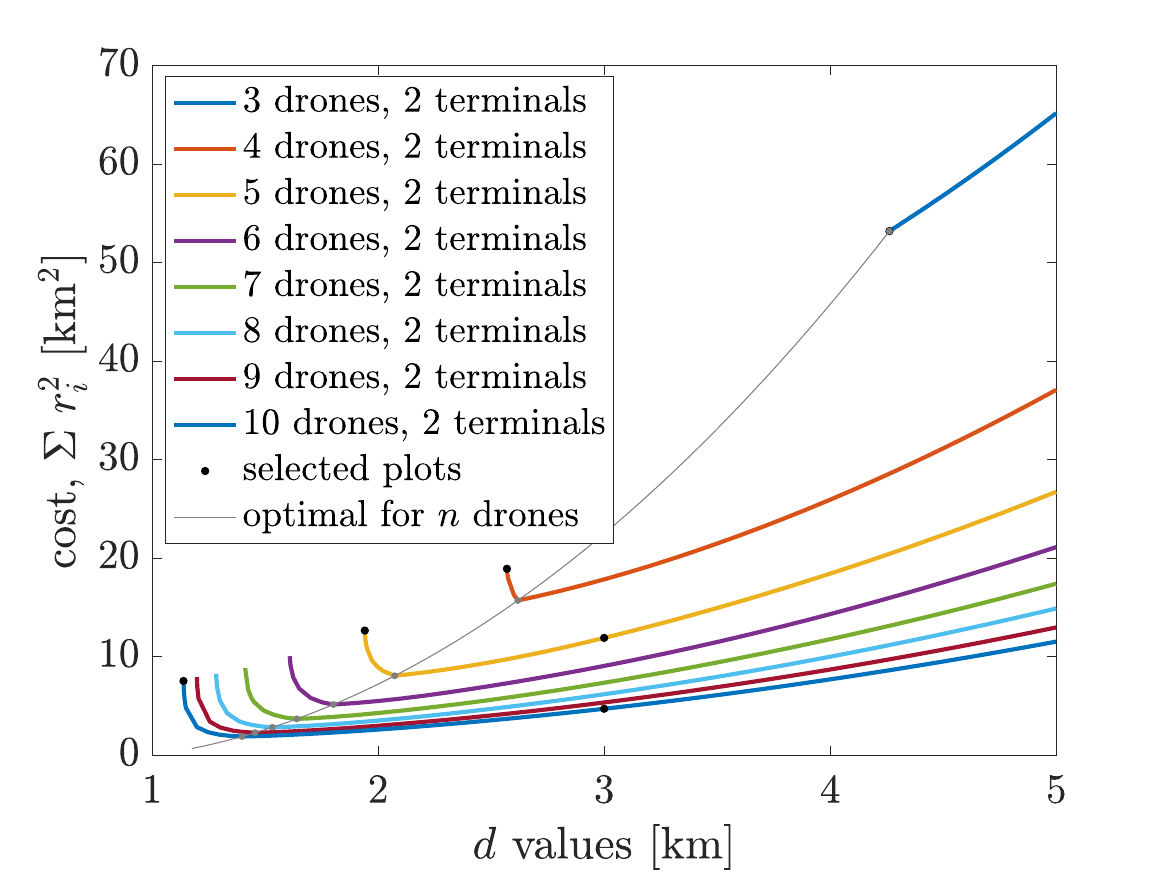}
    \put(72.5,54.5){\textcolor{black}{a)}}
    \put(39.8,25){\textcolor{black}{b)}}
    \put(30.2,22){\textcolor{black}{c)}}
    \put(50.5,21.8){\textcolor{black}{d)}}
    \put(14.5,18){\textcolor{black}{e)}}
    \put(54,11){\textcolor{black}{f)}}
    \end{overpic}
% Plot is generated by 2023Freedom/NetworkingPositionOptimization/RepresentativeOptimizationProblems/plotGatherDataCircleDroneFminconOpt.m
\vspace{-0.7cm}
\caption{Building the lowest cost network between two terminals  at $[\pm d,0]$ separated by a unit radius obstacle, using $n$ relays as shown in Fig.~\ref{fig:threeTerminalsOneDrone}. Increasing the number of relays $n$ always decreases the cost and decreases the minimum $d$ that can be covered.
\label{fig:varyDvalue2Term}}
\end{figure}  

The relays have angular spacing $\frac{\pi}{n+1}$. 
As shown in Fig.~\ref{fig:varyDvalue2Term} the plots of cost as a function of $d$ have a minimum at the optimal solution of Eq.~\eqref{eq:OPTIMALdFORn} (gray line). For smaller $d$ values the terminal transmission ranges must be less than the optimal value and the relays' ranges should be correspondingly larger.  For larger $d$ values all the transmission ranges are identical and the path of the relays forms two straight lines that bend in a circular arc about the obstacle. Three or more relays are required for a solution to exist (see Fig.~\ref{fig:NeedMoreThan2Drones}).

\begin{figure}[b]
\centering
\begin{overpic}[width=1.0\columnwidth]{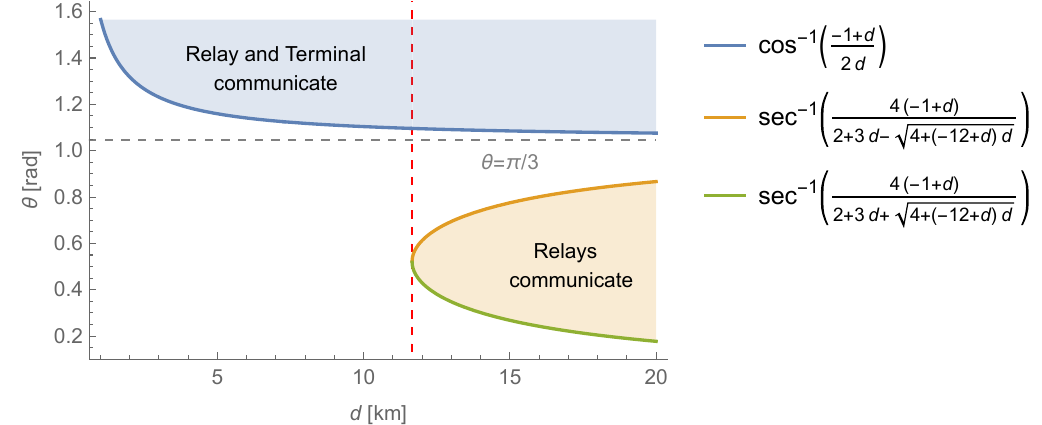}
\put (65,-5){\includegraphics[height=2cm]{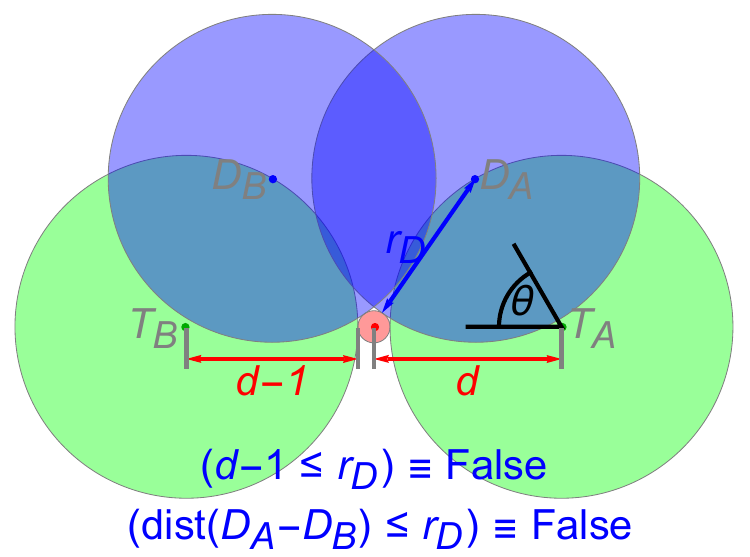}}
\put(0,37){\textcolor{black}{a)}}
\put(96,12){\textcolor{black}{b)}}
\end{overpic}
%\includegraphics[width=0.4\columnwidth]{figures/2DronesNeededSchem.pdf}
% Plot is generated by 2023Freedom/NetworkingPositionOptimization/RepresentativeOptimizationProblems/optimizeDronePositions.nb
\caption{We require at least three relays to build networks around a unit-radius obstacle with terminals $d$ units on either side of the obstacle as shown in Fig.~\ref{fig:varyDvalue2TermPlots}. With two relays placed symmetrically at $[d,0] + (d-1) [-\cos{\theta},\sin{\theta}]$ and  $[-d,0] + (d-1) [\cos{\theta},\sin{\theta}]$ the transmission disk must not overlap the obstacle disk (b).  (a) plots $\theta$ values such that the relays communicate in light orange and $\theta$ values such that the relays communicate with their nearest terminal in light blue. These regions have no union and converge as $d\rightarrow \infty$ to $\theta = \pi/3$ (gray dashed line).
\label{fig:NeedMoreThan2Drones}}
\end{figure}  

% \subsection[Two terminals at 2d separated by two unit radius obstacles]{Two terminals at $[\pm d,0]$ separated by unit radius obstacles at $[0,\pm s]$ }

% \todo{I think this might be a fun numerical simulation to set up. (ATB) }

These optimization problems are relatively simple since we know the communication graph topology \emph{a priori}. When we do not know the communication graph topology before the optimization, then the solution must include it as part of the optimization. Even without obstacles the problem of assigning the minimum area ranges to each terminal is NP-complete~\cite{clementi2004power}. For the simple case of three terminals and one relay we can find the optimal solution to the problem.
\vspace{-1mm}
\subsection{Optimal solution: three terminals,  one relay\label{sec:1ter1rel}}
Given a triangle with vertices $A$, $B$ \& $C$, the relay location $D$ that minimizes the cost for a strongly connected network has multiple candidate solutions as shown in Figs.~\ref{fig:OptimizeDronePlacement3TermCB} and \ref{fig:threeTerminalsOneDrone}.

\begin{figure}[b]
\centering
\begin{overpic}[width=0.45\columnwidth]{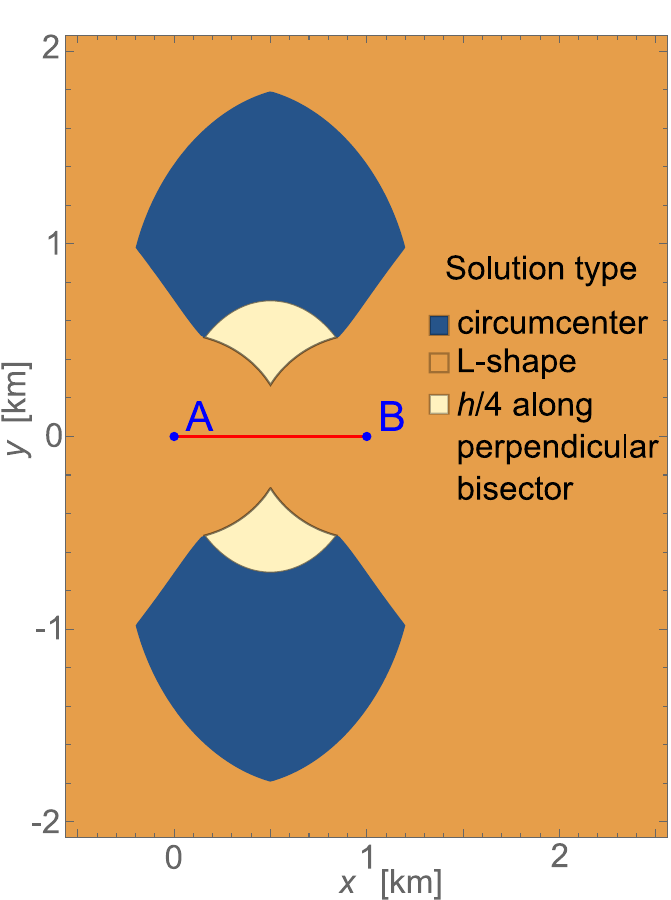}\put(0,86){\textcolor{black}{a)}}\end{overpic}
\begin{overpic}[width=0.53\columnwidth]{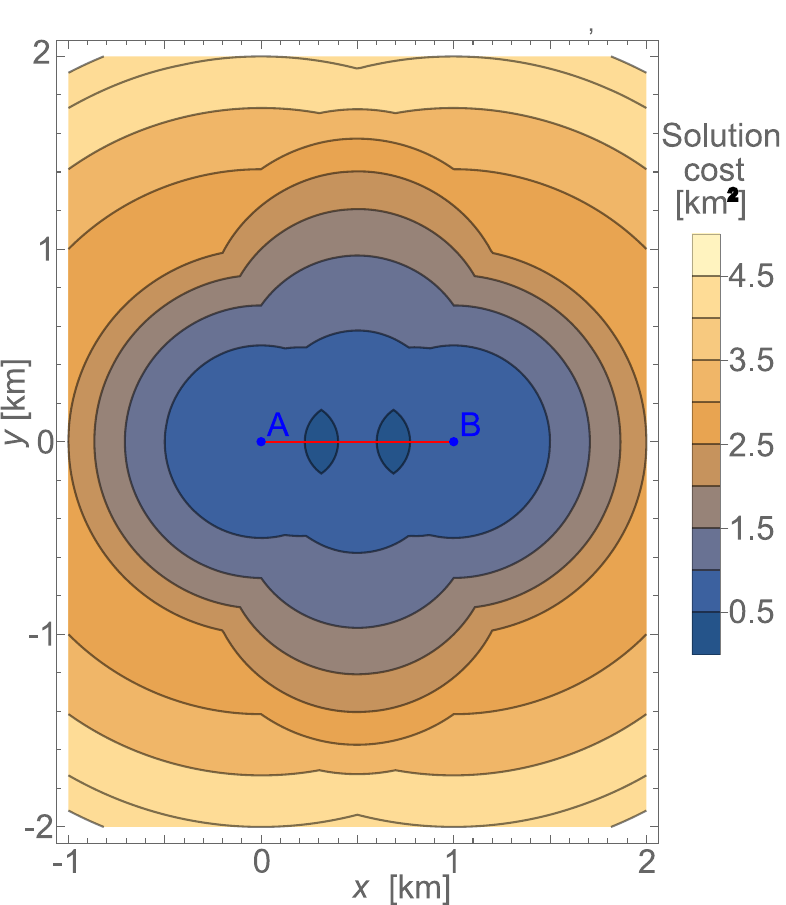}\put(0,86){\textcolor{black}{b)}}\end{overpic}
\vspace{-0.2cm}
% Plot is generated by 2023Freedom/NetworkingPositionOptimization/RepresentativeOptimizationProblems/calculationsForThreeTerminalsOneDrone.nb
\caption{Given three terminals at $A$, $B$ \& $C$ there are three candidate solutions for the optimal relay placement to minimize the cost of the network. The solution depends on the shape of the triangle.  
In the above plots $A=[0,0]$, $B=[0,1]$ and $C=[x,y]$.
\label{fig:OptimizeDronePlacement3TermCB}}
\end{figure}  
\begin{figure*}[bht]
\centering
\includegraphics[height=0.22\textwidth,trim={0.8cm 1.2cm 1.1cm  1.6cm},clip]{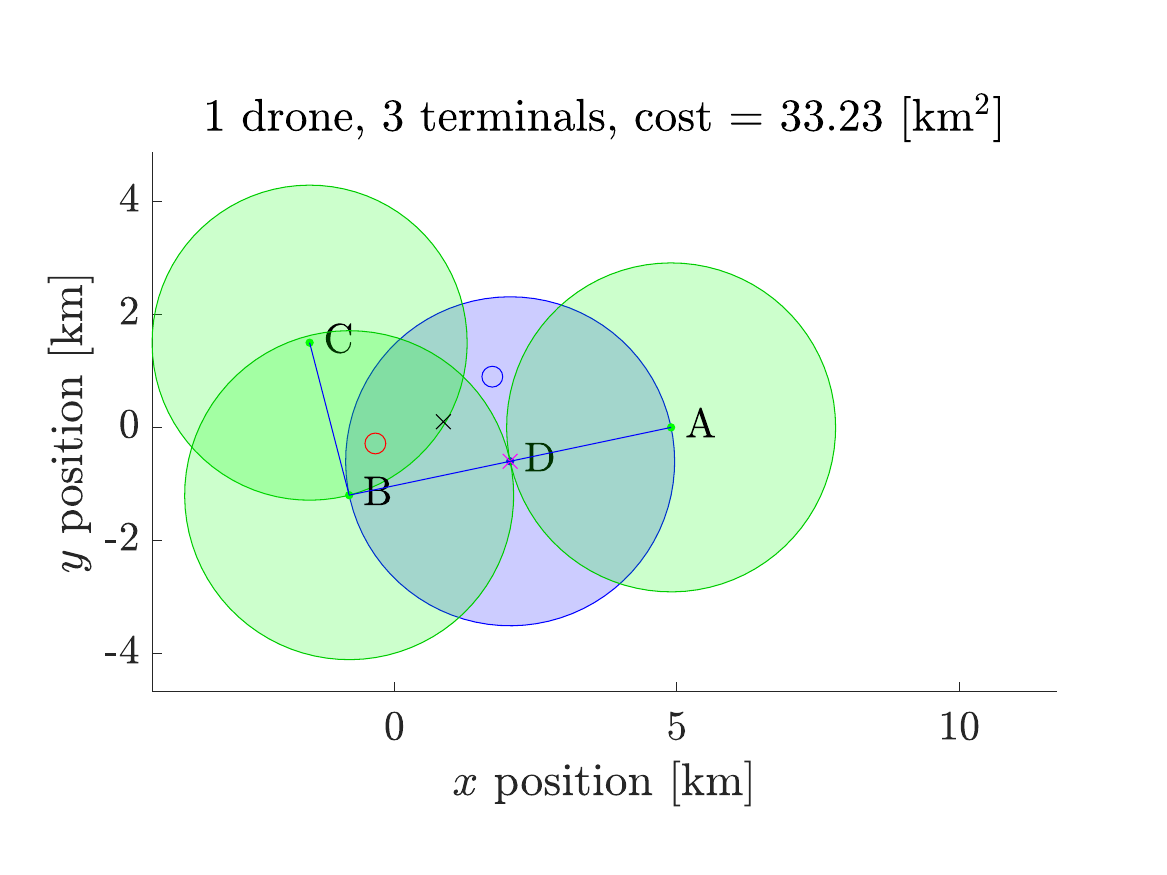}
\includegraphics[height=0.22\textwidth,trim={1.5cm 1.2cm 1.1cm  1.6cm},clip]{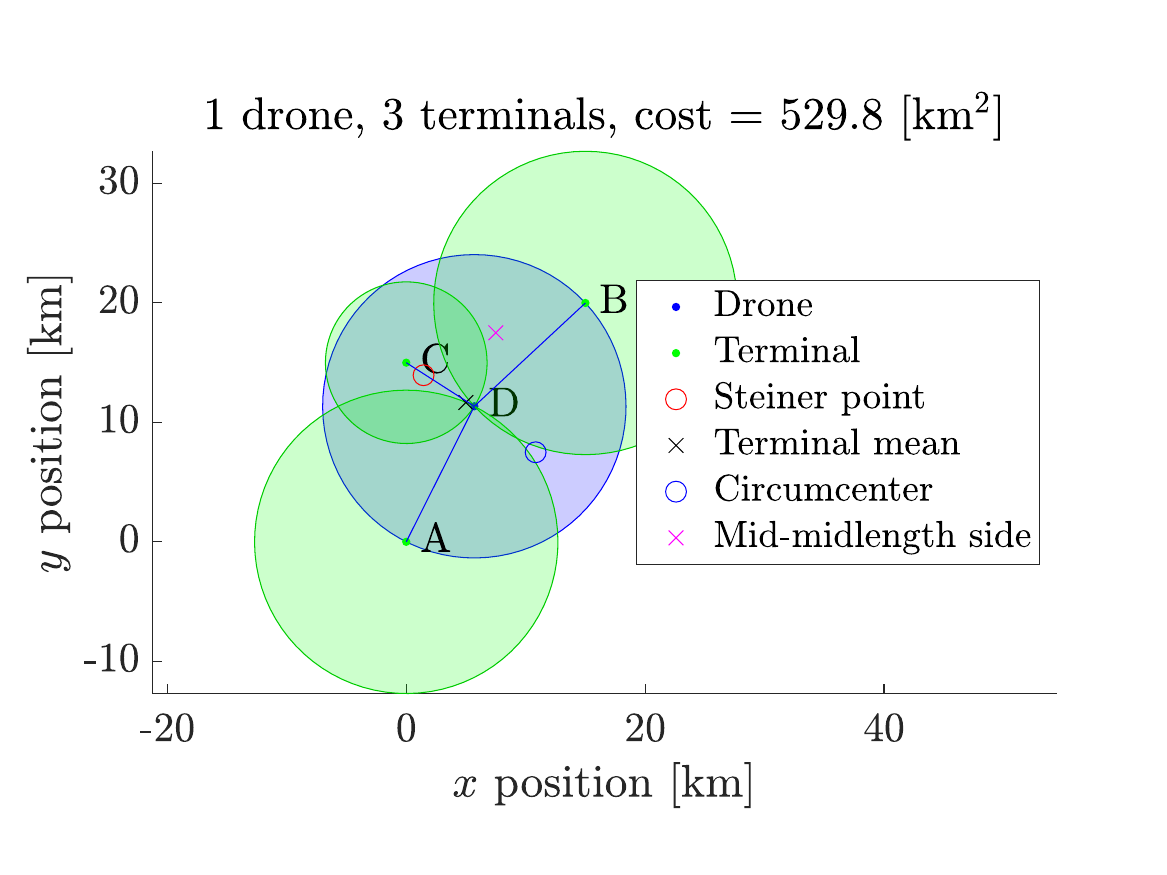}
\includegraphics[height=0.22\textwidth,trim={1.5cm 1.2cm 1.2cm 1.6cm},clip]{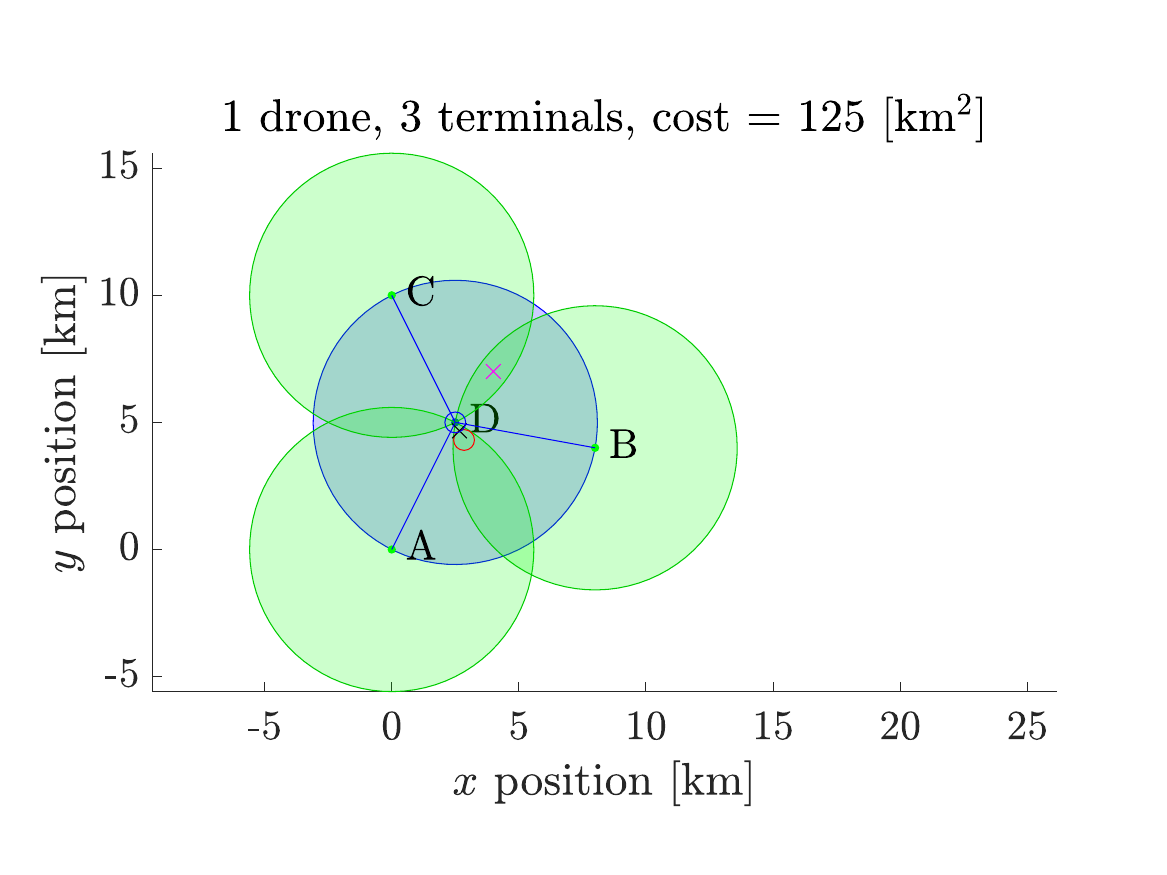}
% Plot is generated by 2023Freedom/NetworkingPositionOptimization/RepresentativeOptimizationProblems/generatePlotsthreeTerminalsOneDroneFmincon
\vspace{-2mm}
\caption{With three terminals and one relay there are three types of solutions for the relay position: a) on the midpoint of the second shortest side, b) on the perpendicular bisector of the longest side 1/4 the height of the triangle or c) at the circumcenter of the triangle. 
\label{fig:threeTerminalsOneDrone}}
\vspace{-0.7cm}
\end{figure*}  
%Given a triangle with vertices $A$, $B$, $C$, what is the location $D$ for a relay that minimizes the cost for a strongly connected network?
%In general, it is not optimal to place $D$ at the Steiner point, the mean of the terminals, or the circumcenter, as shown in Fig.~\ref{fig:threeTerminalsOneDrone}. 
% However, we show that an optimal solution is among three candidate solutions.
%The optimal solutions  either create an L-shaped network or a star-shaped network. The placement of the vertices determines which candidate solutions is optimal. 
%With three terminals and one relay, there are two candidate solutions for placing the relay. 
%In general the relay is not placed at the Steiner point, the mean of the terminals, nor the circumcenter, as shown in Fig.~\ref{fig:threeTerminalsOneDrone}

\begin{theorem}
For three terminals at points $A,B,C$ forming the triangle $\triangle(ABC)$ a relay is  placed optimally on one of the following three locations:
\begin{enumerate}
\item the midpoint of the second largest edge of $\triangle(ABC)$,
\item 25\% of the perpendicular bisector of the longest edge of $\triangle(ABC)$,
\item the circumcenter of $\triangle(ABC)$.
\end{enumerate}
\end{theorem}
\begin{proof}
Without loss of generality, we assume that  $A = [0,0]$, $B = [1,0]$ \& $C = [p,q]$; otherwise we rotate and scale. Let $D=[x,y]$ denote the location of the relay. %$p\leq 1/2$ and that $\overline{AB}$ is the longest side; otherwise we rotate, scale, mirror or rename the vertices. 
In an optimal network the relay is transmitting to two or all three terminals; otherwise it has no added value.

If $D$ has two neighbors then it is beneficial to transmit via the shortest side (wlog $AC$) and place $D$ on the second shortest side (wlog $AB$). The resulting costs are
\begin{align}
   \mathcal C_2 =& \phantom{+~\,} \underbrace{|AC|^2}_{{r_C}^2} +\underbrace{\max( \{ |AD|^2, |AC|^2\})}_{{r_A}^2} \label{eq:LineCost} \\
   & + \underbrace{|BD|^2}_{{r_B}^2} + \underbrace{\max( \{ |AD|^2 , |BD|^2\})}_{{r_D}^2}.\nonumber
\end{align}
%It is easy to check by a small case distinction that $\text{cost}_2$ is minimized when $D$ is placed on the midpoint of $AB$. %\linda{argument in comments}
If $|AD|\geq |AC|$ then $|AD|^2+|BD|^2$ as well as $\max( \{ |AD|^2, |BD|^2\}) $ is minimized for $|AD|=|BD|$. 
Otherwise we have $|AD|< |AC|$. If $|BD|>|AD|$ then decreasing $|BD|$ and increasing $|AD|$ improves; here we either end in the situation that $|AD|=|BD|$ (claim) or $|AD|= |AC|$ (case 1). Hence we have $|BD|<|AD|<|AC|$. In this case, $|BD|^2+\max( \{ |AD|^2, |BD|^2\})=x^2+(1-x)^2$ is minimized for $x=1/2$ (i.e., $|AD|=|BD|$).

Now we consider the case of $D$ having three neighbors. The optimal solution $D$ lies within $\triangle(ABC)$. The resulting costs are 
\begin{align}
   \mathcal C_3 =&  \underbrace{|AD|^2}_{{r_A}^2} + \underbrace{|BD|^2}_{{r_B}^2}+ \underbrace{|CD|^2}_{{r_C}^2} \label{eq:Cost3Star} \\
   &+ \underbrace{\max( \{ |AD|^2 , |BD|^2, |CD|^2\})}_{{r_D}^2}. \nonumber
\end{align}
 We argue that the maximum is attained by two distances; otherwise we may optimize.  Suppose the maximum is attained only for $|AD|$; this  implies that $x> 1/2$ by $|AD|>|BD|$. 
 We consider the derivative with respect to $x$, namely $4x+2(x-1)+2(x-p)=8x-2(1+p)$. If it does not vanish then we may slightly change $x$ and hence improve the cost. If it vanishes then $x=(1+p)/4>1/2$ implying that $p>1$ and $y\leq q/4$ (as $D$ lies in $ABC$). This yields a contradiction to the fact that $|AD|>|CD|$:
 \begin{align*}
 16 |AD|^2&=16(x^2+y^2)\leq (1+p)^2+q^4\\
 &\leq (1+p)^2+q^4+8p(p-1)
 \leq(1-3p)^2+(3q)^2\\
 &<16(x-p)^2+16(y-q)^2=16 |CD|^2.
 \end{align*}
 Thus, we conclude that the maximum in Eq.~\eqref{eq:Cost3Star} is attained by at least two distances.
If it is attained by exactly two distances then we may assume without loss of generality, that $|AD|=|BD|$ (i.e., $D$ is placed on the perpendicular bisector). Then the cost function simplifies to 
$$3|AD|+|CD|=3(x^2+y^2)+(x-p)^2+(y-q)^2 \, ,$$
and its derivative with respect to $y$ reads as $6y+2(y-q)$ which vanishes exactly if $y=q/4$.
If the maximum is attained for all three distances then $|AD|=|BD|=|CD|$ and $D$ is the circumcenter of the triangle.
\end{proof}
% 

% The cost for the star-shaped network is found numerically, but is bounded above by placing the relay at the middle of the longest side.
% \begin{align}
% \text{cost}_\text{longestSide} = \begin{cases}
% 2*\ell_\text{off}^2 + 2*\frac{\ell_\text{long}}{2}^2 & \ell_\text{off}>2*\frac{\ell_\text{long}}{2}\\
% \ell_\text{off}^2 + 3*\frac{\ell_\text{long}}{2}^2  & \text{else}
% \end{cases}
% \end{align}
% where $\ell_\text{long}$ is the length of the longest side and $\ell_\text{off}$ is the distance from the midpoint of the longest side to the remaining terminal.

%#################################################################
\section{Algorithms for multiple terminals and multiple obstacles}\label{Sec:AlgsForMultTerminalsAndObstacles}
%#################################################################
Solutions to the problem of placing movable relays to enable communications between fixed terminals are explored. We begin by adding obstacles between two terminals and finding a solution strategy. The problem's complexity is increased by adding additional terminals and obstacles.
% //// Rewrite for previous line ///
%   To expand the flying ad hoc network to include movable relays, and multiple obstacles, the system is first localized to only two terminals. The solutions found to this system are applied to incleasingly complex networks with varying quantities of terminals, relays, and obstacles. 

% Note: we go from A to C in the italized labeling
% A. 2 terminals n relay x obstacles 1. bi-tangent 2. k optimal path 3. 2d links 
%C. m terminals n relays, no obstacles 
% D. m terminals, n relays, x obstacles 
%#################################################################
\subsection[Solving for m=2 terminals, n relays with obstacles]{Solving for $m=2$ terminals, $n$ relays with $\phi$ obstacles\label{sec:2term}}
%#################################################################
Given two points $A$ and $B$ on a plane the shortest path that connects them is a straight segment. 
If the plane contains obstacles then a shortest path that avoids the obstacles may have a different shape.
% If the plane contains obstacles, a generalization of the concept of distance is required: \linda{what is important about the following definition? it seems to not add much but maybe I overlook something? Probably saying something like a shortest path avoiding the obstacles may be of very different shape.}
% \begin{align}
% &d(\vec x,\vec y) = \begin{cases}
%     \norm{\vec x,\vec y}\,\,(\vec y-\vec x \text{ does not intersect obstacle})\\
%     \tilde d(\vec x,\vec y)\,\,(\vec y-\vec x \text{ intersects obstacle})
% \end{cases}   \nonumber\\
%     &\forall\vec x,\vec y\in\mathbb R^2/\{\textbf{obstacles}\}\;,\label{eq:piecewise}
% \end{align}where the function $\tilde d$ depends on the specific formulation of the obstacle. 
% In presence of obstacles, the notion of shortest path depends on the position and shape of the obstacles, and on the dimensionality of the network.
% \todo{not sure how next sentence relates:}
%  Indeed, the links may be simple segments or curves, as a laser beam connecting two points of an optical grid, or might have a planar extension, as a circular wave front caused by a drop of water in a pond, which connects the center of the front to each radial point. In the following paragraph we present a well-known method to determine shortest paths when the obstacles have circular shape.

\subsubsection{Bitangents method\label{sec:bitangents}}
For simplicity we assume the obstacles are $\phi$ circles with centers $\vec O_i$ and radii $r_{O,i}$, $i \in \{1,...,\phi\}$, whose area is excluded from the set of possible coordinates for the points of a path. In this case, the shortest path between two points is given by an alternating sequence of straight-line
segments and arcs along the circumference of  obstacles~\cite{chang2005shortest}. We begin by finding the  \emph{bitangents} for each pair of circles, and the tangents from the terminals $A$ and $B$ to each individual circle. If two circles do not overlap four bitangents exist. These bitangent lines are tangent to both circles.
If two circles partially overlap then only bitangents that touch the circles externally exist. 
However  if one circle is contained inside the other then no bitangents of any kind exist. 

Among all the bitangents and tangents determined through this procedure we keep only those whose line of sight (LoS) is not obstructed by (i.e. do not cross) other obstacles. 
Next, we add the circular arcs that connect bitangent points on each circle. A pair of points on a circle is always connected by two arcs, which will be different for non antipodal points. If we add to the set of segments the set of all the shortest arcs connecting pairs of points on the circles then we can cast the present system as a graph $G = (V,E,w)$:
\begin{itemize}
    \item The terminals, the tangent points, and the bitangent points are the nodes $V$ of the graph;
    \item The tangents, the bitangent segments, and the circle arcs are the edges $E$;
    \item The lengths of each arc or segment are the weights $w$.
\end{itemize} Once the graph $G$ is constructed then select a graph search algorithm to determine the shortest path between $A$ and $B$. Repeating this procedure for every pair of $A$ and $B$ in the accessible domain defines the generalized distance. Figure \ref{fig:yen}a shows, in green, the shortest path among $A$ and $B$ determined by the bitangent method.

\subsubsection[k-optimal paths and Yen's algorithm]{$k$-optimal paths and Yen's algorithm\label{sec:yen}}
% When  constraints other than the length of the path are used, the shortest path might no longer be the optimal solution. %a tautology?
Introducing obstacles into a simply connected region of the plane generates a connected region where nontrivial \emph{homotopy} classes might exist. Pairs of paths having endpoints in common are said to be ``homotopic''. If there exists a \emph{continuous transformation} that could bring one to the other or vice versa. Intuitively, if the two paths enclose an obstacle, no such continuous function can be defined and the two paths are not homotopic.

% for example, a closed path $\gamma$ drawn around an obstacle cannot be continuously shrunk to a point, so $\gamma$ is said to be not \emph{homotopic} to a point. 
%Homotopic:  if one shape can be continuously deformed into the other
% This notion can be extended to open paths in the following way: given any two non-coincident points $A$ and $B$ on a closed path, consider the two open curves $\gamma_1$ and $\gamma_2$ that have the two points as endpoints and are oriented from $A$ to $B$. Therefore, if ($-\gamma_i$) is the same curve but oriented from $B$ to $A$, we have $(\pm\gamma_1)\cup(\mp\gamma_2) = \gamma$, where the sign depends on the orientation of $\gamma$. If $\gamma$ encloses an obstacle and therefore is not homotopic to a point, equivalently $\gamma_1$ and $\gamma_2$ cannot be continuously deformed one to the other, and are said to be not homotopic. 

\begin{figure}[tb]
\centering
\begin{overpic}[width=.23\textwidth]{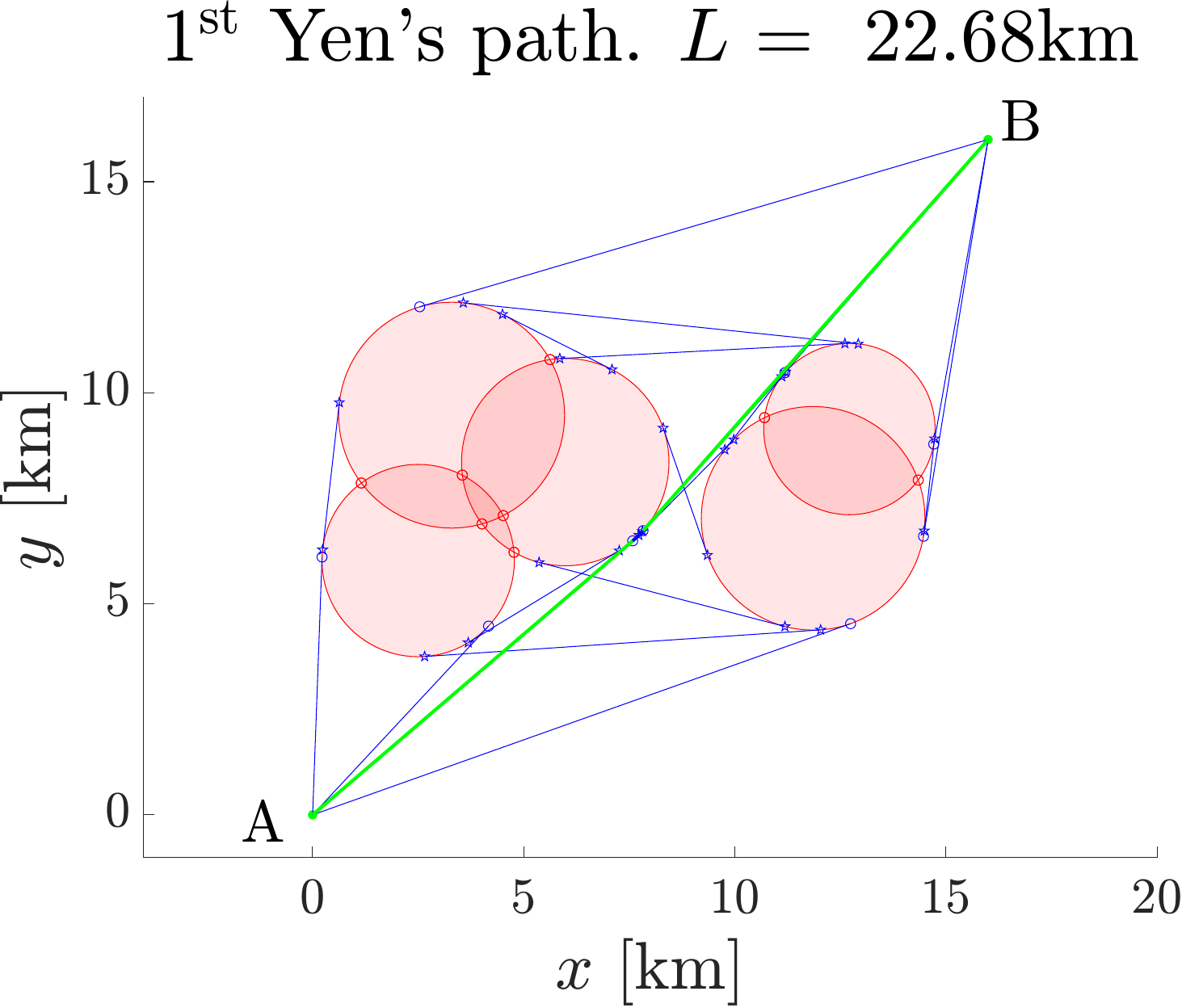}\put(20,65){\textcolor{black}{a)}}\end{overpic}
\begin{overpic}[width=.23\textwidth]{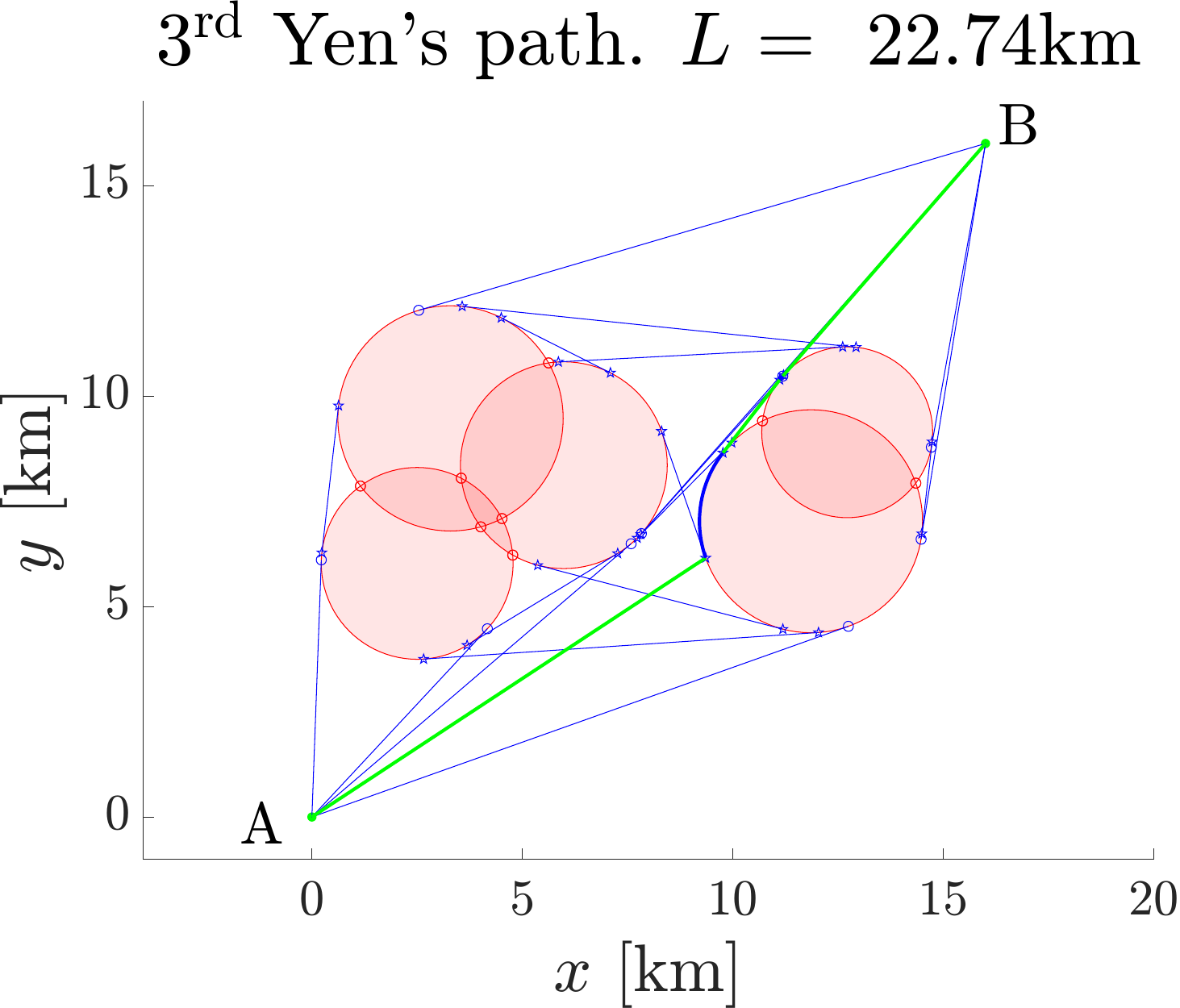}\put(20,65){\textcolor{black}{b)}}\end{overpic}\\
\vspace{1mm}
\begin{overpic}[width=.23\textwidth]{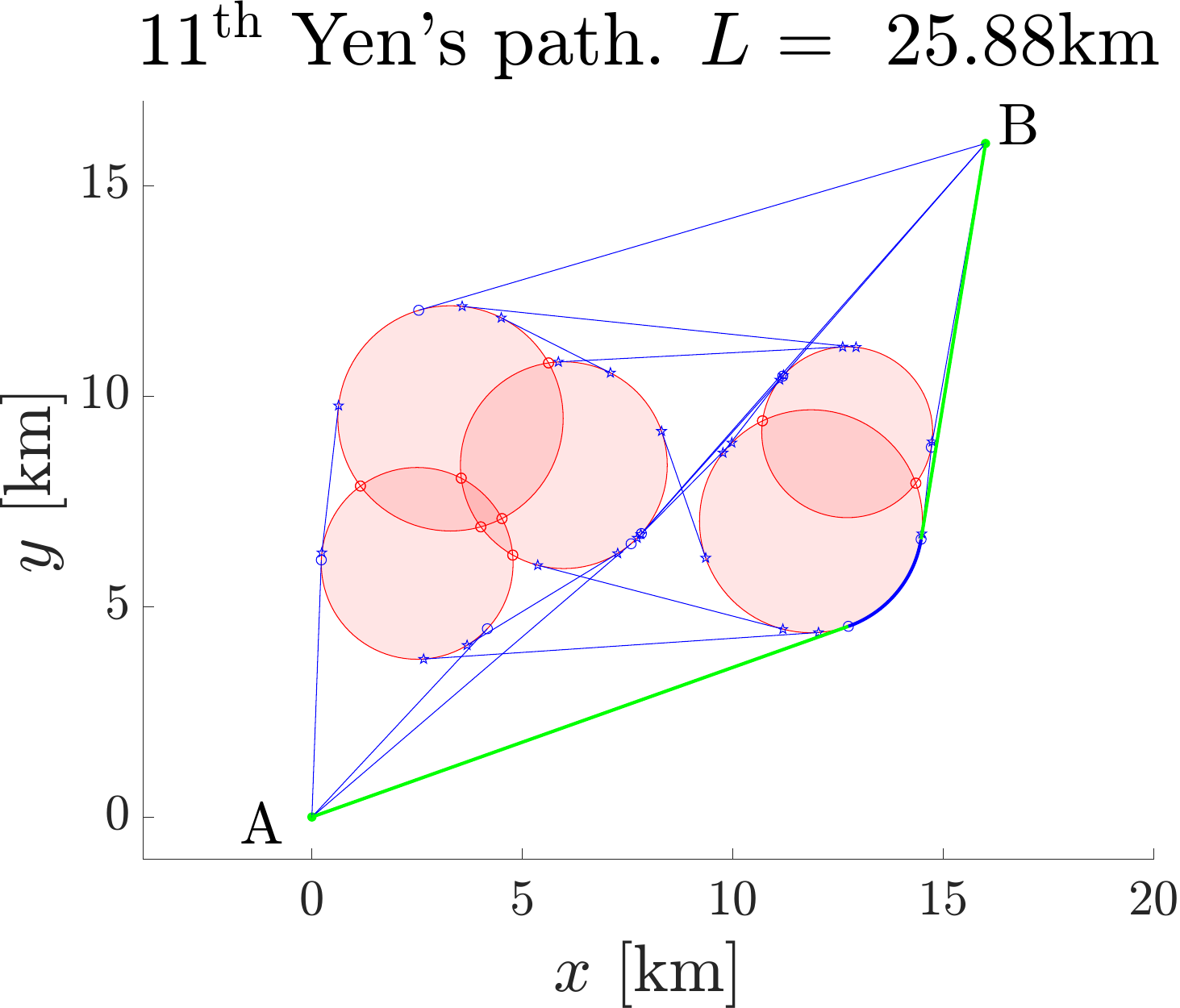}\put(20,65){\textcolor{black}{c)}}\end{overpic}\begin{overpic}[width=.23\textwidth]{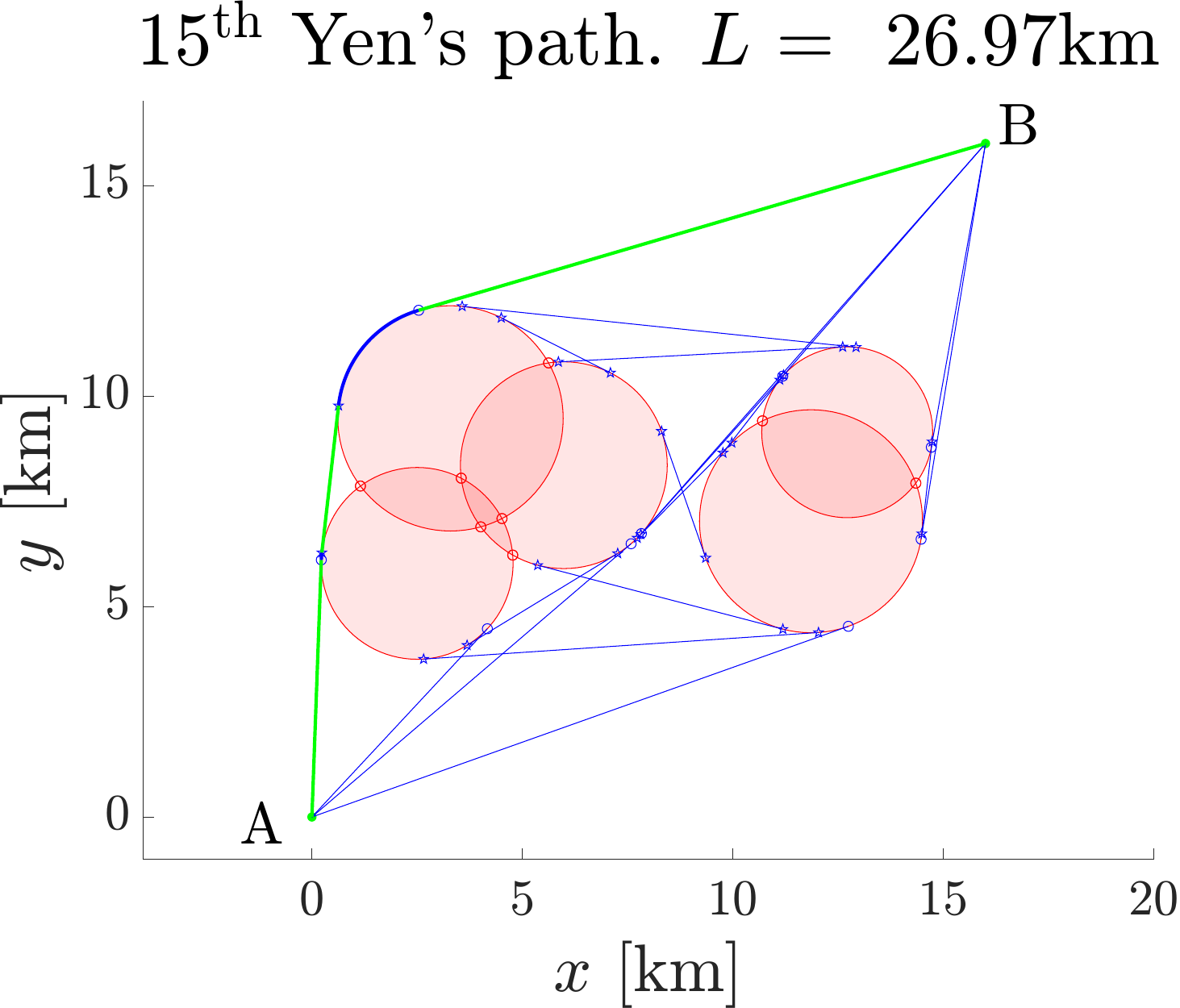}\put(20,65){\textcolor{black}{d)}}\end{overpic}
% pics created using code found at DroneSwarmEasterIsland\17_BitangentsSteiner\YenPaths.m
\caption{Construction of Yen's $k$-optimal paths (green) for a certain configuration of terminals ($A$ and $B$) and obstacles (red) with the bitangents method. 
%The $k$-optimal paths are determined by scanning the graph through Yen's algorithm. 
The $1^{\textrm{st}}$ Yen's path (a) corresponds to the absolute shortest path, particularly the shortest in its homotopy class.  (b) is in the same homotopy class of  (a) but it is not the shortest of its class. (c\&d) belong to different homotopy classes and are also not the shortest paths, rather they are the shortest paths within their respective classes.
\label{fig:yen}
}\vspace{-3mm}
\end{figure} 
In our example of Fig.~\ref{fig:yen}, the green paths in  Fig.~\ref{fig:yen}(a\&b) are homotopic as they do not enclose any obstacle. Conversely, the green paths in  Fig.~\ref{fig:yen}(b\&c) are not homotopic to each other or to those in  Fig.~\ref{fig:yen}(a\&b). 
%(SAVE FOR THESIS) Although homotopy and path length are not equivalent, homotopy is generally tied to a notion of \emph{distance} between two paths: due to the obstacle enclosed between two non-homotopic curves, some of their points will be as far apart as the linear dimensions of the obstacle.
Paths of different homotopy are found by considering the \emph{$k-$optimal} paths between $A$ and $B$ (i.e. a set of paths ordered by their length). Using Yen's algorithm~\cite{yen1970algorithm} applied to the edges of the graph $G$ defined in Sec.~\ref{sec:bitangents} such a set can be built. If $k$ is sufficiently large then all the possible homotopy classes will be visited. 

\subsubsection{2-dimensional links\label{sec:fmincon}}
The paths found in Sec.~\ref{sec:yen} are 1-dimensional. Given two nodes they are linked by either segments or arcs. Now consider nodes that have circular shape and have a variable transmission radius $r_{D,i}$. Such a radius defines the \emph{coverage} region of the relay. Given $n$ relays, the problem is their optimal placement to minimize the transmission cost between two points $A$ and $B$ . 

% From a practical perspective, \textcolor{red}{this amounts to assume that the transverse dimensions of what passes through these links are negligible with respect to the longitudinal ones.}\textcolor{blue}{I'm referring to how signals are moving across the trajectories: if it's a laser you can consider just the line of sight from a point to the other, whereas if you have a circular or a dipolar antenna which transmits in lobes, you have to consider the 2d or 3d extension of this channel and compare it to your environment, for example the small space between two obstacles.} %Aaron asks, "what does this mean?"
% In some cases this might not be applicable: if we consider a spherical  and we project its field on a plane, this will be described by a circular field of intensity decreasing with the distance from its center, and we can assume that the signal will be detectable up to some

The system is defined by Eq.~\eqref{eqn:TransmissionCostForNetwork} with $m=2$.
% We consider  $i$ to be able to receive from  $i-1$ if its center is contained within $(i-1)$'s coverage region:
% \begin{equation}
%      f_{i}^{\rightarrow}=\norm{\vec C_{i}^{(a)}(t)-\vec C_{i-1}^{(a)}(t)}-R_{i-1}^{(a)}(t)\leq0\;,
% \end{equation}If the communication is bidirectional this holds also in the opposite direction:
% \begin{equation}
%      f_{i}^{\leftarrow}=\norm{\vec C_{i}^{(a)}(t)-\vec C_{i-1}^{(a)}(t)}-R_{i}^{(a)}(t)\leq0\;,
% \end{equation}and it can be summarized as
% \begin{equation}
%      f_{i}^{\leftrightarrow}=\norm{\vec C_{i}^{(a)}(t)-\vec C_{i-1}^{(a)}(t)}-\min\{R_{i}^{(a)}(t),R_{i-1}^{(a)}(t)\}\leq0\;,
% \end{equation}
This restricts relays from transmitting into obstacle regions. %Their positions $\vec D_{i}(t)$ and transmission radii $r_{D,i}(t)$ must respect the following inequalities:
% \begin{equation}
%     f_{i,j}=\norm{\vec D_{i}(t)-\vec O_{j}}-r_{D,i}(t)-r_{O,j}\geq0,\,\forall i,j, \forall t\;.\label{eq:constraint}
%  \end{equation}
%where $\vec C_{j}$ and $R_{j}$, $j\in\{1,...,N\}$ are the positions and radii of the obstacles defined in Sec.~\ref{sec:bitangents}, and $i\in\{1,...,n\}$. 
We consider the inter-node distances as a measure of the cost of transmission for each pair and pursue the goal of minimizing this cost. This is a nonlinear multi-objective optimization problem with constraints. We consider a scalarizing function\cite{chugh2020scalarizing} as in Eq.~(\ref{eqn:TransmissionCostForNetwork}).

Given an initial number of relays and an initial guess on their positions the MATLAB function \texttt{fmincon} attempts solving the optimization problem while respecting the constraints. Fig.~\ref{fig:homotopy} shows the results of the optimization:
\begin{itemize}
    \item Fig.~\ref{fig:homotopy}a  shows the shortest paths of each homotopy class found in Sec.~\ref{sec:yen} where 20 relays have been placed evenly along each class. This is not an acceptable solution of the new problem since the coverage regions of the relays overlap the obstacles (violating the constraints).
    
    \item Fig.~\ref{fig:homotopy}(b-d) represent the solutions found by \texttt{fmincon} by moving the relays with overlapping discs away from the obstacles and adjusting the coverage radii to establish the optimal link between each pair of relays. 
\end{itemize}
The existence of a solution  depends on the initial number of mobile relays $n$. If $n$ is too small then the relays will not be able to link the two terminals without overlapping the obstacles. Let $n_0$ be the initial number of available relays, and $L$ the length of one of the paths in Fig.~\ref{fig:homotopy}a. Then an average radius $R_{\textrm{avg}}$ can be defined as
\begin{equation}
    R_{\textrm{avg}}=\frac{L}{n_0+1}\label{eq:avgrad}\,\,.
\end{equation}If the solution passes through obstacles separated by a distance larger than $R_{\textrm{avg}}$ then Eq.~\eqref{eq:avgrad} is a close estimate of the actual radii of the relays (Fig.~\ref{fig:homotopy}(c\&d)). Conversely corridors narrower than $R_{\textrm{avg}}$ result in regions of larger or smaller local densities of relays (Fig.~\ref{fig:homotopy}b).
% If $R_{\textrm{avg}}$ is the maximum coverage radius of a node, and $L$ is the length of one of the paths in Fig.~\ref{fig:homotopy}a, a lower bound on the number of relays is
% \begin{equation}
%     N_{\textrm{min}} = \left\lceil\frac{L}{R_{\textrm{avg}}} \right\rceil -1 \, .
% \end{equation} 
% However, the paths in Fig.~\ref{fig:homotopy}b-d have a length $L'$ which is generally larger than that of the corresponding homotopic path of Fig.~\ref{fig:homotopy}a. 
A new path of length $L'$ may then be covered if and only if the sum of  relays required to cover the different regions is still smaller or equal to $n_0$.

\begin{figure}[h]
\centering
\begin{overpic}[width=.22\textwidth]{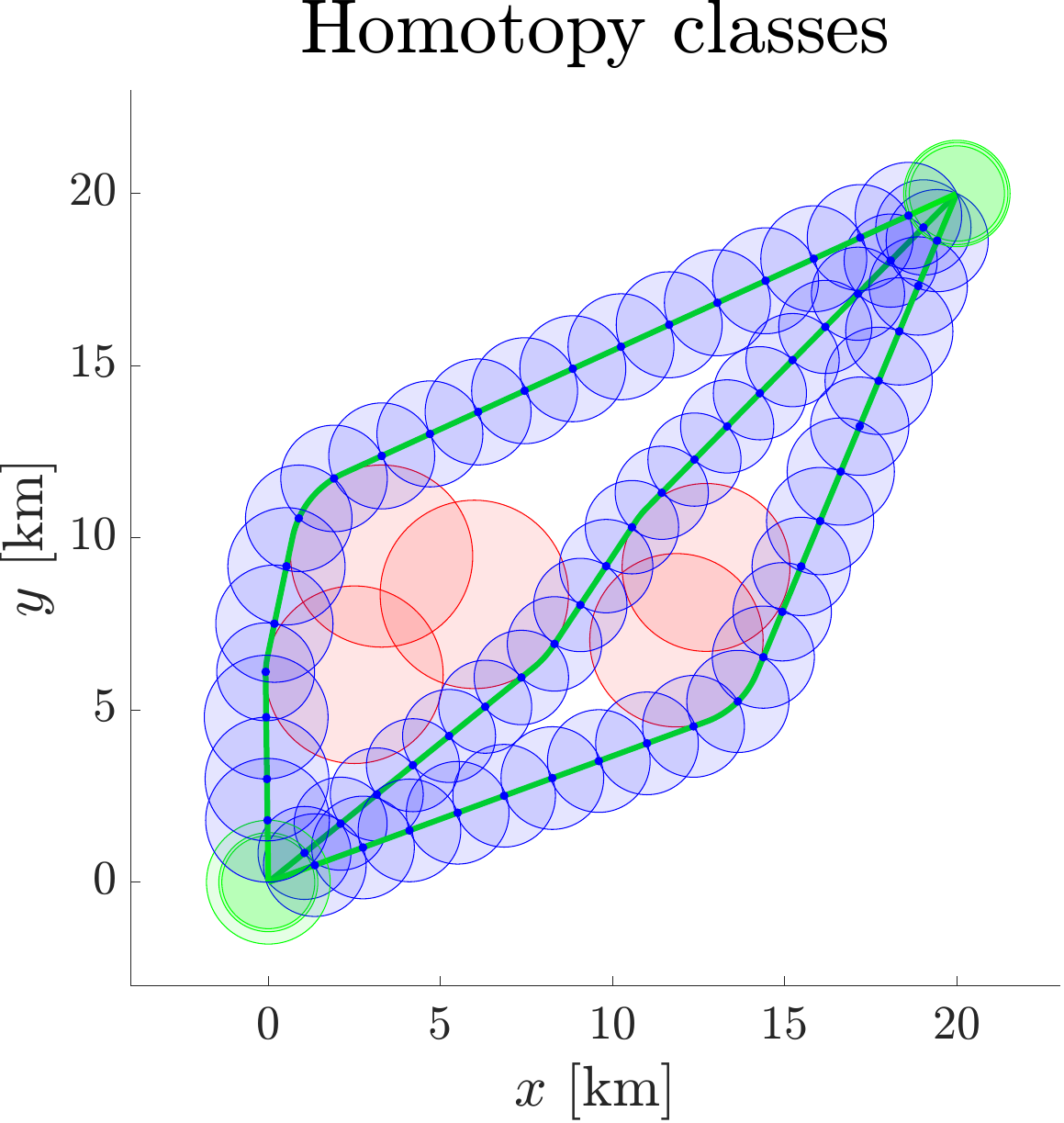}\put(15,90){a)}
\end{overpic}~\begin{overpic}[width=.22\textwidth]{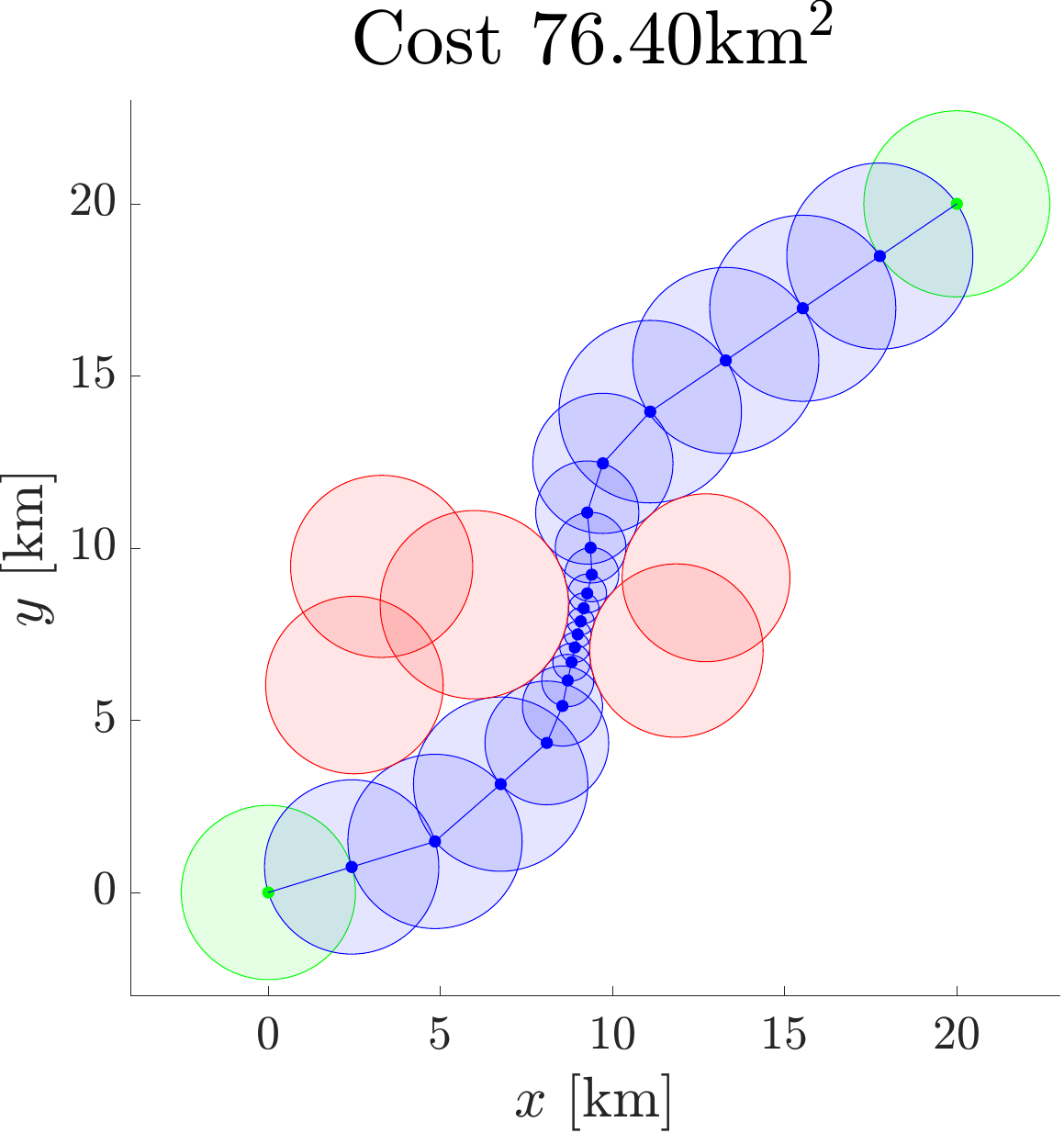}\put(15,90){b)}
\end{overpic}\\
\vspace{1mm}
\begin{overpic}[width=.22\textwidth]{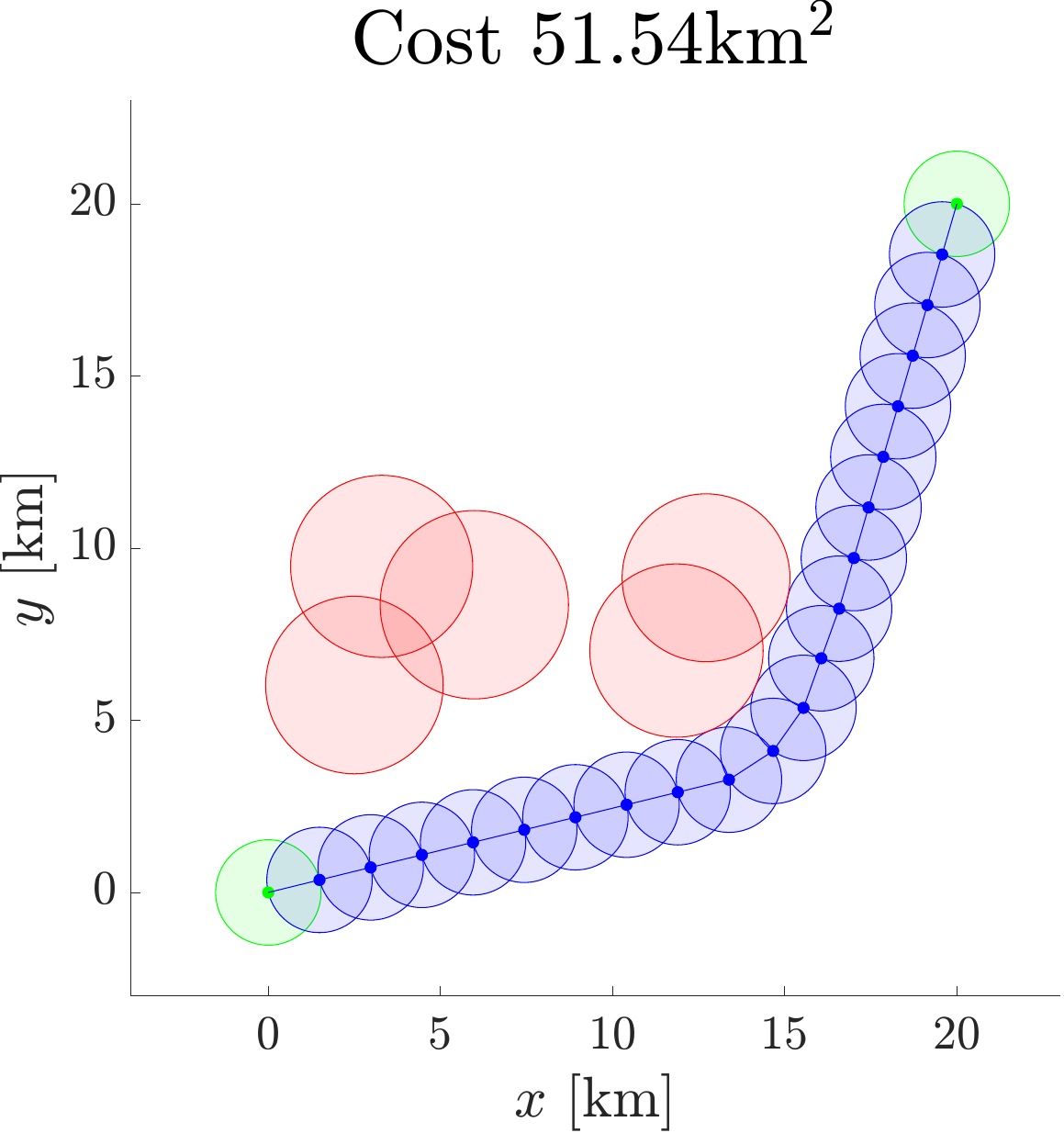}\put(15,90){c)}
\end{overpic}~\begin{overpic}[width=.22\textwidth]{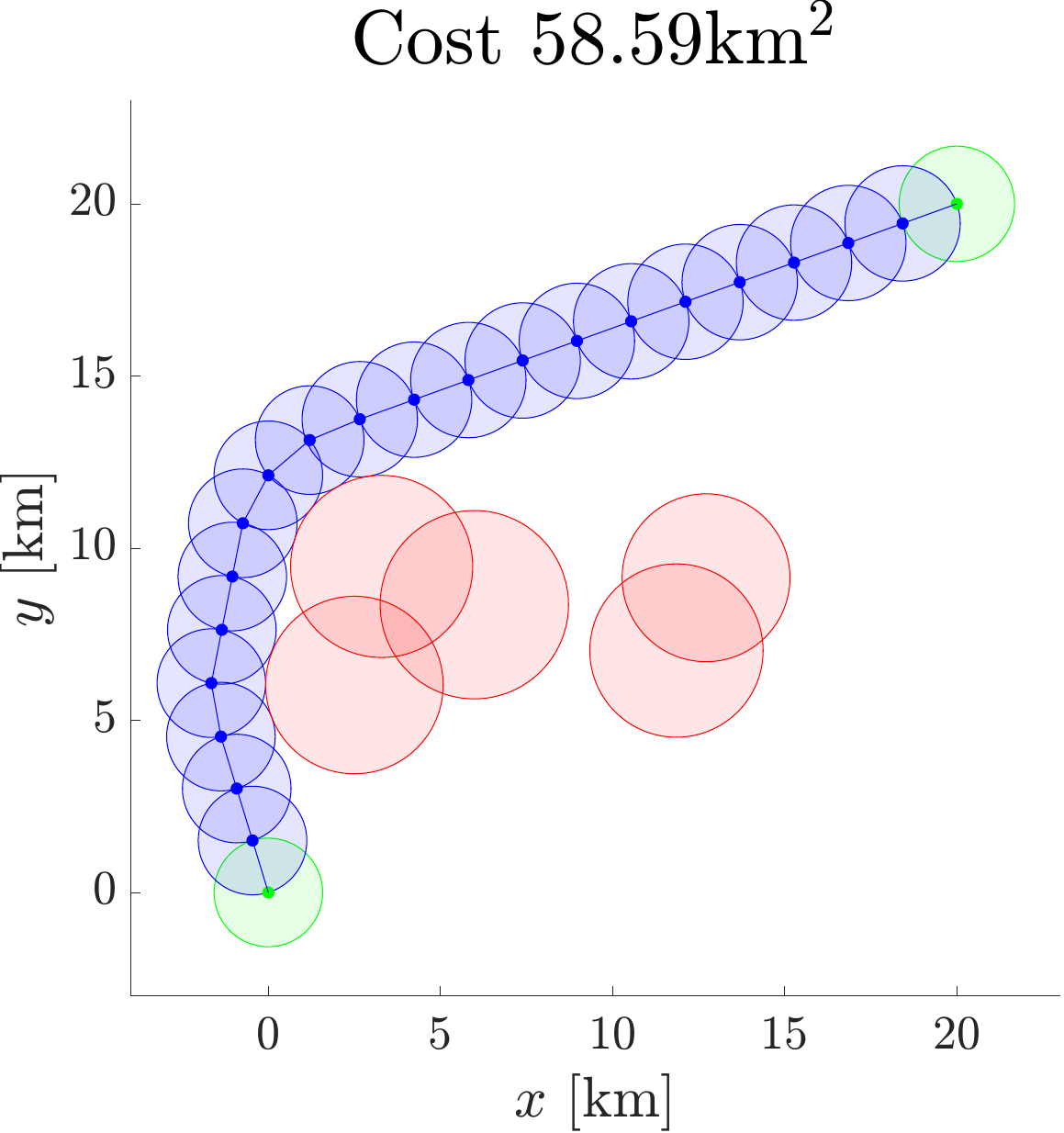}\put(15,90){d)}
\end{overpic}

%Plots generated using the code at DroneSwarmEasterIsland\17_BitangentsSteiner\HomotopyRevisited\CombineWfmincon.m
\caption{Homotopy classes (a) for a given configuration of terminals (green) and obstacles (red). Homotopy classes are found by applying Yen's algorithm. A uniform coverage is provided as initial condition using $n=20$ nodes (blue). A multi-objective optimization (b-d) removes the overlaps between nodes and obstacles seen in a).} 
\label{fig:homotopy}
\end{figure} 
The minimum cost path in Fig.~\ref{fig:homotopy}c
does not belong to the same homotopy class as the shortest path in Fig.~\ref{fig:yen}. This is due to the choice of Eq.~(\ref{eqn:TransmissionCostForNetwork}) as the scalarization function and the additional no-overlap constraint of Eq.~\ref{eqn:TransmissionCostForNetwork}, so it depends on both the mathematical model chosen, the number of available relays $n_0$, \emph{and} the specific obstacle distribution considered: if $n_0$ is sufficiently large for $R_{\textrm{avg}}$ to be smaller than the narrowest passage between obstacles, then the homotopy of the 1D and 2D solutions coincide. This result is expected to extend also to the general case of more than 2 terminals, as we will show in Sec.~\ref{sec:withobs}. 

%Homotopy can be used to categorize solutions of both optimization problems.
% \subsection{Review these todos and see if they're covered}
% Discussion on two terminals with $n$ relays with obstacles.

% \todo{Define Homotopy class.}

%   -- fmincon() - Global optimization with constraints (provide topology)
%   -- shortest path through forest + Yen's algorithm / topological sorting

%   \todo{Figure showing A* search with circular obstacles, and k-best paths}

%   \todo{Figure showing bubbles going through same set of obstacles, but different homotopy classes }

\subsubsection{Estimating homotopies\label{sec:estHomotopy}}
The dependence of the homotopy class of the optimal solution on the optimization problem raises the question of how many inequivalent homotopy classes $N_{\textrm{hom}}^{\textrm{max}}$ can exist, given two terminals and a set of obstacles. 

In general, the number of homotopy classes is infinite, as a path can loop around an obstacle an arbitrary number of times. However, 
%from the perspective of an optimization problem, 
solutions involving loops are always suboptimal and can be excluded from the computation. %Solutions that take unnecessarily long routes around the obstacles are also suboptimal. 

Ruling out loops %and expensive routes, 
$N_{\textrm{hom}}^{\textrm{max}}$ is interpreted as the number of \emph{reasonable} homotopy classes for an optimization problem, and it is expected to be finite. To obtain an upper bound on $N_{\textrm{hom}}^{\textrm{max}}$, we look for a map of the set of the reasonable homotopies onto some other known finite set. 

Assuming an obstacle is placed close enough to the terminals to obstruct their line of sight, a path connecting the two terminals can have the obstacle on one side or the other. The same reasoning extends to the case of a number $\phi>1$ of obstacles, that can individually and independently be placed on either side of the path. This poses the problem as the distribution of $\phi$ distinguishable objects in two bins. The number of ways this can be done is given by\begin{equation}
    N_{\textrm{hom}}^{\textrm{max}}(\phi)=2^\phi\,\,.\label{eq:Nmax}
\end{equation} Loops %and intricate routes
can affect the length of the path, but will not change the number of ways the distribution can be performed.

\begin{figure}[tb]
\vspace{-3mm}
\begin{overpic}[width=\columnwidth]{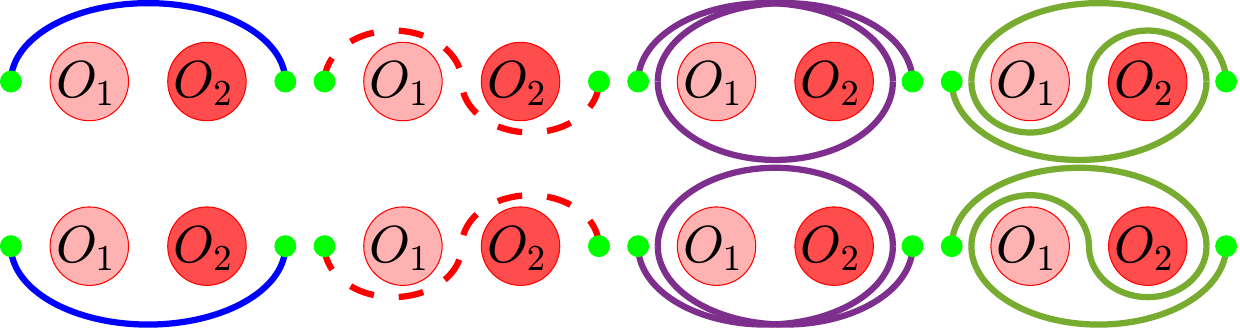}
\put(0,13){\textcolor{black}{($a_1$)}}
\put(25,13){\textcolor{black}{($a_2$)}}
\put(0,-2){\textcolor{black}{($a_3$)}}
\put(25,-2){\textcolor{black}{($a_4$)}}
\put(49,13){\textcolor{black}{($b_1$)}}
\put(74,13){\textcolor{black}{($b_2$)}}
\put(49,-2){\textcolor{black}{($b_3$)}}\put(74,-2){\textcolor{black}{($b_4$)}}
\end{overpic}
\caption{Map between the set of homotopies and the distributions of $\phi$ objects in 2 bins $[B_1;B_2]$:  ($a_1$)--($a_4$) map to $[O_1,O_2;\emptyset]$, $[O_2;O_1]$, $[\emptyset;O_1,O_2]$, $[O_1;O_2]$ respectively. The mapping is not injective, as ($b_1$)--($b_4$) realize the same distributions (albeit at the cost of using loops and expensive routes). However, if we limit ourselves to the homotopies that realize a distribution optimally (i.e. those of group $a$), the mapping is bijective.
\label{fig:allhomotopies}}%Plot made using \2023Freedom\DroneSwarmEasterIsland\19_BitangentSteinerJournal\HomotopyTrees\plotAllHomotopies.m
\vspace{-6mm}
\end{figure}
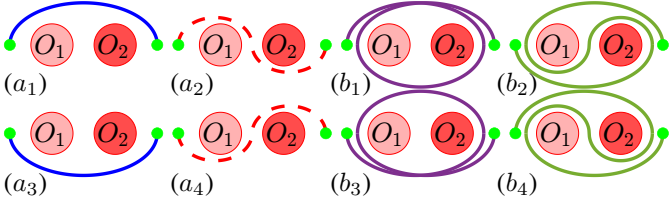

%as each obstacle can \linda{explain}where the $\phi!$ accounts for the possible permutations of the obstacle
To clarify the concept, we solve explicitly the problem with $m=\phi=2$: the goal is to determine the homotopies that map onto the $N_{\textrm{hom}}^{\textrm{max}}=2^\phi = 4$ possible distributions of two objects $O_1,O_2$ in two bins $\{B_1;B_2\}$, identified with the two sides of the path. Figure \ref{fig:allhomotopies} describes the process: homotopies ($a_1$)--($a_4$) map to $[O_1,O_2;\emptyset]$, $[O_2;O_1]$, $[\emptyset;O_1,O_2]$, $[O_1;O_2]$ respectively, so the mapping is realizable. Also ($b_1$)--($b_4$) realize the same mapping, so the mapping would not be injective. However, ($b_1$)--($b_4$) use loops, which lead to non-optimal paths. This suggests to identify a reasonable homotopy with a distribution realized optimally. So, the mapping becomes bijective once limited to candidate optimal realizations.

This concept can be extended to trees, i.e. when the number $m$ of terminals is larger than 2. Consider for simplicity a convex distribution of $m$ terminals, and any spanning tree with $n_s\leq m-2$ Steiner points. If we augment the tree by adding $m$ convex edges among pairs of adjacent terminals, we obtain a connected graph with $v =m+n_s$ vertices and $e=2m+n_s-1$ edges. By Euler's formula, the graph has a number of faces equal to $f=2-v+e=m+1$. 

Excluding the exterior face of the graph where obstacles do not interfere with the network to lie, the remaining $m$ faces represent the bins where to distribute the obstacles. Therefore \eqref{eq:Nmax} generalizes to
\vspace{-2mm}\begin{equation}
    N_{\textrm{hom}}^{\textrm{max}}(f,\phi)=f^\phi\,\,,\label{eq:Nmaxgen}
\end{equation}where $f\leq m$, with the equal sign only in case of a convex distribution of terminals (the convex hull has exactly $m$ vertices). The procedure is shown in Fig.~\ref{fig:treefaces}.

For what concerns the lower bound, the minimum number of different homotopies corresponds to the partitions of a set of $\phi$ objects, provided by the Bell numbers\cite{bell1938iterated}:
\begin{equation}
    N_{\textrm{hom}}^{\textrm{min}}=B_\phi\,\,.
\end{equation}For example, for $\phi = 4$ we expect at least $N_{\textrm{hom}}^{\textrm{min}}=B_4=15$ different partitions, and corresponding homotopies.
\begin{figure}[tb]
\centering
\begin{overpic}[width=.22\columnwidth]{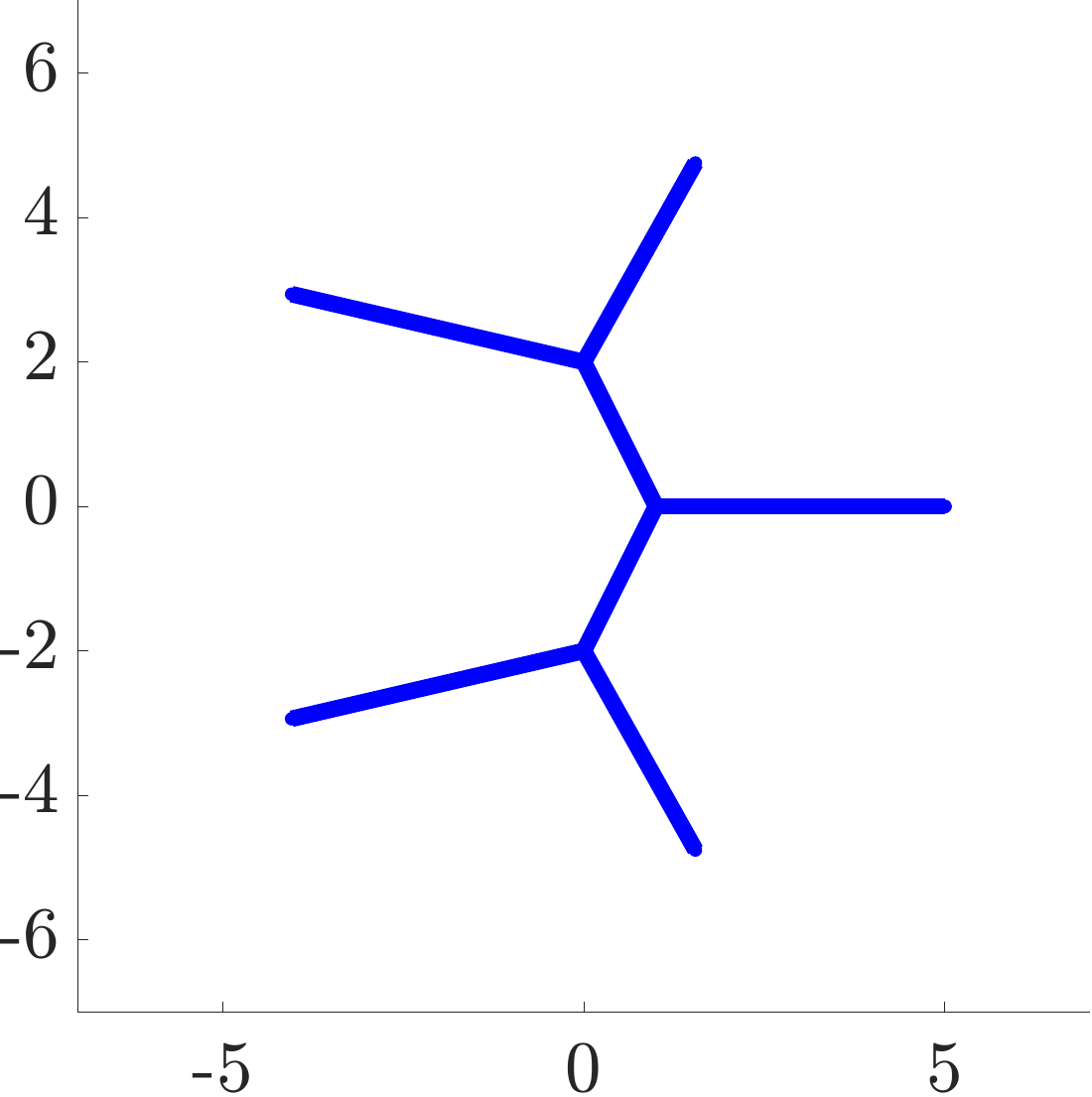}\put(0,0){($a_1$)}
\end{overpic}\begin{overpic}[width=.22\columnwidth]{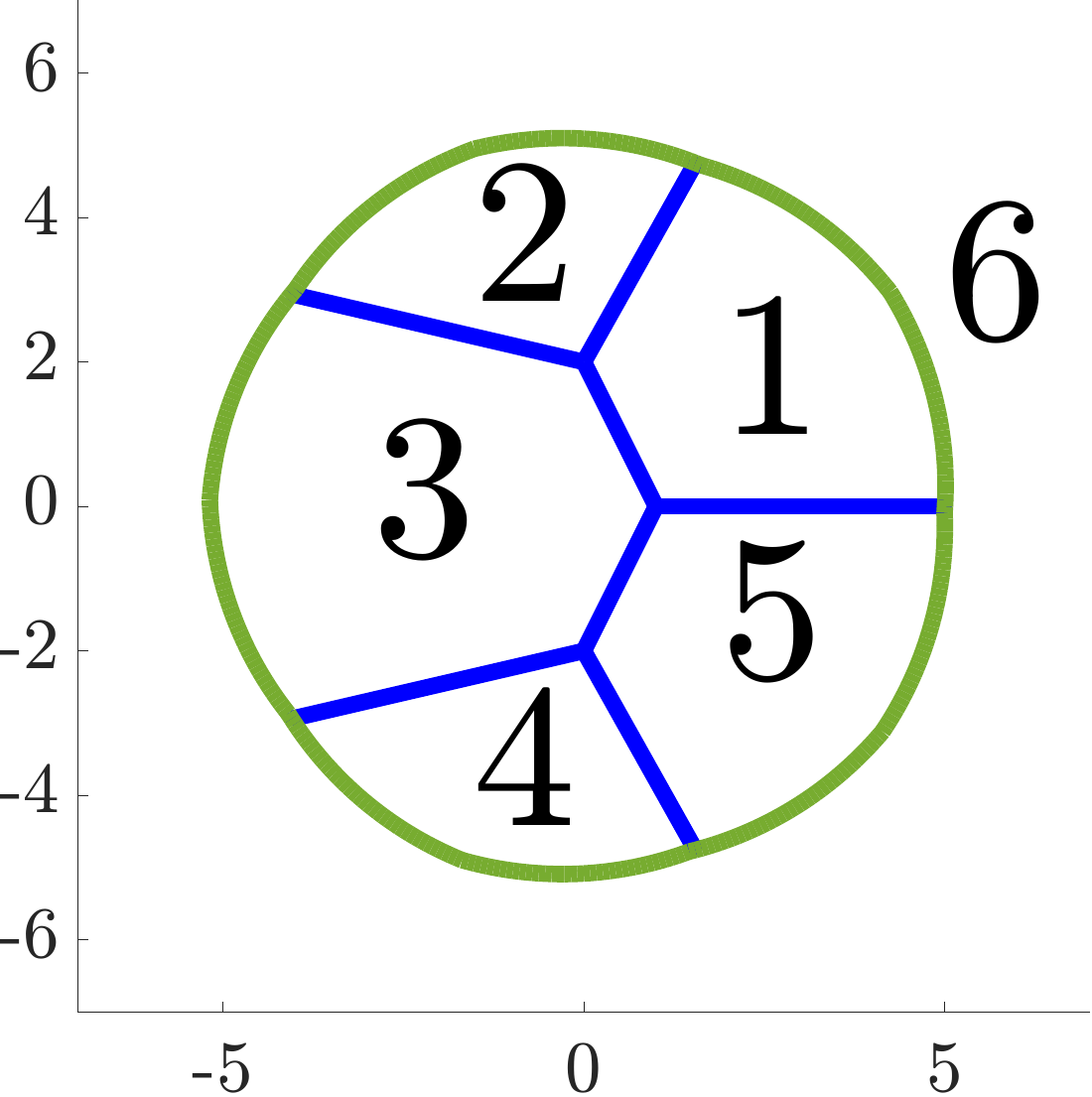}\put(0,0){($a_2$)}
\end{overpic}\begin{overpic}[width=.22\columnwidth]{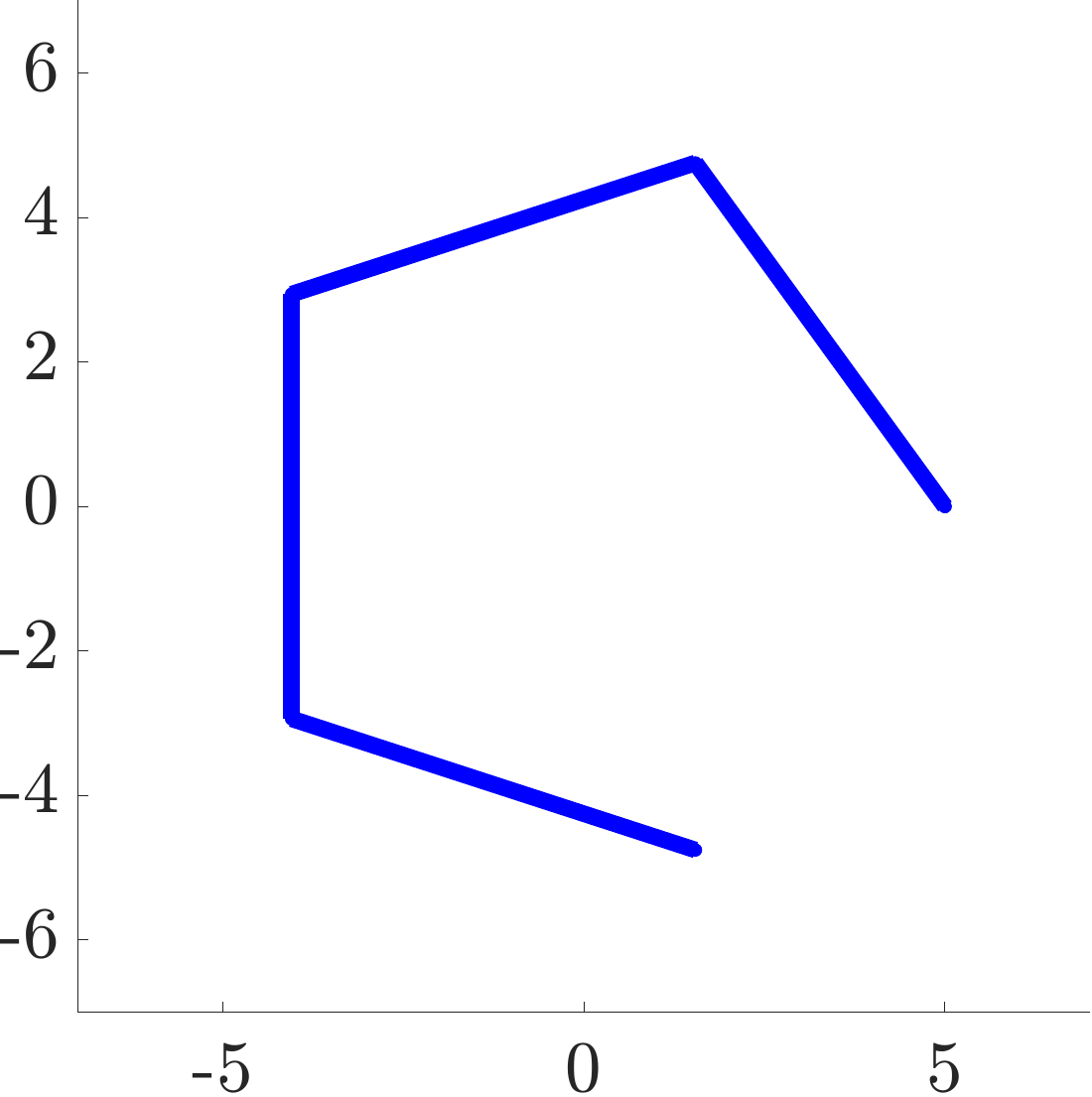}\put(0,0){($b_1$)}
\end{overpic}\begin{overpic}[width=.22\columnwidth]{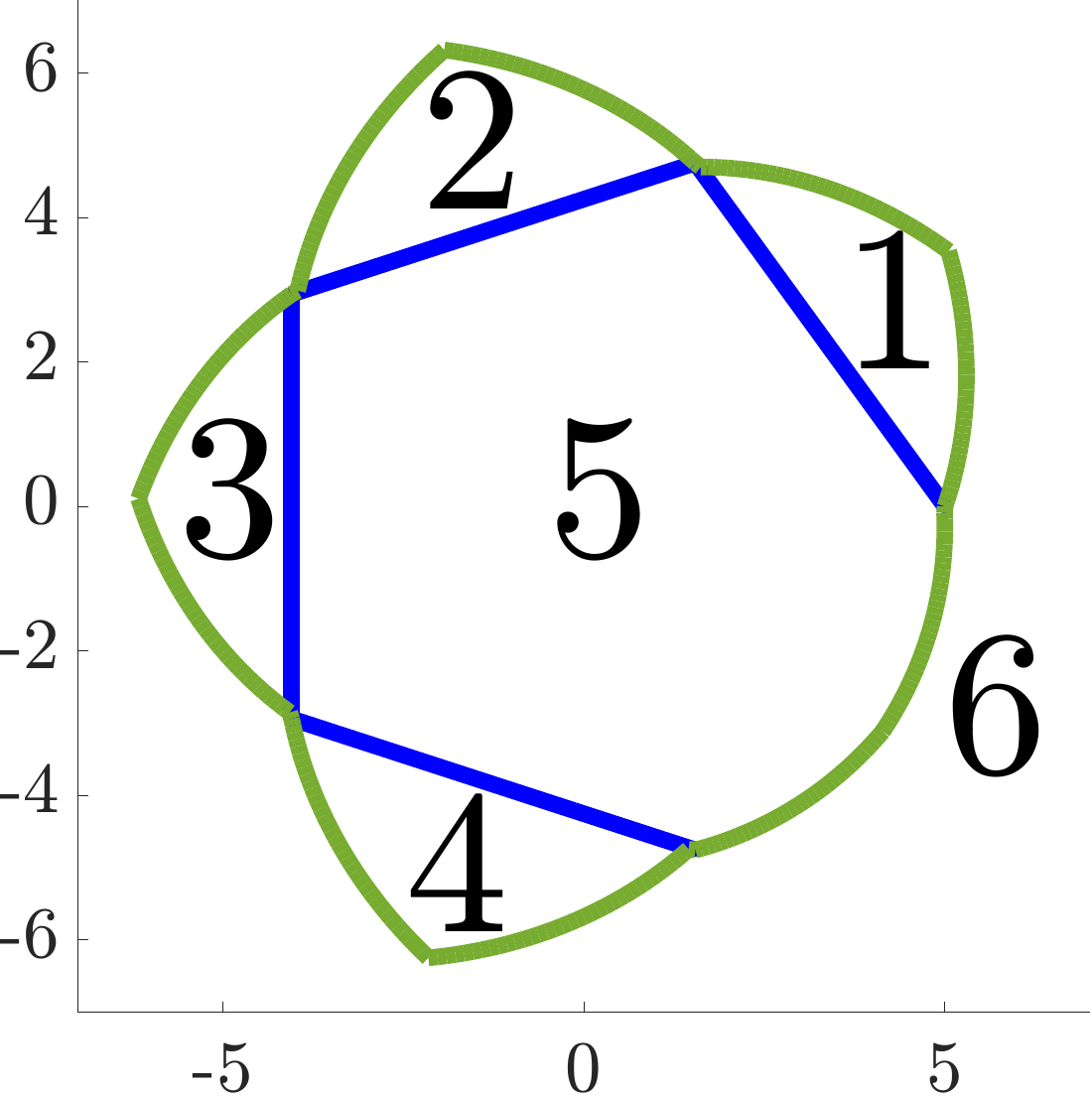}\put(-5,0){($b_2$)}
\end{overpic}

%Plots generated using the code at DroneSwarmEasterIsland\19_BitangentSteinerJournal\HomotopyTrees\plotTreeFaces.m. Crop each of the four files using Adobe Pro with the following margins: [0.115,0.755,0.788,0.12]in
\caption{Definition of  bins for obstacle distribution for cases with $m>2$. Starting from a convex distribution of $m=5$ terminals and any tree interconnecting them ($a_1$,$b_1$), we add $m$ convex edges and create a graph with exactly $f = m + 1$ faces ($a_2$,$b_2$). One of these faces represents the region where obstacles do not obstruct the network, so it can be discarded. The $m$ faces left represent the $m$ bins for the calculation of the distributions.} 
\label{fig:treefaces}
\vspace{-5mm}
\end{figure} 
\subsubsection{Homotopy classification\label{sec:homotopyclassif}}
Homotopy is a well established concept in topology, and formal methods exist to determine if two paths belong to the same homotopy class. 
Implementations of these methods can be found in \cite{bhattacharya2010search} and \cite{jenkins1991shortest}. Relying on exact results such as line integrals, these methods often prove tedious to implement in code. Furthermore, the ultimate focus of this work is \textbf{trees}: the extension of the examined references to more complex structures like trees has not been carried out yet, but can be prevented if the underlying assumptions rely on features of paths that are not shared by trees. For example, the algorithm of \cite{jenkins1991shortest} classifies the homotopy of a path by defining a suitably defined set of segments related to the obstacles, and checking if the path crosses them. The underlying assumption is that crossing a segment \emph{an even number of times} is equivalent to not crossing it at all (see Fig.~\ref{fig:homotopyrule}a). While this rule is perfectly reasonable for paths, it cannot be directly extended to trees: as shown in Fig.~\ref{fig:homotopyrule}b, in case of a node\footnote{A node of degree $n$ is a node where $n$ edges converge.} of degree 3 the departing branches\footnote{In this work, a branch of a tree is defined as a path whose ends are either terminals or nodes of degree $>2$.} may cross the segment two (blue, solid), one (red, dashed) or zero (black, dotted) times. If the red circles represent obstacles in the environment that cannot be continuously overcome, the black structure is not homotopically equivalent to the other two, in particular to the one that crosses the segment twice. Individual branches of a tree, however, are paths so the rule applies to them.

\begin{figure}[tb]
\centering
\begin{overpic}[width=0.40\columnwidth]{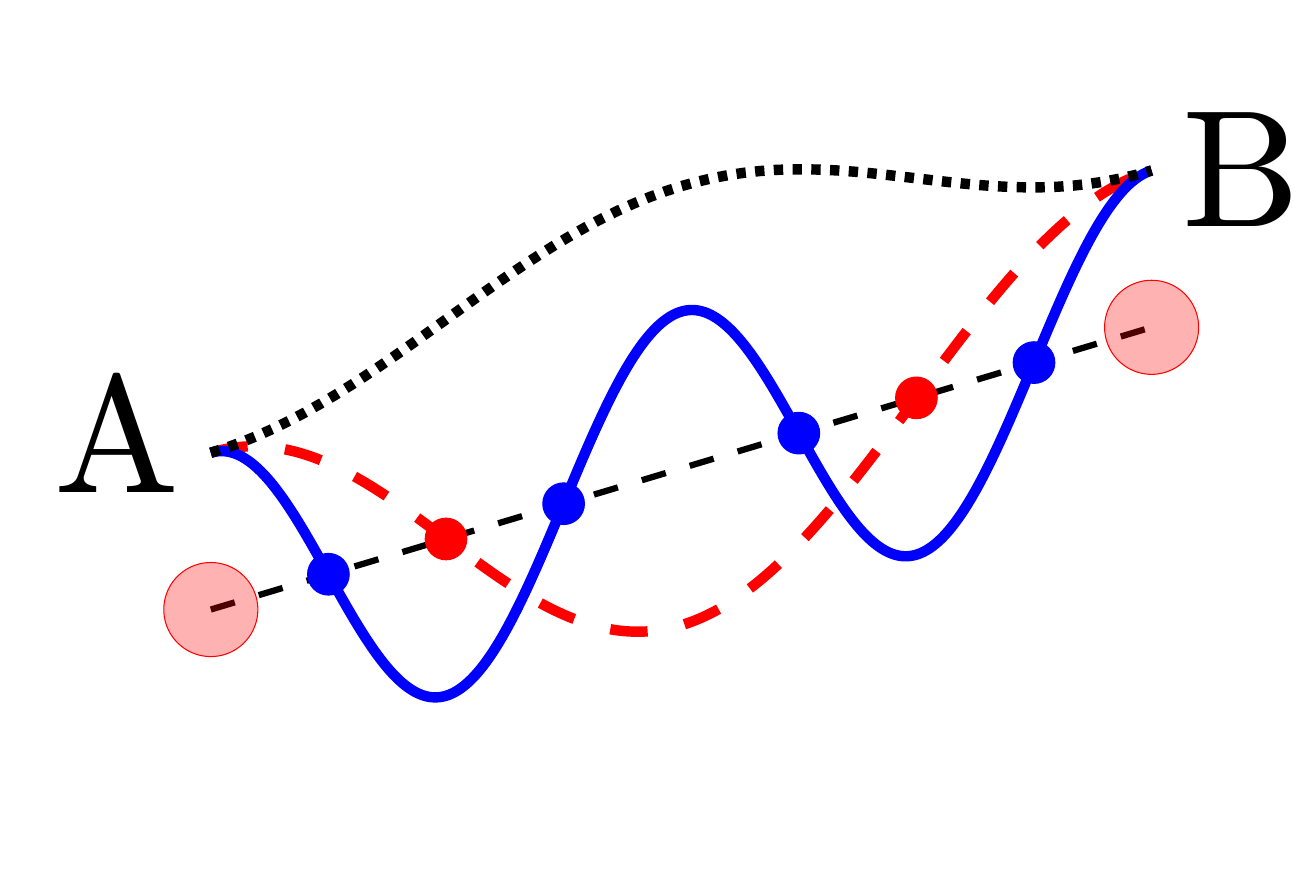}\put(40,0){\textcolor{black}{(a)}}
\end{overpic}\begin{overpic}[width=0.40\columnwidth]{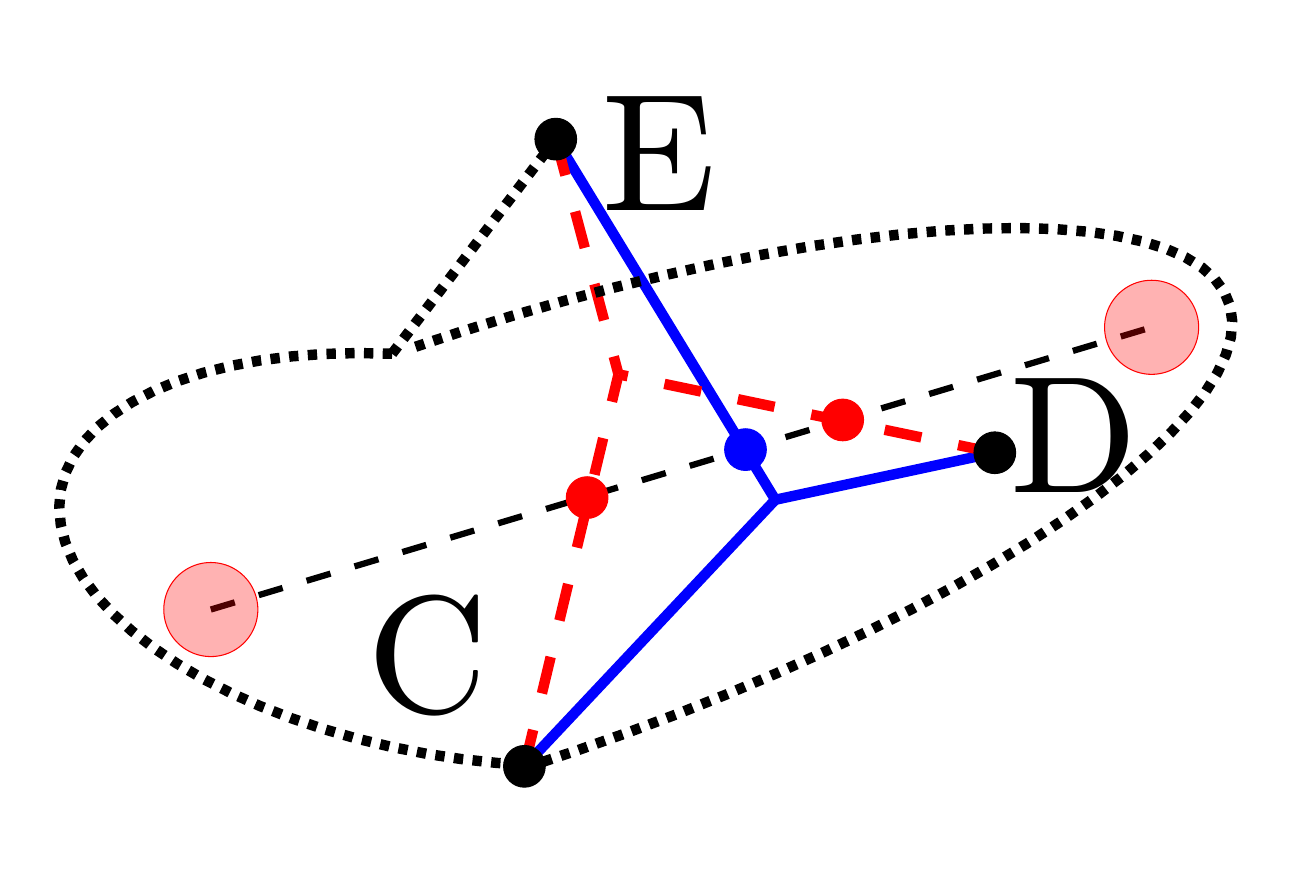}\put(40,0){\textcolor{black}{(b)}}
\end{overpic}
\caption{Different interpretation of number of crossings for path and trees: for paths (a) all occurrences with an even number of crossings correspond to the same homotopy class: zero (blue, solid), two (red, dashed) and four (black, dotted). For trees (b) in general this does not hold: in the case of a node of degree 3 we may have one (red, dashed), two (blue, solid), or none (black, dotted) of the three branches cross, and the latter will not be homotopically equivalent to the other two if there are obstacles (red circles) in the environment.
\label{fig:homotopyrule}
}%Images created with scripts DroneSwarmEasterIsland\19_BitangentSteinerJournal\HomotopyTrees\plotPathCrossings.m and plotTreeCrossings.m. Crop with Adobe Reader at [0.123, 1.682, 1.189, 0.14]in
\vspace{-5mm}
\end{figure} 
In the following paragraph we will develop a classification procedure for paths/trees (\emph{networks}, in the following) that approximates the classification by actual homotopy. This requires abstracting general characteristics of homotopy:

%It is immediate to show that such a choice would lead to paths that are in general longer than their generating paths (see Fig.~\ref{fig:loops}). In an optimization problem we can therefore rule out loops, as by definition they represent suboptimal paths.
% \begin{figure}[tb]
% \centering
% \begin{overpic}[width=0.40\columnwidth]{figures/JournalFigs/Crossings/NoLoopCrossings.pdf}\put(0,30){\textcolor{black}{a)}}
% \end{overpic}\begin{overpic}[width=0.40\columnwidth]{figures/JournalFigs/Crossings/LoopCrossings.pdf}\put(0,30){\textcolor{black}{b)}}
% \end{overpic}
% \caption{A path that coasts around an obstacle without looping around it (a) is not homotopically equivalent to a path that loops around it (b). From an optimization perspective, a path containing loops can always be made shorter by removing the loops.
% \label{fig:loops}
% }%Images created with scripts DroneSwarmEasterIsland\19_BitangentSteinerJournal\HomotopyTrees\plotNoLoops.m and plotLoops.m. Crop with Adobe Reader at [0.123, 1.682, 1.189, 0.14]in
% \end{figure} 
\begin{enumerate}
    \item A network may or may not traverse the segment $d_{ij}\equiv|O_i-O_j|$ that joins the $i$\textsuperscript{th} and $j$\textsuperscript{th} obstacles' centers. These segments are $\phi(\phi-1)/2$ %\linda{$\phi\choose 2$ pairs or $\phi(\phi-1)$ ordered pairs}
    , the number of non-ordered pairs between $\phi$ obstacles. We can define a $\phi(\phi-1)/2$-dimensional vector 
    \begin{equation}
        \vec h^{(1)}\equiv\{n_{12},...,n_{1\phi},...,n_{\phi-1\,\phi}\}
    \end{equation}where, by the considerations of Fig.~\ref{fig:homotopyrule}, $n_{ij}$ are either 1 if the segment is crossed by at least one branch, an odd number of time, and 0 otherwise.
    \item As per Sec.~\ref{sec:estHomotopy}, obstacles can be on either side of a network's branch. For each obstacle we define four orthogonal segments $d_{iK}\equiv\{d_{iN},d_{iE},d_{iS},d_{iW}\}$, obtained by joining the center of the $i$-th obstacle with its cardinal points $N$, $E$, $S$, $W$ placed ``at infinite\footnote{For our purposes, ``infinite'' is defined as outside of the convex hull of obstacles and the path/tree's endpoints.}''. For $\phi$ obstacles, there are $4\phi$ of such segments, that define a $4\phi$-dimensional vector 
    \begin{equation}
        \vec h^{(2)}\equiv\{NESW_1;...;NESW_\phi\}\,\,,
    \end{equation}where $N$, $E$, $S$, $W$ are 1 or 0 if the corresponding segment is crossed or not.
\end{enumerate}
Let us test the procedure against the homotopies of Fig.~\ref{fig:allhomotopies}, by computing $\vec h=\{n_{12};NESW_1;NESW_2\}$ for the eigth paths shown:
\begin{align}
    \vec h_{a1}&=\color{blue}{\{0;1101;1101\}}\,,\,
    &\vec h_{a3}&=\color{purple}{\{0;0111;0111\}}\,,\,\nonumber\\
    \vec h_{a2}&=\color{red}{\{1;1101;1101\}}\,,\,&\vec h_{a4}&=\color{black}{\{1;0111;0111\}}\,,\,\nonumber\\
        \vec h_{b1}&=\color{brown}{\{0;1111;1111\}}\,,\,
    &\vec h_{b2}&=\color{gray}{\{1;1111;1111\}}\,,\,\nonumber\\
    \vec h_{b3}&=\color{brown}{\{0;1111;1111\}}\,,\,
    &\vec h_{b4}&=\color{gray}{\{1;1111;1111\}}\,,\,\label{eq:exampleclass}\end{align}This example shows that the classification procedure performs well with the optimal homotopies of group a), but due to the booleanization of $\vec h$, it fails at recognizing those of group b). 
    
    This is not a huge limitation, as trajectories with loops are highly suboptimal and unlikely to be found by a solver in an optimization problem. Moreover, as we will show in Sec.~\ref{sec:withobs}, our algorithms automatically removes loops from the network, so the classification procedure can be confidently applied to the final results where no loops can be found.  

Similarly, we test the procedure on the trees of Fig.\ref{fig:ambiguity}, with $m=5$ terminals and $\phi=4$ obstacles:
\begin{align}
    \vec h \equiv&\{n_{12},n_{13},...,n_{34};NWSE_1;;...;NWSE_4\}\,\,,\nonumber\\
    \vec h_{a}=&\{\mathbf{111011};0011;1110;1110;1110\}\nonumber\\
        \vec h_{b}=&\{\mathbf{111011};0011;0111;0111;1001\}\nonumber\,\,,\\
        \vec h_{c}=&\{001011;\mathbf{0011;1011;1111;1011}\}\nonumber\\
        \vec h_{d}=&\{011110;\mathbf{0011;1011;1111;1011}\}\,\,,\label{eq:treeclass}
\end{align} Both $\vec h^{(1)}$ and $\vec h^{(2)}$ are necessary for the classification: indeed, the trees of Fig.\ref{fig:ambiguity}a-b have same $\vec h^{(1)}$, while those of Fig.\ref{fig:ambiguity}c-d share the same $\vec h^{(2)}$.
\begin{figure}[tb]
\centering
\begin{overpic}[width=0.24\columnwidth]{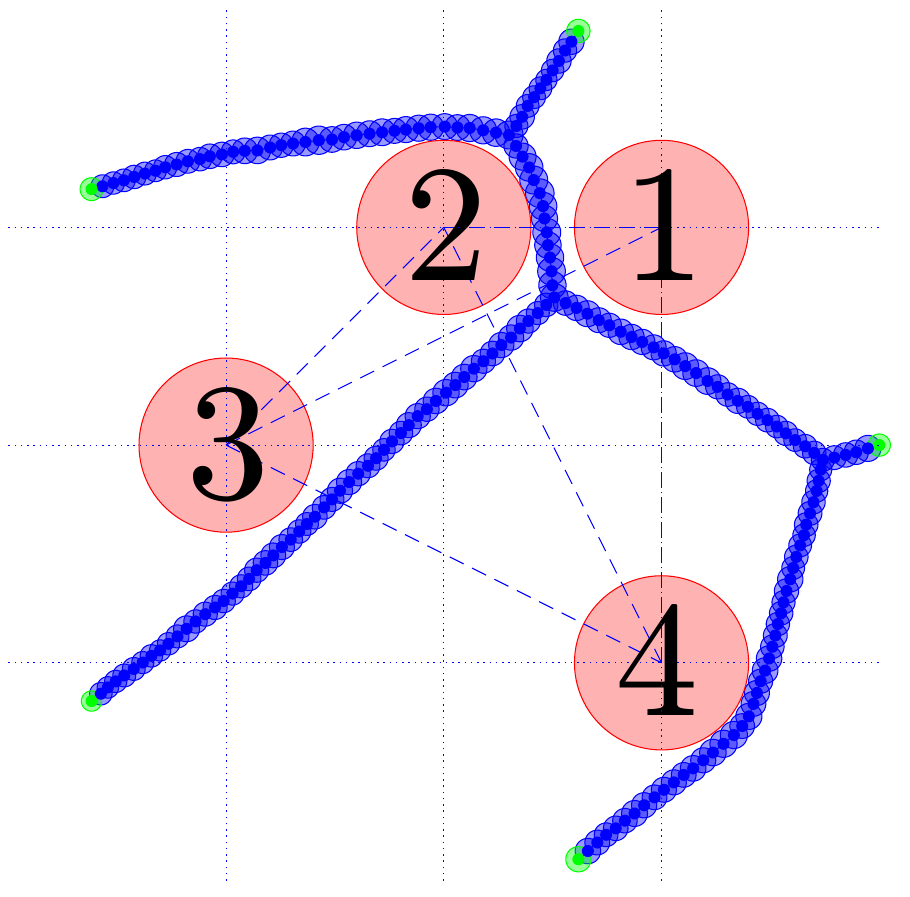}\put(40,-10){\textcolor{black}{a)}}
\end{overpic}\begin{overpic}[width=0.24\columnwidth]{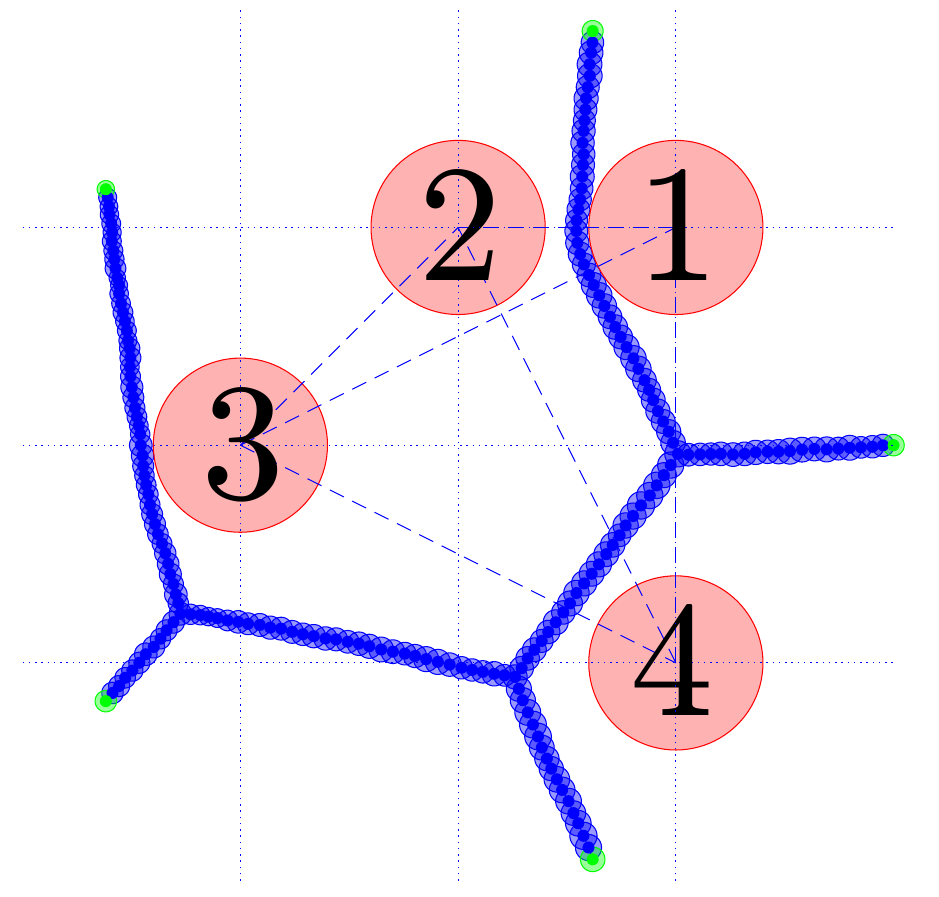}\put(40,-10){\textcolor{black}{b)}}
\end{overpic}\begin{overpic}[width=0.24\columnwidth]{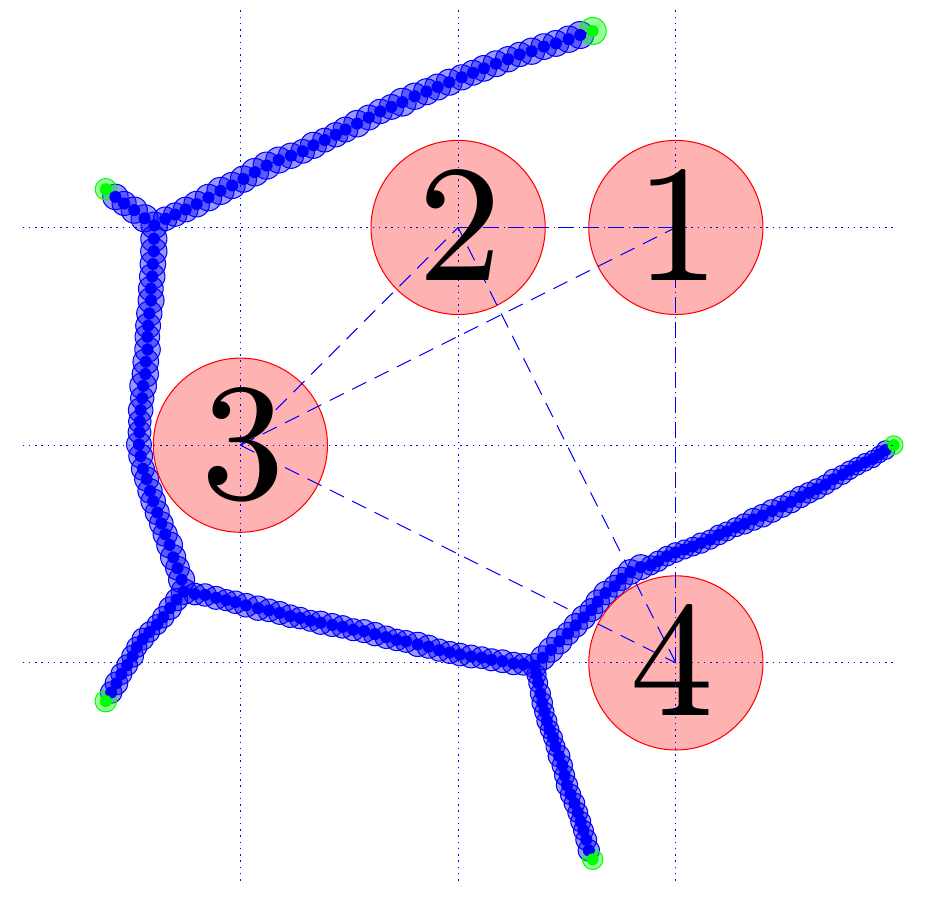}\put(40,-10){\textcolor{black}{c)}}
\end{overpic}\begin{overpic}[width=0.24\columnwidth]{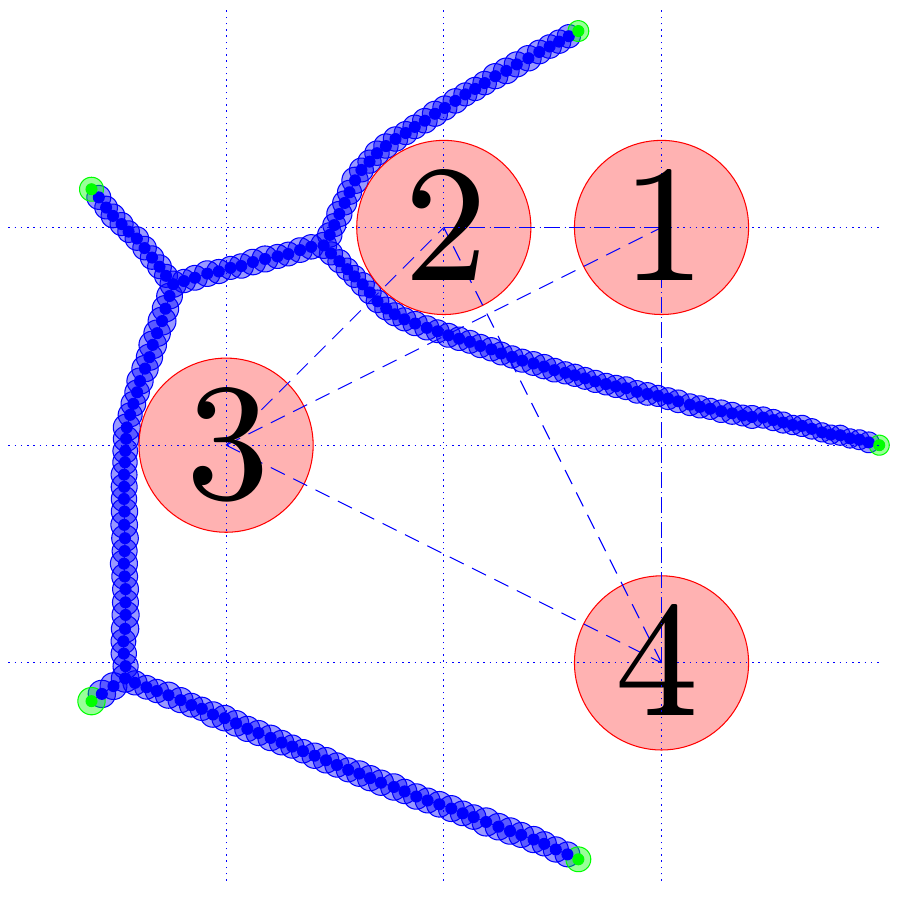}\put(40,-10){\textcolor{black}{d)}}
\end{overpic}
\caption{Classification procedure applied to trees with $m=5$ terminals and $\phi=4$ obstacles. The results of \eqref{eq:treeclass} show that both $\vec h^{(1)}$ and $\vec h^{(2)}$ are necessary for the classification.
\label{fig:ambiguity}}%Figures created using the script DroneSwarmEasterIsland\19_BitangentSteinerJournal\HomotopyTrees\MAIN_HomotopyBuilderTerminalIsolation.m and selecting solutions [14,30,5,6] from the AllGen results set. Crop using Adobe Pro with dimensions [0.749, 1.469, 1.976, 1.557]inches
\vspace{-5mm}
\end{figure}

\subsubsection{Redundancy of $\vec h$}
The combination of $\vec h^{(1)}$ and $\vec h^{(2)}$ is a $\phi(\phi+7)/2$-dimensional boolean vector, which allows for $2^{\phi(\phi+7)/2}$ possible different classes. 

However, not all the components of $\vec h$ are independent:
\begin{itemize}
    \item Crossing of a $d_{ij}$ segment may imply the crossing of one or more other $d_{km}$ segments. %For example, if three obstacles are lined up, crossing $d_{12}$ implies crossing $d_{13}$.
    \item Unless an obstacle is far away (i.e. it can be ignored), the network will cross at least one of its cardinal directions will be crossed. This rules out $\{NESW\}=\{0000\}$.
    \item The crossing of a $d_{ij}$ segment may imply the crossing of one or more cardinal segments of the $i$\textsuperscript{th} and $j$\textsuperscript{th} obstacles.
    \item We consider continuous and connected networks, so if the pairs $N/S$ or $W/E$ of an obstacle are crossed, then at least another cardinal direction is crossed. This rules out $\{NESW\}=\{1010\}$ and $\{0101\}$. 
\end{itemize}These facts imply that only a subset of the $\phi(\phi+7)/2$-dimensional space corresponds to actual classes. 

Even though the exact number depends on the specific distribution of obstacles and terminals considered, we will assume the method to be able to classify at least $N_{\textrm{hom}}^{\textrm{max}} = 2^\phi$ homotopies, including the optimal ones of the type of Fig.~\ref{fig:allhomotopies}a. For trees, $N_{\textrm{hom}}^{\textrm{max}}=f^\phi\leq 2^{\phi(\phi+7)/2}$ so we expect the procedure to break down when the number of faces (at most, the number of terminals), is comparable to $2^{\phi/2}$, so that $f^\phi\sim 2^{\phi^2/2}$.

%#################################################################
\subsection[Solving for m terminals, n relays, no obstacles]{Solving for $m$ terminals, $n$ relays, no obstacles\label{sec:mtermnoobs}}
%#################################################################
When using more terminals the primary goal remains to connect them through a network at a minimum cost. As in the 1D case the addition of more relays $n$ affects the nature of the problem: the case $n=0$ and $n>0$ generalize the MST and the Steiner tree problem but with the non-Euclidean cost function of Eq.~(\ref{eqn:TransmissionCostForNetwork}). This also explains why the solutions need not share the properties of the 1D counterparts (e.g. the \SI{120}{\degree} rule as shown in Sec.~\ref{sec:1ter1rel}). 

 With $m$ terminals, the network is a collection of branches $\{\mathcal B\}$ interconnecting nodes: the iterative algorithm described in Alg.~\ref{alg:MSTPlusLeaf} attempts to optimize the cost of this network.
\begin{figure}[tb]
\vspace*{-4mm}
\begin{algorithm}[H]
\small
\caption{\small\sc MST-Based Network Optimization}\label{alg:MSTPlusLeaf}
\begin{algorithmic}[1] 
%\label{alg:MSTPlusLeaf}
\While{{\sc True}}
\State Compute MST on set of terminals and relays.
\State Call Advanced Network Optimization\algorithmiccomment{Alg.~\ref{alg:advOpt}}
\For{relay  \textbf{in} relays}
\If{relay in a leaf branch of MST} 
\State Steer relay to parent node with degree > 2.

\Else
\State Compute movement to average position of neighbors
\If{obstacles}
\State Call No Overlap Constraint\algorithmiccomment{Alg.~\ref{alg:nooverlap}}
\EndIf
\State Perform movement
\EndIf 
\EndFor
\If{Cost stable \textbf{or} reached max steps}
\State Return network \& exit
\EndIf
\EndWhile
\end{algorithmic}
\end{algorithm}
\vspace*{-6mm}
\end{figure}
%
% \begin{enumerate}
%     \item Considers the combined set of the terminals and the nodes, and computes an MST over it;
%     \item Looks for ``leaf'' branches of the MST, i.e. chains of nodes which do not connect more than one terminal, and steers them towards that terminal;
%     \item Steers each non-terminal node of the MST towards the average positions of its nearest neighbors.
% \end{enumerate}

This procedure involves local optimization of the cost and therefore is prone to produce local minima: the local rules of the algorithm produce branches with locally uniform densities, but as they do not allow relays to move from a denser branch to a less dense one densities may not be uniform globally.

To mitigate this issue assume the system is in a local minimum and that the radii of the nodes involved in each branch are uniform and equal to $R = L/(N+1)$, where $L$ is the length of the branch and $N$ the number of its nodes \emph{excluding the endpoints}. Let $\mathcal B_{\textrm{sml}}$ and $\mathcal B_{\textrm{lrg}}$ be two neighboring branches (i.e. with one endpoint in common) with higher and lower density of relays and thus \textbf{smaller} and \textbf{larger} radii. 

Then we move one relay from the branch with the smaller $R$ ($R_{\textrm{sml}}$) to the branch with the larger $R$ ($R_{\textrm{lrg}}$), halfway between the branches' common endpoint and the first relay.
% \begin{enumerate}
%     \item[1)] Consider each node with degree 3 and each terminal with degree 2; being in a local minimum implies leaf branches have been removed so if a terminal has degree 2, it must be connected with another;
%     % a terminal could have 3 neighbors 
%     \item[2)] Sort the radii to the neighbors and determine $R_{\textrm{min}}$ and $R_{\textrm{avg}}$, along with $N_{\textrm{min}}$, $N_{\textrm{max}}$, $L_{\textrm{min}}$, $L_{\textrm{max}}$ the populations and the lengths of the respective branches.
%     \item[3)] Remove the link from the node to the nearest neighbor with $R_{\textrm{min}}$, and replace it with a link to the next-to-nearest neighbor;
%     \item[4)] Place the previously-removed node halfway between the node and the nearest neighbor with $R_{\textrm{avg}}$, removing the link between the two and replacing it with two new links:
%     \begin{itemize}
%         \item between the displaced neighbor and the node;
%         \item between the displaced neighbor and the nearest neighbor with $R_{\textrm{avg}}$ (which then becomes a next-to-nearest neighbor for the node);
%     \end{itemize} 
% \end{enumerate}
\begin{figure}[tb]
\vspace{-4mm}
\begin{algorithm}[H]
\small
\caption{\small\sc Advanced Network Optimization}\label{alg:advOpt}
\begin{algorithmic}[1] 
%\label{alg:MSTPlusLeaf}
\For{node  \textbf{in} network}
\State Find degree of node.
\If {node == terminal \textbf{and} degree == 2}
    \If {star improves cost}
        \State Create star.
    \Else
        \State Equilibrate radii.
    \EndIf
\ElsIf {degree == 3}
\State Equilibrate radii. 
\EndIf 
\EndFor
\State Return updated node positions
\end{algorithmic}
\end{algorithm}
\vspace{-7mm}
\end{figure}

\begin{figure}[tb]
\vspace{-4mm}
\begin{algorithm}[H]
\small
\caption{\small\sc ``No-overlap'' constraint}\label{alg:nooverlap}
\begin{algorithmic}[1] 
%\label{alg:MSTPlusLeaf}
\If {node \textbf{overlaps} obstacles}
   \State Ignore movement from Alg.~\ref{alg:MSTPlusLeaf}
   \State Compute radial movement away from obstacles
   \If{radial movement increases overlaps}
        \If{No moves for $5$ consecutive steps in the past}
            \State Perform radial movement
        \Else
            \State Perform no movement
        \EndIf
    \EndIf
\Else
    \If{Movement from Alg.~\ref{alg:MSTPlusLeaf} causes overlaps}
        \State Cancel movement from Alg.~\ref{alg:MSTPlusLeaf}
        \Else
        \State Confirm movement from Alg.~\ref{alg:MSTPlusLeaf}
    \EndIf
\EndIf 
\State Return updated movement
\end{algorithmic}
\end{algorithm}
\vspace{-8mm}
\end{figure}
On the former denser branch the common endpoint is now  distant $R'=2R_{\textrm{sml}}$  from its nearest neighbor while on the former less dense side it is distant $R'=R_{\textrm{lrg}}/2$, both different from their branch average. After the algorithm reaches a new equilibrium  the branches will have lengths $L'$ and one more or one less relay:
\begin{equation}
 R_{\textrm{sml}}'=\frac{L_{\textrm{sml}}'}{(N_{\textrm{Rsml}}+1)-1},\,\,R_{\textrm{lrg}}^{'}=\frac{L_{\textrm{lrg}}'}{(N_{\textrm{Rlrg}}+1)+1}\;.\label{eq:avgradii}
 \end{equation}
     %For brevity we relabel $N+1\rightarrow N$ in \eqref{eq:avgradii}. 
     The cost of each branch is given by $C=(N+1)R^2=L^2/(N+1)$ therefore the cost change $ \Delta C= C'-C $ is:
\begin{equation}
   \Delta C=\frac{L_{\textrm{sml}}^{'2}}{N_{\textrm{Rsml}}}+\frac{L_{\textrm{lrg}}^{'2}}{N_{\textrm{lrg}}+2}-\frac{L_{\textrm{sml}}^2}{N_{\textrm{Rsml}}+1}-\frac{L_{\textrm{lrg}}^2}{N_{\textrm{Rlrg}}+1}\;,
\end{equation}Assuming $L_i\sim L_i'$ (reasonable for large $N_i$) and $L_1=L_2$ then we have
\begin{align}
    \Delta C &\equiv\alpha (1+N_{\textrm{Rlrg}}-N_{\textrm{Rsml}})\leq0\;,
\end{align}where $\alpha$ is a positive quantity and the inequality holds for $N_{\textrm{Rsml}}\geq N_{\textrm{Rlrg}}+1$. This occurs when the difference between $R_{\textrm{sml}}$ and $R_{\textrm{lrg}}$ is such that $\left\lceil L_{\textrm{lrg}}/R_{\textrm{lrg}}\right\rceil<\left\lceil L_{\textrm{sml}}/R_{\textrm{sml}}\right\rceil$. Therefore locally equilibrating radii may result in a cost improvement as long as the density imbalance between two branches is significant and they have comparable length.

Fig.~\ref{fig:MSTnoobs} presents the outcome of the optimization under the combined action of Algs.~\ref{alg:MSTPlusLeaf}--\ref{alg:advOpt} at various steps. In particular Fig.~\ref{fig:MSTnoobs}d represents the optimal configuration and resembles a full Steiner topology with $m=5$ terminals and $n=3$ Steiner points. 
% \todo{make a figure more with homogeneous radii across branches}
Due to the high likelihood of local minima, employing only Alg.~\ref{alg:MSTPlusLeaf} does not guarantee the result of Fig.~\ref{fig:MSTnoobs}d. Several simulations were performed employing only Alg.~\ref{alg:MSTPlusLeaf} with $m=5$ terminals and $N=100$ different initial placements of $n=40$ relays: although 66\% of the cases had the right topology, none of them was optimal. The remaining 44\% presented one to three terminals of degree 2, not optimal for this instance. 

\begin{figure}[tb]
\centering
\begin{overpic}[width=.22\columnwidth]{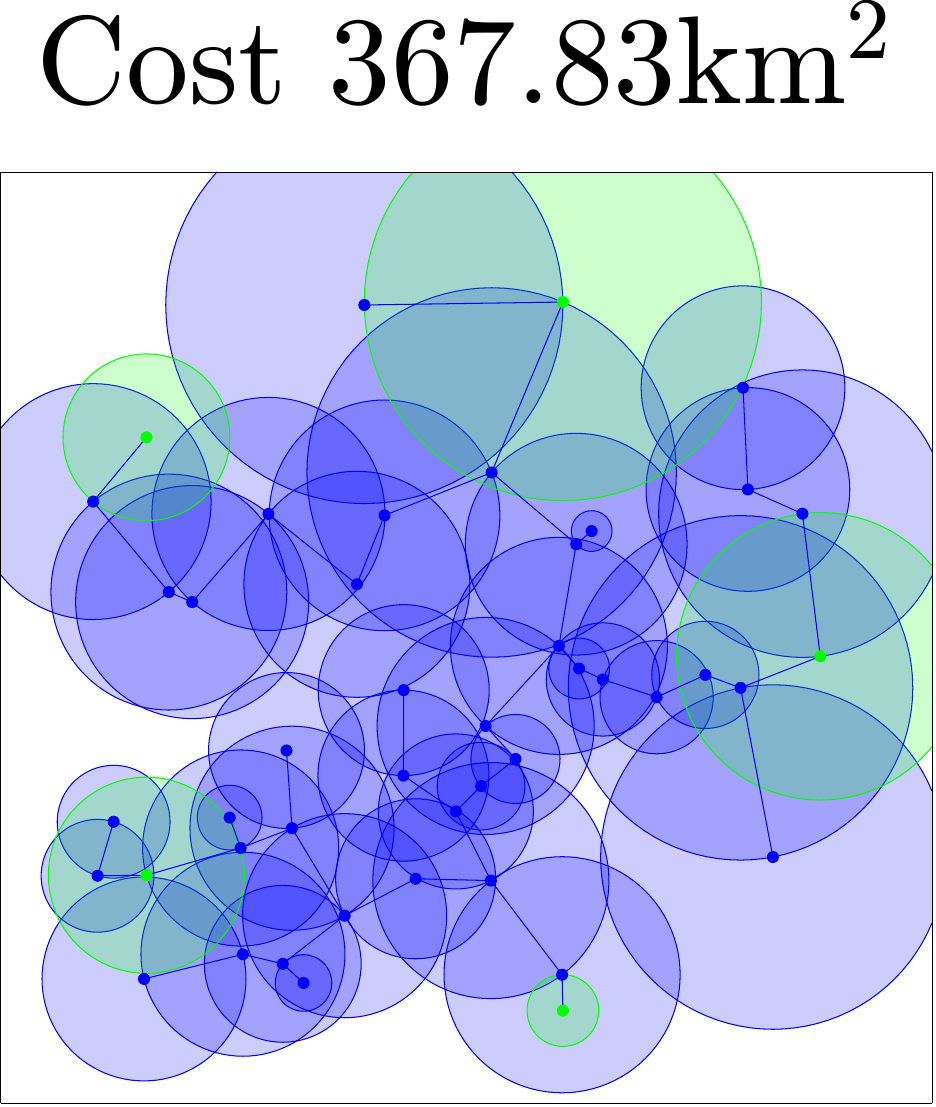}\put(2,72){a)}
\end{overpic}~\begin{overpic}[width=.22\columnwidth]{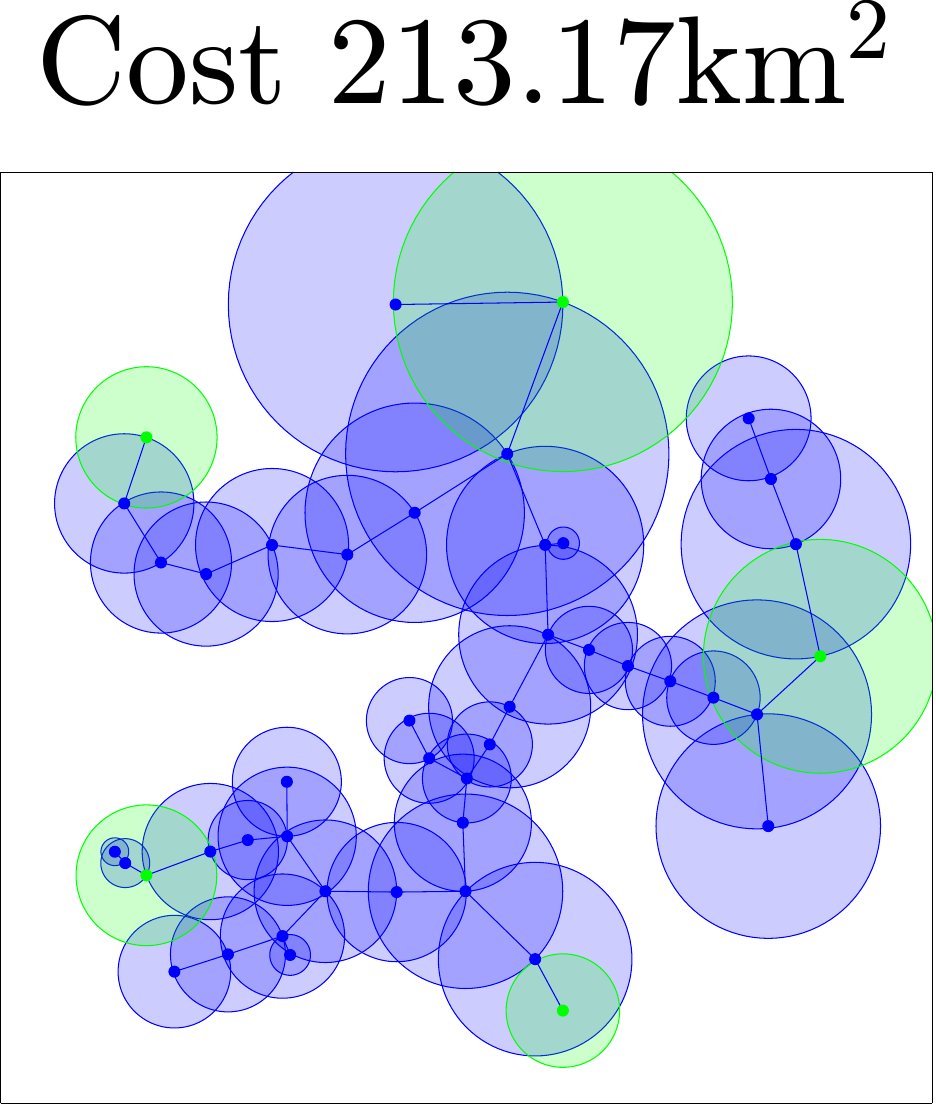}\put(2,72){b)}
\end{overpic}
\begin{overpic}[width=.22\columnwidth]{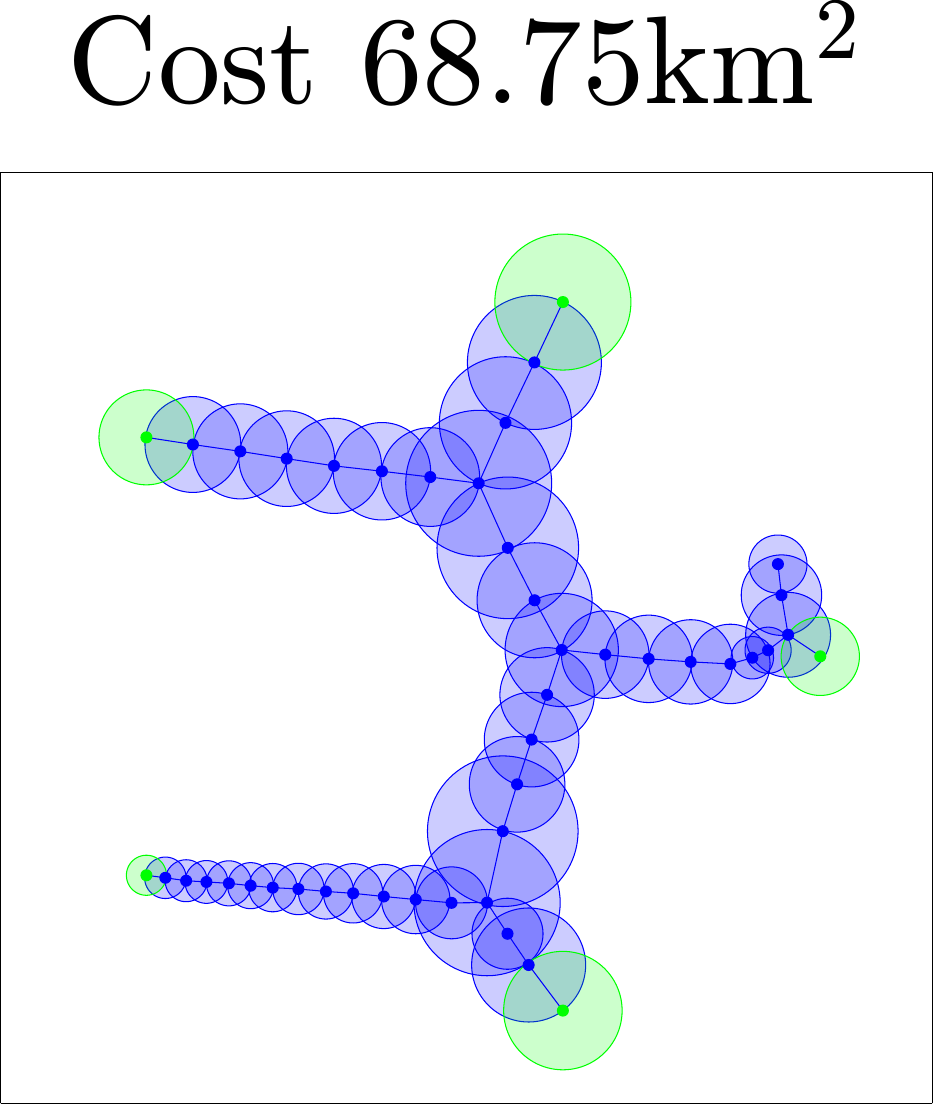}\put(2,72){c)}
\end{overpic}~\begin{overpic}[width=.22\columnwidth]{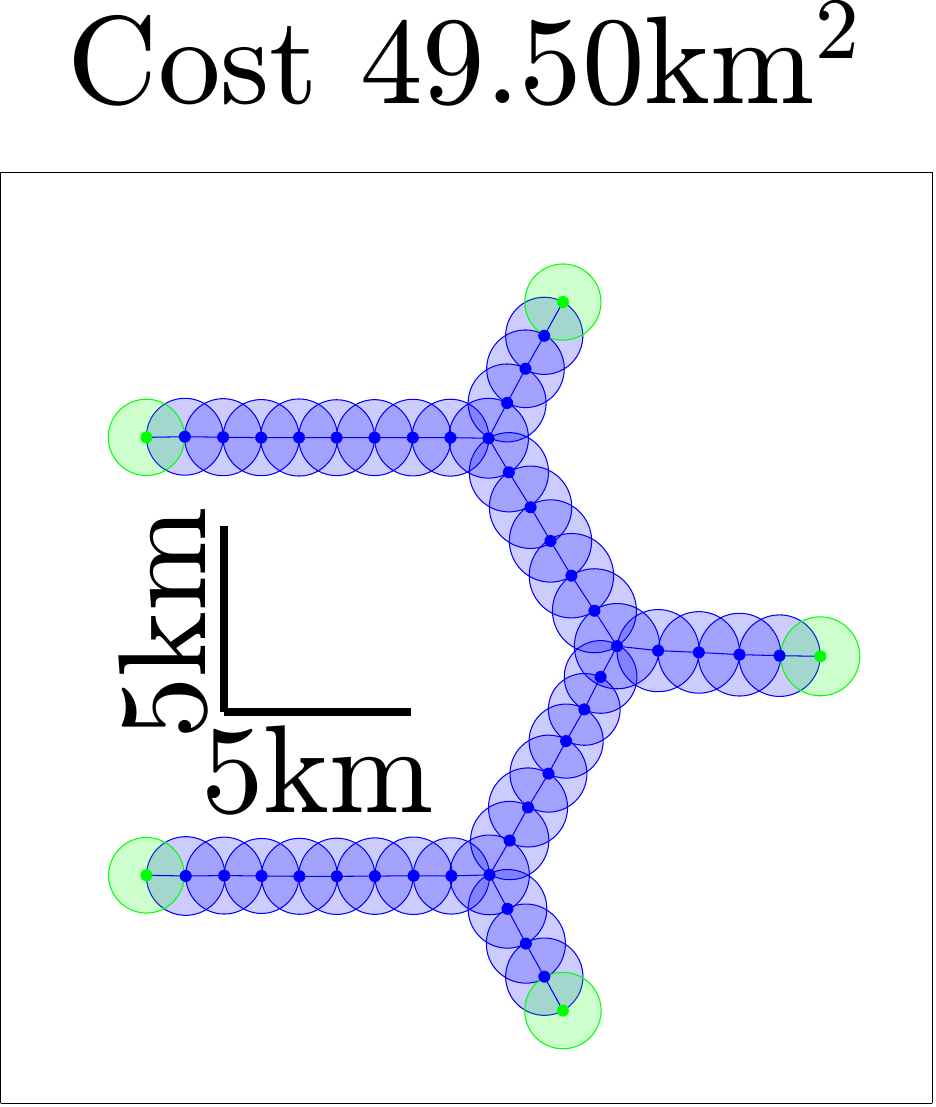}\put(2,72){d)}
\end{overpic}

%Pics generated using code found at \DroneSwarmEasterIsland\17_BitangentsSteiner\LauncherNoObs.m
\caption{%Iterating between calculating the MST and moving mobile relays to the average position of their neighbors approximates a Steiner tree.
From an initial random placement of $n=40$ relays (a), Algs.~\ref{alg:MSTPlusLeaf}-\ref{alg:advOpt} develop a network interconnecting $m=5$ terminals. Alg.~\ref{alg:advOpt} equilibrates radii across branches (b-c) . The final structure shows pseudo-Steiner points (d) with angles $\sim$\SI{120}{\degree}.} 
\label{fig:MSTnoobs}
\end{figure} 

A second round of simulations, that also utilized Alg.~\ref{alg:advOpt}, was undertaken and resulted and it resulted in 100\% of optimal results (up to a rotation of \SI{72}{\degree},  the terminals are the vertices of a regular pentagon).
%, thus the problem presents symmetry under the dihedral group $D_5$. 
Figure~\ref{fig:MSTnoobs}d shows one of the possible solutions while Fig.~\ref{fig:CostTrends}a presents trends and cost
distributions.

\begin{figure}[t]
\centering
\begin{overpic}[width=0.49\columnwidth]{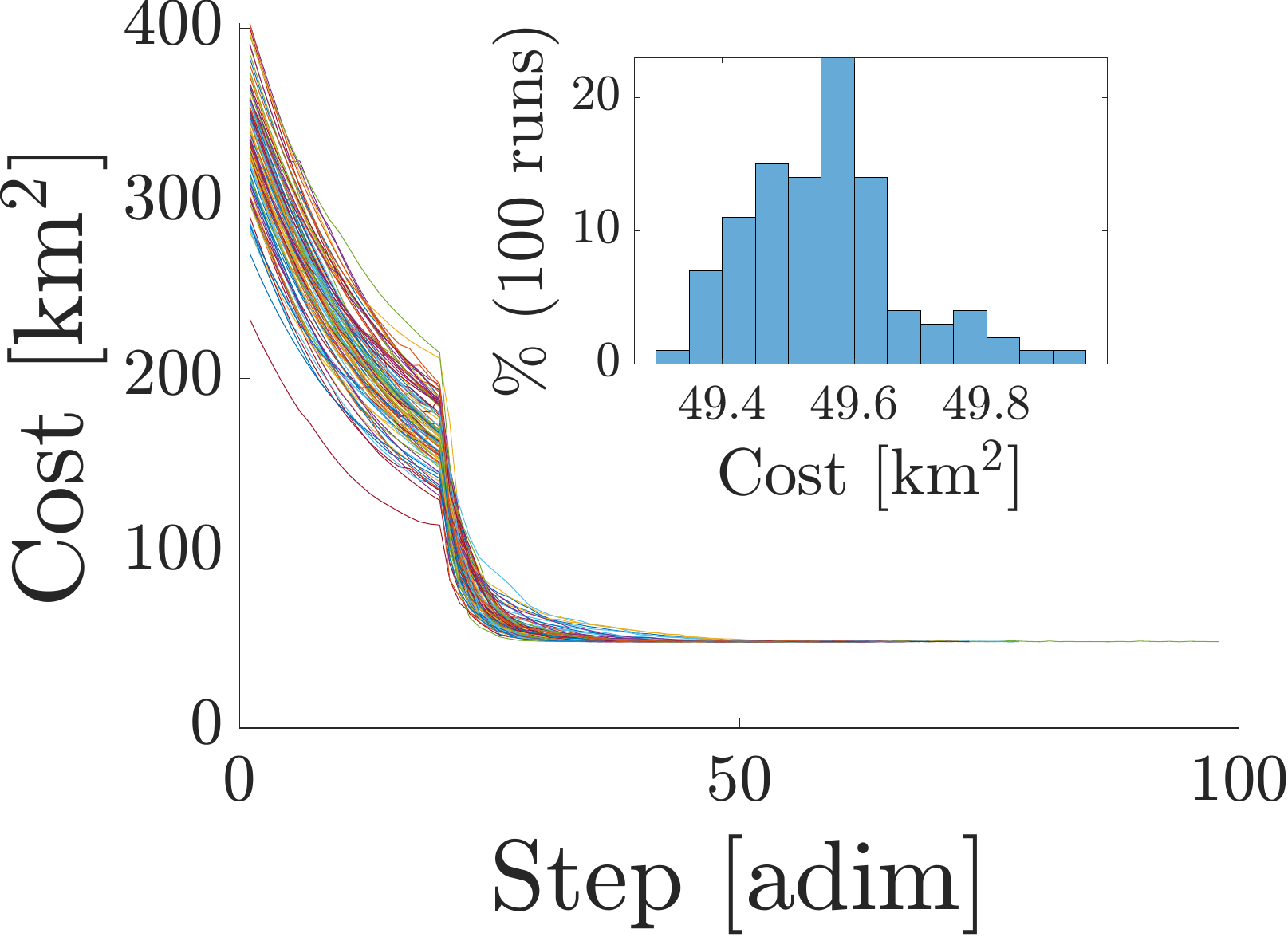}\put(90,70){a)}\end{overpic}\begin{overpic}[width=0.49\columnwidth]{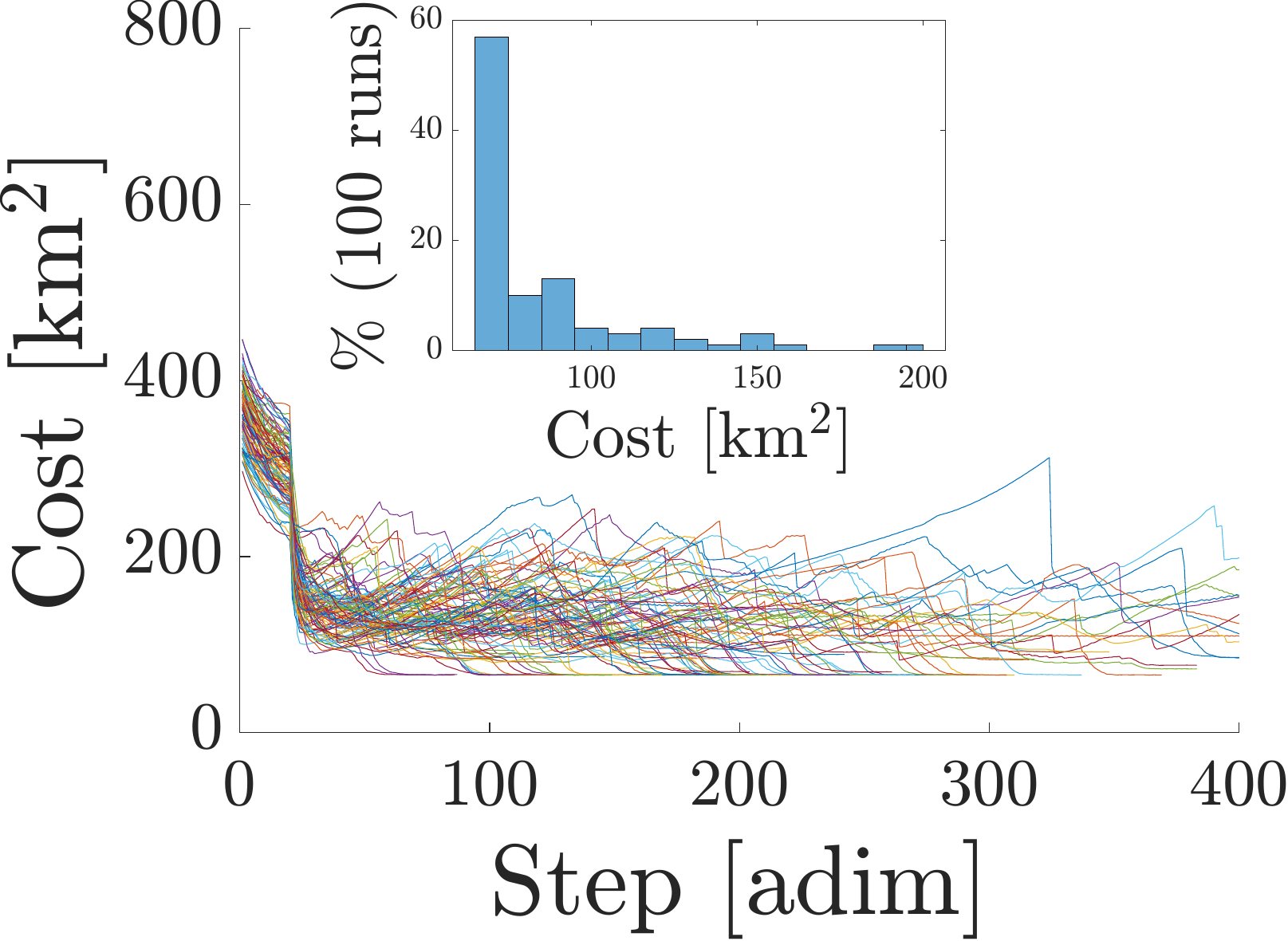}\put(90,70){b)}\end{overpic}
% Plot is generated by 2023Freedom/blob/main/DroneSwarmEasterIsland/7_MultiNodes/100runs/plots.m

\caption{Cost trend of the $100$ simulations of  Fig.~\ref{fig:MSTnoobs} a with no obstacles and Fig.~\ref{fig:MSTobs40D}b with obstacles using  Algs.~\ref{alg:MSTPlusLeaf}--\ref{alg:advOpt}. In the case with no obstacles all simulations converged to the optimal configuration of Fig.~\ref{fig:MSTnoobs}d, up to a rotation of \SI{72}{\degree} (the terminals are the vertices of a regular pentagon). In the case with obstacles 53\% of the instances obtained the optimal solution of Fig.~\ref{fig:MSTobs40D}a, whereas the rest either converged to a sub-optimal configuration (21\%) similar to those of Fig.~\ref{fig:MSTobs40D}(b--d), or did not converge at all (26\%).
\label{fig:CostTrends}}
\vspace{-5mm}
\end{figure}  

%#################################################################
\subsection[Solving for m terminals, n relays, and phi obstacles]{Solving for $m$ terminals, $n$ relays, and $\phi$  obstacles\label{sec:withobs}}
%#################################################################
By combining the heuristics of Algs.~\ref{alg:MSTPlusLeaf}--\ref{alg:advOpt} with the no-overlap constraint Eq.~(\ref{eqn:TransmissionCostForNetwork}) detailed in Alg.~\ref{alg:nooverlap}, the relays can be steered away from a situation where they overlap obstacles. 

We refer to \cite{Bernardini2024CASE} for a detailed discussion of the results of the simulations in the case of obstacles. Here we only recall the main outcome of the discussion: the nature of the local minima generated by the obstacles makes Algs.~\ref{alg:MSTPlusLeaf}--\ref{alg:nooverlap} inefficient if paired with random initial conditions. This can be seen comparing the cost distributions of Fig.~\ref{fig:CostTrends}a-b, where the costs relative to the case with no obstacle are narrowly peaked around the optimal value, while in the case with obstacles the spread is much more significant. Some of the solutions obtained with the method are shown in Fig.~\ref{fig:MSTobs40D}. 
\begin{figure}[h]
\centering
\begin{overpic}[width=.24\columnwidth]{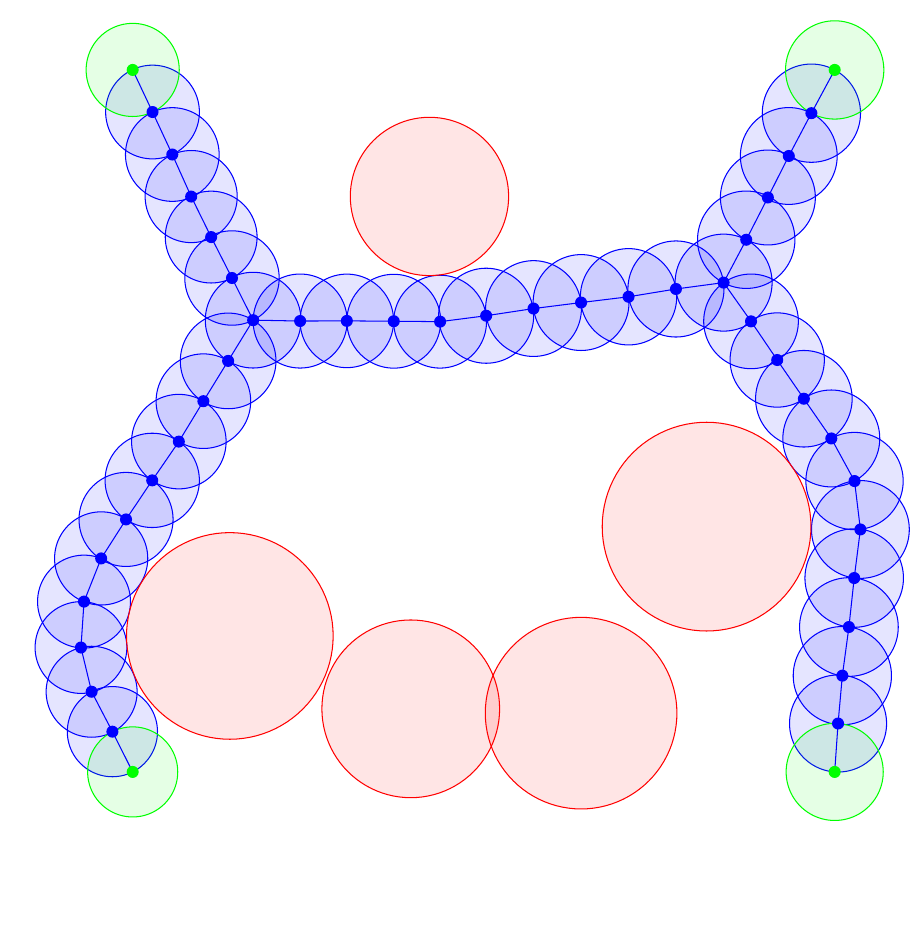} \put(45,40){a)}
\put(15,105){\scriptsize Cost 65.31km$^2$}
\end{overpic}\begin{overpic}[width=.24\columnwidth]{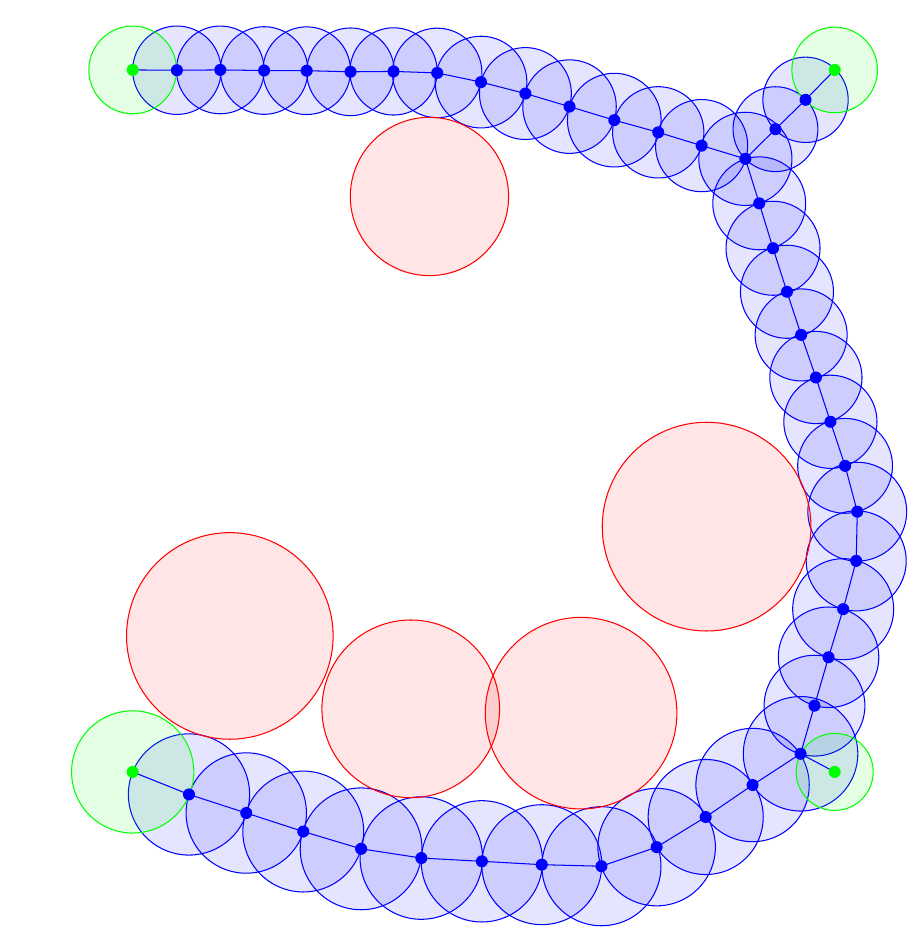} \put(45,40){b)}\put(10,105){\scriptsize Cost \SI{73.20}{\kilo\meter\squared}}
\end{overpic}\begin{overpic}[width=.24\columnwidth]{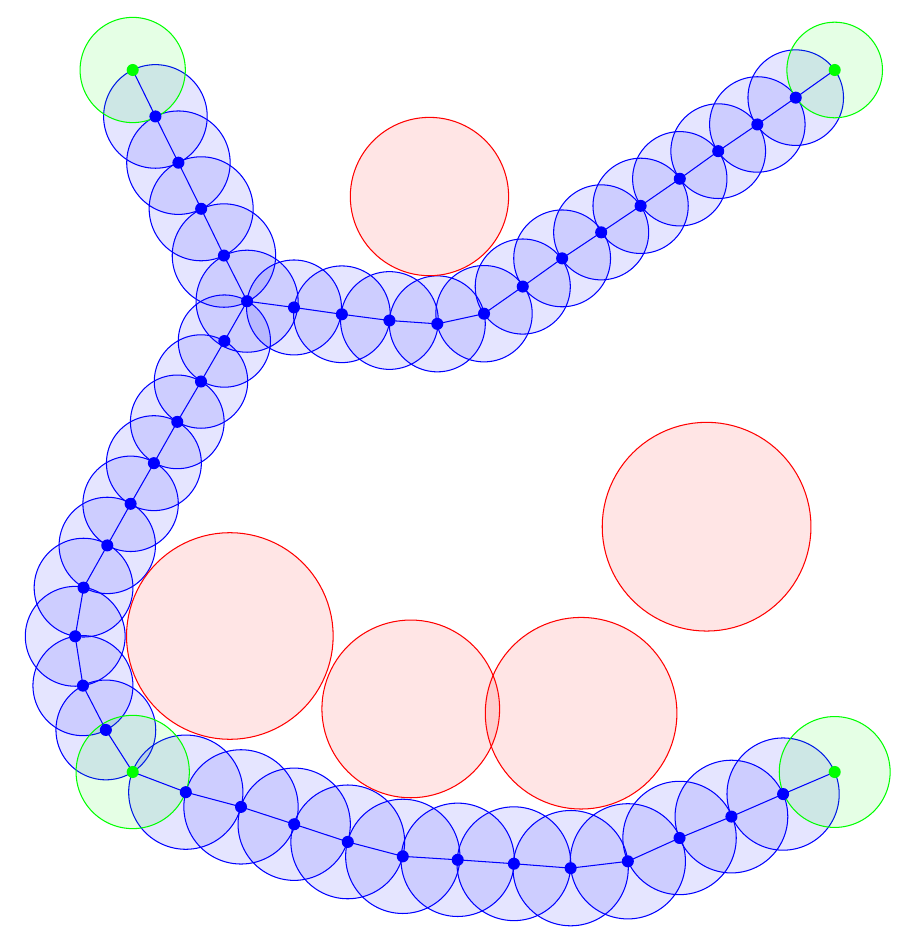} \put(45,40){c)}\put(15,105){\scriptsize Cost 76.61km$^2$}
\end{overpic}\begin{overpic}[width=.24\columnwidth]{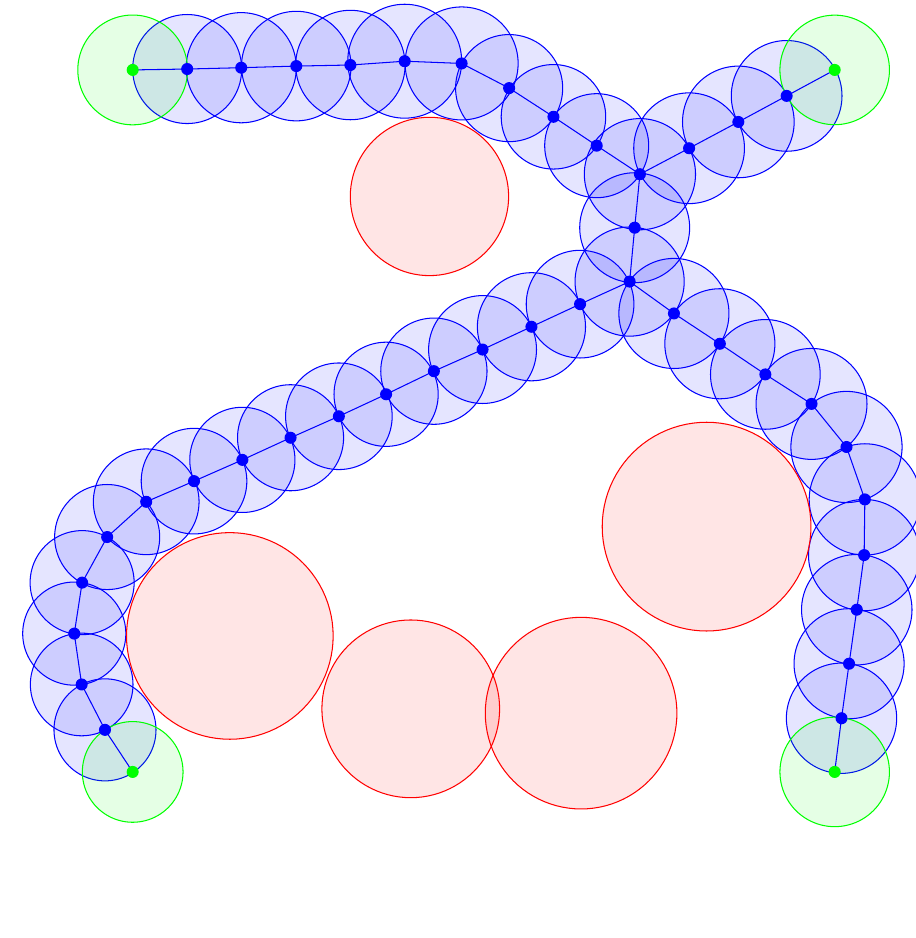} \put(45,40){d)}\put(15,105){\scriptsize Cost 84.69km$^2$}
\end{overpic}
%These pics are made with code found at DroneSwarmEasterIsland\17_BitangentsSteiner\100runsObs\plotsObs.m
\caption{Converged solutions of four topologies for the same problem set with $m=4$ terminals, $\phi = 5$ obstacles, and $n=40$ nodes. Although the solution with the least cost (a) resembles a full Steiner topology, this is not a sufficient condition, as shown in (d).
\label{fig:MSTobs40D}}
\end{figure}

Algorithms \ref{alg:MSTPlusLeaf}--\ref{alg:nooverlap}  are capable of dissolving loops, regardless of the complexity of the initial condition provided. An example with $\phi = 4$ and $m=5$ is provided in Fig.~\ref{fig:2Tex}: even though the initial placement of the relays (Fig.~\ref{fig:2Tex}a) contains loops, the network evolves towards a simpler structure (Fig.~\ref{fig:2Tex}b-d). 

%%[COMMENTED OUT FOR THE JOURNAL VERSION]
%$100$ simulations with random initial conditions were performed where the optimal topology was found in only 53\% of the cases, while in 21\% of the cases the simulation yielded sub-optimal configurations and in 26\% of the cases convergence failed.
% Figure~\ref{fig:CostTrends}b presents the trends and distributions of costs for the case with obstacles. Figure~\ref{fig:MSTobs40D}a shows the optimal configuration while Fig.~\ref{fig:MSTobs40D}(b--d) show some of the sub-optimal configurations obtained by the algorithm. 

% This degradation in performances with respect to the case with no obstacles is imputable to a limitation of the greedy heuristics of Algs.~\ref{alg:MSTPlusLeaf}--\ref{alg:advOpt}: indeed, it can only overcome local minima related to \emph{local} features such as density imbalances across adjacent branches. Conversely obstacles introduce local minima of \emph{topological} nature which, due also to the no-overlap constraint of Alg.~\ref{alg:nooverlap}, are only overcome by significant changes in the network structure, incompatible with the greedy heuristics' dynamics.

% Moreover Algs.~\ref{alg:MSTPlusLeaf}--\ref{alg:advOpt} tend to form uniform branches, thus cannot handle narrow passages like those obtained by \texttt{fmincon} in Fig.~\ref{fig:homotopy}b. 
% Fortunately, generating a connection through a narrow passage is expensive, and more frugal solutions usually exist that take different routes.
%%[END COMMENTED OUT FOR THE JOURNAL VERSION]

\section{Pre-scan algorithm}\label{sec:PreScan}
The previous section started with random initial conditions, and may miss the optimal solution. 
This section presents a deterministic method to generate network homotopies, and a heuristics to discard homotopies with a low chance of success.   
%In  \cite{Bernardini2024CASE} the authors suggested a deterministic approach to the problem, focusing on the specific obstacle distribution, pre-searching the possible homotopies, and using them as initial conditions for the relays in place of the random ones. 

The algorithm is described in Alg.~\ref{alg:pre-scan}: let $S_0$ be the optimal solution of the network in absence of obstacles. If obstacles are present, it is necessary to check if $S_0$ is compatible with their distribution. As in general this will not be the case, the resto of the algorithm can be subdivided in three actions: homotopy generation, convergence likelihood estimation, and evolution of the qualified initial conditions. The evolution was presented in Secs.~\ref{sec:mtermnoobs}--\ref{sec:withobs}. Homotopy generation is described in Sec.~\ref{sec:homgen} and convergence likelihood estimation in Sec.~\ref{sec:convlike}.

\begin{figure}[tb]
\vspace{-4mm}
\begin{algorithm}[H]
\small
\caption{\small\sc Pre-scan algorithm}\label{alg:pre-scan}
\begin{algorithmic}[1] 
\State Compute solution $S_0$ assuming no obstacles
\If{$S_0$ is compatible with obstacles}
    \State $S_0$ is the solution
    \Else
    \State Solve STPG and obtain $S$ 
    \State $\mathbf{S}_{\textrm{pre}}\leftarrow$ \textsc{HomotopyGeneration}($\{S\}$)\algorithmiccomment{Alg.~\ref{alg:homgen}} 
    \For{$S$ in $\mathbf{S}_{\textrm{pre}}$} 
        \If{$CL(S)>10\%$}\algorithmiccomment{See Sec.~\ref{sec:convlike}, \eqref{eq:CL}}
            \State Initialize Algs.~\ref{alg:MSTPlusLeaf}--\ref{alg:nooverlap}
        \Else
            \State Discard $S$
        \EndIf
        
    \EndFor
\EndIf
\end{algorithmic}
\end{algorithm}
\vspace{-6mm}
\end{figure}

\begin{figure}[h]
\centering
\begin{overpic}[width=0.24\columnwidth]{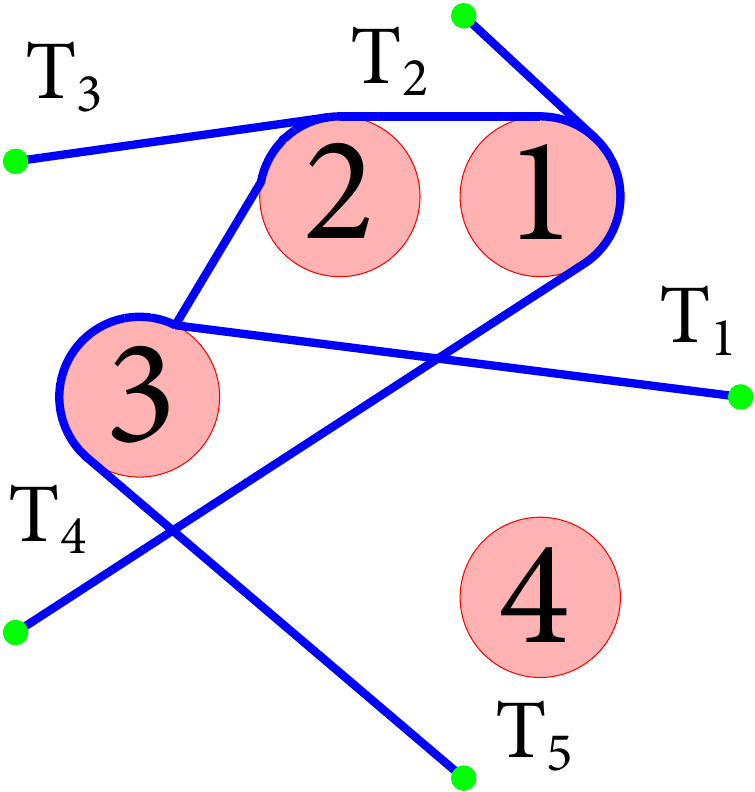}
\put(13,10){{\small\textcolor{black}{a)}}}
\end{overpic}\begin{overpic}[width=0.24\columnwidth]{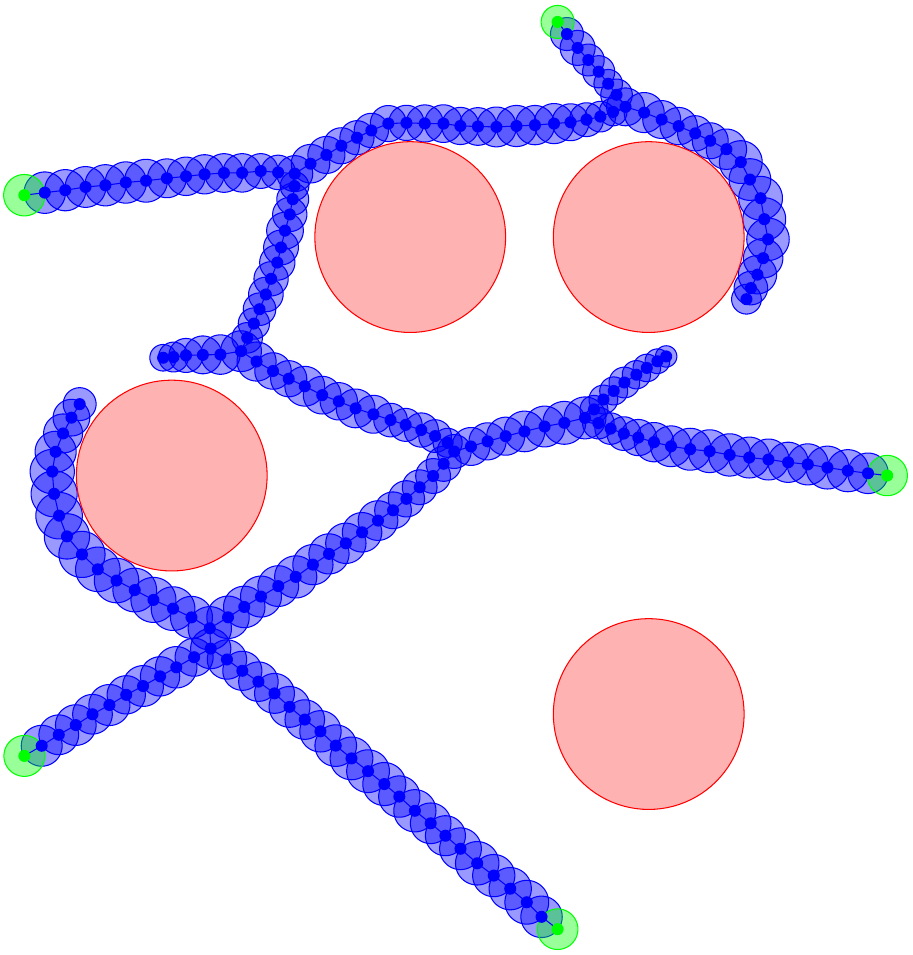}
\put(13,10){{\small\textcolor{black}{b)}}}
\end{overpic}\begin{overpic}[width=0.24\columnwidth]{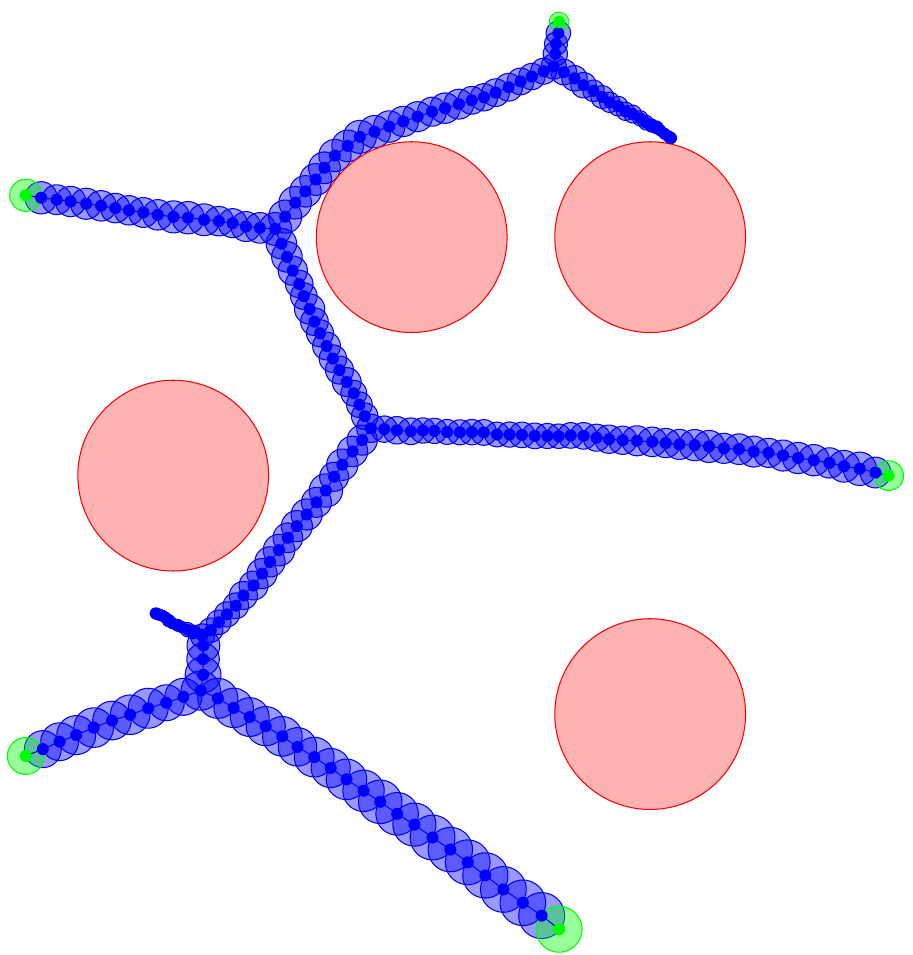}
\put(13,10){{\small\textcolor{black}{c)}}}
\end{overpic}\begin{overpic}[width=0.24\columnwidth]{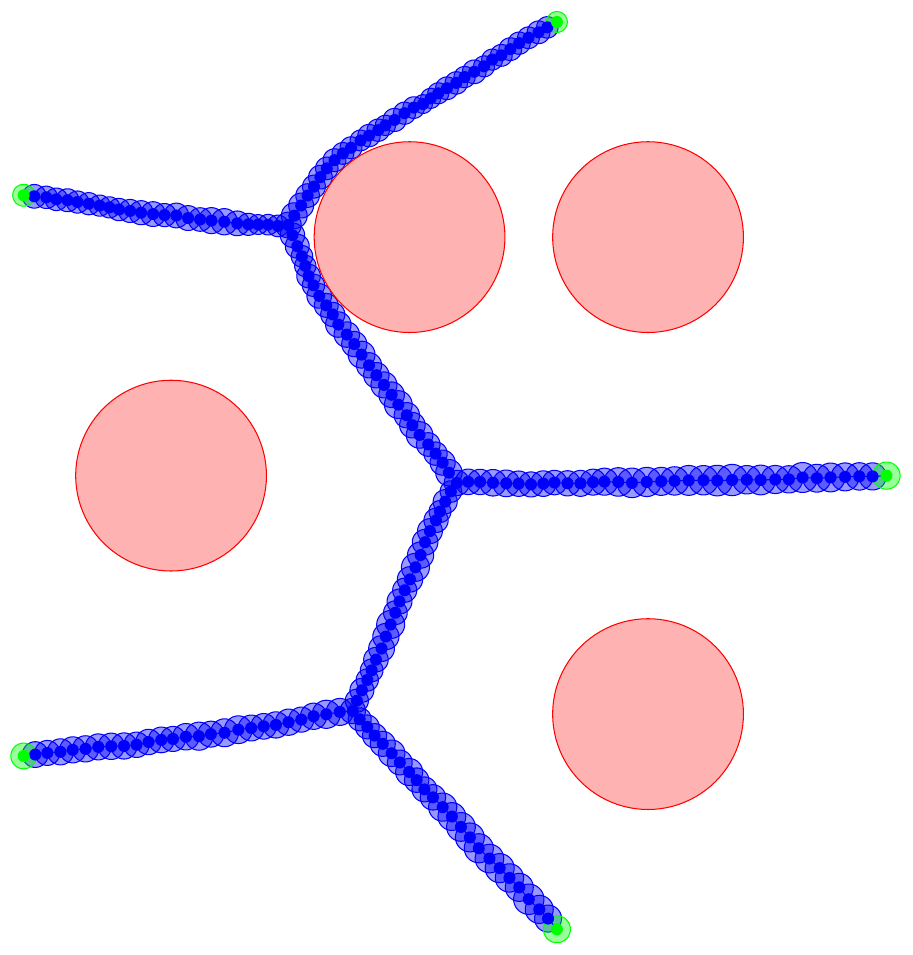}
\put(13,10){{\small\textcolor{black}{d)}}}
\end{overpic}
% Plots are generated by script 
\vspace{-5mm}
\caption{Evolution from a homotopy class involving loops. The homotopy-generation procedure generates a 1-dimensional tree (a). Relays are placed along the tree, but their recombination via Algs.~\ref{alg:MSTPlusLeaf}--\ref{alg:nooverlap} breaks down the initial structure (b-c), until a simpler homotopy is obtained (d).
\label{fig:2Tex}}
\vspace{-5mm}
\end{figure}

\subsection{Homotopy generation algorithm (\textsc{HomGen})\label{sec:homgen}}
The homotopy generation algorithm (\textsc{HomGen}) builds upon the graph $G$ defined in Sec.~\ref{sec:bitangents}: this graph is augmented to a larger graph $\tilde G$ to include  also all the edges connecting pairs of nodes in line of sight\footnote{This has been shown in \cite{bernardini2024multi} to enable cheaper solutions to be found by the solver.}. The \textbf{Steiner tree problem for graphs}\cite{dreyfus1971steiner} (STPG) is relevant to our purposes: the solution of such problem is a 1-dimensional tree with shortest length on $\tilde G$ that connects all terminals. We denote with $S$ the homotopy class of this solution, and with $\mathbf{S}_{\textrm{pre}}$ the set of all homotopies. 
%: clearly $S\in \mathbf{S}_{\textrm{pre}}$.
To solve the STPG we implemented the SCIP-Jack~\cite{RehfeldtKoch2023} solver into custom {\sc Matlab} routines.

Other homotopies can be generated through alterations of the underlying graph $\tilde G$. A simplified version of the method shown in \cite{bernardini2024multi} can be employed, focusing on terminals only. The procedure is shown graphically in Fig.~\ref{fig:Hmaxexample}: starting from the full graph $\tilde G$ (\ref{fig:Hmaxexample}a), the initial solution $S$ is found (\ref{fig:Hmaxexample}b). A terminal is selected  and one edge connecting it to the rest of the tree is removed, obtaining a new graph $\tilde G'$ (\ref{fig:Hmaxexample}c). 
Recomputing the STPG (\ref{fig:Hmaxexample}e) obtains a \emph{1\textsuperscript{st} generation} (1G) solution $S'_1$. $S'_1$ is classified using the method of Sec.~\ref{sec:homotopyclassif}: new homotopies are saved to $\mathbf{S}_{\textrm{pre}}$ as a pair $\{S'_1,\tilde G'\}$, otherwise they are discarded.

The procedure is repeated for other terminals (\ref{fig:Hmaxexample}d) obtaining corresponding solutions $S_1''$ (\ref{fig:Hmaxexample}f). When all terminals of $S$ are visited, we restart from the first available $\{S_1',\tilde G'\}$ saved in $\mathbf{S}_{\textrm{pre}}$ (\ref{fig:Hmaxexample}g), obtaining \emph{2\textsuperscript{nd} generation} (2G) solutions $S_2$ (\ref{fig:Hmaxexample}h). The procedure iterates until all solutions are scanned, or until a predefined generation limit is reached.

\begin{figure}[h]
\centering
\begin{overpic}[width=\columnwidth]{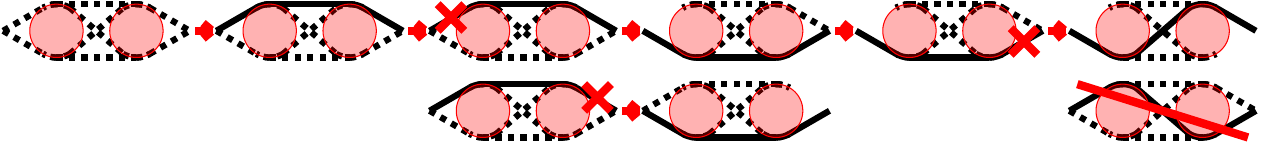}
\put(0,4.5){{\small\textcolor{black}{(a)}}}
\put(16,4.5){{\small\textcolor{black}{(b)}}}
\put(19,12){{\small\textcolor{black}{(initial)}}}
\put(32,4.5){{\small\textcolor{black}{(c)}}}
\put(32,-1){{\small\textcolor{black}{(d)}}}
\put(49,4.5){{\small\textcolor{black}{(e)}}}
\put(49,-1){{\small\textcolor{black}{(f)}}}
\put(55,12){{\small\textcolor{black}{(1G)}}}
\put(66,4.5){{\small\textcolor{black}{(g)}}}
\put(81,4.5){{\small\textcolor{black}{(h)}}}
\put(81,-1){{\small\textcolor{black}{(i)}}}
\put(89,12){{\small\textcolor{black}{(2G)}}}
\end{overpic}
% Plots are generated by scripts DroneSwarmEasterIsland\19_BitangentSteinerJournal\HomotopyTrees\plotHmax.m
\caption{Algorithm \ref{alg:homgen}: full graph $\tilde G$ (a). First STPG solution  (b). Edge removal from left (c) and right (d) terminal. 1\textsuperscript{st} generation solution found and saved (e), identical solution found and discarded (f). Proceed to obtain a 2\textsuperscript{nd} generation solution (h). Discarded solutions may prevent homotopies to be found (i). 
\label{fig:Hmaxexample}}
\end{figure}
This algorithm is not expected to produce all the $N_{\textrm{hom}}^{\textrm{max}}$ homotopies of \eqref{eq:Nmaxgen}, but only a limited subset: the removal of a terminal's edge can only affect locally the direction of a tree branch, but cannot control its curvature among pairs of nodes that are not terminals. If a dense distribution of obstacles is inside the convex hull of the terminals, the algorithm will unravel the homotopy among the outermost shell of obstacles, but will not partition   inner shells (shells not connected directly to the terminals). This limitation is intrinsic to the model and cannot be avoided.

% The number of retrievable homotopies depends on how many edges stem from each terminal. However, these may include loops (as in the example of Fig.~\ref{fig:2Tex}) so this estimation is not suitable for a direct comparison with $N_{\textrm{hom}}^{\textrm{max}}$. 

Another limitation is that the current version of the algorithm discards redundant homotopies, regardless of the underlying graph that generated them. For example Fig.~\ref{fig:Hmaxexample}e and Fig.~\ref{fig:Hmaxexample}f are identical, but only Fig.~\ref{fig:Hmaxexample}e is saved, and can initialize another iteration of the algorithm with itself and its own graph as initial conditions. This causes only the homotopy of Fig.~\ref{fig:Hmaxexample}h to be found, but not its symmetrical one Fig.~\ref{fig:Hmaxexample}i. This limitation is only temporary and will be corrected in an imminent revision of the algorithm.
 Pseudocode for \textsc{HomGen} is listed in Alg.~\ref{alg:homgen}.
\begin{figure}[tb]
\vspace{-4mm}
\begin{algorithm}[H]
\small
\caption{\small\sc \textsc{HomGen} (\textbf{S})}\label{alg:homgen}
\begin{algorithmic}[1] 
\For{$S$ in \textbf{S}}
\For{Terminal $T$ \textbf{in} $S$}
\State Remove edge \& run STPG solver on reduced graph
\If{$S'$ exists}
    \If{$S'$ is new homotopy class}
        \State \textbf{S} $ \leftarrow \{\textbf{S},S'\}$
        \Else
        \State Discard $S'$
    \EndIf
\EndIf
\EndFor
\EndFor
\end{algorithmic}
\end{algorithm}
\vspace{-6mm}
\end{figure}

\vspace{-6mm}
\subsection{Convergence likelihood estimation\label{sec:convlike}}
When an homotopy from $\mathbf{S}_{\textrm{pre}}$ is provided as an initial condition and $n$ relays to Algs.~\ref{alg:MSTPlusLeaf}--\ref{alg:nooverlap}, the subsequent evolution will produce a 2D network that may or may not converge to the same homotopy. This depends on the number of available relays, which allows to define an average radius $R_{\textrm{avg}}\simeq L/n$ analogous to Eq.~\eqref{eq:avgrad}, where $L$ length of the STPG solution. 

If the solution passes between obstacle-obstacle or terminal-obstacle pairs whose distances
\begin{equation}
   \small D_{OO'} \equiv \norm{O-O'}-R_O-R_{O'}\,\,,\,\,  D_{TO}\equiv\norm{T-O}-R_O\,\,,
\end{equation}are smaller than $2R_{\textrm{avg}}$ ($D_{OO'}$) or $R_{\textrm{avg}}$ ($D_{TO}$), the solution is likely to either be suboptimal or not converge at all under Algs.~\ref{alg:MSTPlusLeaf}--\ref{alg:nooverlap}. Therefore we define CL as
\begin{equation}
CL\equiv    \min_{\{OO',TO\} \text{pairs}}\left\{ 1-\frac{2R_{\textrm{avg}}}{D_{OO'}},1-\frac{R_{\textrm{avg}}}{D_{TO}}\right\}\times 100\%
\label{eq:CL}
\end{equation} among all $OO'$ and $TO$ pairs intersected by a given network.
\vspace{-6mm}
\subsection{Additional simulations\label{sec:secondset}}
The pre-scan algorithm was tested with two  additional campaigns of simulations, with the same configuration of terminals as Sec.~\ref{sec:mtermnoobs} but with $\phi=4$ obstacles. The first campaign concluded naturally after all possible generations were scanned, while the second was stopped after the first.

As with $\mathbf{S}_{\textrm{pre}}$, we denote with $\mathbf{S_{\textrm{fin}}}$ the set of final homotopies of the 2D networks obtained evolving $\mathbf{S}_{\textrm{pre}}$ under Algs.~\ref{alg:MSTPlusLeaf}--\ref{alg:nooverlap}:
\begin{itemize}
    \item The first simulation produced 9560 trees, further classified in 141 independent classes $\mathbf{S_{pre}^{all}}$ and 7 generations. Several of these involved loops, so they will not survive to the final stage.
        
    \item The second simulation stopped at the first generation, producing 1650 trees distributed in 35 classes $\mathbf{S_{pre}^{1gen}}$.
\end{itemize}
\begin{table}
    \centering
    \caption{Simulation Results using Algs.~\ref{alg:pre-scan}--\ref{alg:homgen} and 1 or all generations.}
    \begin{tabular}{c|c|c|c|c|c|c|c|}
         $n$ & Gen. & $|\mathbf{S_{\textrm{pre}}}|$ & $CL>10\%$ & $|\mathbf{S_{\textrm{fin}}}|$ &  Min. [km$^2$]\\
          \hline\hline
        180&  $\begin{array}{c}
             \text{All} \\
              1
        \end{array}$ & $\begin{array}{c}
             141 \\
             35 
        \end{array}$ & $\begin{array}{c}
             141 \\
             35 
        \end{array}$ & $\begin{array}{c}
             59 \\
             29 
        \end{array}$ &$\begin{array}{c}
             12.3 \\
             12.3 
        \end{array}$ \\
        \hline
              60&  $\begin{array}{c}
             \text{All} \\
              1
        \end{array}$ & $\begin{array}{c}
             141 \\
             35 
        \end{array}$ & $\begin{array}{c}
             52 \\
             17 
        \end{array}$ & $\begin{array}{c}
             28 \\
             16 
        \end{array}$ & $\begin{array}{c}
             34.9 \\
             34.9 
        \end{array}$\\
        \hline
              30&  $\begin{array}{c}
             \text{All} \\
              1
        \end{array}$ & $\begin{array}{c}
             141 \\
             35 
        \end{array}$ & $\begin{array}{c}
             17 \\
             10 
        \end{array}$ & $\begin{array}{c}
             10 \\
              9
        \end{array}$ & $\begin{array}{c}
             70.8 \\
             70.8 
        \end{array}$\\
        \hline\hline
    \end{tabular}
    
    \label{tab:results}
    \vspace{-3mm}
\end{table}
$\mathbf{S_{pre}^{all}}$ and $\mathbf{S_{pre}^{1gen}}$ were passed as initial conditions to Algs.~\ref{alg:MSTPlusLeaf}--\ref{alg:nooverlap}, and evolved with $n=180$, $n=60$ and $n=30$ relays. The $CL$ threshold for acceptance of an initial condition  was to 10\%. The results are collected in Table \ref{tab:results}, and inspire some considerations. Over $N_{\textrm{hom}}^{\textrm{max}}=m^\phi=625$ possible homotopies only 59 were determined in the final stage.
This is mainly due to the discarded redundant homotopies mentioned in Sec.~\ref{sec:homgen}. 
 
\begin{figure}[h]
\centering
\begin{overpic}[width=0.6\columnwidth]{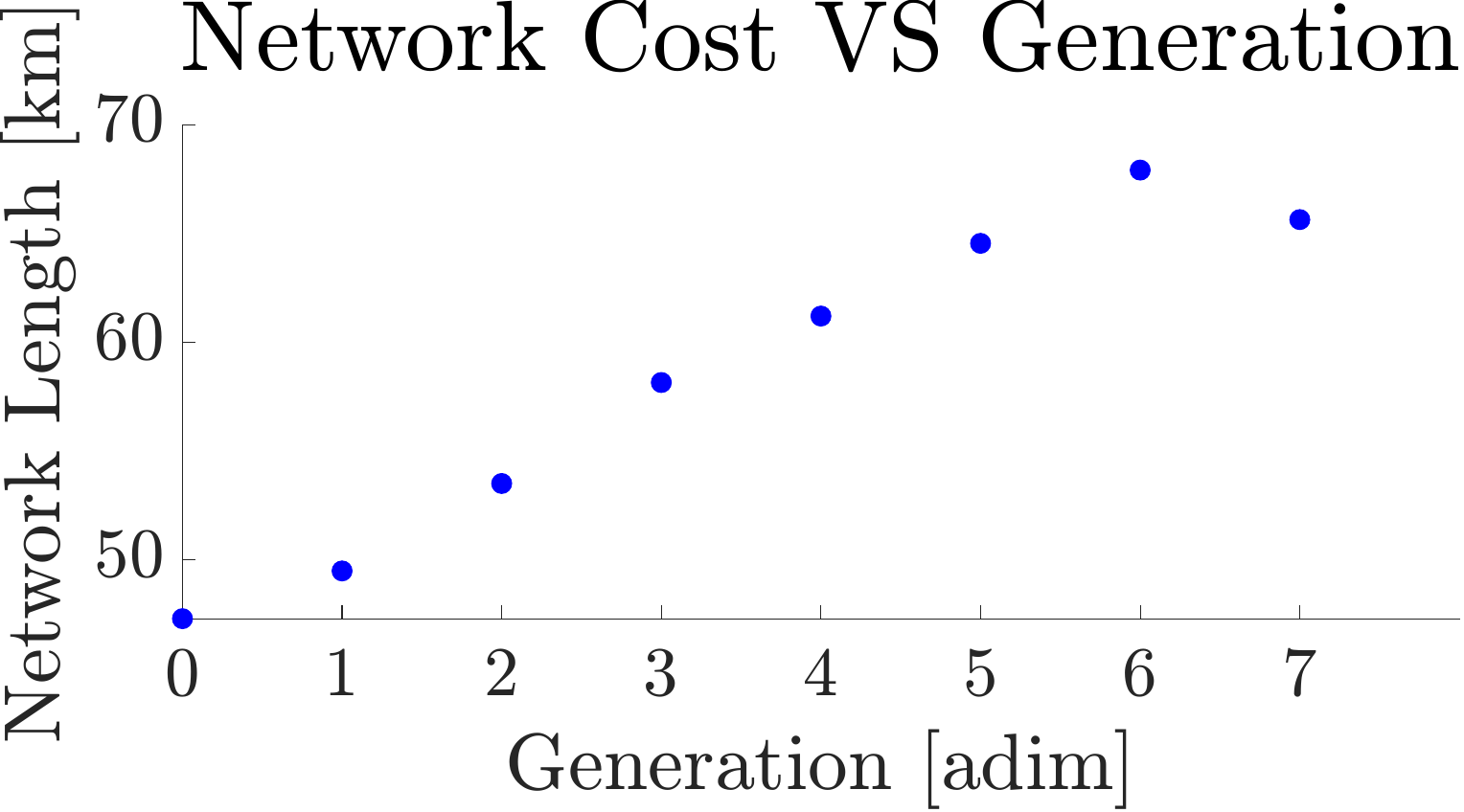}
%\put(0,4.5){{\small\textcolor{black}{(a)}}}

\end{overpic}
% Plots are generated by scripts DroneSwarmEasterIsland\19_BitangentSteinerJournal\HomotopyTrees\CostGenerations.m
\caption{Average length per generation of the solutions found by Alg.~\ref{alg:homgen}. Later generations appear to have larger lengths.
\label{fig:gencosts}}
\end{figure}However, discarded homotopies appear from the 2\textsuperscript{nd} generation on, and the optimal results in all cases belonged at most to the 1\textsuperscript{st} generation. The fact that on average, later generations have larger costs (as shown in Fig.~\ref{fig:gencosts}) makes the exclusion of the redundant homotopies less likely to have altered the result.
    
Fig.~\ref{fig:homcosts} shows the optimal homotopies for the three values of $n$ considered: as $n$ is inversely proportional to the average radius $R_{\textrm{avg}}$, larger radii ($n=60$ and $n=180$) allow the 2D network to extend through narrower obstacle pairs. In the present case, $n=60$ and $n=180$ allowed network to evolve to a solution very close to that of the case without obstacles (Fig.~\ref{fig:homcosts}a-b), while for $n=30$ this was not allowed and an alternate solution was chosen (Fig.~\ref{fig:homcosts}c).

\begin{figure}[tb]
\vspace{-2mm}
\centering
\begin{overpic}[width=0.32\columnwidth]{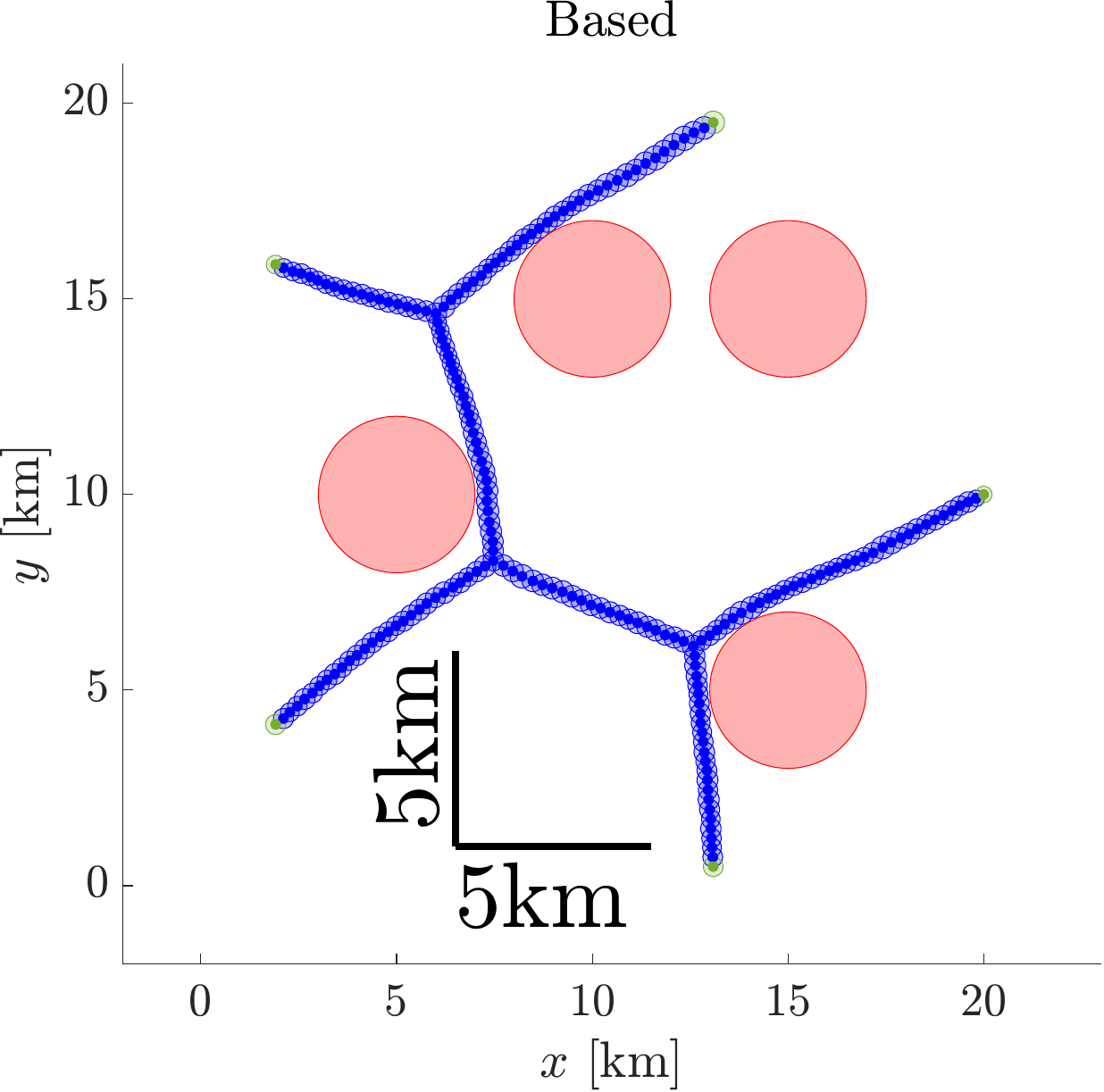}
\put(75,0){{\small\textcolor{black}{(a)}}}
\end{overpic}\begin{overpic}[width=0.32\columnwidth]{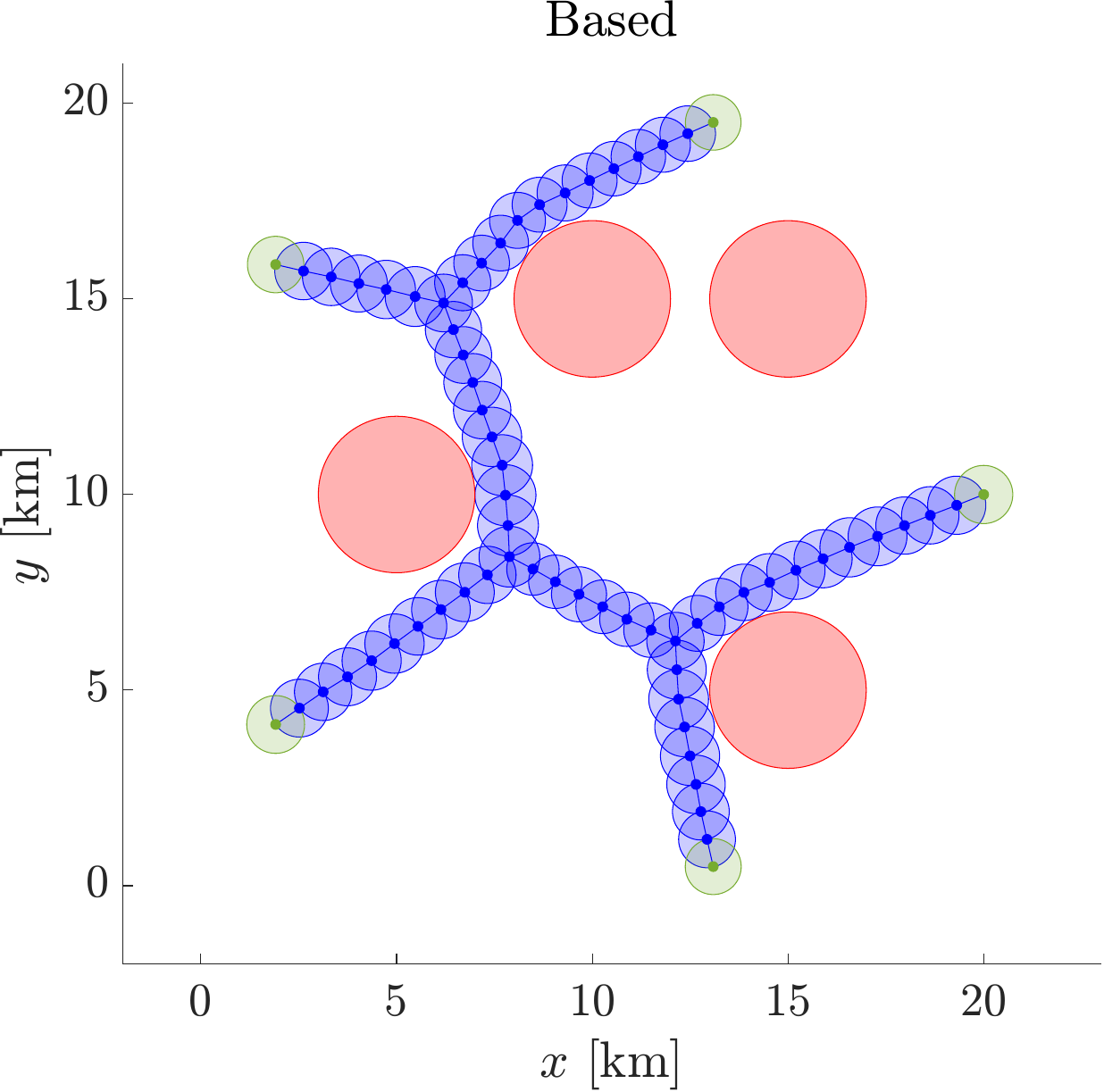}\put(75,0){{\small\textcolor{black}{(b)}}}
\end{overpic}\begin{overpic}[width=0.32\columnwidth]{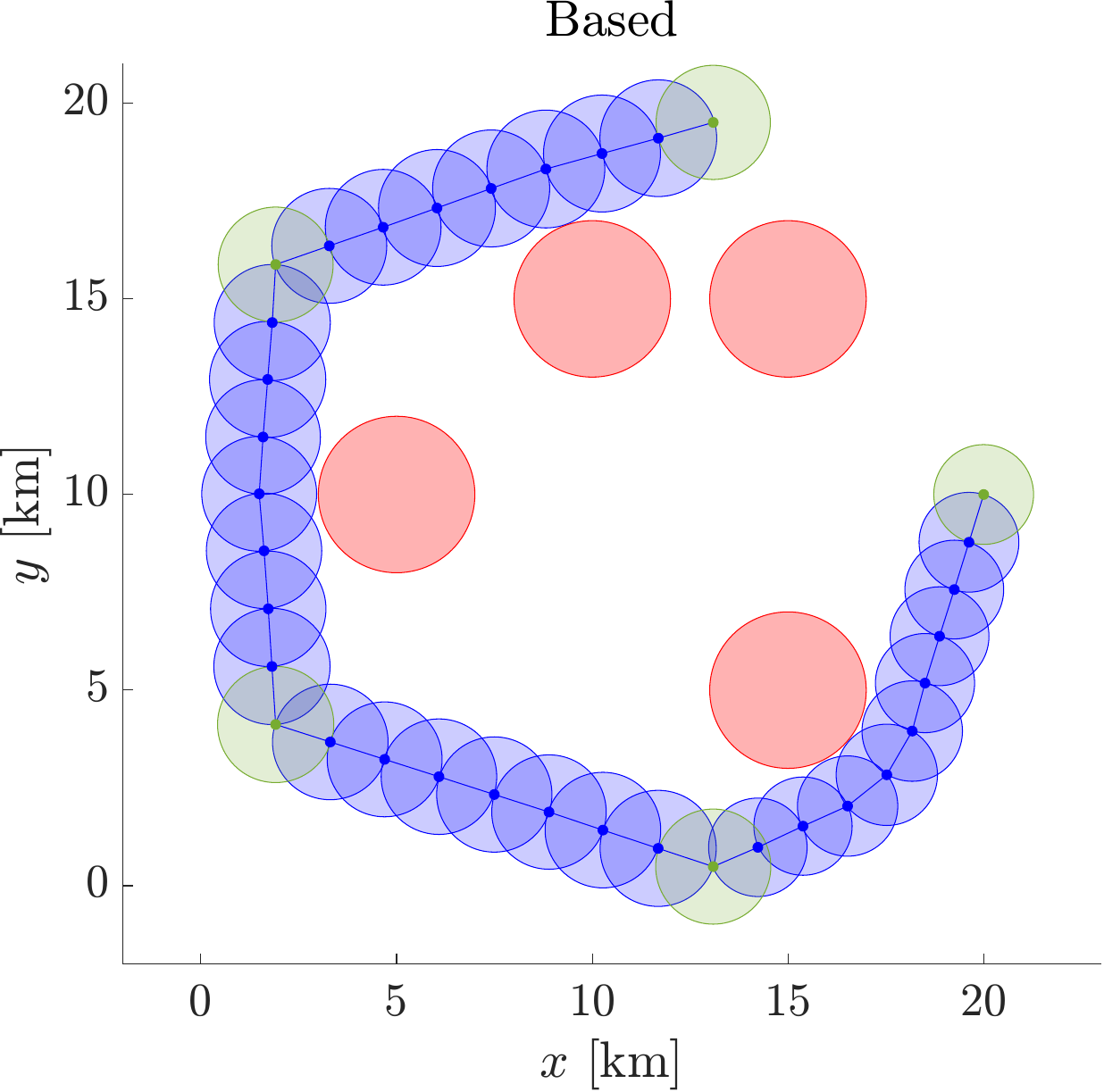}\put(75,0){{\small\textcolor{black}{(c)}}}
\end{overpic}
% Plots are generated by scripts DroneSwarmEasterIsland\19_BitangentSteinerJournal\HomotopyTrees\MAIN_HomotopyBuilderTerminalIsolation.m, selecting only result 1 (n=180,60) and result 3 (n=30), and cropping the picture with Adobe Pro with magins [0.358,1.148,1.177,0.403]in
\caption{Optimal homotopies for $n=180$ (a), $n=60$ (b) and $n=30$ (c). 
\label{fig:homcosts}}
\vspace{-5mm}
\end{figure}
 
%  For the present simulation, we have instantiated the removal procedure in three consecutive phases, which differ only for the subset of edges that are removed at each step of the procedure. Each phase starts from the full graph and the corresponding solution of the STPG, that consists of $N_e$ edges, arbitrarily numbered. Then, half of the edges are chosen and removed in the following way:
% \begin{enumerate}
%     \item In the first phase, we remove from the base graph the edges $\{1,\lfloor N_e/2\rfloor\}$ of the solutions of each STPG. 
%     \item In the second phase, we remove the edges $\{\lfloor N_e/2\rfloor+1,N_e\}$.
%     \item In the third phase, we remove edges $\{\lfloor N_e/4\rfloor+1,\lfloor 3N_e/4\rfloor-1\}$ (assuming $N_e\geq4$). 
% \end{enumerate}

\vspace{-2mm}
\section{Conclusions and paths forward\label{sec:concl}}
% We described a framework to solve the minimum range assignment problem amidst obstacles, assigning to nodes of the network radii equal to the largest incident edge and assuming a cost function quadratic in these radii. The network is optimized with the constraint that network nodes and obstacles, both modeled as disks, may not overlap.
We revised and extended the conclusions provided in \cite{Bernardini2024CASE} regarding a framework to solve the minimum range assignment problem amidst obstacles, modeled as circular regions inaccessible by terminals' and relays' transmission disks. We defined and implemented an automatic procedure to instantiate the pre-scan method of Alg.~\ref{alg:pre-scan}, and discussed its possible interpretation in terms of abstract distributions of distinct objects in distinct bins. The results obtained with the new procedure confirm the findings of \cite{Bernardini2024CASE}, for example that the homotopy of the solution of the 2D case depends on the number of available relays $n$, and may not coincide with the 1D solution, unless $n\rightarrow\infty$. 

% The heuristics defined in Algs.~\ref{alg:MSTPlusLeaf}--\ref{alg:advOpt} were first tested without obstacles, where they proved to be independent from initial conditions and yielded the optimal result in 100\% of the trials for a specific configuration. Introducing obstacles spoiled this independence: trials converged to several stable configurations, but only found the optimal configuration in 53\% of the trials. These configurations are topological local minima which cannot be overcome by the heuristic rules of Algs.~\ref{alg:MSTPlusLeaf}--\ref{alg:advOpt}. A procedure was devised to systematically explore homotopically different networks and use them as guidance for initial relay placement. Evolution under Algs.~\ref{alg:MSTPlusLeaf}--\ref{alg:advOpt} was then biased to produce solutions of the range assignment problem in the same homotopy class of the initial placement. 
The final number of homotopies determined by feeding the results of \textsc{HomGen} to Algs.~\ref{alg:MSTPlusLeaf}--\ref{alg:nooverlap} was $|\mathbf{S_{\textrm{fin}}}|=59$ (Fig.~\ref{fig:homotopylist}), at least one order of magnitude lower than $N_{\textrm{hom}}^{\textrm{max}}=625$. This is due to the intrinsic inability of \textsc{HomGen} of influencing the behavior of the network far from the terminals. Including the redundant homotopies could partially fix the problem for configurations where terminals come in contact with a high number of obstacles. However, 14 of the 15 expected Bell partitions of the obstacles' set has been found. The 15\textsuperscript{th} is complex enough to be outside the capabilities of \textsc{HomGen}, and a possible realization has been generated manually for completeness: however, evolution with Algs.~\ref{alg:MSTPlusLeaf}--\ref{alg:nooverlap} caused its decay to a different partition.

\textsc{HomGen}'s ratio of homotopy discovery per tree generated was $\sim 1-2\%$: to determine $N_{\textrm{hom}}^{\textrm{max}}=5^4=625$ homotopies it should generate at least 50-100 times as many trees. This is likely due to the complexity of the graph $\tilde G$: more than edge stemming from a terminal may lead to the same homotopy, but the algorithm still visits each of them. This suggests to direct efforts towards the study of more minimal graphs, that would still capture the essential features of an obstacle distribution, but with less edges and vertices. Further developments of these aspects are already under study and will be the topic of a future work.
Increasing the performances of \textsc{HomGen} would also enable its application to dynamic scenarios, where both obstacles and terminals are allowed to reconfigure.

Natural applications of these ideas are in distributed networks communicating via line of sight (1D) or coverage (2D) methods: for example, a network might be composed of unmanned aerial vehicles (UAVs), ground/underwater remotely operated vehicles (ROVs), sensors or satellites. 
\vspace{-4mm}

  %Currently optimizing positions and transmission radii is easier than assigning the network topology. 
 % \todo{is it?}  no proof...
%Determining the optimal topology.  Given a topology, Dominik Krupke has a linear program that quickly solves for the optimal positions, but figuring out who talks to whom (optimal topology or even an approximation) is hard.
%Handling obstacles.  We don’t have any good theory for this, and are iterating between various global searches to find the best solution.
%\clearpage
%\mbox{}
%\clearpage
%\newpage
%\balance
\bibliographystyle{IEEEtranS}
 \bibliography{IEEEabrv,biblio.bib}

\begin{thebibliography}{10}
\providecommand{\url}[1]{#1}
\csname url@rmstyle\endcsname
\providecommand{\newblock}{\relax}
\providecommand{\bibinfo}[2]{#2}
\providecommand\BIBentrySTDinterwordspacing{\spaceskip=0pt\relax}
\providecommand\BIBentryALTinterwordstretchfactor{4}
\providecommand\BIBentryALTinterwordspacing{\spaceskip=\fontdimen2\font plus
\BIBentryALTinterwordstretchfactor\fontdimen3\font minus \fontdimen4\font\relax}
\providecommand\BIBforeignlanguage[2]{{%
\expandafter\ifx\csname l@#1\endcsname\relax
\typeout{** WARNING: IEEEtran.bst: No hyphenation pattern has been}%
\typeout{** loaded for the language `#1'. Using the pattern for}%
\typeout{** the default language instead.}%
\else
\language=\csname l@#1\endcsname
\fi
#2}}

\bibitem{alt2006minimum}
H.~Alt, E.~M. Arkin, H.~Br{\"o}nnimann, J.~Erickson, S.~P. Fekete, C.~Knauer, J.~Lenchner, J.~S. Mitchell, and K.~Whittlesey, ``Minimum-cost coverage of point sets by disks,'' in \emph{Proceedings of the twenty-second annual symposium on Computational geometry}, 2006, pp. 449--458.

\bibitem{10139317}
A.~Altaweel, H.~Mukkath, and I.~Kamel, ``{GPS} spoofing attacks in {FANET}s: A systematic literature review,'' \emph{IEEE Access}, pp. 1--1, 2023.

\bibitem{bell1938iterated}
E.~T. Bell, ``The iterated exponential integers,'' \emph{Annals of Mathematics}, vol.~39, no.~3, pp. 539--557, 1938.

\bibitem{bernardini2024multi}
F.~Bernardini and A.~T. Becker, ``Multi-objective heuristics for network construction in an obstacle-dense environment,'' in \emph{IEEE International Conference on Automation Science and Engineering}, 2024.

\bibitem{Bernardini2024CASE}
F.~Bernardini, D.~Biediger, I.~Pineda, L.~Kleist, and A.~T. Becker, ``Strongly-connected minimal-cost radio-networks among fixed terminals using mobile relays and avoiding no-transmission zones,'' in \emph{2024 IEEE 20th International Conference on Automation Science and Engineering (CASE)}, Aug 2024, pp. 2059--2066.

\bibitem{bhattacharya2010search}
S.~Bhattacharya, ``Search-based path planning with homotopy class constraints,'' in \emph{Proceedings of the AAAI conference on artificial intelligence}, vol.~24, no.~1, 2010, pp. 1230--1237.

\bibitem{chang2005shortest}
E.-C. Chang, S.~W. Choi, D.~Kwon, H.~Park, and C.~K. Yap, ``Shortest path amidst disc obstacles is computable,'' in \emph{Proceedings of the twenty-first annual symposium on Computational geometry}, 2005, pp. 116--125.

\bibitem{chugh2020scalarizing}
T.~Chugh, ``Scalarizing functions in bayesian multiobjective optimization,'' in \emph{2020 IEEE Congress on Evolutionary Computation (CEC)}.\hskip 1em plus 0.5em minus 0.4em\relax IEEE, 2020, pp. 1--8.

\bibitem{clementi2004power}
A.~E. Clementi, P.~Penna, and R.~Silvestri, ``On the power assignment problem in radio networks,'' \emph{Mobile Networks and Applications}, vol.~9, pp. 125--140, 2004.

\bibitem{dreyfus1971steiner}
S.~E. Dreyfus and R.~A. Wagner, ``The steiner problem in graphs,'' \emph{Networks}, vol.~1, no.~3, pp. 195--207, 1971.

\bibitem{efrat2008improved}
A.~Efrat, S.~P. Fekete, P.~R. Gaddehosur, J.~S. Mitchell, V.~Polishchuk, and J.~Suomela, ``Improved approximation algorithms for relay placement,'' in \emph{Algorithms-ESA 2008: 16th Annual European Symposium, Karlsruhe, Germany, September 15-17, 2008. Proceedings 16}.\hskip 1em plus 0.5em minus 0.4em\relax Springer, 2008, pp. 356--367.

\bibitem{efrat2015improved}
A.~Efrat, S.~P. Fekete, J.~S. Mitchell, V.~Polishchuk, and J.~Suomela, ``Improved approximation algorithms for relay placement,'' \emph{ACM Transactions on Algorithms (TALG)}, vol.~12, no.~2, pp. 1--28, 2015.

\bibitem{fuchs2008hardness}
B.~Fuchs, ``On the hardness of range assignment problems,'' \emph{Networks: An International Journal}, vol.~52, no.~4, pp. 183--195, 2008.

\bibitem{GUIMARAES2021101771}
\BIBentryALTinterwordspacing
D.~A. Guimar{\~a}es and A.~{Salles da Cunha}, ``The minimum area spanning tree problem: Formulations, benders decomposition and branch-and-cut algorithms,'' \emph{Computational Geometry}, vol.~97, p. 101771, 2021. [Online]. Available: \url{https://www.sciencedirect.com/science/article/pii/S0925772121000274}
\BIBentrySTDinterwordspacing

\bibitem{jenkins1991shortest}
K.~D. Jenkins, ``The shortest path problem in the plane with obstacles: A graph modeling approach to producing finite search lists of homotopy classes,'' Ph.D. dissertation, Monterey, California. Naval Postgraduate School, 1991.

\bibitem{kirousis2000power}
L.~M. Kirousis, E.~Kranakis, D.~Krizanc, and A.~Pelc, ``Power consumption in packet radio networks,'' \emph{Theoretical Computer Science}, vol. 243, no. 1-2, pp. 289--305, 2000.

\bibitem{RehfeldtKoch2023}
D.~Rehfeldt and T.~Koch, ``Implications, conflicts, and reductions for {Steiner} trees,'' \emph{Mathematical Programming}, vol. 197, pp. 903 -- 966, 2023.

\bibitem{yen1970algorithm}
J.~Y. Yen, ``An algorithm for finding shortest routes from all source nodes to a given destination in general networks,'' \emph{Quarterly of applied mathematics}, vol.~27, no.~4, pp. 526--530, 1970.

\end{thebibliography}
% \bibliography{biblio.bib}
% \bibliographystyle{IEEEtran}

\begin{figure*}
    \centering
    \begin{tabular}{m{0.08\textwidth}m{0.08\textwidth}m{0.08\textwidth}m{0.08\textwidth}m{0.08\textwidth}m{0.08\textwidth}m{0.08\textwidth}m{0.08\textwidth}m{0.08\textwidth}m{0.08\textwidth}}
        \boxed{\text{No obs.}} &\includegraphics[width=0.09\textwidth]{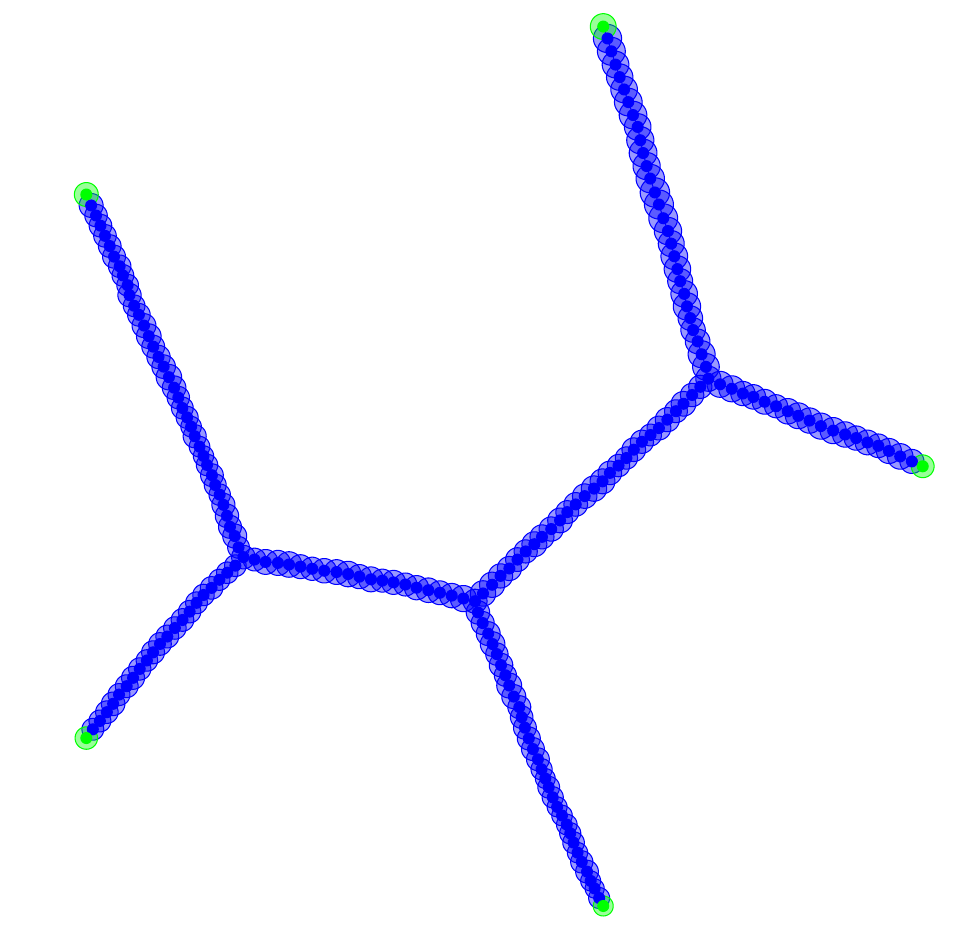}&\includegraphics[width=0.09\textwidth]{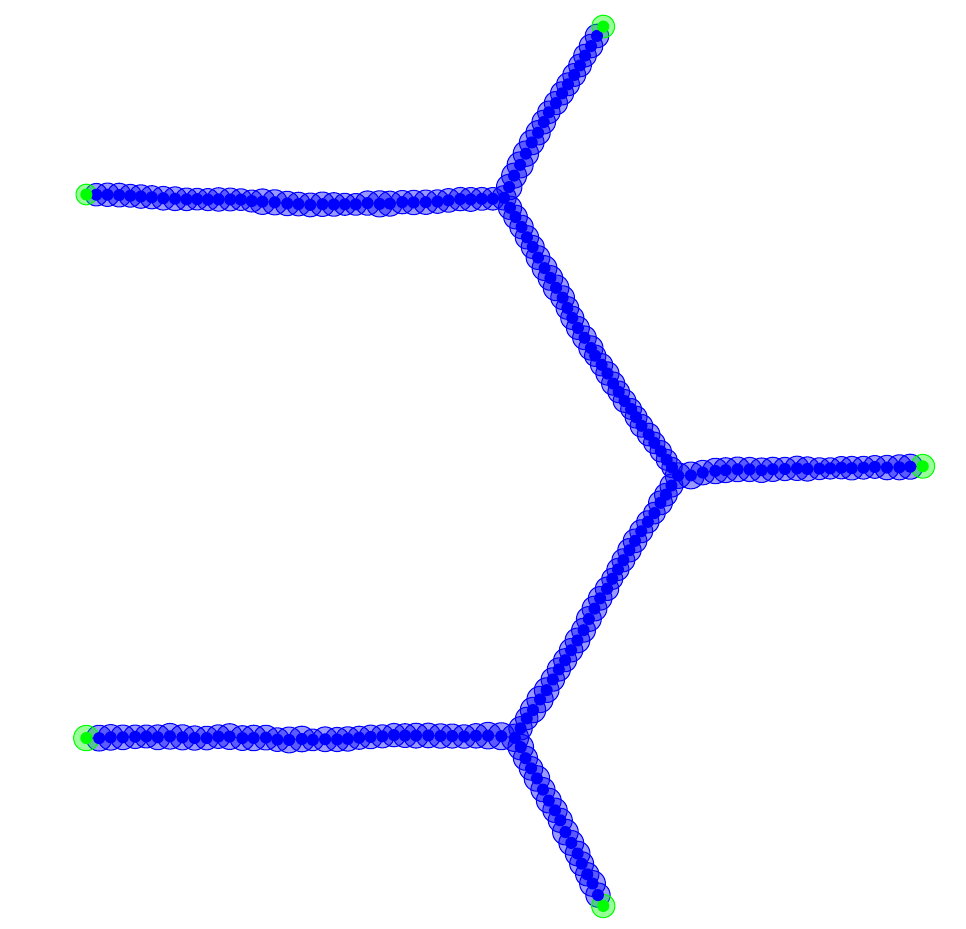}&\includegraphics[width=0.09\textwidth]{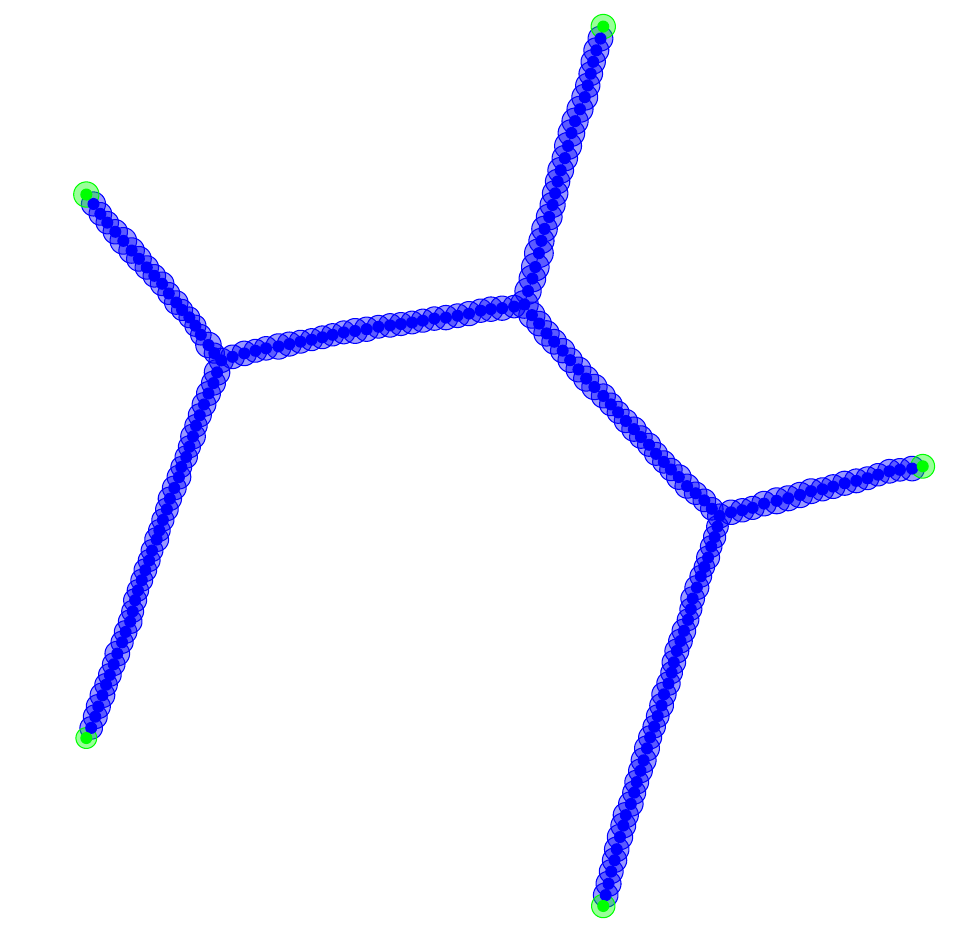}&\includegraphics[width=0.09\textwidth]{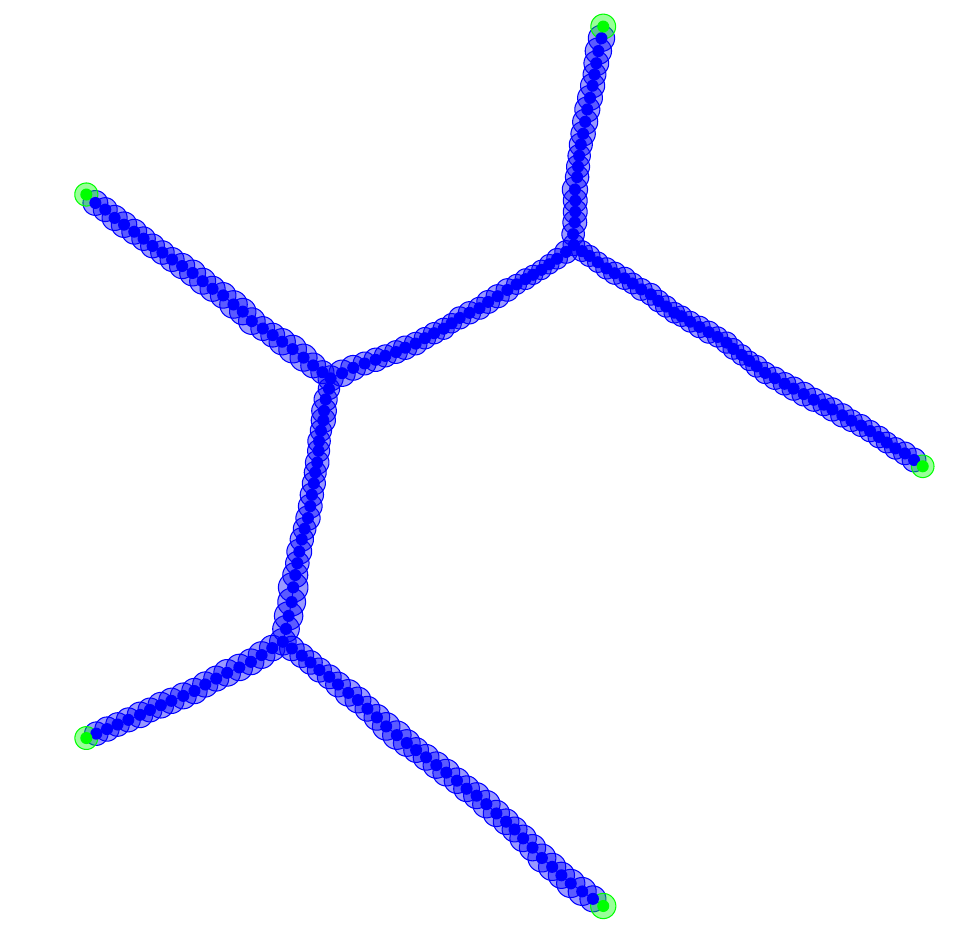}&\includegraphics[width=0.09\textwidth]{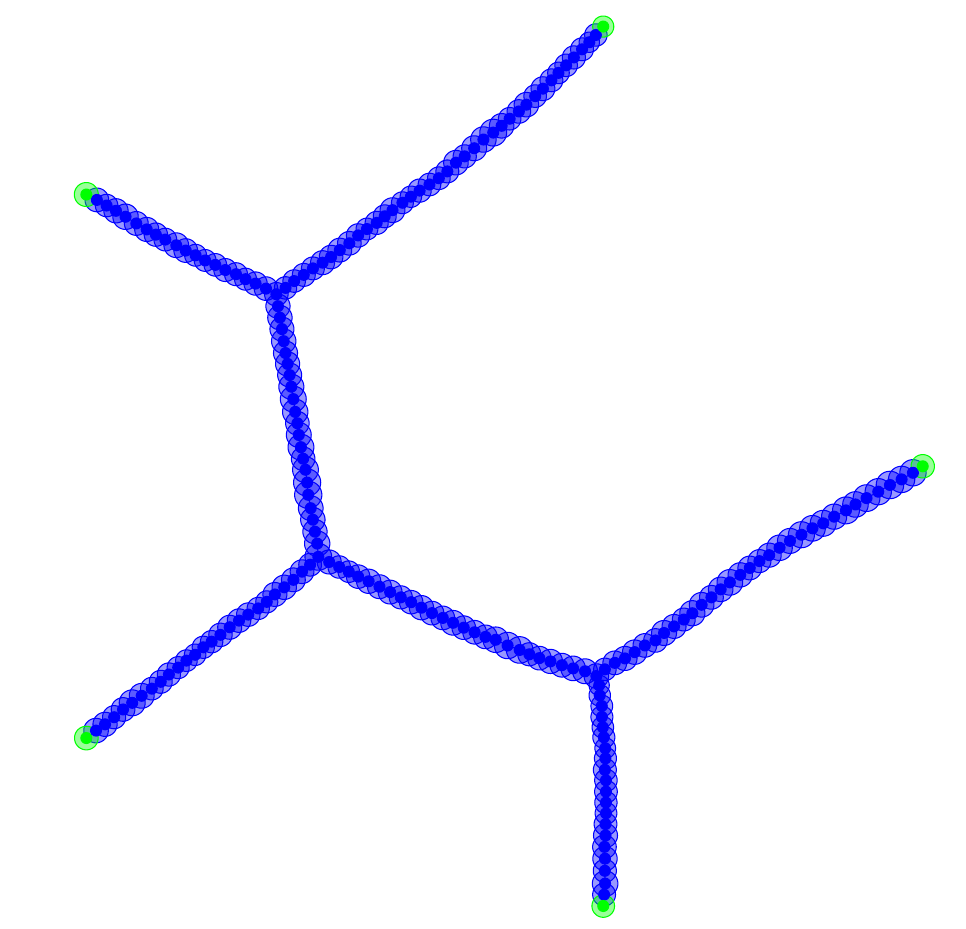}\\
        \boxed{$\{1234\}$}&\includegraphics[width=0.09\textwidth]{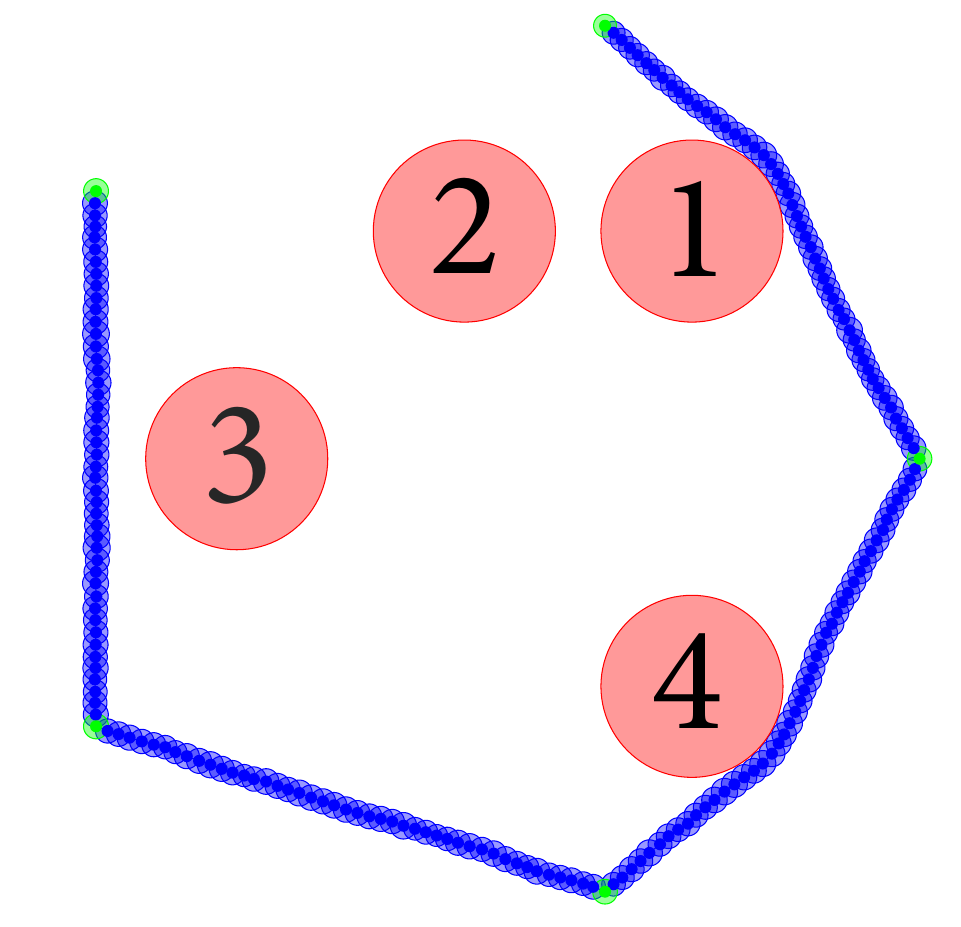}&\includegraphics[width=0.09\textwidth]{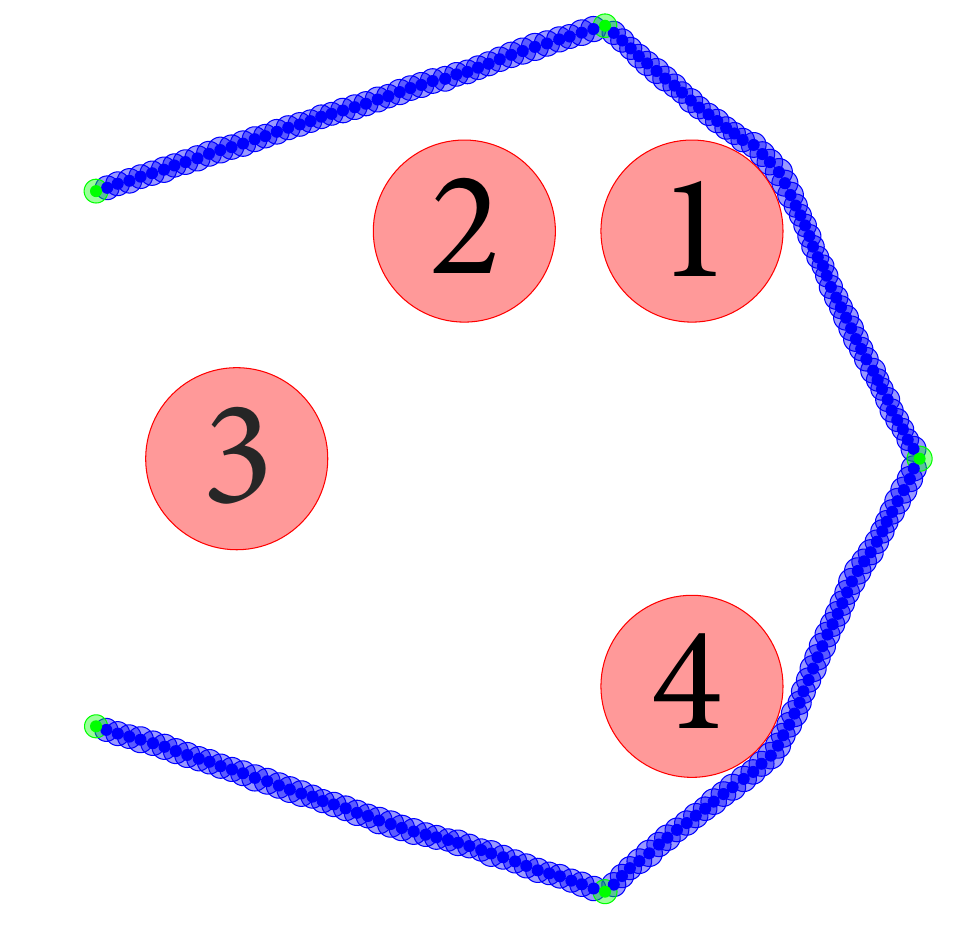}&\includegraphics[width=0.09\textwidth]{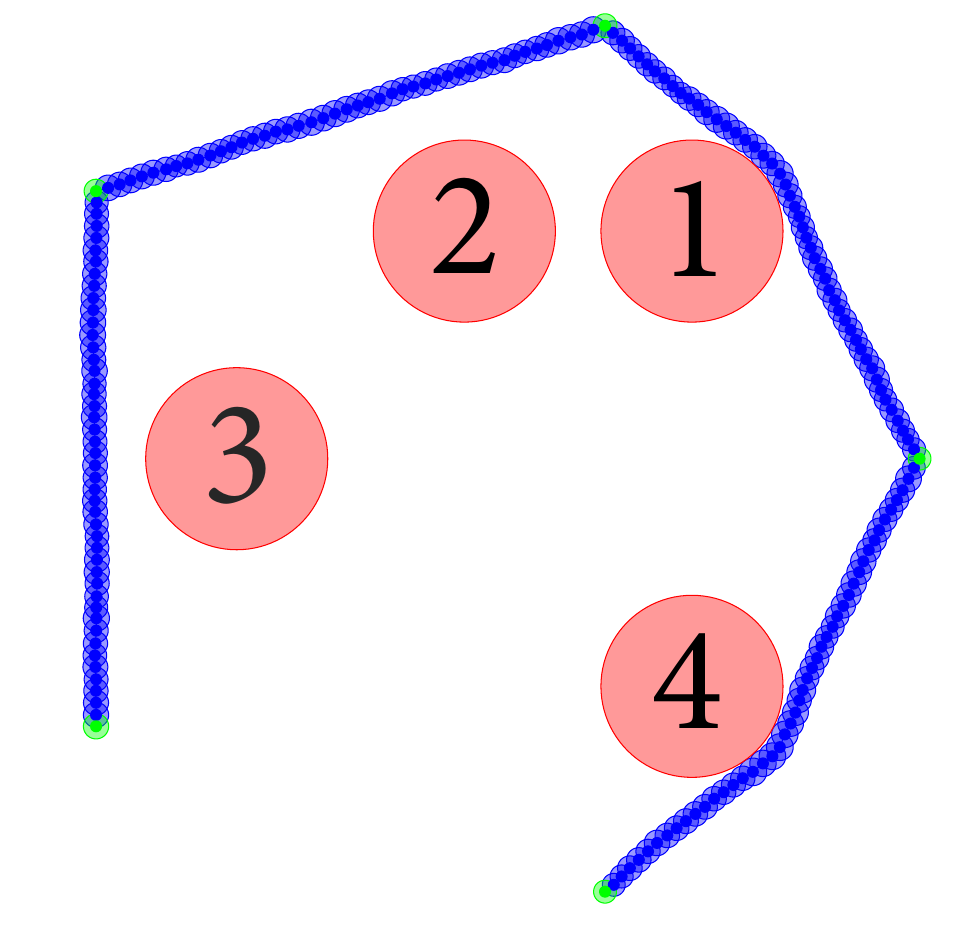}&\includegraphics[width=0.09\textwidth]{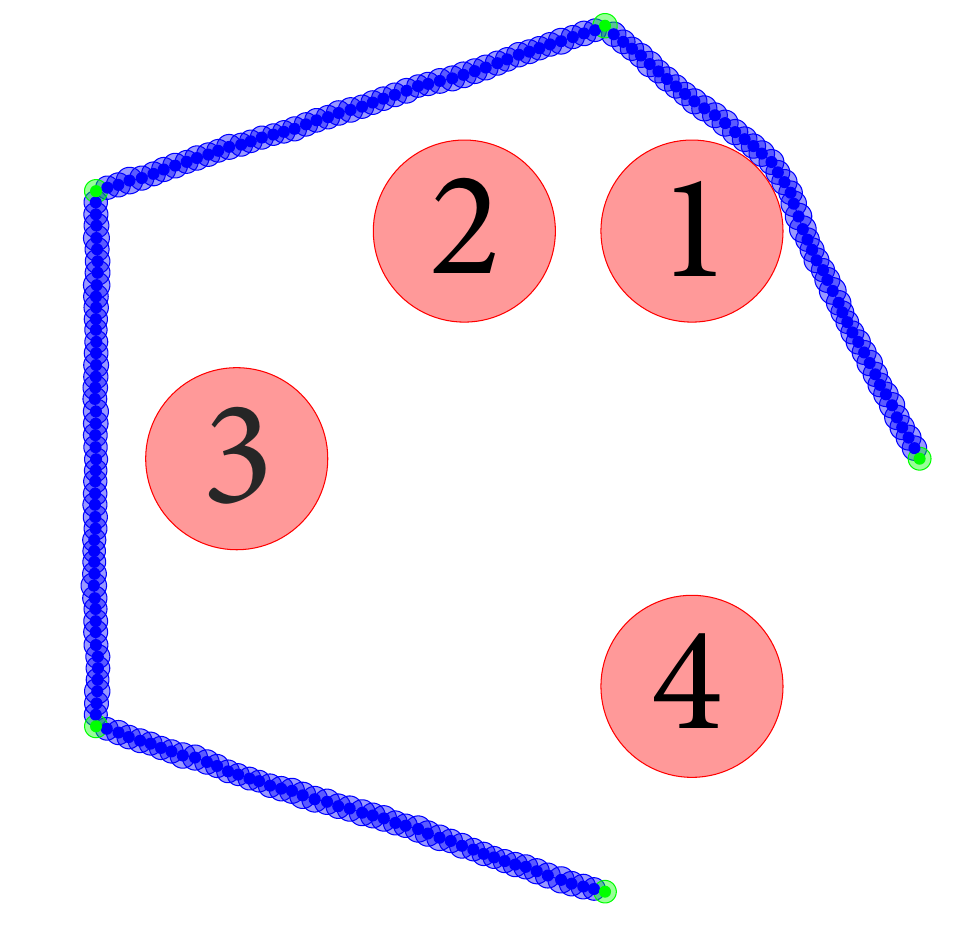}&\includegraphics[width=0.09\textwidth]{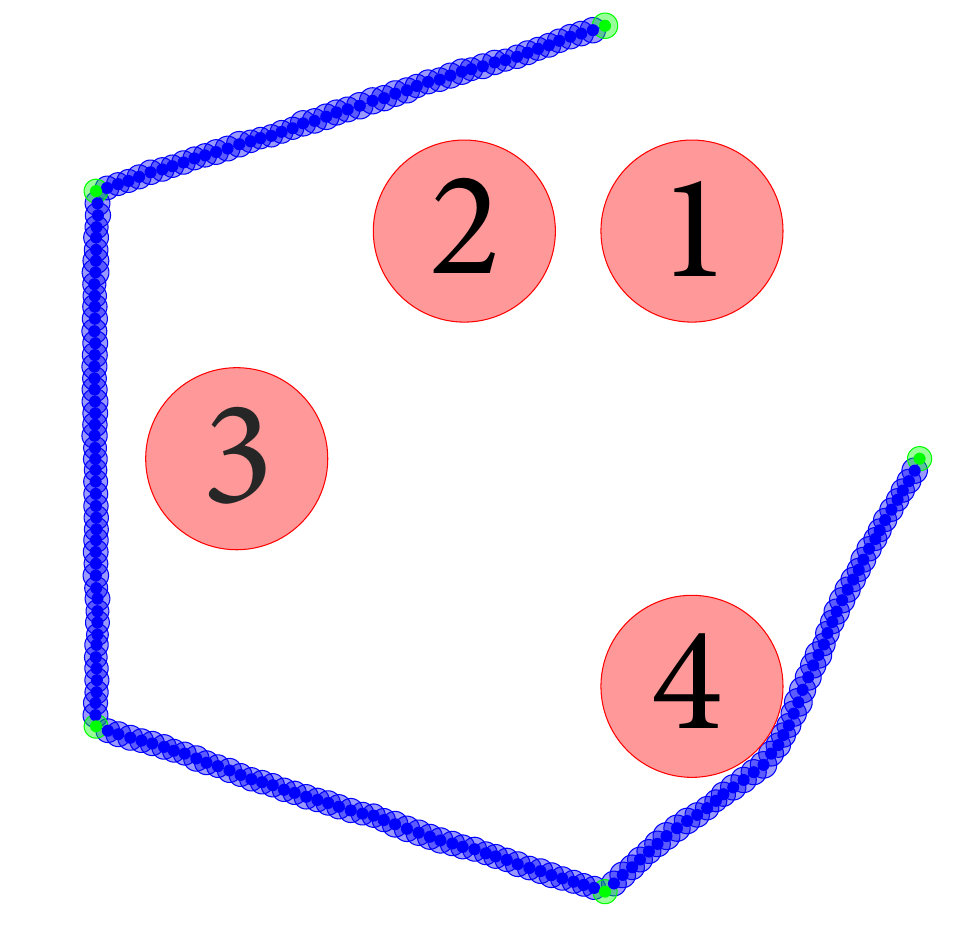}\\
        \boxed{$\{123;4\}$}&
                \includegraphics[width=0.09\textwidth]{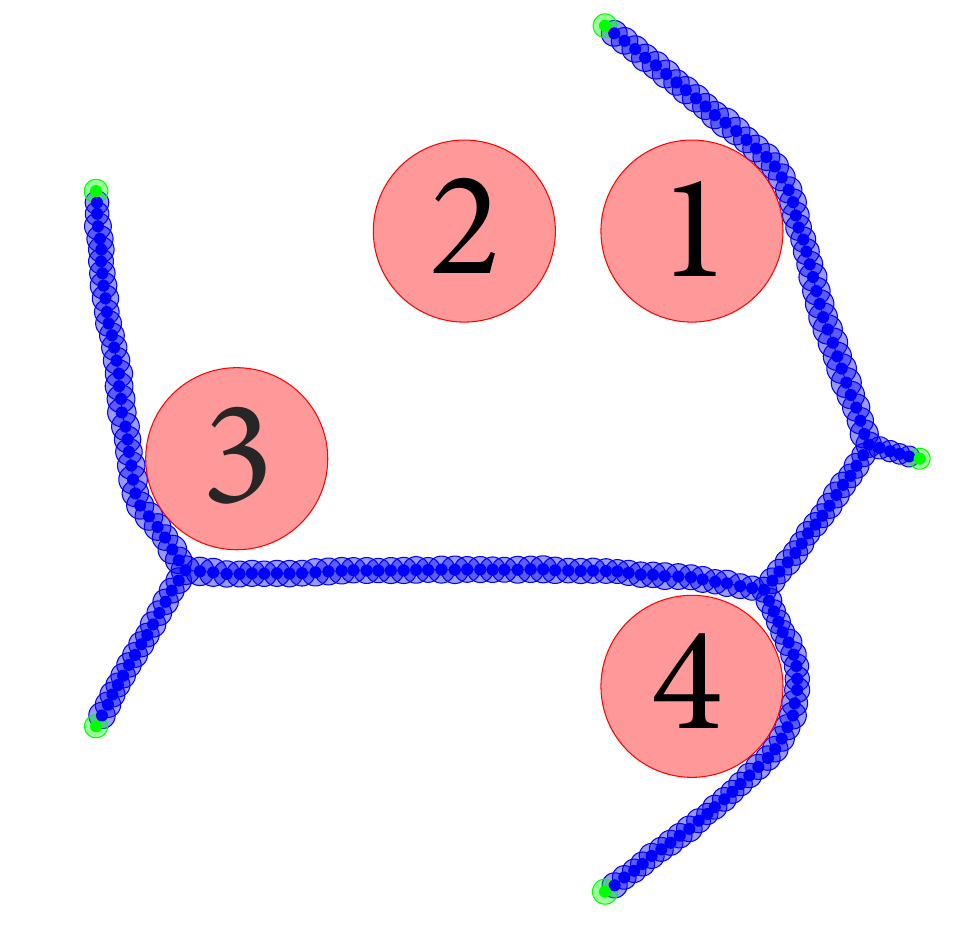}&
                \includegraphics[width=0.09\textwidth]{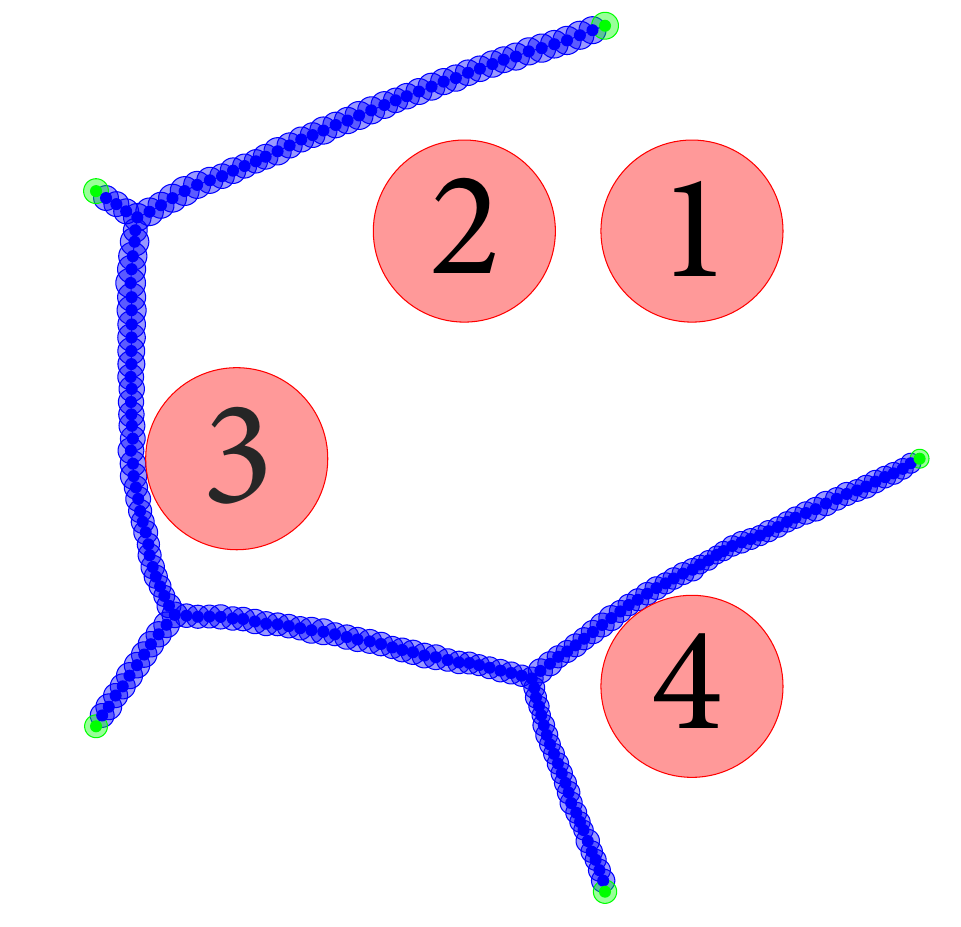}&
                \includegraphics[width=0.09\textwidth]{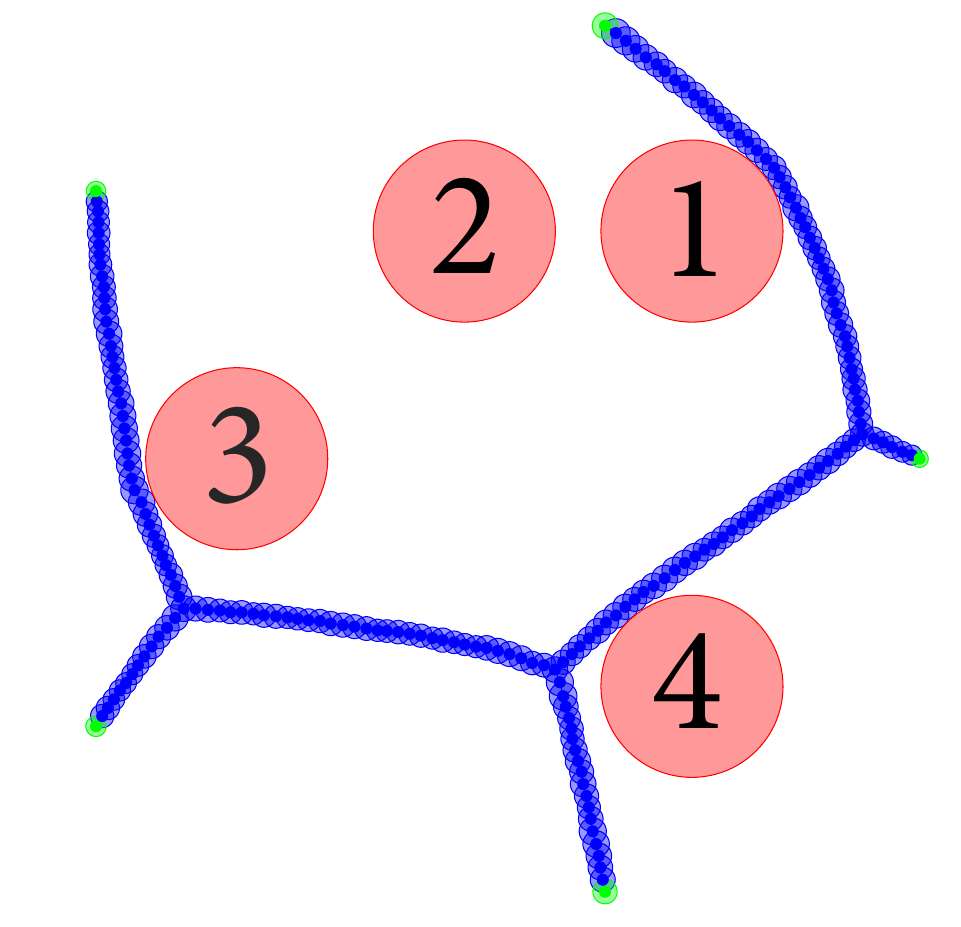}&
                \includegraphics[width=0.09\textwidth]{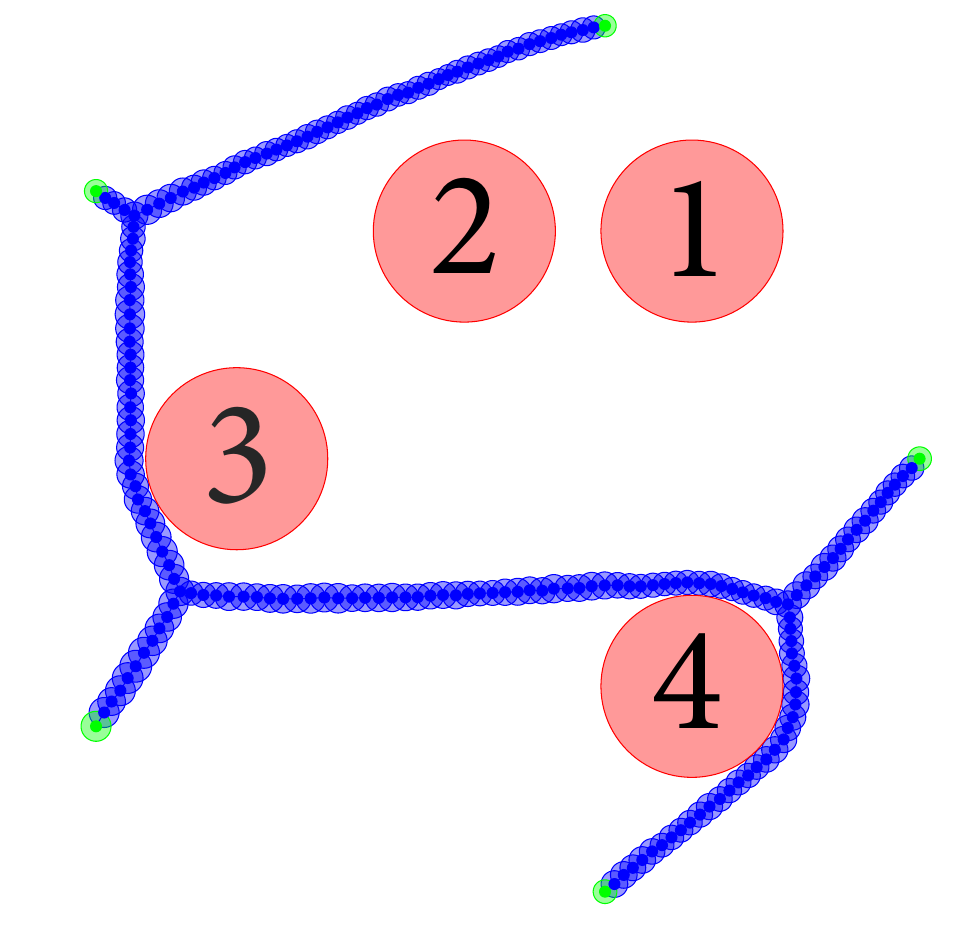}&\includegraphics[width=0.09\textwidth]{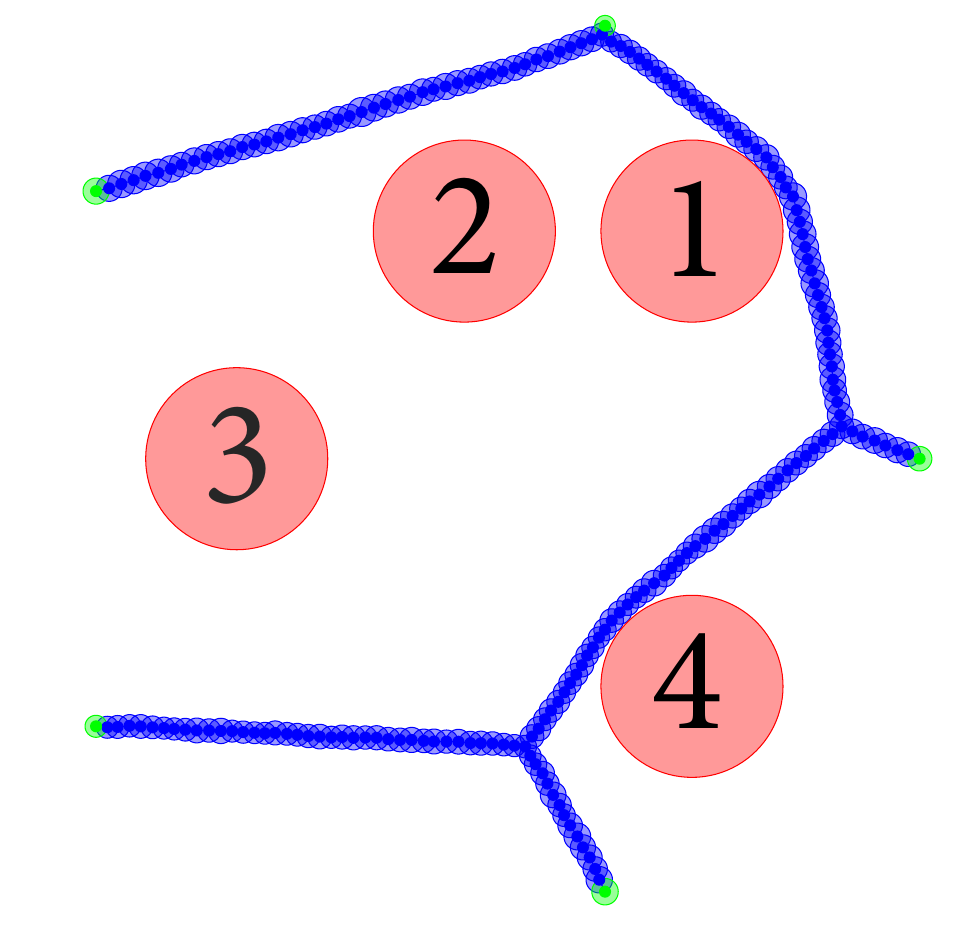}&\includegraphics[width=0.09\textwidth]{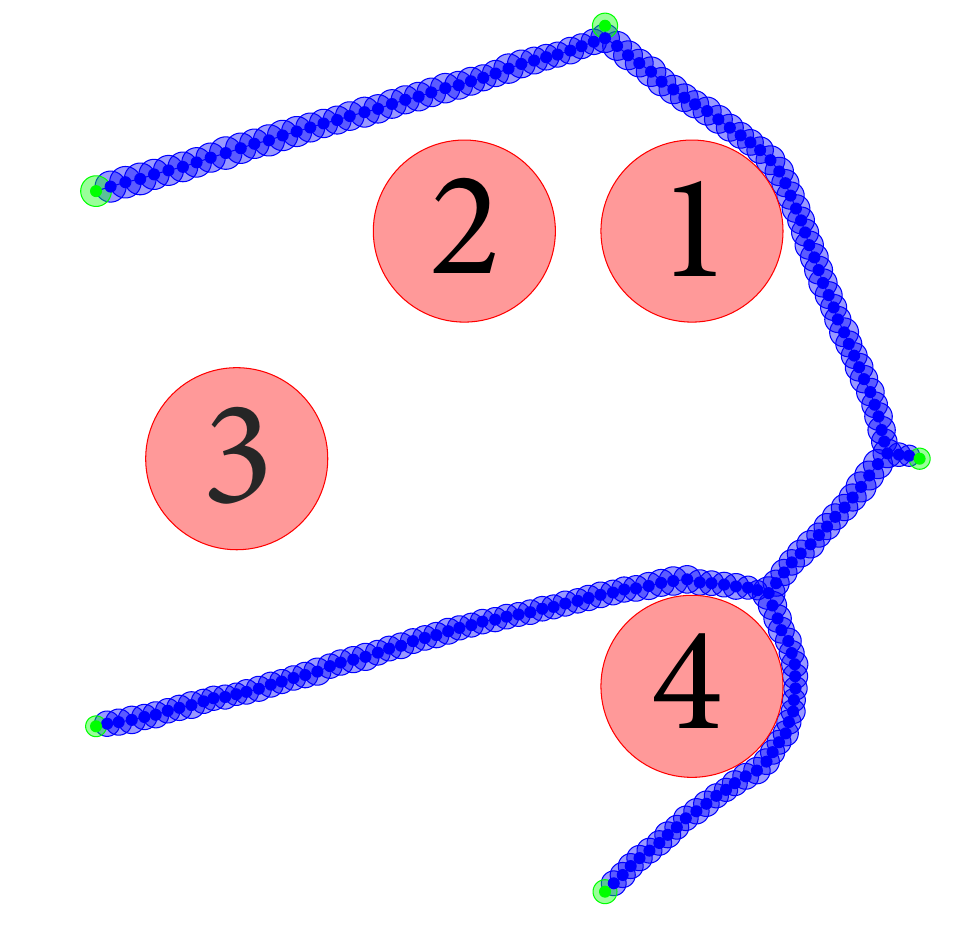}\\
               \boxed{$\{124;3\}$}&\includegraphics[width=0.09\textwidth]{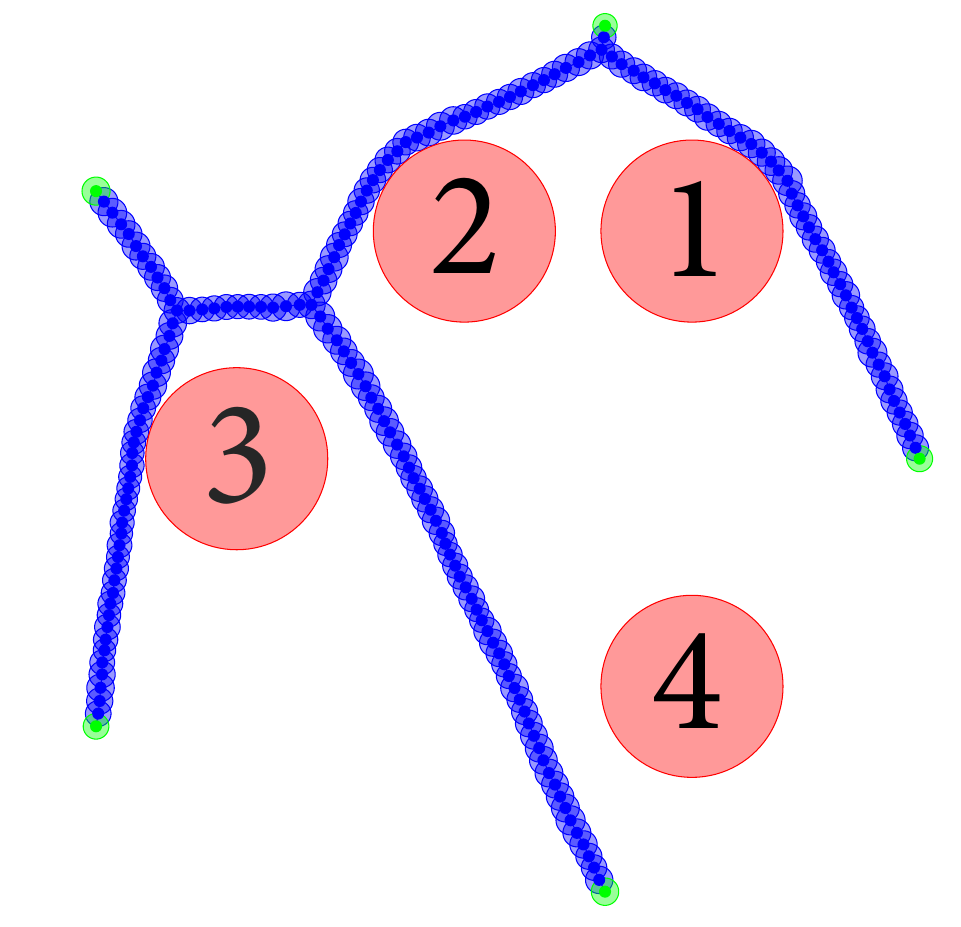}& \includegraphics[width=0.09\textwidth]{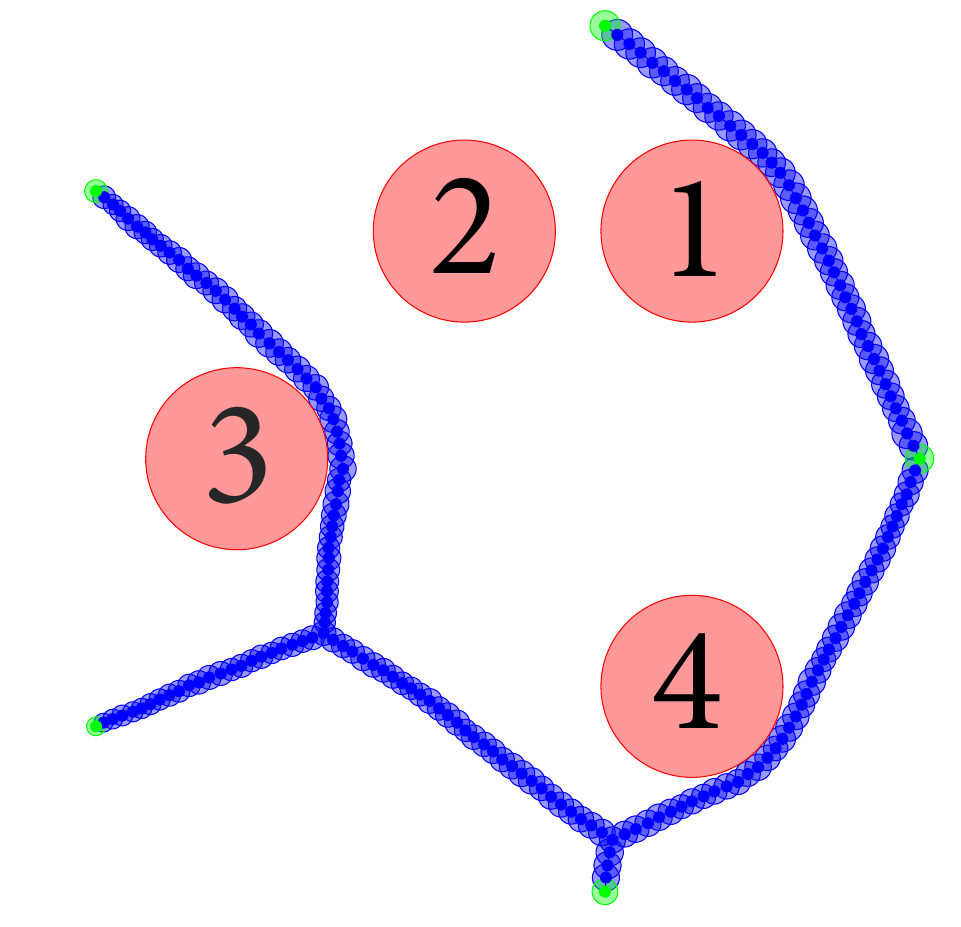}&\includegraphics[width=0.09\textwidth]{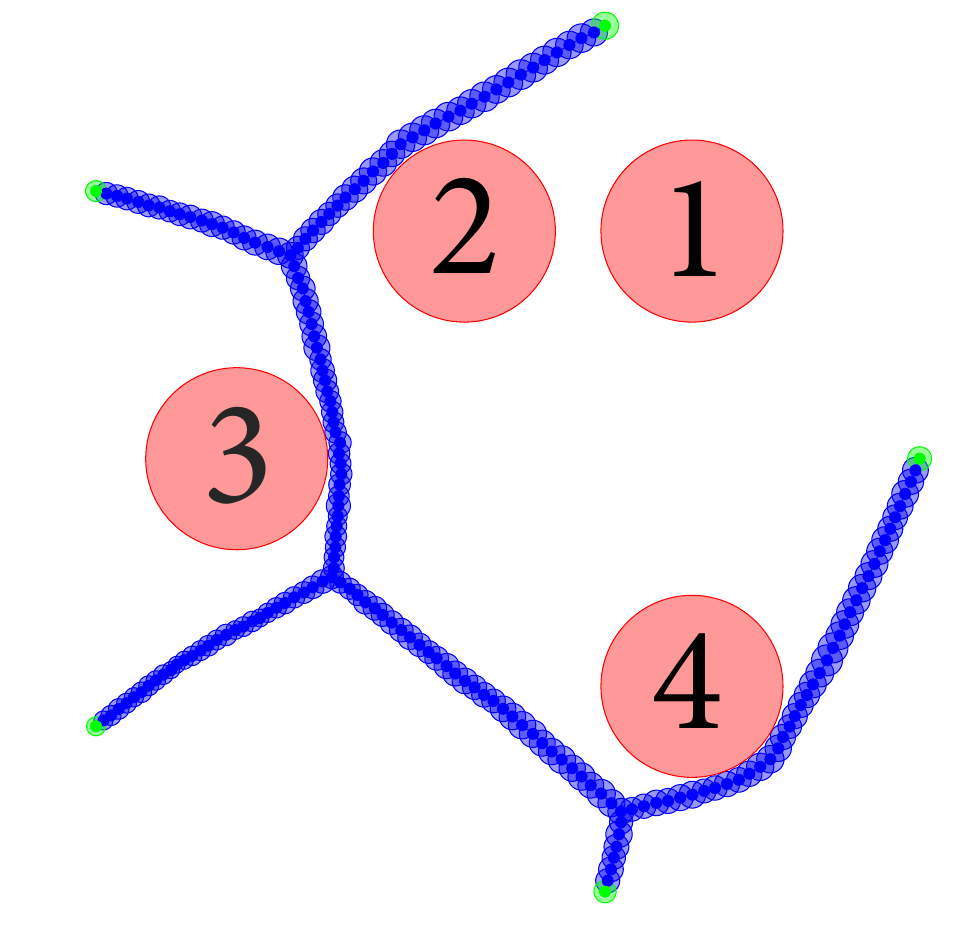}&\includegraphics[width=0.09\textwidth]{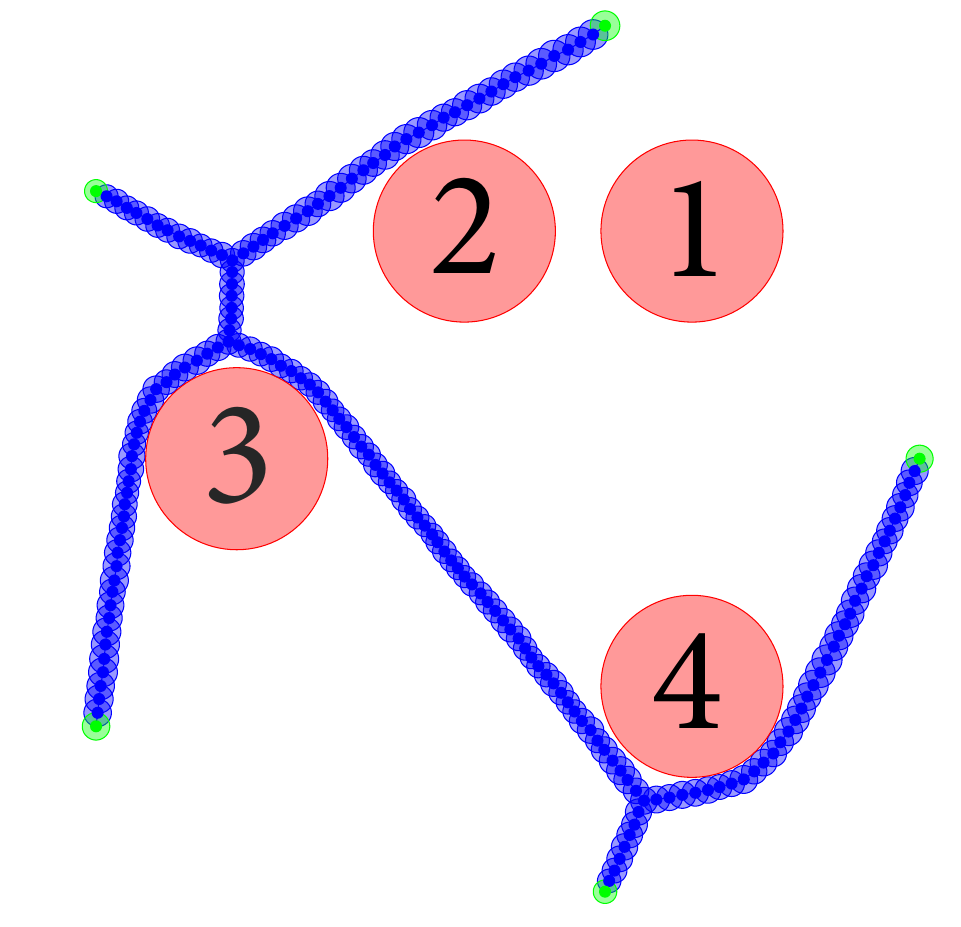}
                &\includegraphics[width=0.09\textwidth]{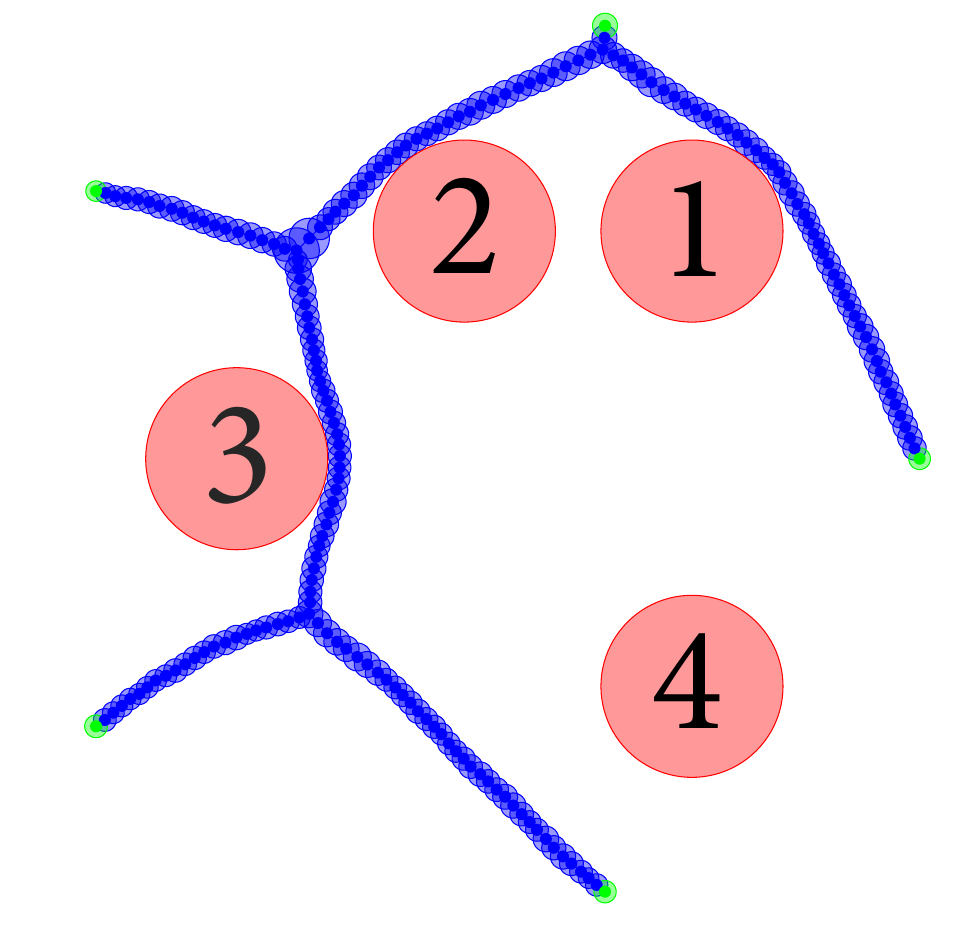}\\
                 
                 \boxed{$\{1;234\}$}&\includegraphics[width=0.09\textwidth]{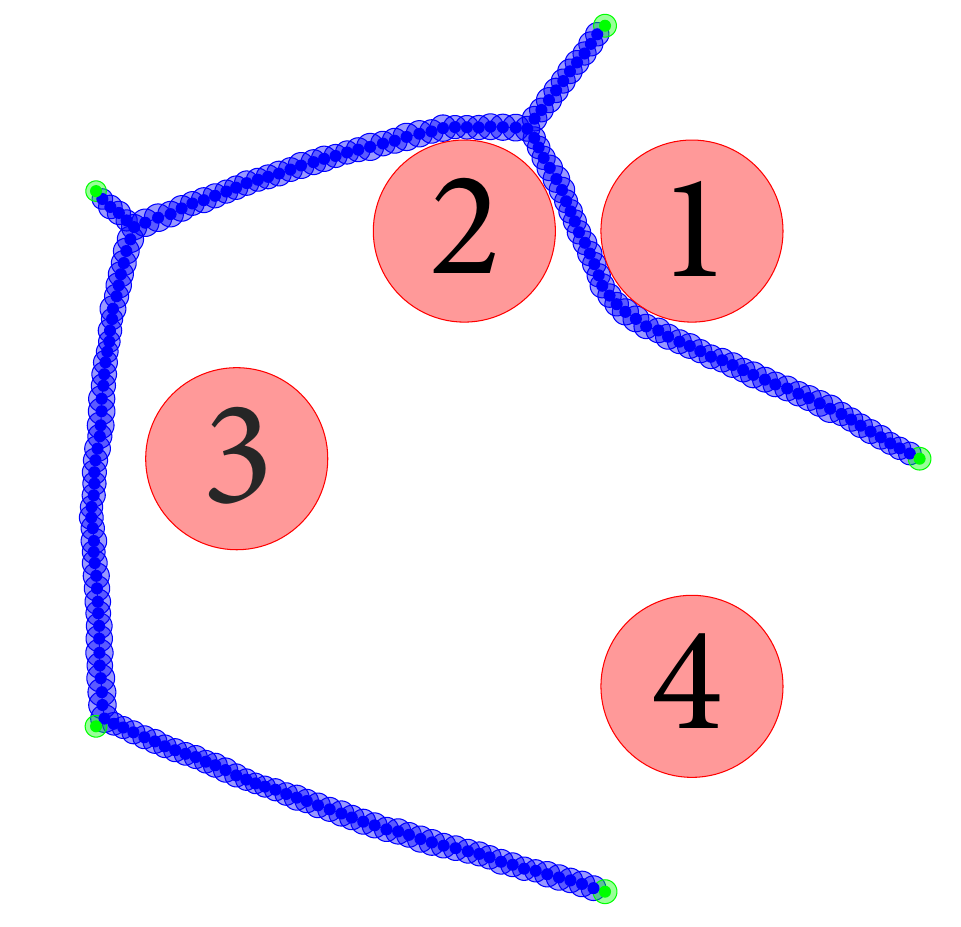}&\includegraphics[width=0.09\textwidth]{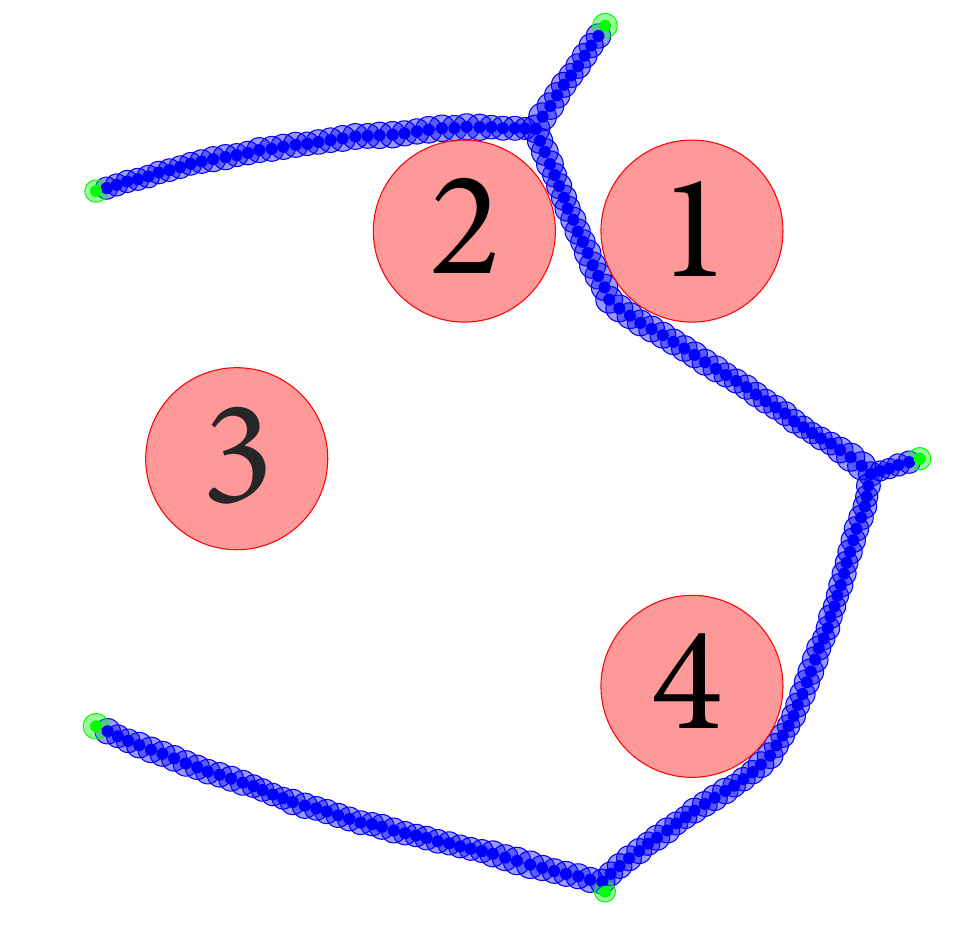}&\includegraphics[width=0.09\textwidth]{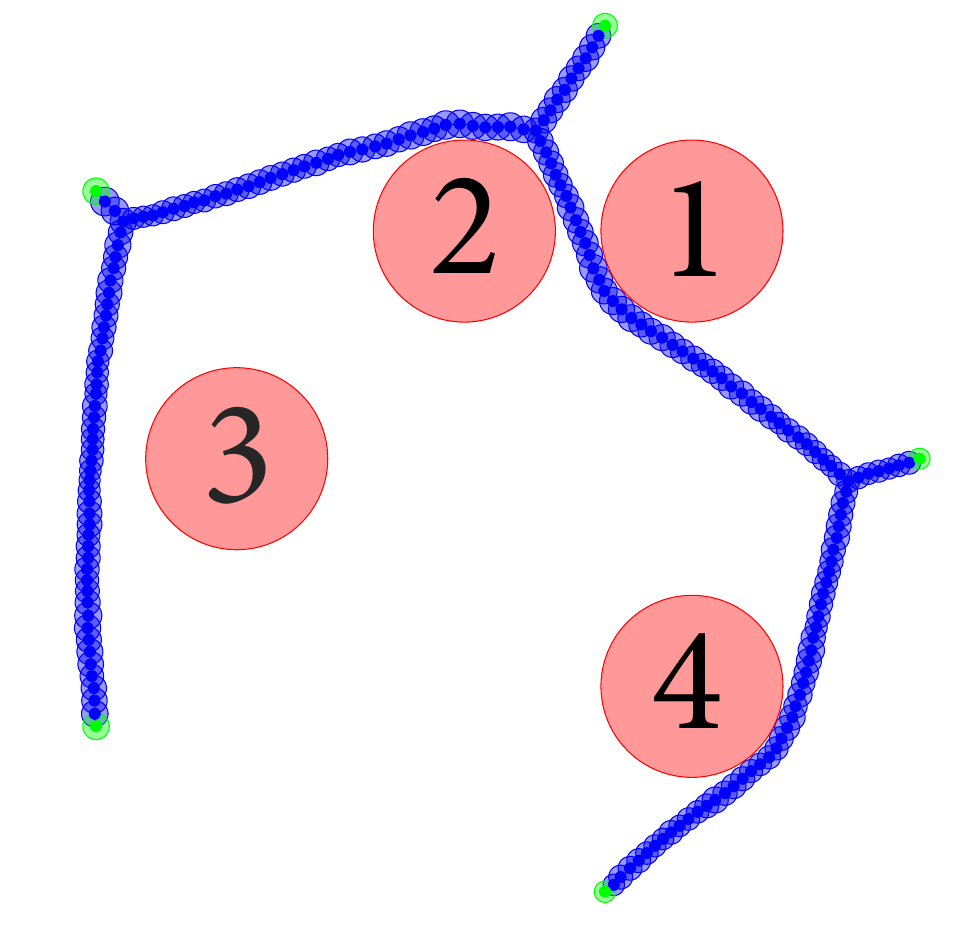}& &\boxed{$\{134;2\}$}&\includegraphics[width=0.09\textwidth]{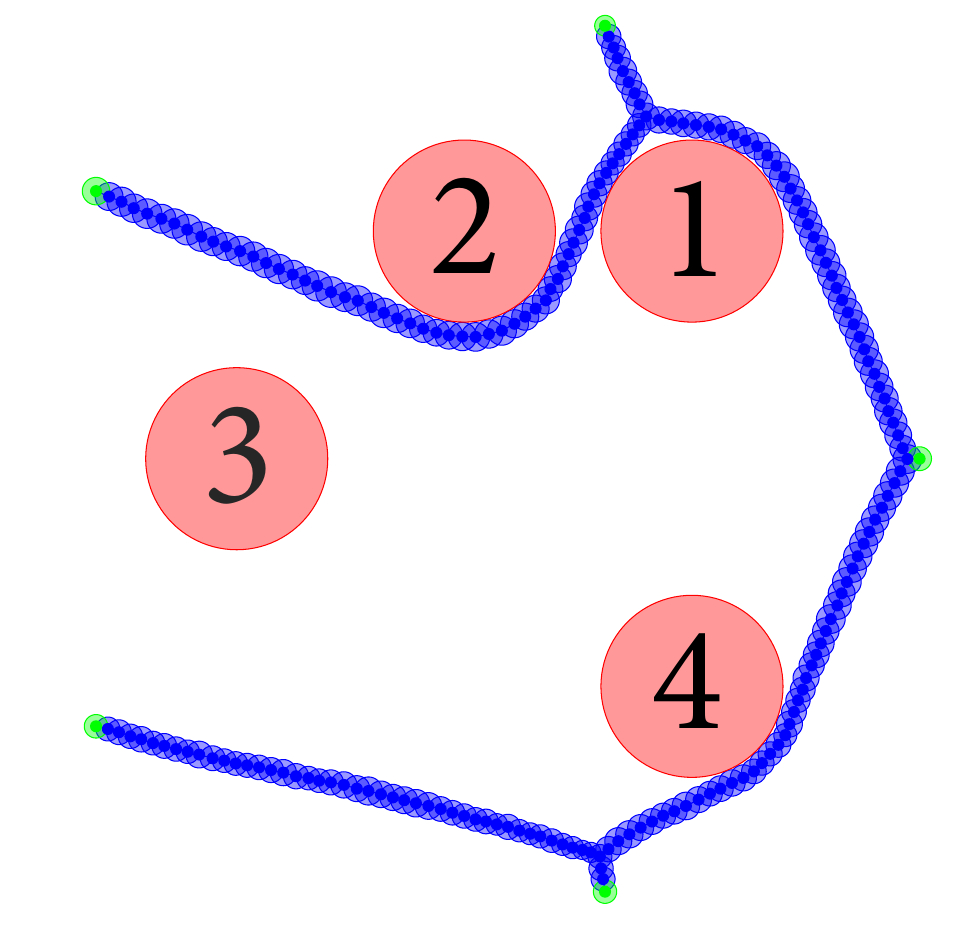}&\includegraphics[width=0.09\textwidth]{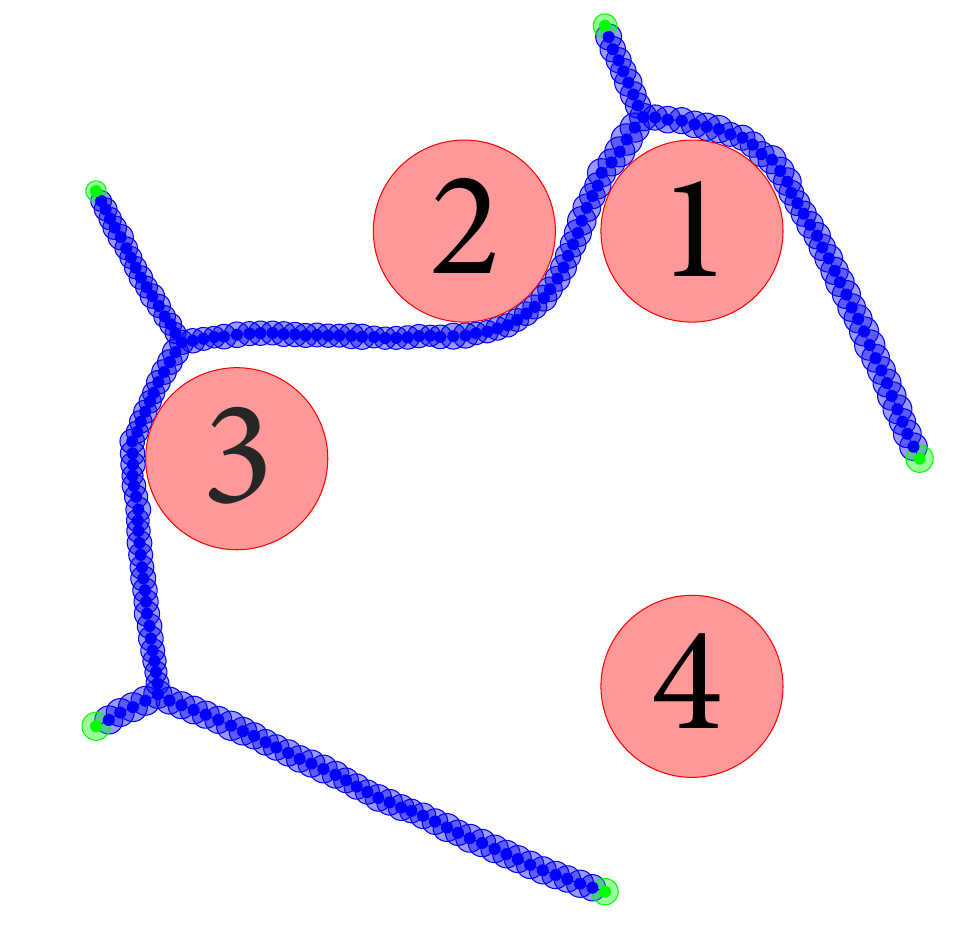}\\
                  \boxed{$\{12;34\}$}&\boxed{\begin{overpic}[width=0.09\textwidth]{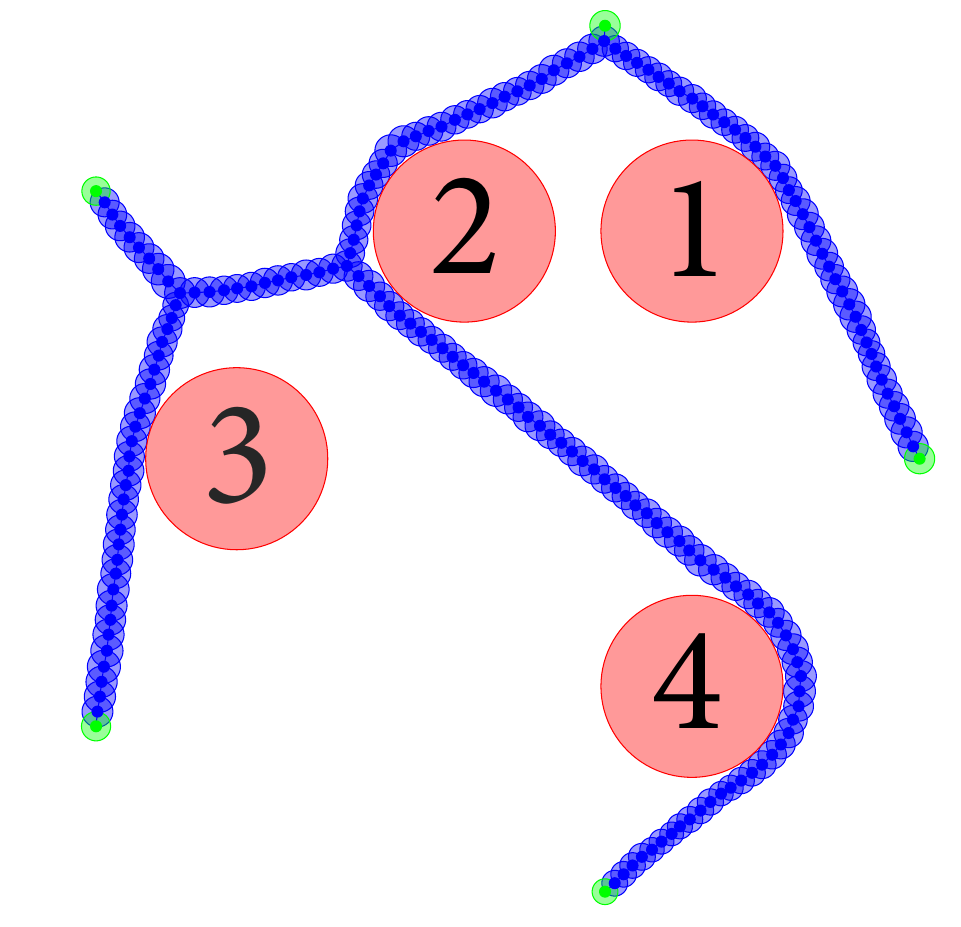}\put(-5,85){\textcolor{red}{Worst}}\end{overpic}}&\includegraphics[width=0.09\textwidth]{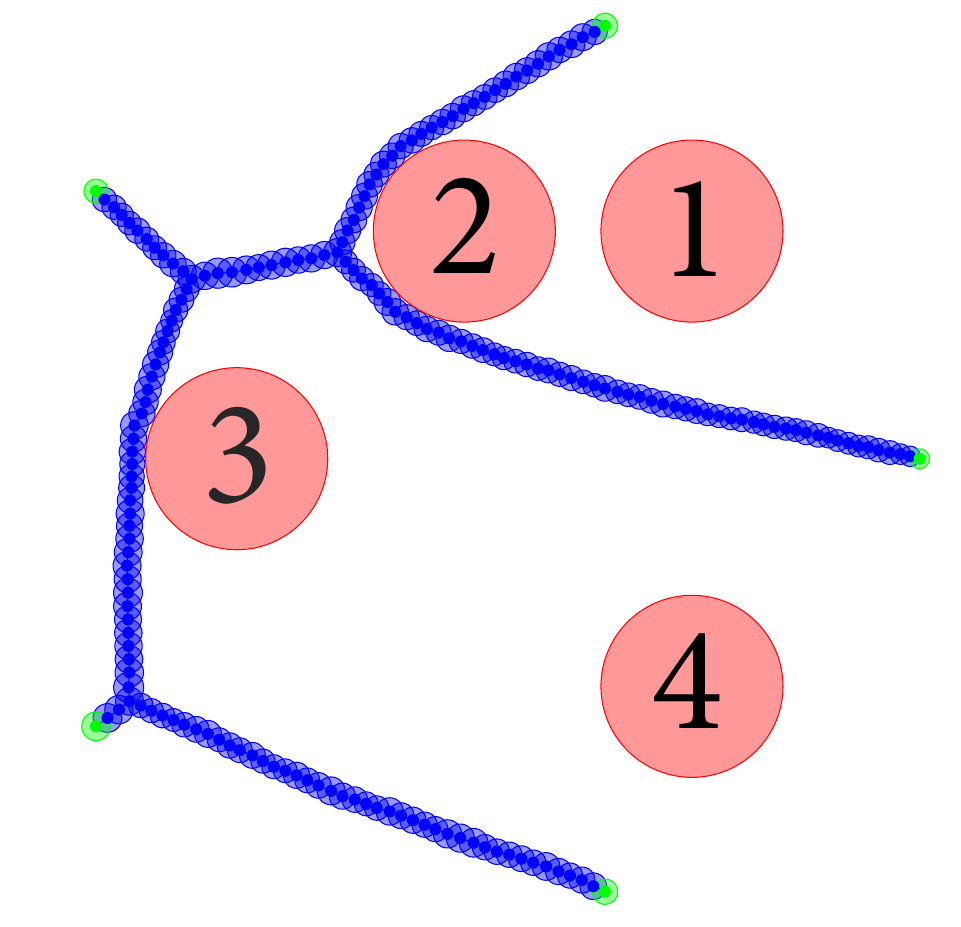}&\includegraphics[width=0.09\textwidth]{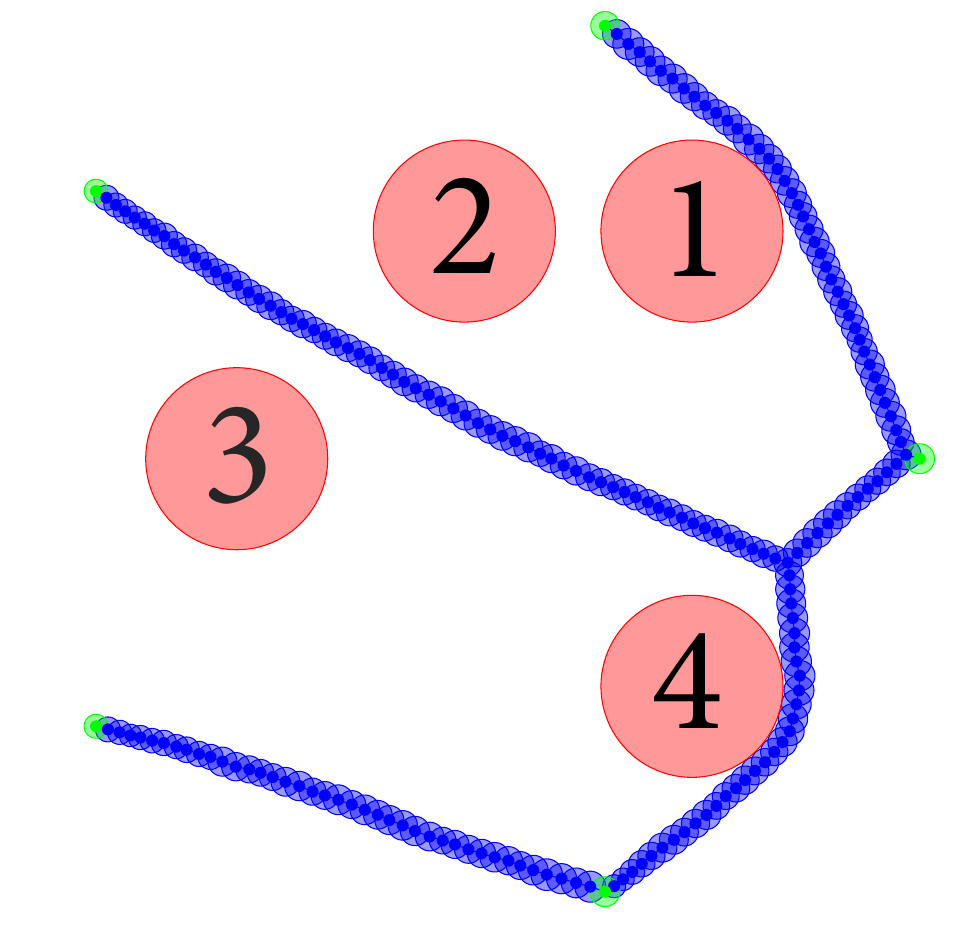}&\includegraphics[width=0.09\textwidth]{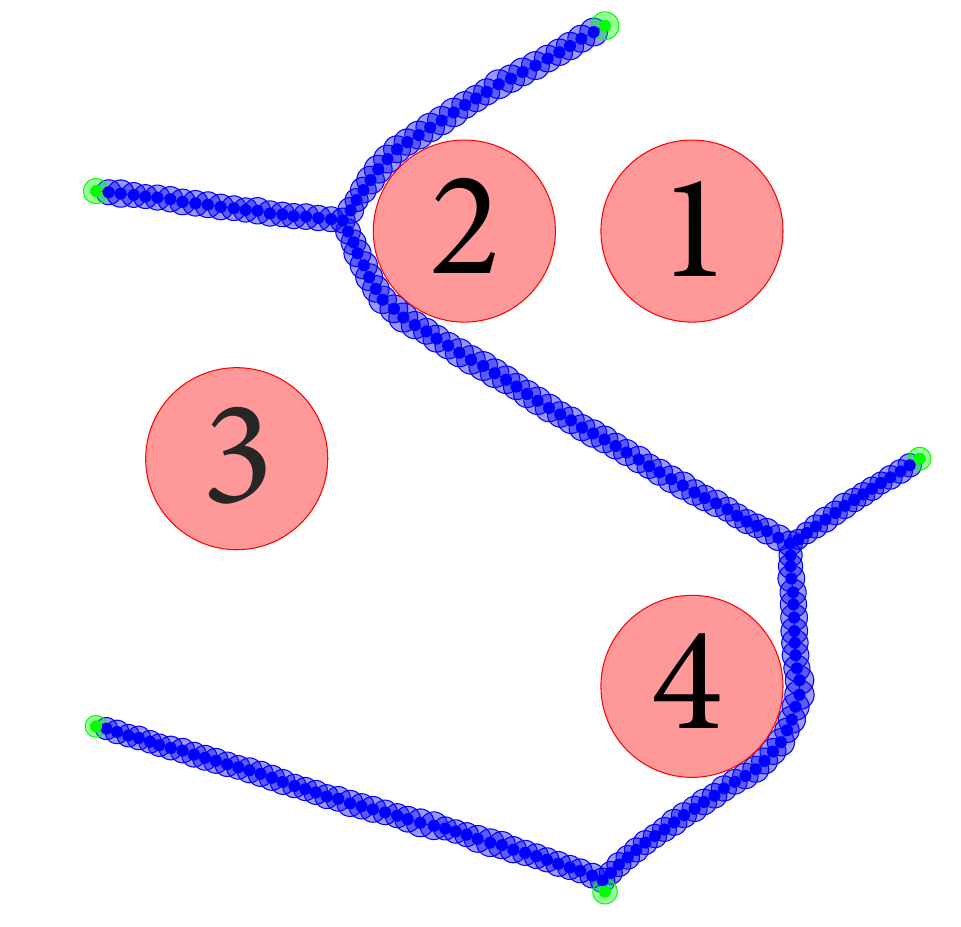}&  \includegraphics[width=0.09\textwidth]{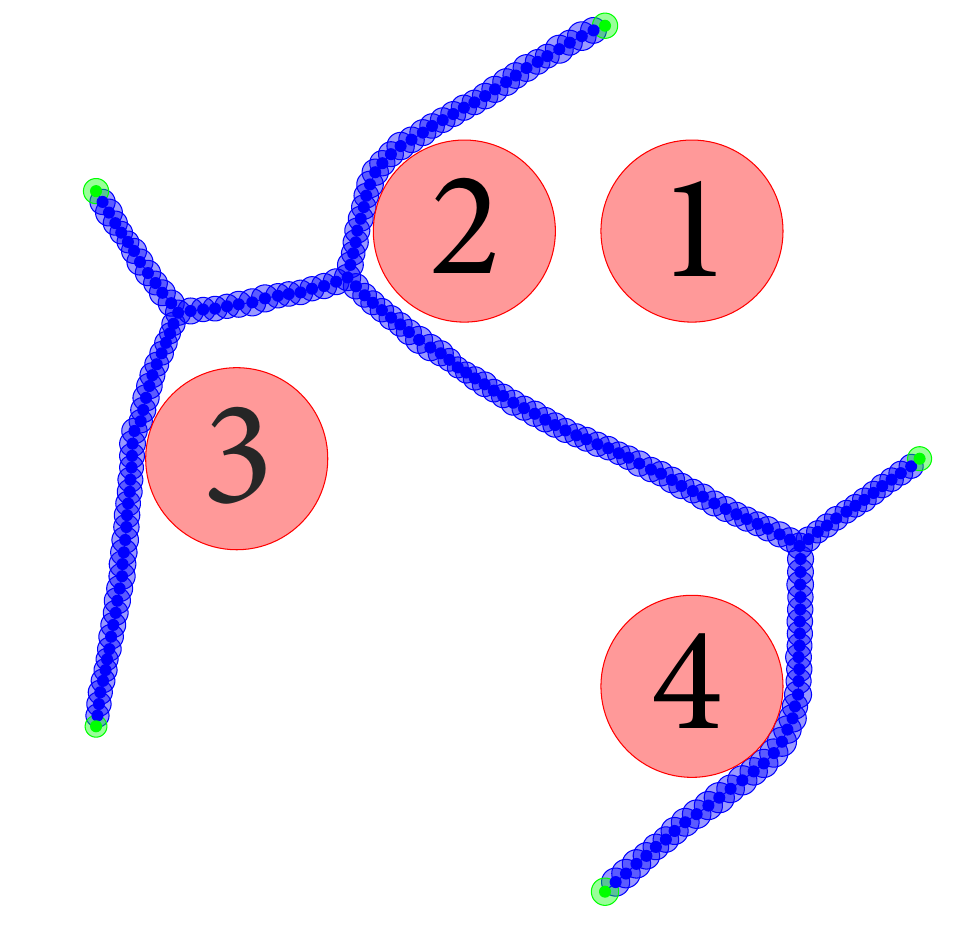}   &\includegraphics[width=0.09\textwidth]{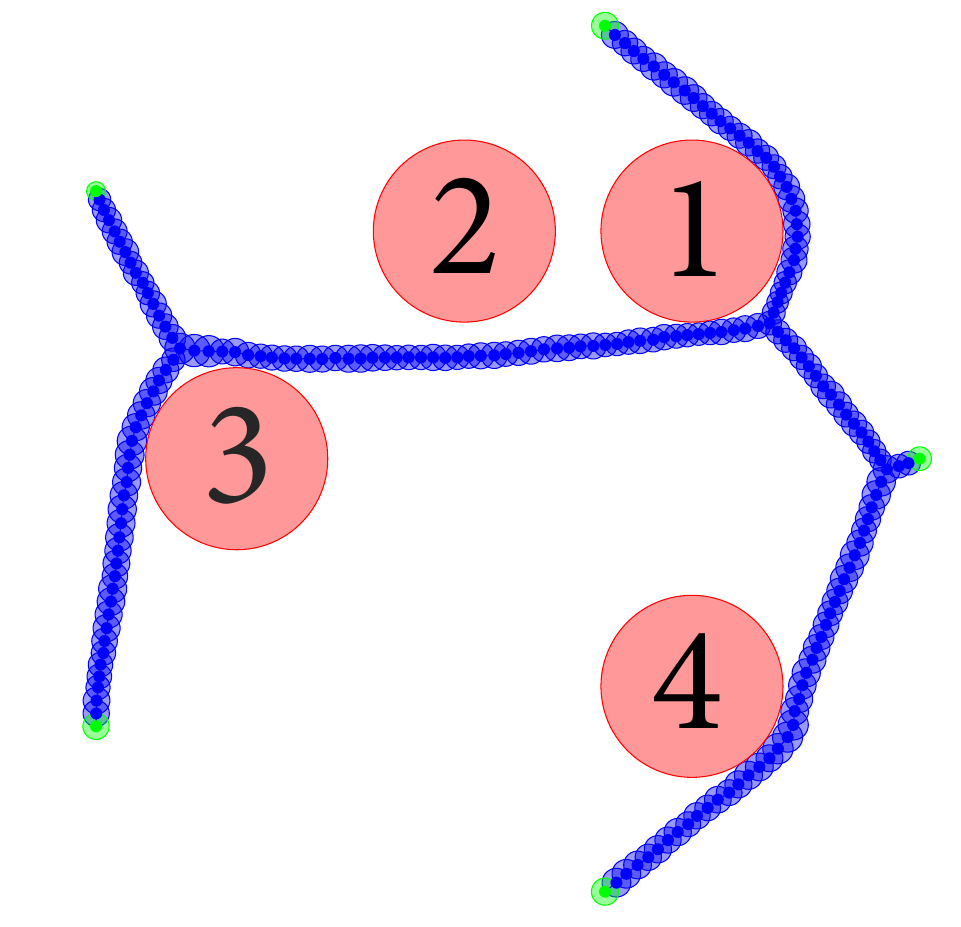}\\
                \boxed{$\{14;23\}$}&
                \includegraphics[width=0.09\textwidth]{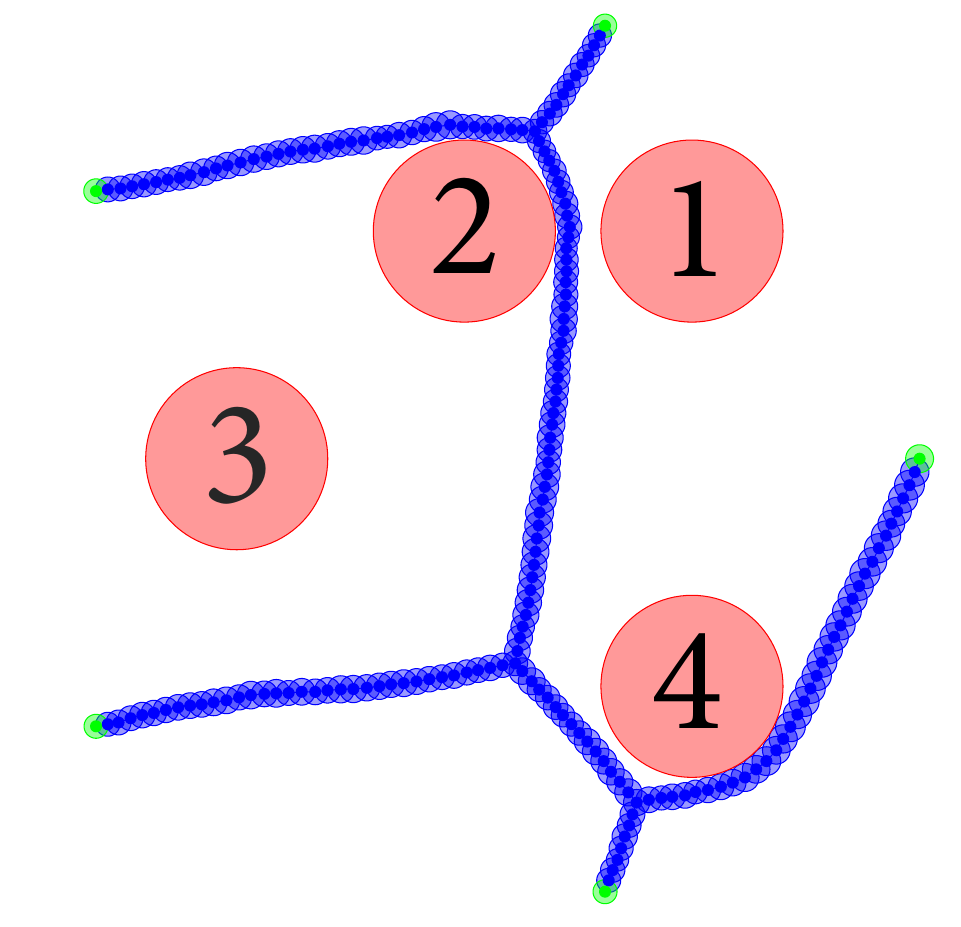}&\includegraphics[width=0.09\textwidth]{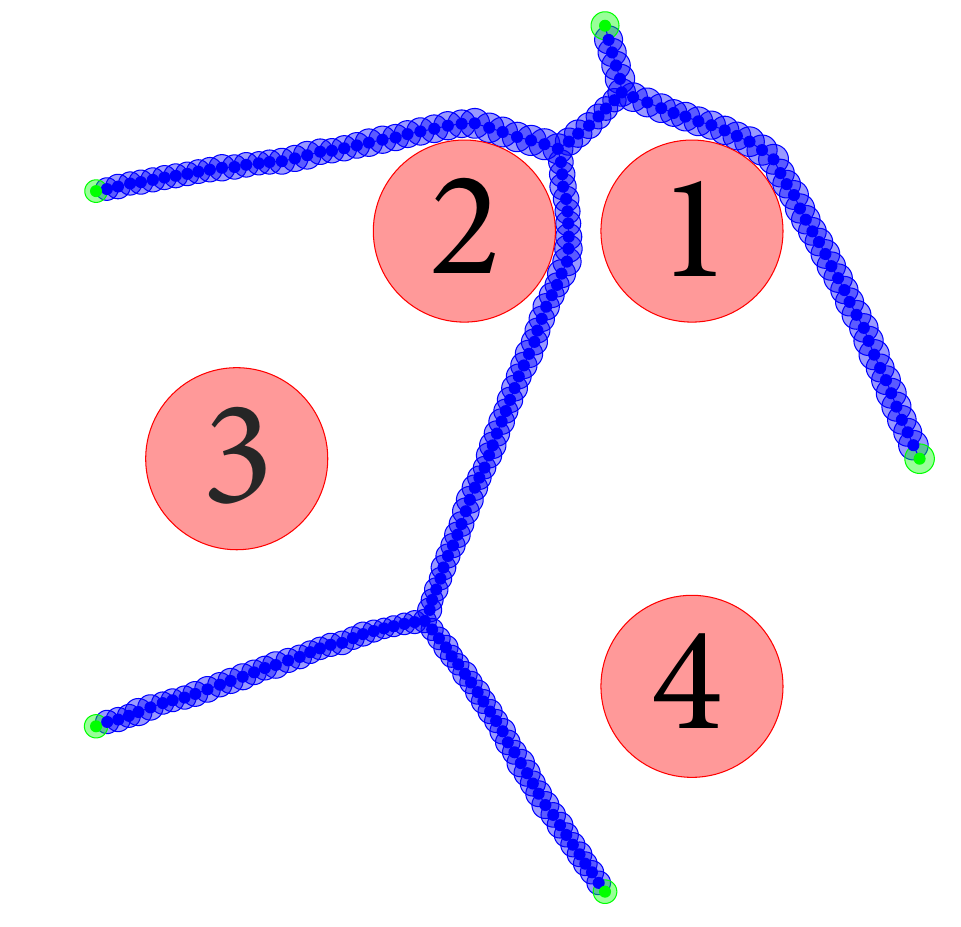}&\includegraphics[width=0.09\textwidth]{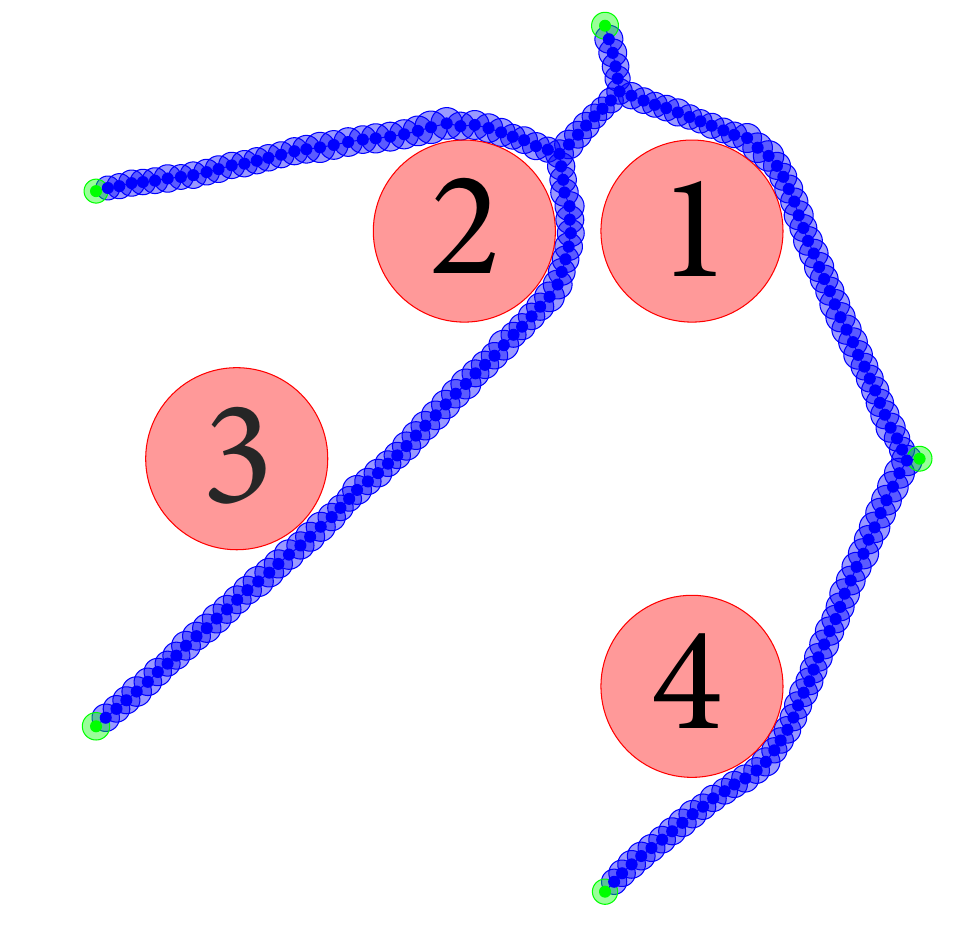}&\includegraphics[width=0.09\textwidth]{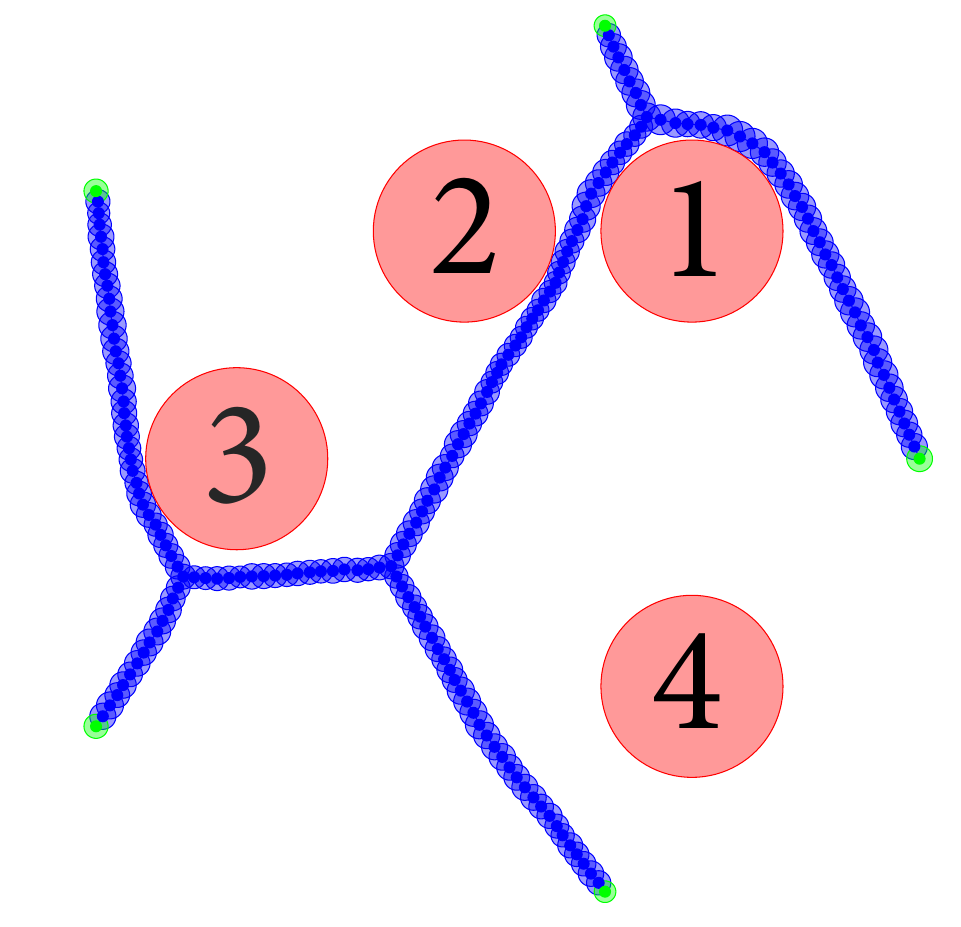}&\includegraphics[width=0.09\textwidth]{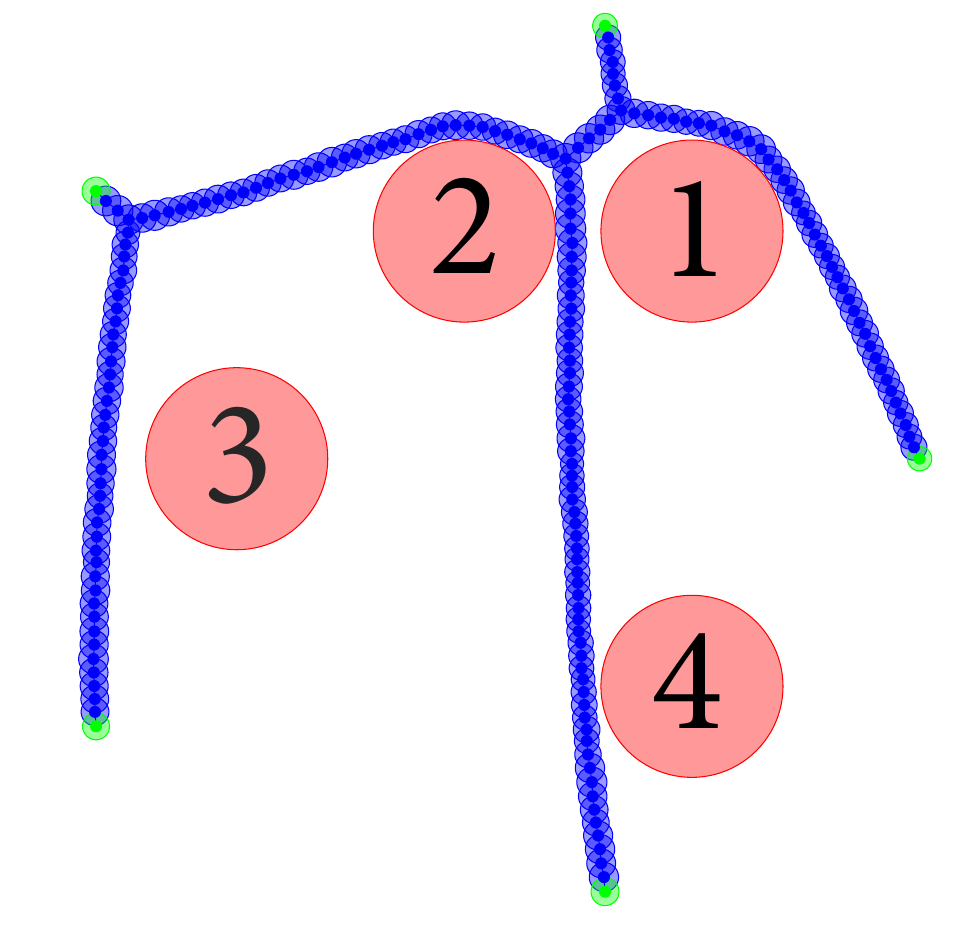}&\includegraphics[width=0.09\textwidth]{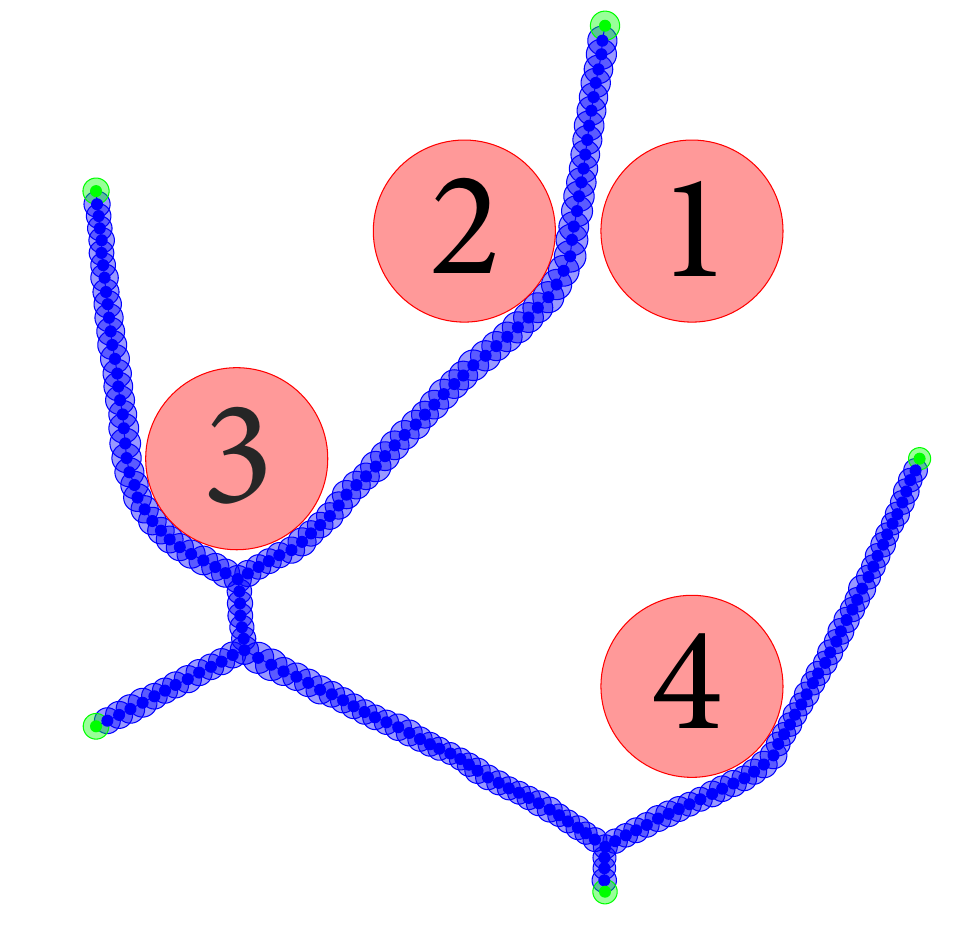}\\
                \boxed{$\{12;3;4\}$}&\boxed{\begin{overpic}[width=0.09\textwidth]{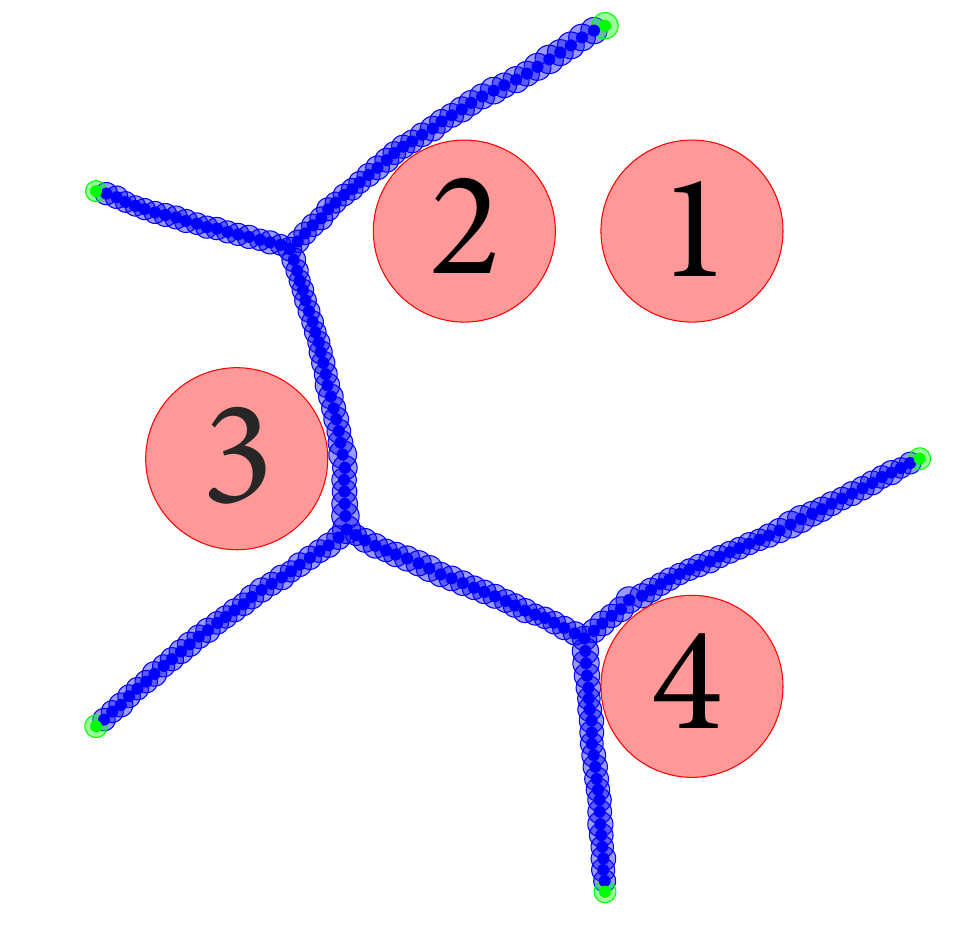}\put(0,83){\textcolor{black}{$1$\textsuperscript{st}}}\end{overpic}}&\includegraphics[width=0.09\textwidth]{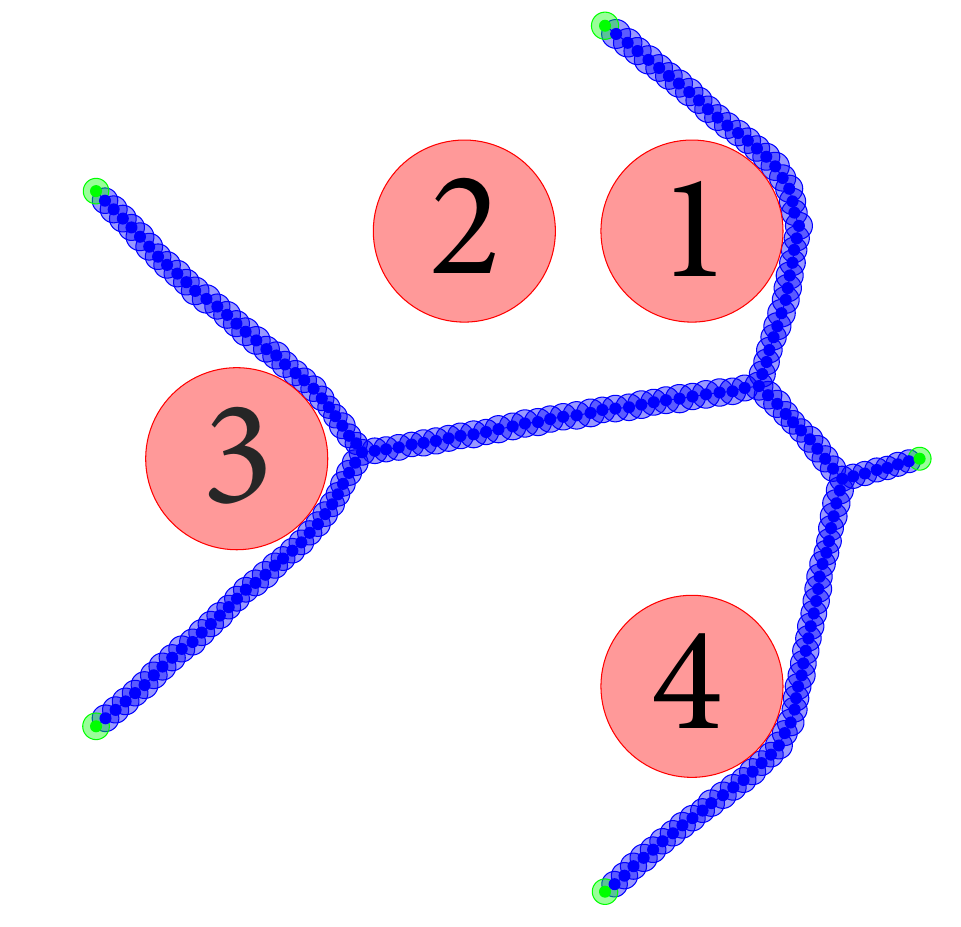}&\includegraphics[width=0.09\textwidth]{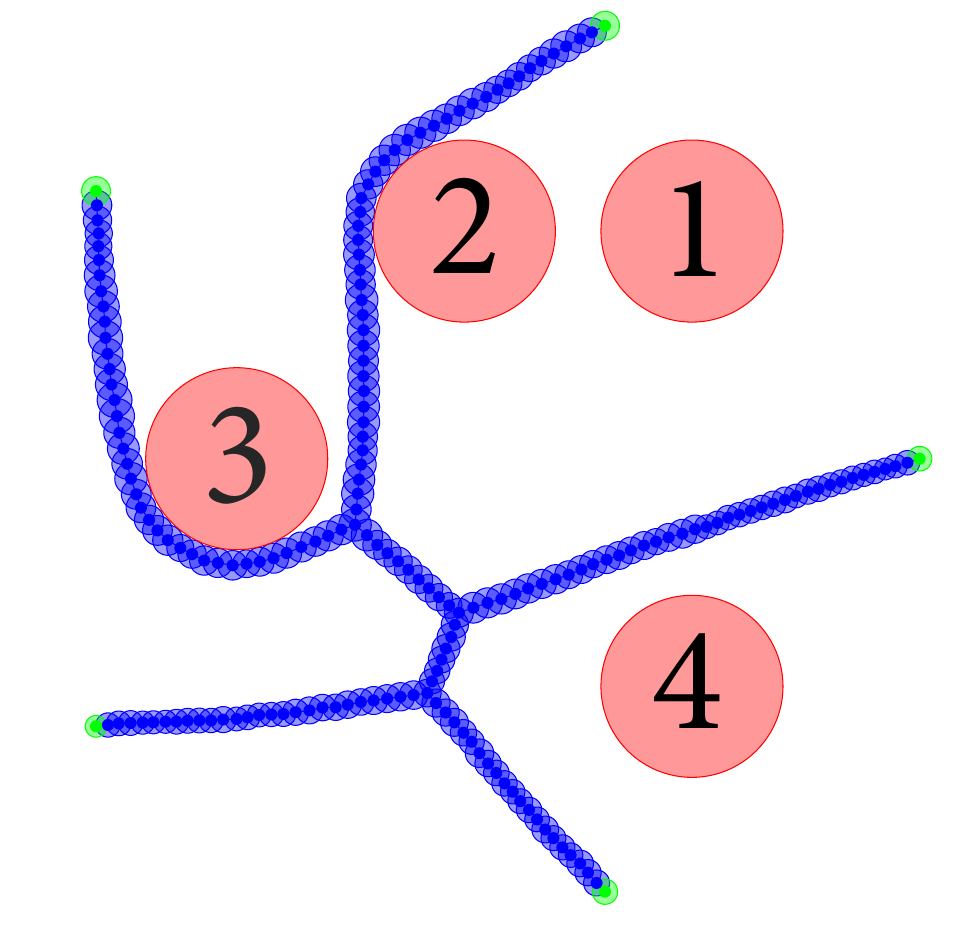}&\includegraphics[width=0.09\textwidth]{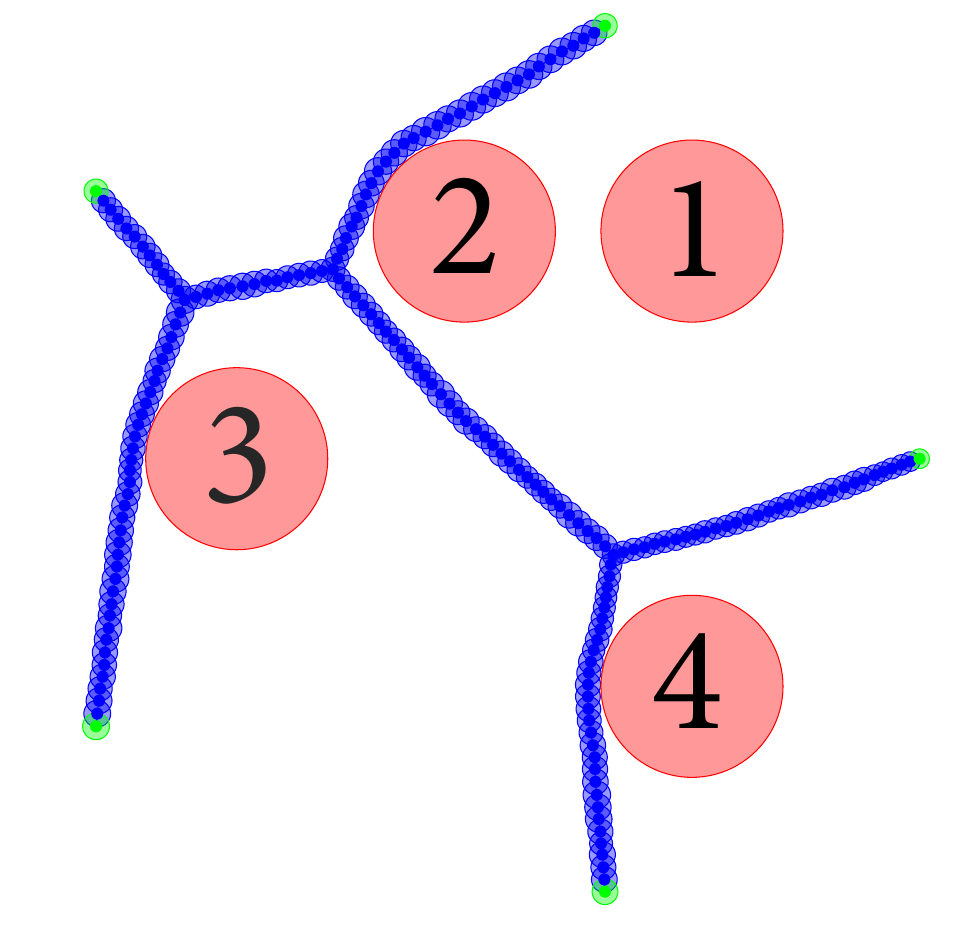}&\includegraphics[width=0.09\textwidth]{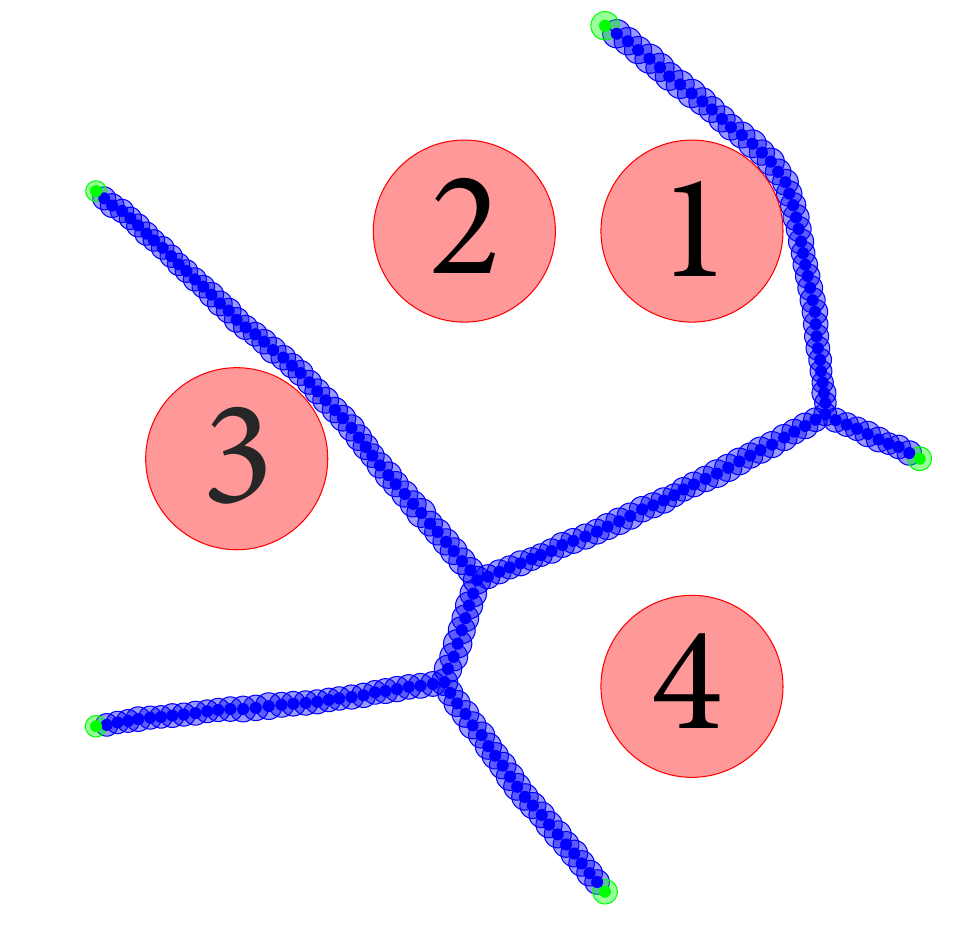}&\includegraphics[width=0.09\textwidth]{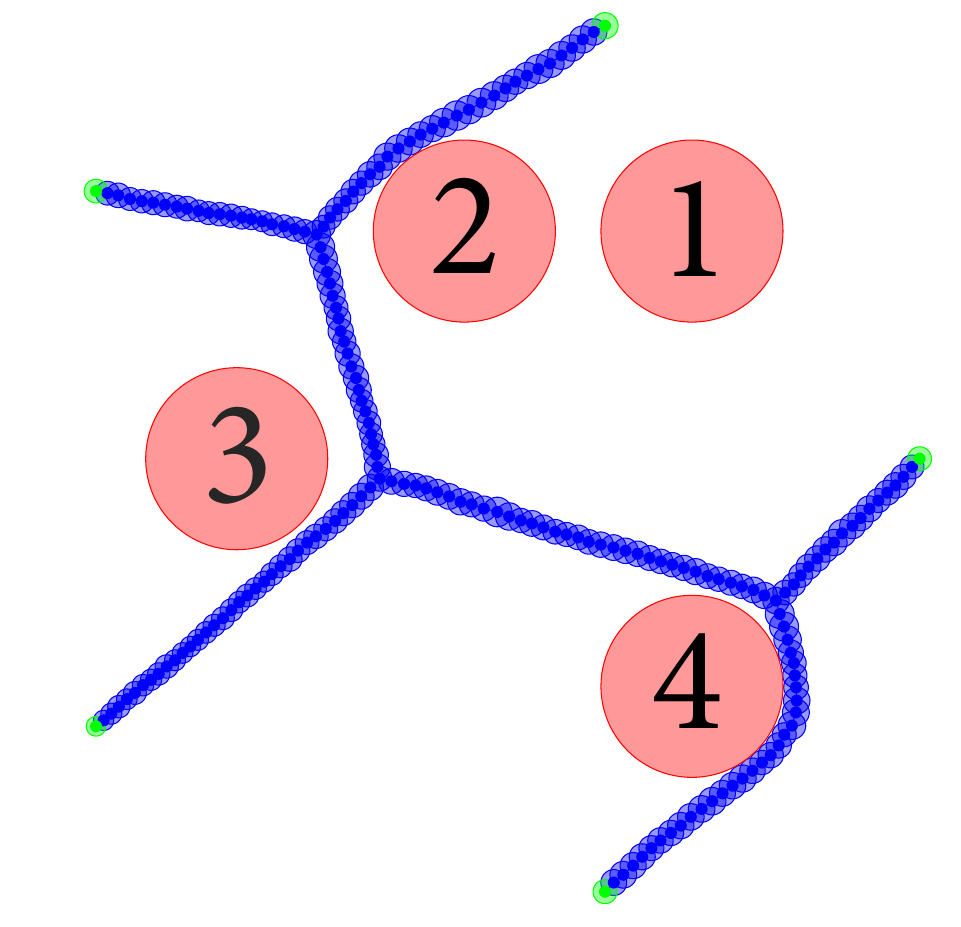}&\includegraphics[width=0.09\textwidth]{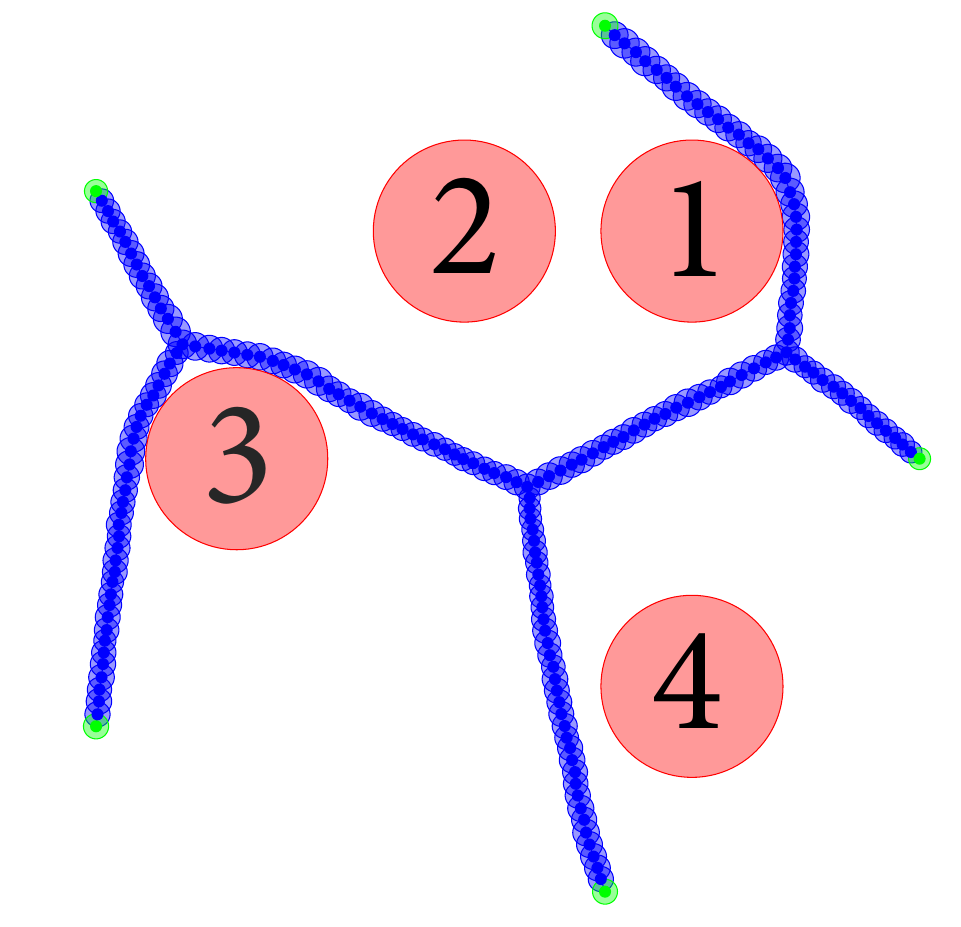}\\
               \boxed{$\{1;23;4\}$}&\boxed{\begin{overpic}[width=0.09\textwidth]{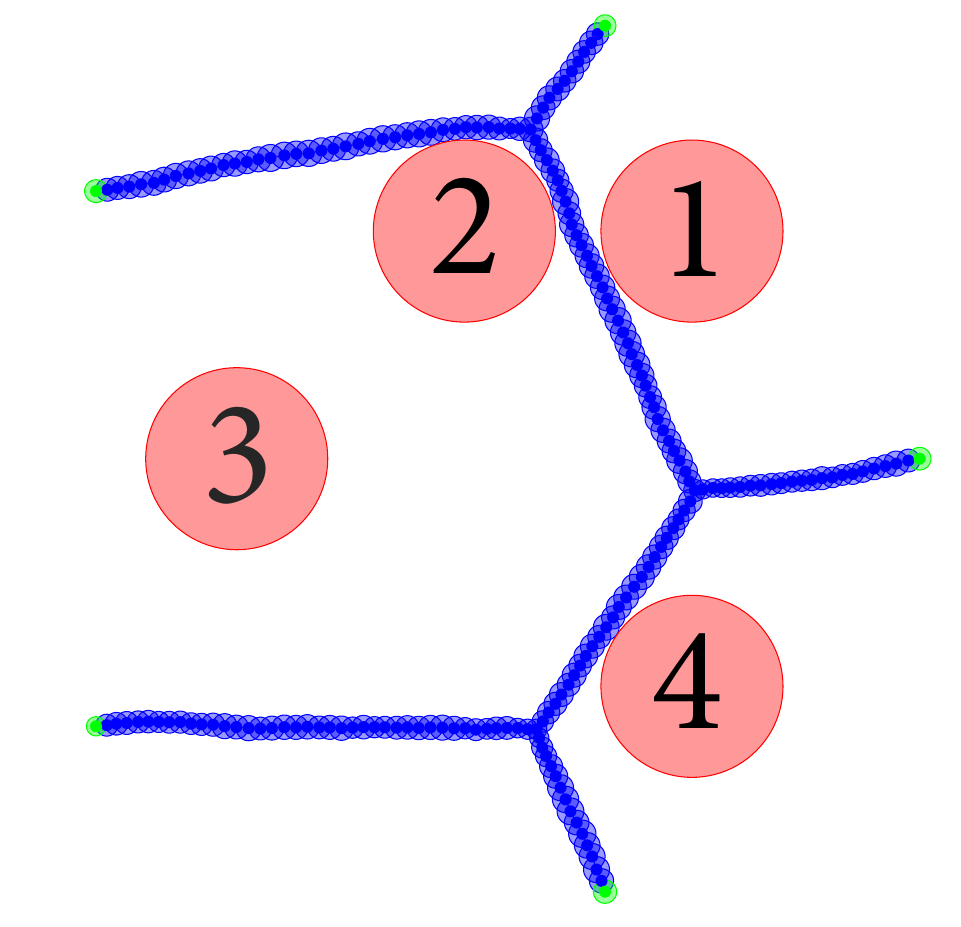}\put(0,83){\textcolor{black}{$2$\textsuperscript{nd}}}\end{overpic}}& 
                \includegraphics[width=0.09\textwidth]{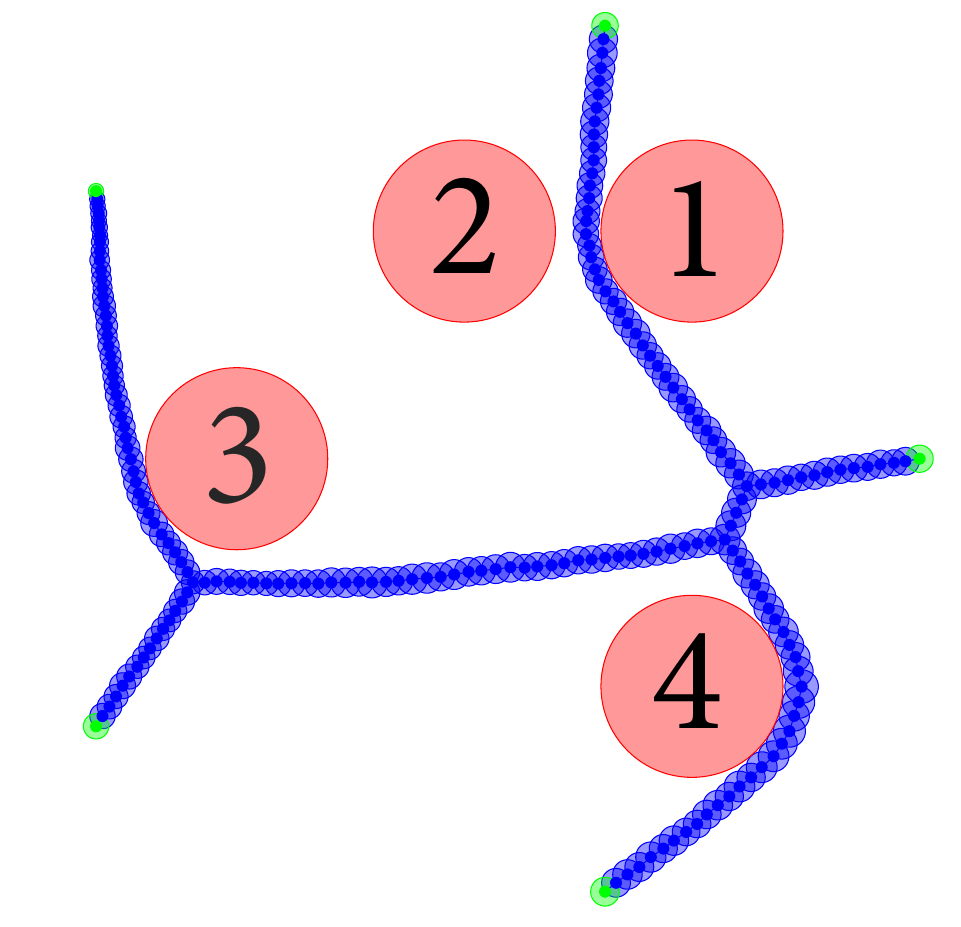}&
                \includegraphics[width=0.09\textwidth]{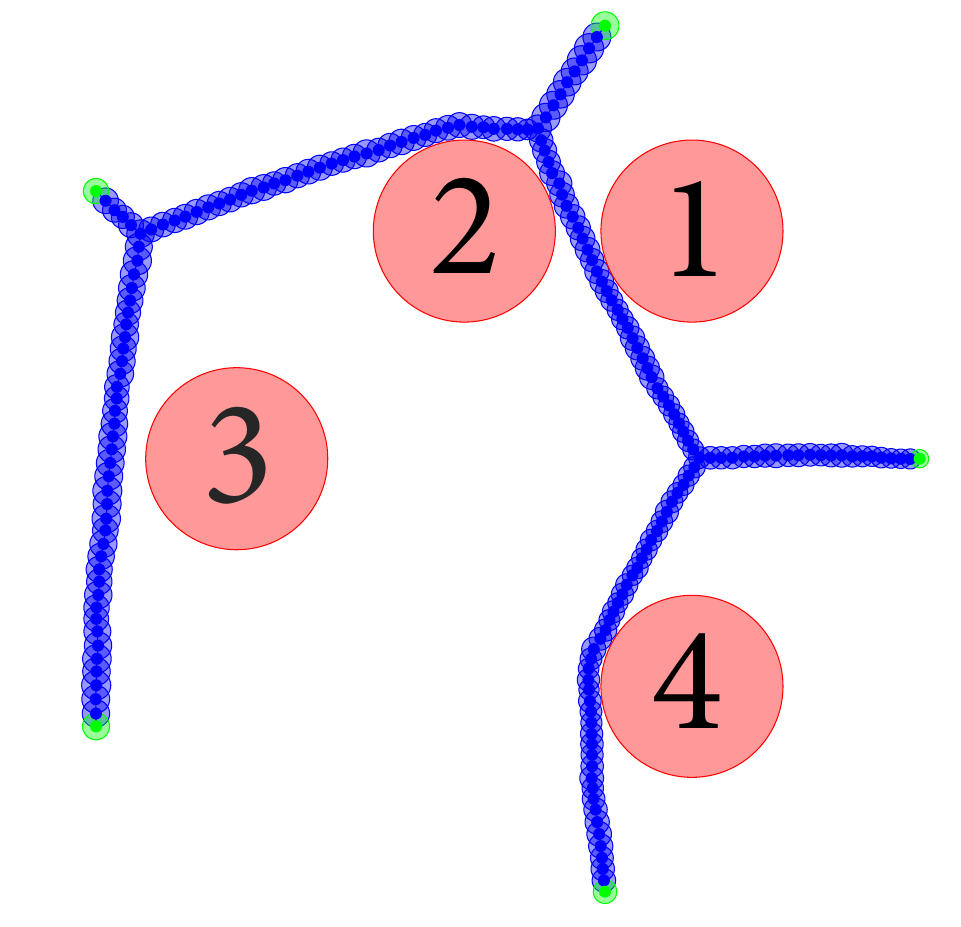}&
                \includegraphics[width=0.09\textwidth]{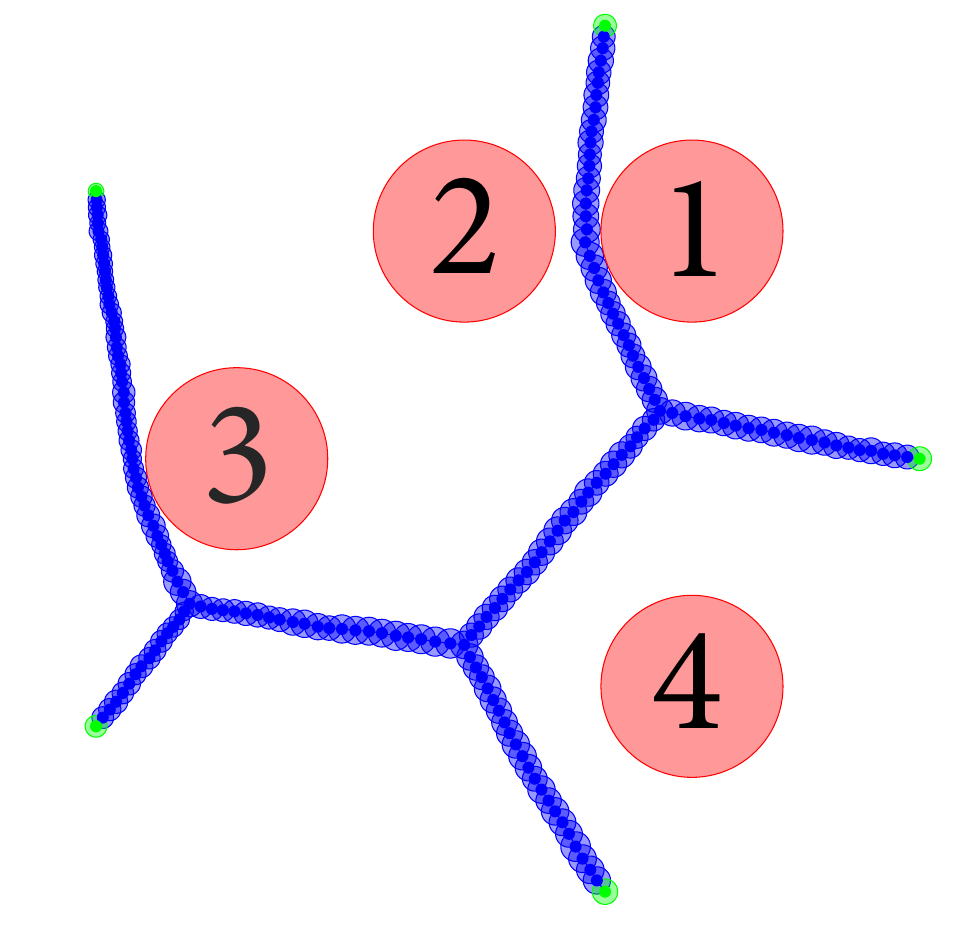}&
                \includegraphics[width=0.09\textwidth]{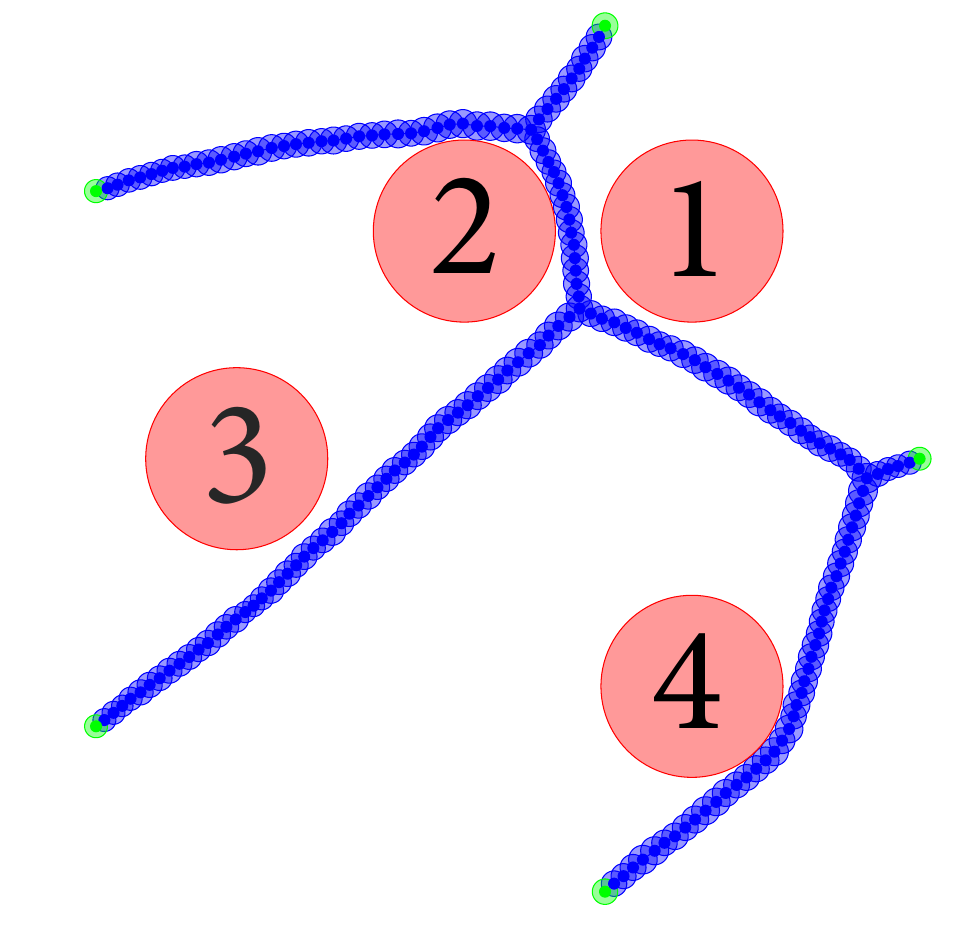}&&\boxed{$\{13;2;4\}$} &
                \includegraphics[width=0.09\textwidth]{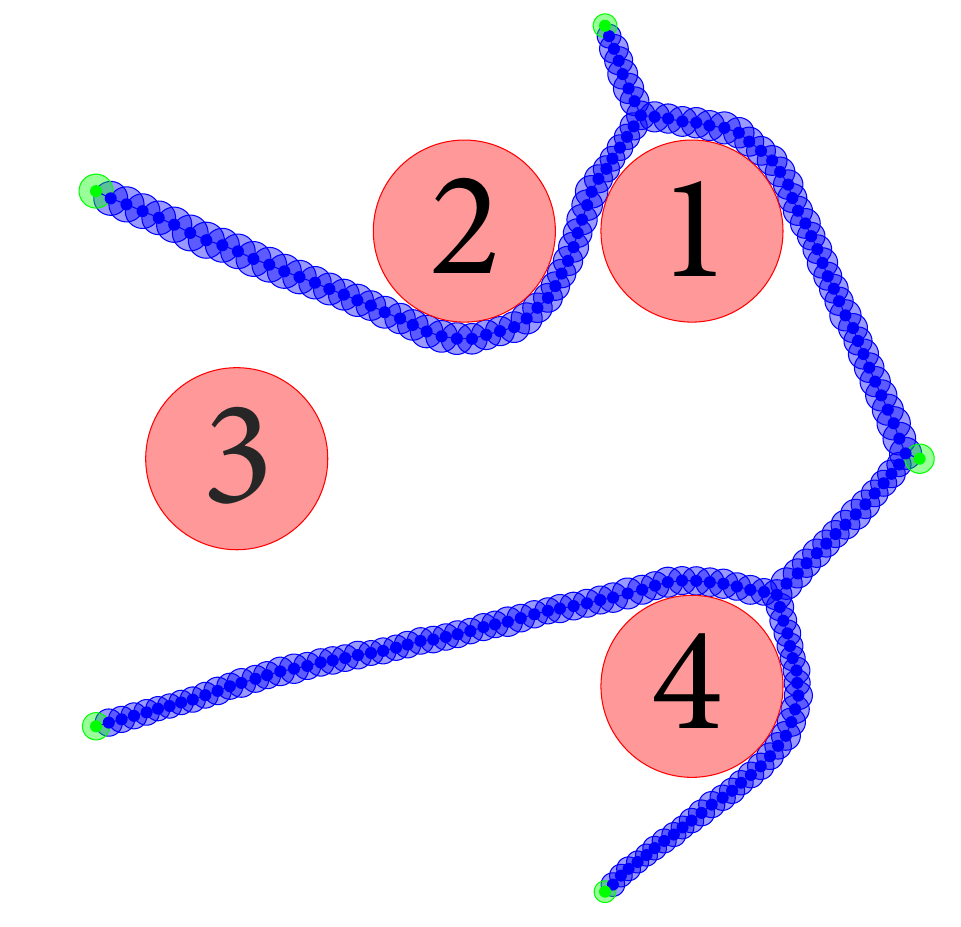}& \\
                
               \boxed{$\{1;2;34\}$}&\includegraphics[width=0.09\textwidth]{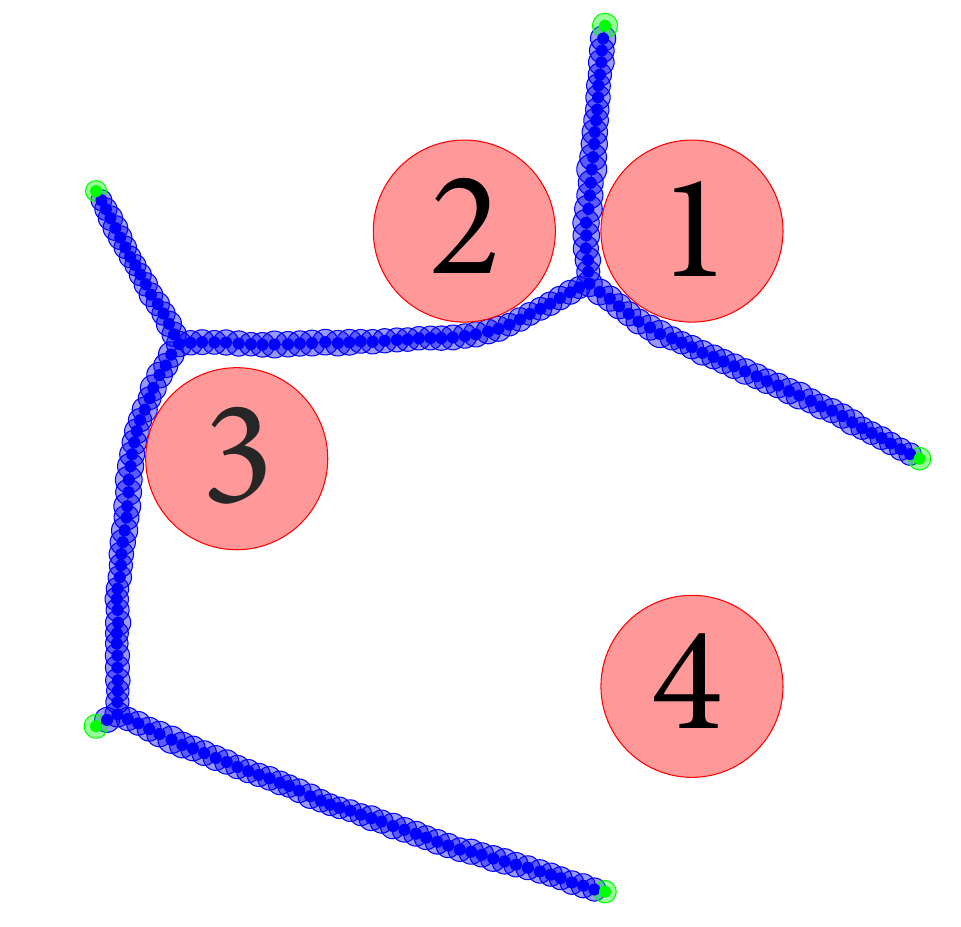}&\includegraphics[width=0.09\textwidth]{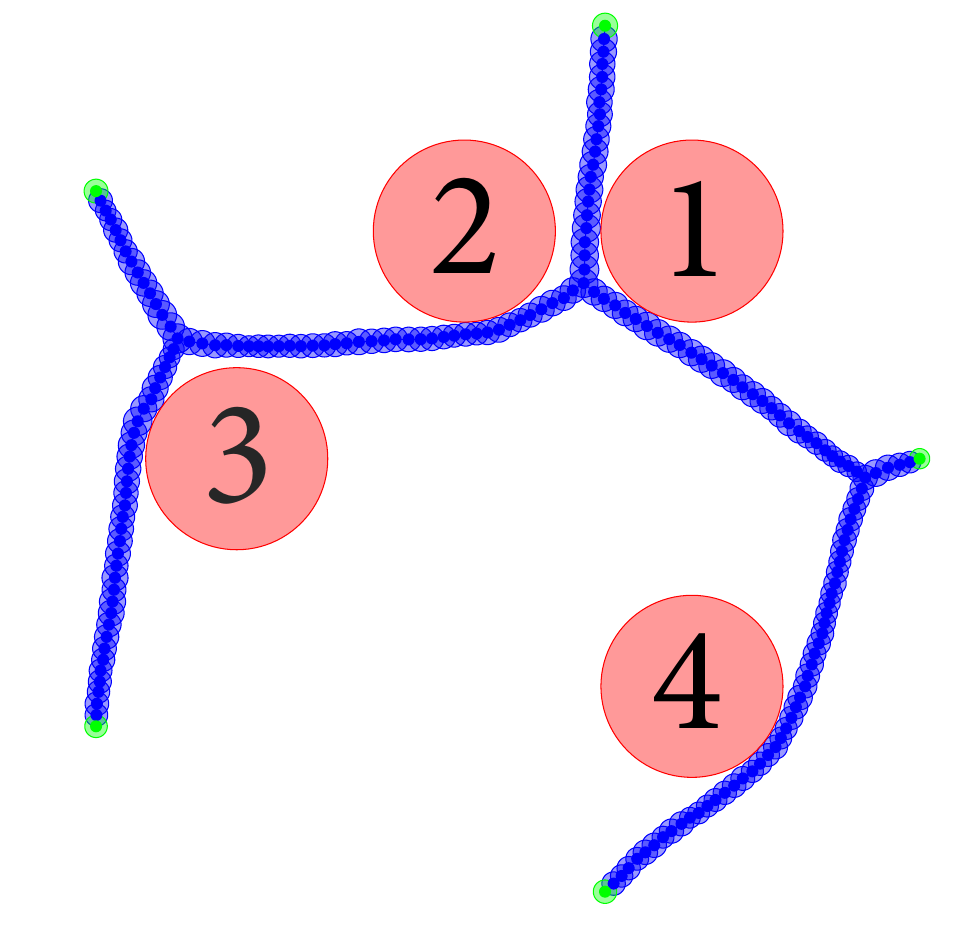}&
                \includegraphics[width=0.09\textwidth]{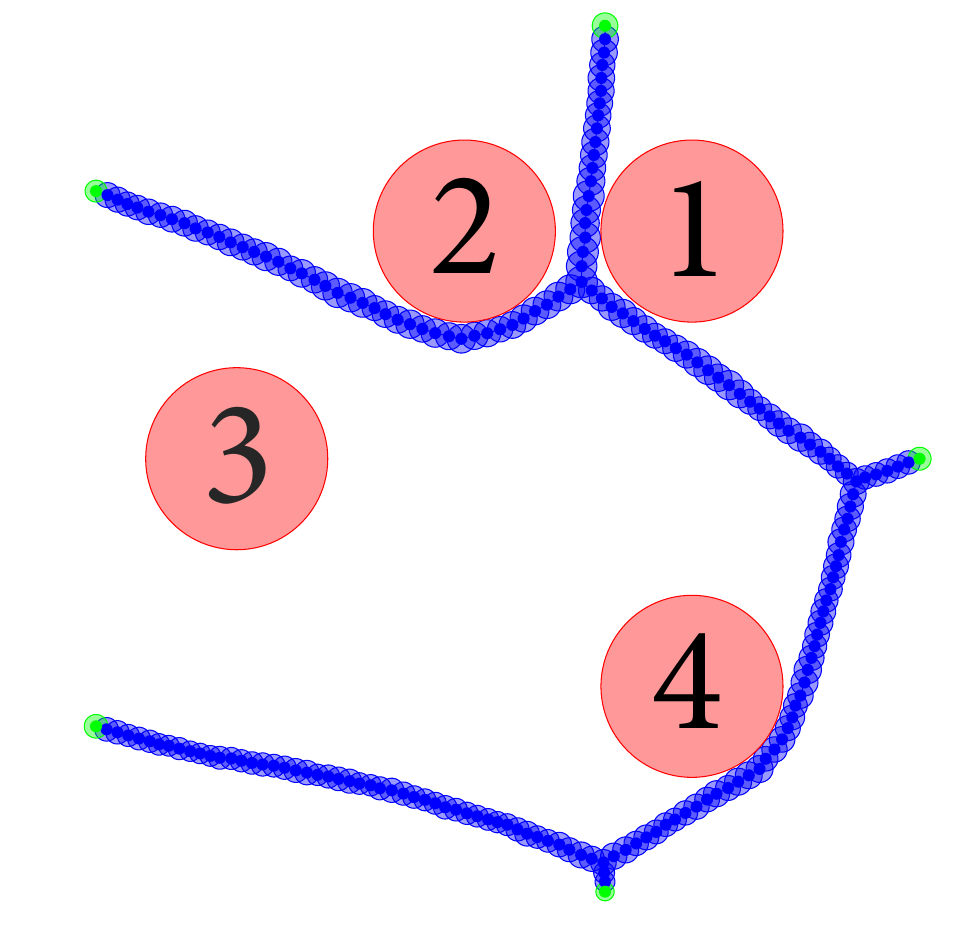}& &\color{red}{$\{13;24\}$} & \begin{overpic}[width=0.09\textwidth]{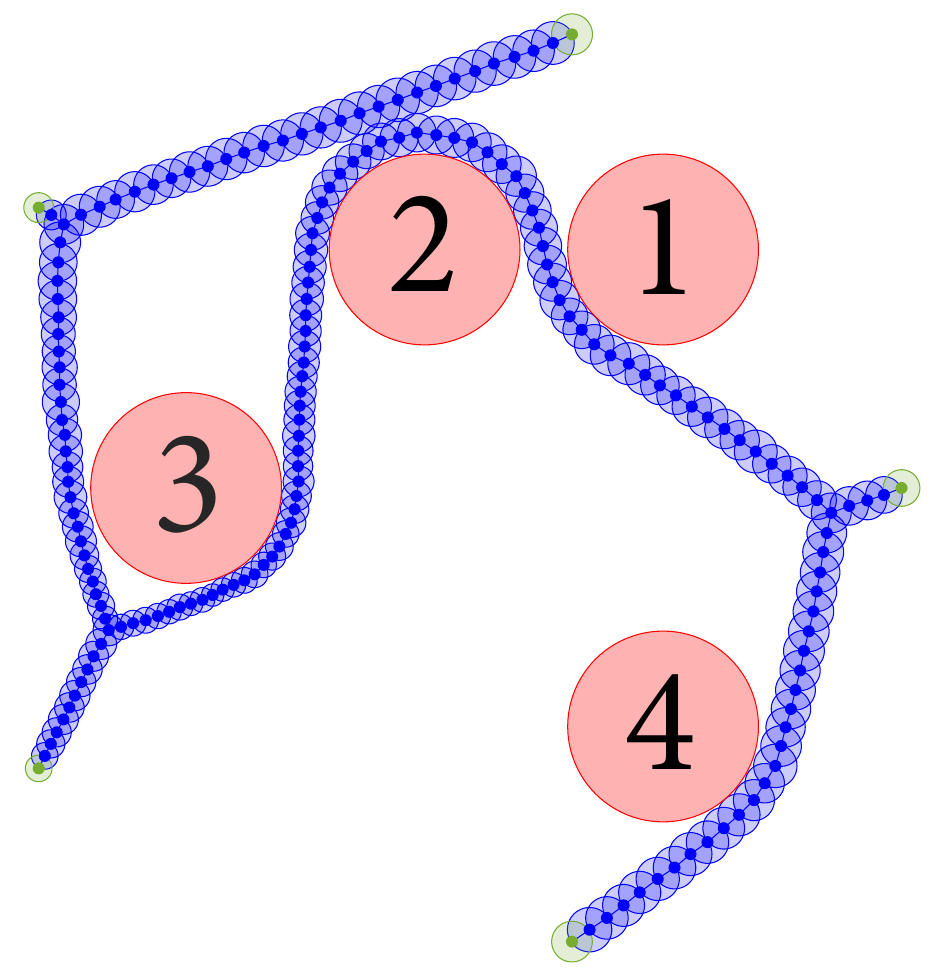}\put(90,50){\textcolor{black}{$\rightarrow$}}\end{overpic} &\begin{overpic}[width=0.09\textwidth]{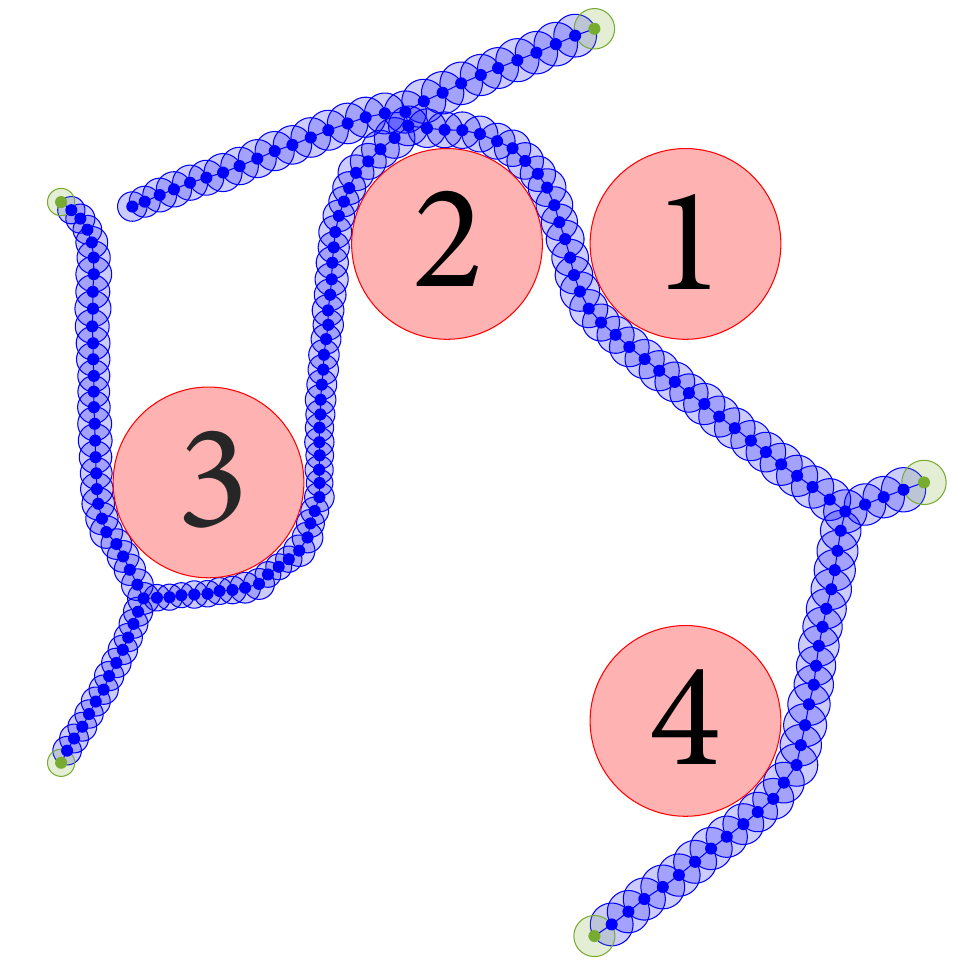}\put(90,50){\textcolor{black}{$\rightarrow$}}\end{overpic}  &\includegraphics[width=0.09\textwidth]{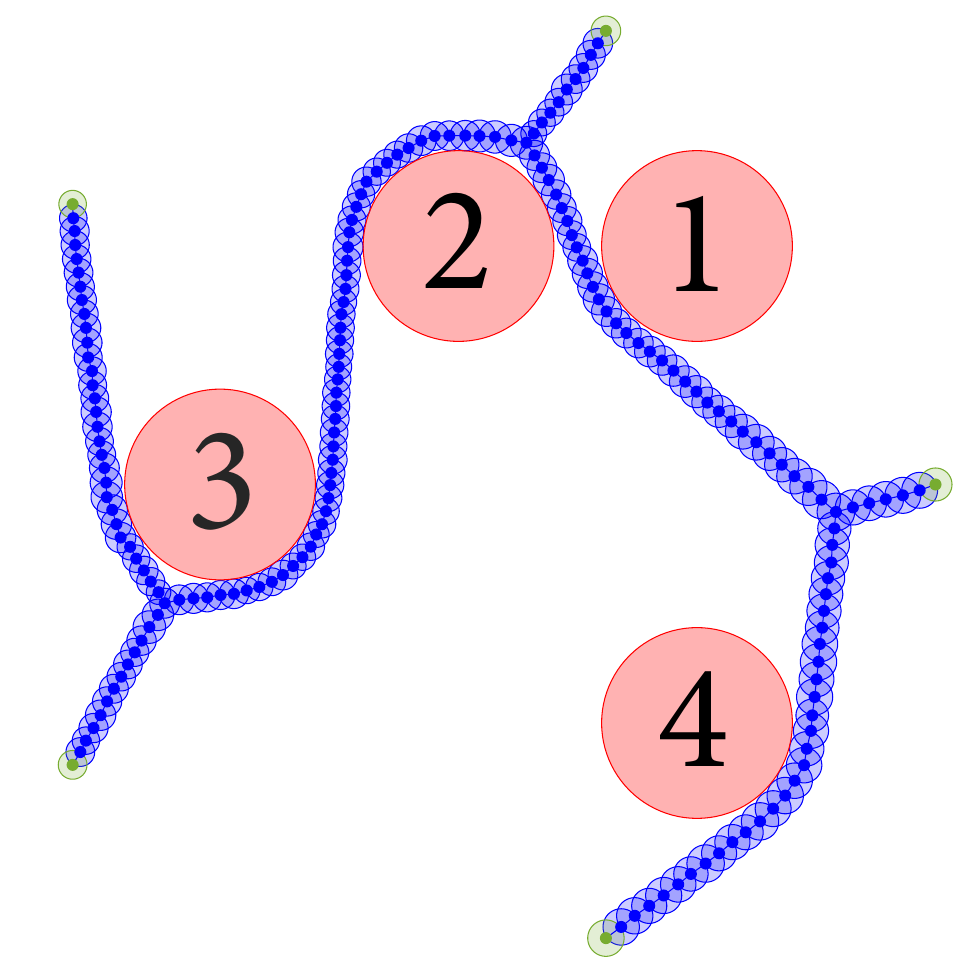}\\
                \boxed{$\{1;24;3\}$} & \includegraphics[width=0.09\textwidth]{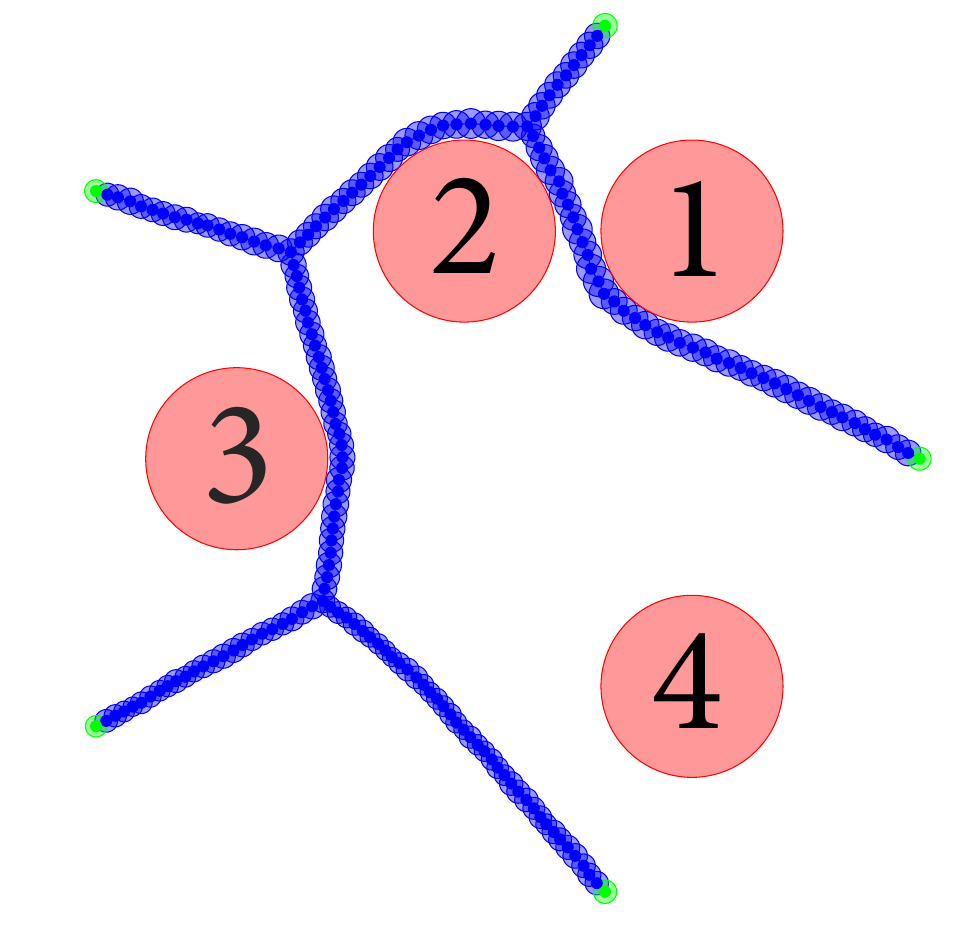}&
                \includegraphics[width=0.09\textwidth]{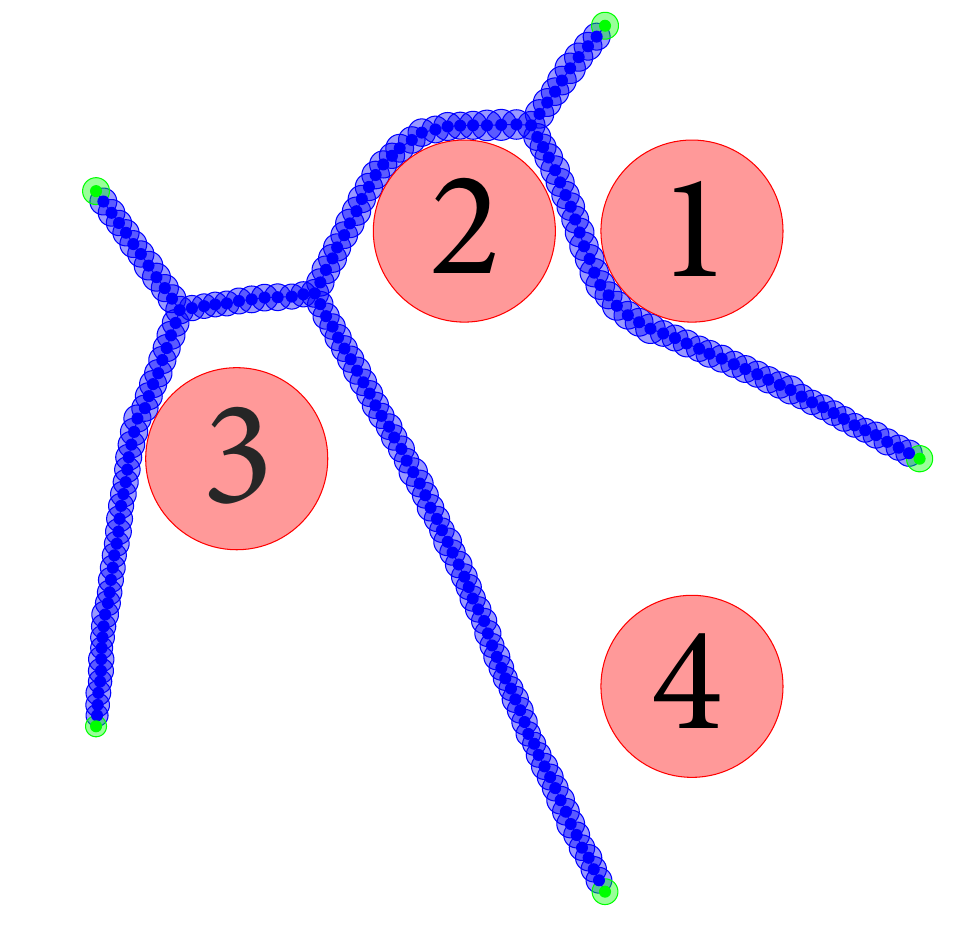} & &\boxed{$\{14;2;3\}$} & \includegraphics[width=0.09\textwidth]{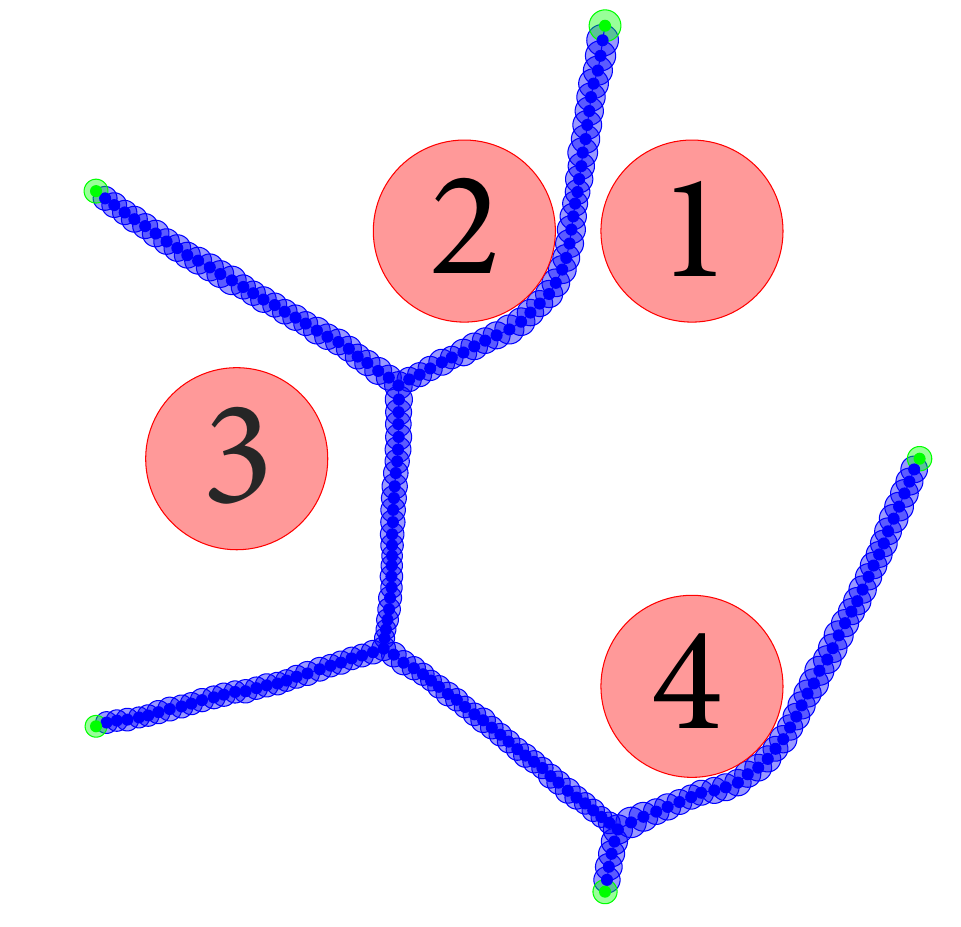}&\includegraphics[width=0.09\textwidth]{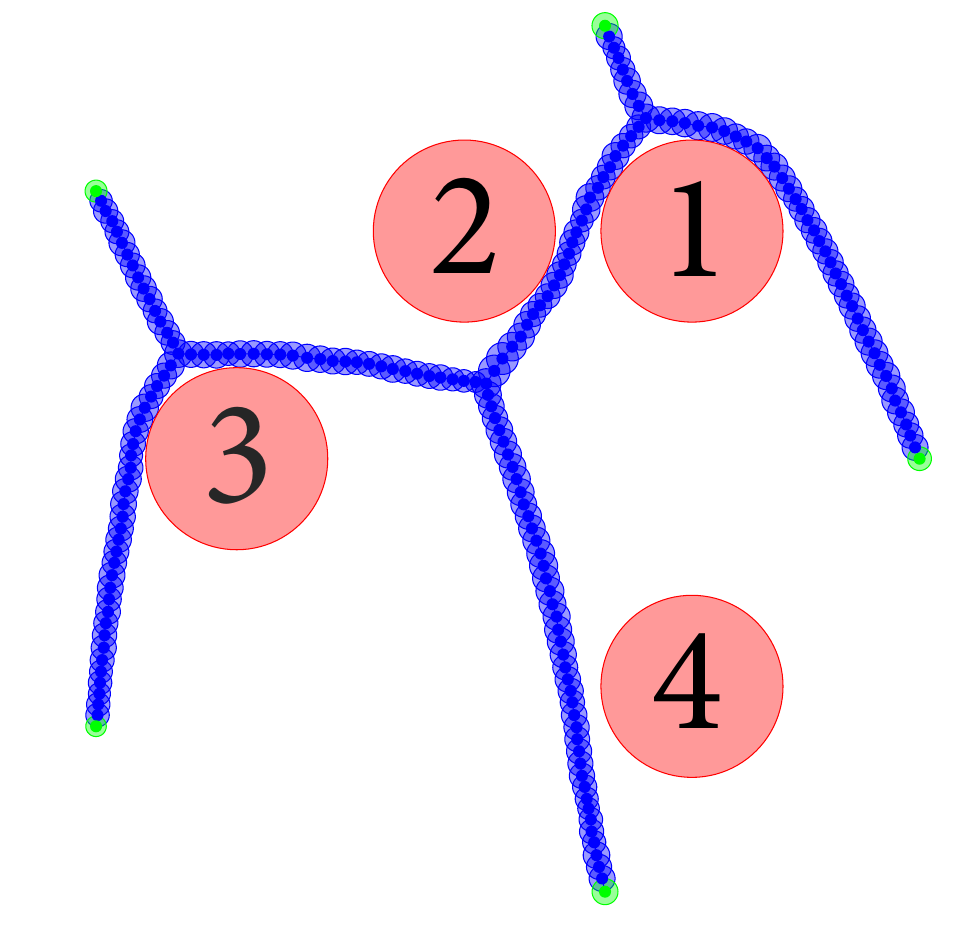}&\includegraphics[width=0.09\textwidth]{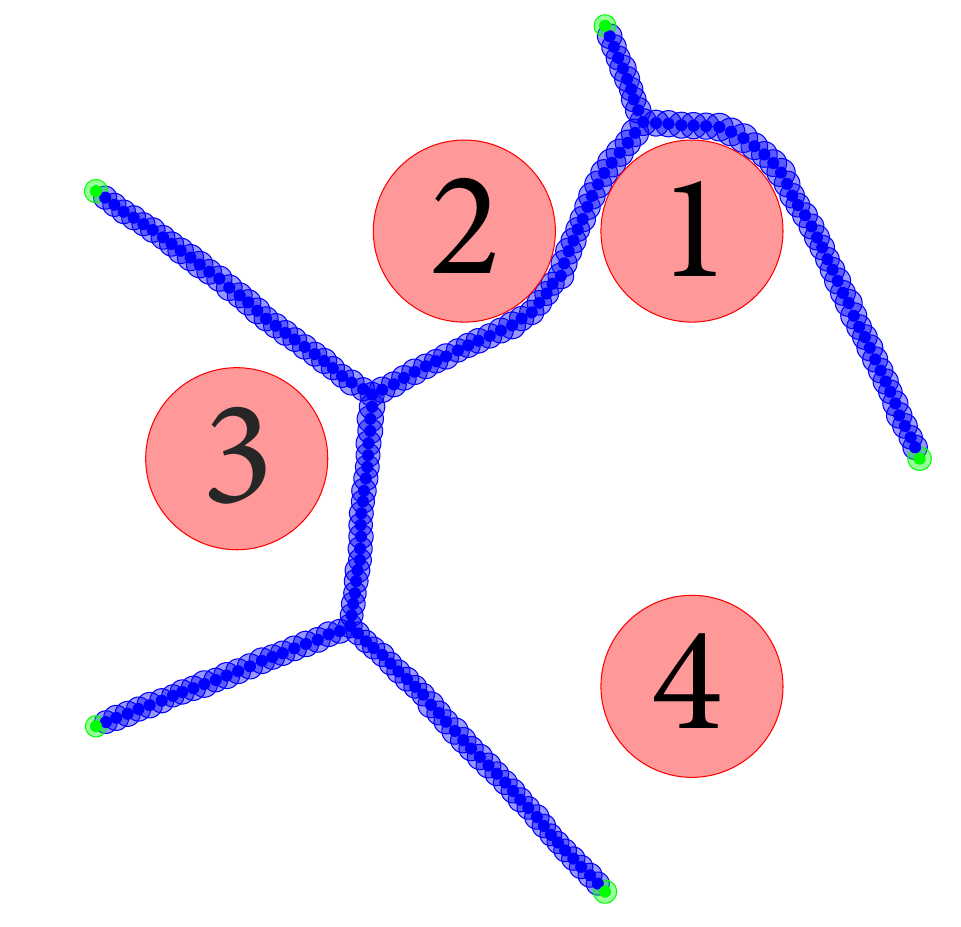}&\includegraphics[width=0.09\textwidth]{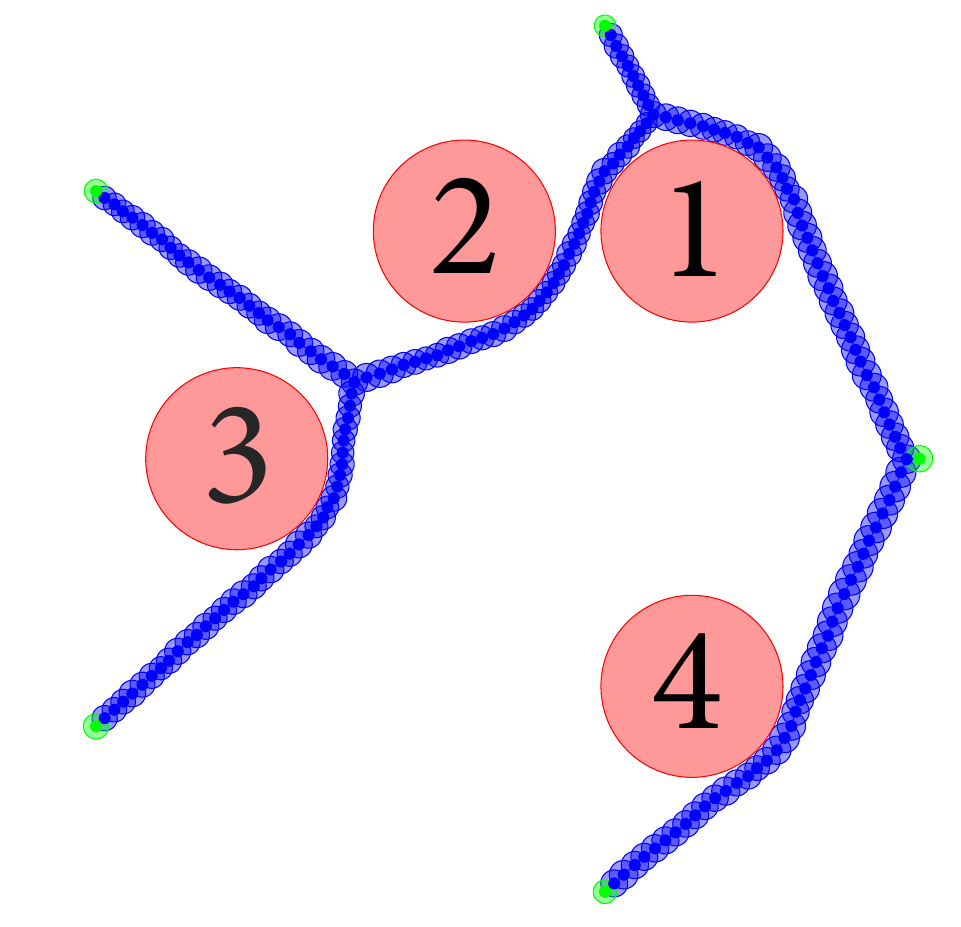} \\%5
                
                \boxed{$\{1;2;3;4\}$}&\boxed{\begin{overpic}[width=0.09\textwidth]{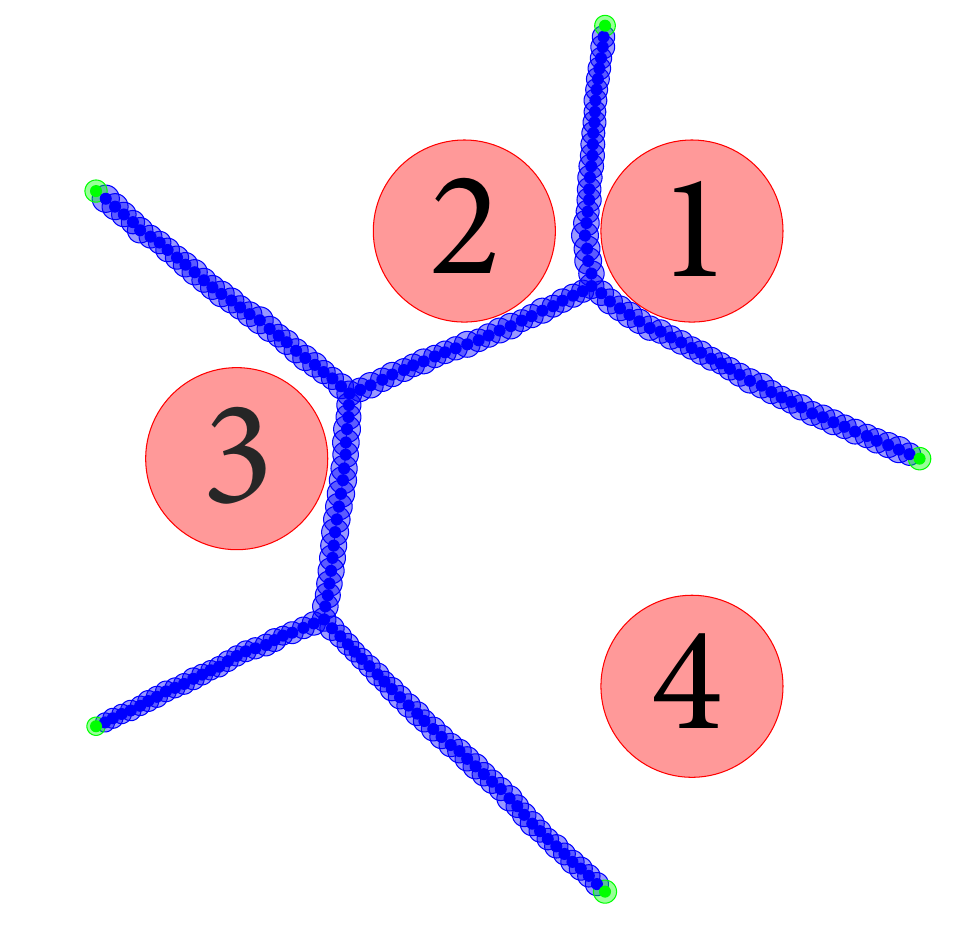}\put(0,83){\textcolor{black}{$3$\textsuperscript{rd}}}\end{overpic}} &\includegraphics[width=0.09\textwidth]{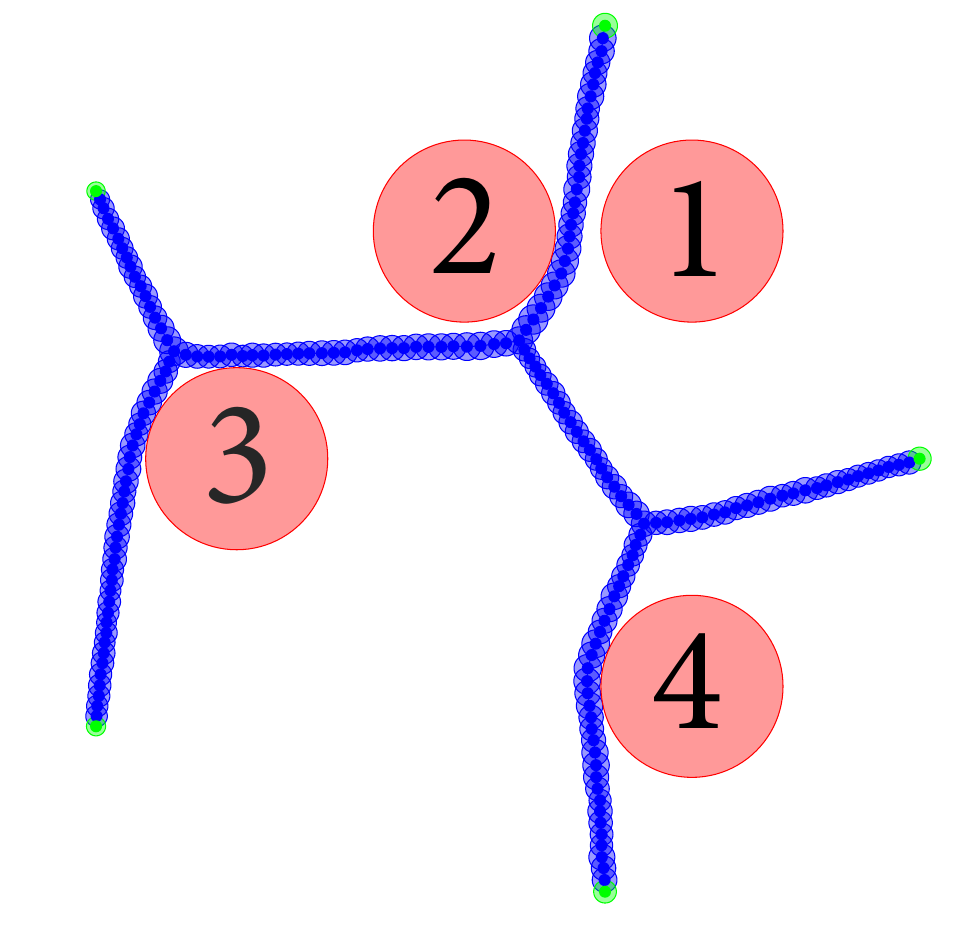}&\includegraphics[width=0.09\textwidth]{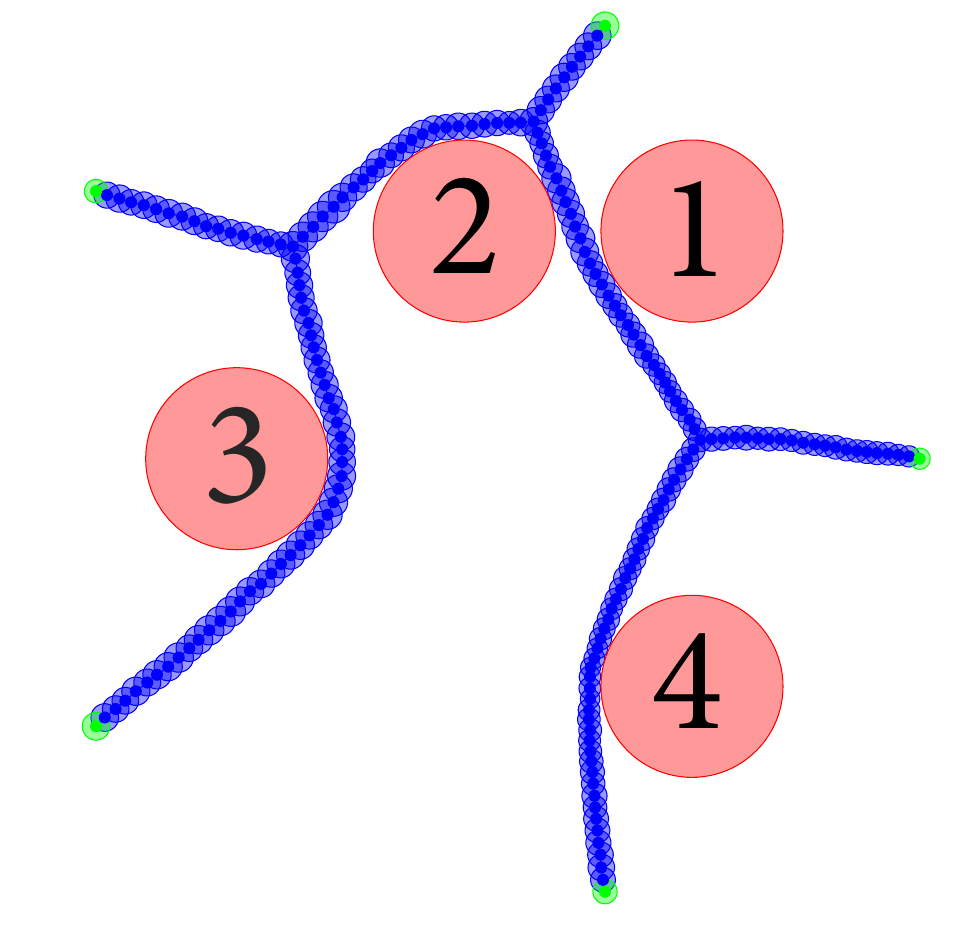}&\includegraphics[width=0.09\textwidth]{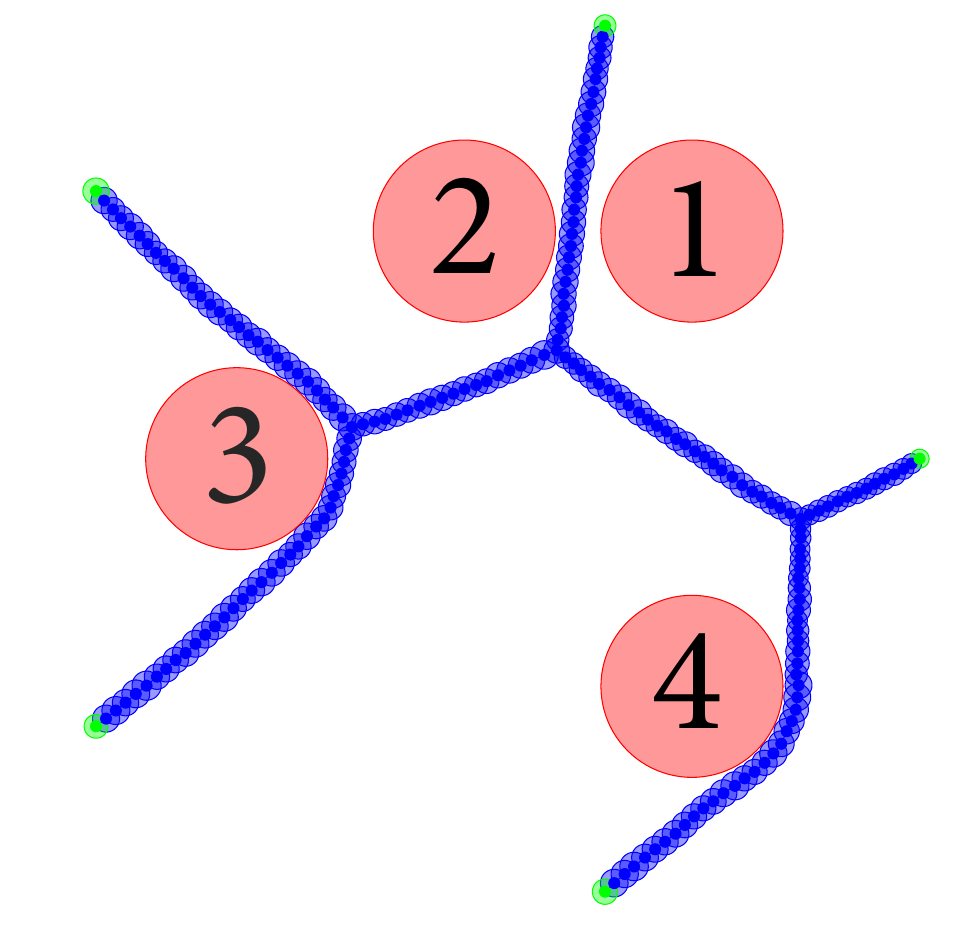}       
    \end{tabular}
    \caption{List of 59 different homotopy classes for $m=5$ terminals, $\phi=4$ obstacles and $n=180$ relays. The first line shows the five possible solutions of the relay placement problem without obstacles: the value of such solutions is considered as a lower bound for the problem with obstacles. 
    As described in Sec.~\ref{sec:homgen}, $m$ terminals define $m$ logical bins where the obstacles may be placed, so the problem can be stated as the placement of $\phi$ distinguishable objects in $m$ bins, which gives $m^\phi=625$ independent homotopies. Most of these homotopies involve longer networks or loops around the obstacles, so they have a low likelihood to be found by the STPG solver. Applying \textsc{HomGen}, \num{9560} trees were examined and classified into 141 independent classes. These classes were used as initial conditions for Algs.~\ref{alg:MSTPlusLeaf}--\ref{alg:nooverlap}, resulting in 59 independent homotopy classes. Of these, 53 were homotopies already found in the pre-scan process, while 6 were derived from the simplification of more complex homotopies, as described in Fig.~\ref{fig:2Tex}.
    For ease of visualization, the results are grouped by the partitions of the obstacle set they define, given by the Bell numbers (in this case $B_{4}=15$). All Bell partitions are found, except for $\{13;24\}$ that has been added manually: the distribution of obstacles forced the network to collapse to a different partition (\{1;24;3\}) under the effect of Algs.~\ref{alg:MSTPlusLeaf}--\ref{alg:nooverlap}. The three cheapest solutions, as well as the most expensive, have been highlighted in the table. 
\label{fig:homotopylist}}
\end{figure*}

\end{document}